\documentclass[12pt,letterpaper,twosided]{report}
\title{A Differential Topological Model for Olfactory Learning and Representation\\ \normalsize{It's Amazing it Works at All}} 
\author{Jack Alexander Cook}
\date{}
\usepackage{authblk}
\usepackage{amssymb}
\usepackage{amsmath}
\usepackage{mathrsfs}
\usepackage{bbm}
\usepackage[T1]{fontenc}
\usepackage{palatino}
\usepackage{mathpazo}
\usepackage{tikz-cd}
\usepackage[margin=1in]{geometry} 
\usepackage{ragged2e}
\usepackage[utf8]{inputenc}
\usepackage[english]{babel}
\usepackage{amsthm}
\usepackage{pgfplots}
\usepackage{tikz-3dplot}
\usepackage{pgffor}
\usepackage{comment}
\usepackage{setspace}
\usepackage{asymptote}
\usepackage{color}   %May be necessary if you want to color links
\usepackage{hyperref}
\hypersetup{
    colorlinks,
    citecolor=blue,
    filecolor=black,
    linkcolor=black,
    urlcolor=black
}
\usepackage[bottom]{footmisc}

\usetikzlibrary{shapes,calc,positioning}
\usetikzlibrary{calc,intersections}
\usetikzlibrary{bending}
	
\newtheorem{theorem}{Theorem}[section]
\newtheorem{claim}[theorem]{Claim}
\newtheorem{lemma}[theorem]{Lemma}
\newtheorem{proposition}[theorem]{Proposition}
\newtheorem{corollary}[theorem]{Corollary}
\newtheorem{dtheorem}[theorem]{Definition/Theorem}
\newtheorem{dproposition}[theorem]{Definition/Proposition}
\theoremstyle{definition}
\newtheorem{definition}[theorem]{Definition}
\newtheorem{example}[theorem]{Example}
\newtheorem{dexample}[theorem]{Definition/Example}
\newtheorem{nexample}[theorem]{Non-Example}
\newtheorem{remark}[theorem]{Remark}

\newcommand{\close}[1]{\overline{#1}}

\newcommand{\ds}{\oplus}
\newcommand{\tensor}{\otimes}
\newcommand{\im}{\operatorname{Im }}
\newcommand{\coker}{\operatorname{coker} }
\newcommand{\comp}{\circ}
\newcommand{\Span}[2] [\mathbb{C}]{ \operatorname{Span}_{#1} \{#2\}  }
\newcommand{\stab}{\operatorname{Stab}}
\newcommand{\orb}{\operatorname{Orb}} 
\newcommand{\C}{\mathbb{C}}
\newcommand{\R}{\mathbb{R}}

\newcommand{\Q}{\mathbb{Q}}
\newcommand{\Z}{\mathbb{Z}}
\newcommand{\N}{\mathbb{N}}
\newcommand{\F}{\mathbb{F}}
\newcommand{\ip}[1]{\left \langle #1 \right \rangle}
\newcommand{\id}{\operatorname{Id}}
\newcommand{\Aut}[1]{\operatorname{Aut}(#1)}
\newcommand{\End}{\operatorname{End}}
\newcommand{\Der}{\operatorname{Der}}
\newcommand{\Mat}[1]{\begin{pmatrix}#1\end{pmatrix}} 
\newcommand{\Hom}{\operatorname{Hom}}
\newcommand{\script}[1]{\mathscr{#1}}
\newcommand{\Ind}{\operatorname{Ind}}
\newcommand{\res}{\operatorname{Res}}
\newcommand{\Char}{\operatorname{char }}
\newcommand{\lie}[1]{\mathfrak{#1}}
\newcommand{\Spec}[1]{\operatorname{Spec}(#1)}
\newcommand{\Obj}{\operatorname{Obj}}
\newcommand{\Ext}{\operatorname{Ext}} 
\newcommand{\Tor}{\operatorname{Tor}} 
\newcommand{\ad}{\operatorname{ad}}
\newcommand{\Ad}{\operatorname{Ad}}

\newcommand{\Gr}{\operatorname{Gr}} 
\newcommand{\Fl}{\operatorname{Fl}}

\usetikzlibrary{hobby}

\begin{document}
\maketitle 

\clearpage
\begin{center}
    \thispagestyle{empty}
    \vspace*{\fill}
    \textit{To Janet and Ralph}
    \vspace*{\fill}
\end{center}
\clearpage

\chapter*{Acknowledgement}
\thispagestyle{empty}

I would like to thank my advisor Thomas Cleland for his patience and guidance throughout my undergraduate career. Without him, this thesis would not have materialized. 

And to my parents who supported me always, thank you.  
\clearpage

\chapter*{Preface} 
This thesis is designed to be a self-contained exposition of the neurobiological and mathematical aspects of sensory perception, memory, and learning with a bias towards olfaction. The final chapters introduce a new approach to modeling focusing more on the geometry of the system as opposed to element wise dynamics. Additionally, we construct an organism independent model for olfactory processing: something which is currently missing from the literature. Chapter $1,$ serves as an introduction to the basic biology, structure, and functions of the olfactory system and the related regions of the brain.   Starting with the nasal cavity, odors excite receptors which in turn relay information to the olfactory bulb(we will often refer to this as bulb). From the bulb information is sent to piriform cortex which projects onto a myriad of structures, some of which are hippocampus, anterior olfactory nucleus, and amygdala. We discuss neuromodulation and some conjectures about higher order processing (post bulb). 

In Chapter $2,$ we take a brief aside to discuss some basic algebra which makes up the first half of the mathematical material needed to understand the later chapters. We begin the tour with set theory where we lay down the preliminaries on functions, set theoretic notation, and various definitions which will appear consistently throughout this text. The next stop is group theory where we study the symmetries of objects and build the notion of an Action on a set. We then pass to Ring Theory where we discuss ideals, morphisms and hidden group structures. Rings show up naturally in chapter $3$ and play an important role in the theory of sections. We end the tour of basic structures with a discussion of fields and polynomial rings with coefficients in a field. This leads to a natural discussion of higher order structures such as Vector Spaces and modules. The latter being an integral component of the model. 

Chapter $3,$ forms the second half of the mathematical underpinnings for chapters $4$ and $5.$ Here we discuss geometry, topology and give a brief introduction to the theory of categories, sheaves, and differentiable stacks. Topology studies the intrinsic properties of a space endowed with a topology. It concerns itself with ideas such as connectedness, compactness, and continuity. Geometry studies calculus on these spaces and very quickly leads to the ideas of flows, geodesics, and Lie groups. The terminal topics are abstractions of the notions of set and function. These provide a convenient language and place to discuss some of the algebraic invariants given to a topological space. 

Chapters $4$ makes up the entirety of the original research of this thesis. We first explore the topological and geometric properties of the physical and perceptual spaces involved in the olfactory system and discuss how the use of vector bundles and non-canonical maps from a bundle to its base space provide insight into the geometry of the system as a whole. We conclude with future directions of research and unanswered questions along with some conjectures about the model. 

Chapter $5$ will focus on potential new areas of investigation. The majority of this chapter covers representation theory and culminates with the Borel-Weil theorem. This gives a realization of representations of certain groups as sections of line bundles. This geometric view of the situation makes it natural to consider sheaves. As the ultimate theorem will tell us, there is some interesting information contained in sheaf cohomology that cannot be accessed through other means. 
\newpage

\tableofcontents
\newpage

\chapter{Olfaction and the Problem of Learning}
%\title{Chapter 1: An Introduction to the Neuroscience of Olfaction}

This section is intended to be a crash-course in the neurobiology of the olfactory system and the various computational aspects of neuroscience. We assume a passing knowledge of general neuroscience. This includes the broad organization of the brain, structure of a neuron, biochemistry of action potentials, existence of neurotransmitters, structure of a synapse, and feedback loops at the level of \cite{Breedlove_Watson_2013}. The main goal of this section is to introduce the idea of categorical perception and apply it to olfaction.

\section{Sensory Systems: Generally} 
Sensory systems are the backbone of human perception and form the only method for which humans (an all animals) can gain information about the outside world. Although the exact number of distinct senses is debated, it is generally agreed upon that humans have 6-7 main ones which govern a occupy a large portion of the brain and almost all of the cortical space devoted to perception  \cite{Breedlove_Watson_2013}. One could spend an enormous number of pages discussing the intricacies of each of these sensory systems and their associated perceptual constructions. As the main focus of this thesis is to understand olfactory processing, we shall only give a broad introduction to the other sensory systems and leave the remaining details to the many references. 
\begin{remark} 
	For the remainder of this chapter, all definitions are operational (may change between researchers)  unless otherwise noted. We shall give some explanation of the definitions in the cases where we deviate from the standard references. 
\end{remark}

The easiest way to begin an analysis of these systems is to understand the basic neurophysiology. 

\begin{definition}  
	A \textbf{sensory system} is a part of the nervous system consisting of sensory neurons, a neural pathway, and a cortical area.  
\end{definition}
\noindent Sensory systems play a key role in every action the body performs: from simple things like standing up straight and picking up a glass of water to more complex tasks like skiing or identifying someone's face in a dimly lit room. To better understand these objects, lets investigate a few well known examples. 
\begin{example}\text{}
	\begin{enumerate}
		\item \textbf{Vision}: In this case, sensory neurons are rods and cones. These light sensitive cells transmit information to the optic nerve which relays this information to visual cortex. In fact, different cells along the pathways from the retina to the occipital lobe have varying receptor fields. This variance contributes to the processing of an image. 
			
		\item \textbf{Audition}: Vibrations in the basilar membrane due to sound waves coming into contact with the ear drum, lead to the vibration of hair cells. This motion induces action potentials in the auditory nerve. From here the signal partially decussates to the temporal lobes where further processing occurs. 
	\end{enumerate}
\end{example}	
We can further divide up the sensory systems into those which have chemical stimuli and those which do not. Those chemical senses depend on molecular interactions in the sensory neurons to facilitate the transformation from stimulus to perception. In the case of gustation, there are five distinguished "tastes": salty, sweet, sour, bitter, umami. These all correspond to different molecules interacting with the papillae on the tongue. Something which tastes more "salty" is directly related to the Na$^+$ ions present in the solution of saliva and food. In contrast to this, we have audition. The pertinent objects here are pressure waves in the air which vibrate the ear drum which in turn vibrates the bones of the inner ear and causes waves in the basilar membrane. These waves cause the ends of the hair cells to be perturbed and induce an action potential. 

The point of these examples is to show that sensory systems have a wide array of possible stimuli.  

\subsection{Function}
Now that we have the basic (and grossly vague) description of the structure of a sensory system. One may ask, ``what is the purpose of such a system?" Beyond the obvious answer (we need a way to interact with our environment), there are some subtle and incredibly important operations that sensory systems accomplish. The main references for this section are $\cite{Harnad1987}$ and $\cite{Cohen_Lefebvre2005}$.  

The main operations we will discuss are learning, representation, and categorization.\footnote{We will discuss this at length in the next section. The intent here is to get some intuition for the problems we will be attacking in the later chapters.} They are closely related and in some perspectives, are even intertwined. In fact we can think of learning and representation as disparate ideas, whereas categorical perception seeks to, in some sense, unify these ideas. The motivation for studying such a construction first originated in vision and speech with color perception and categorization of various speech patterns. 

There is an obvious evolutionary advantage to the construction of categories. Typically, stimuli are continuous, or at least abundant enough that any model would function adequately considering them continuous. Categorical perception transforms this continuity into a discrete spectrum of perceptions organized by similarity with respect to some metric. We take the following example from $\cite{Harnad1987}.$ Consider a digital clock which presents the time in 12hr increments with the use of $am$ and $pm.$ Then twice a day, the clock would present $12:00$ with the only difference being which signifier is present. In this way, we can categorize the time on the clock as either times marked by $am$ or times marked by $pm.$ 

Now consider another construction of categories which will be revisited in chapter 3. If asked to classify the capital letters in the english alphabet, what is an appropriate choice of category? Suppose we choose to split them by the number of holes: in that any letter with a closed loop has a hole, and having multiple closed loops should split up the categories further. In this schema(which is font specific), the letters are grouped as follows: \begin{align*}
	\{ C,E,F,G,H,I,J,K,L,M,N,S,T,U,V,W,X,Y,Z\} && \{ A,D,O,P,Q,R\} && \{ B\}
\end{align*}
So we have three distinct categories: No holes, One hole, two holes. Notice that we can be a bit more general about this however. If we consider only letters that have a hole and letters which do not, we only have two categories. Inside of the holed category we get a sub-category consisting of letters which have multiple holes. 

The point of these examples is to illuminate the idea that categories seek to simplify the stimuli. It is much easier to think about letters with or without holes than all of the letters simultaneously. Due to this, it should be no surprise that a majority of current research in sensory systems is devoted to understanding the process of categorization. This is precisely what we will investigate in the upcoming sections and is the topic of Chapter 4. We can think of completing a category, $C$, by adding to it all of the points which are "infinitesimally close" to $C.$ If we think about this in a geometric way, this amounts to the circle which bounds a disc. 

Before diving into the world of olfaction, we need one more general function of sensory systems: generalization. In 1987 Shepard introduced the idea of \textit{generalization} for perceptual spaces. To this end, it is a different method of categorization, but one which depends on minimal learning. We shall call this perceptual generalization. Here is a more formal definition.  
\begin{definition}
	\textbf{Perceptual generalization} is the process by which a sensory system (in particular the neural pathway) constructs a broad category for a given stimulus, based solely on the learning of one (or a few) stimulus. 
\end{definition}

In the figure below (Figure 1.1), we show the first examples of perceptual generalization. This process seems to be a method of producing categories for unlearned/partially learned stimuli. The method of generalization is extrapolate information from one stimulus and use this to "learn" something about its nearest neighbors in the perceptual space. 
\begin{figure}[!htb]
    \centering
    \includegraphics[width=.8 \linewidth]{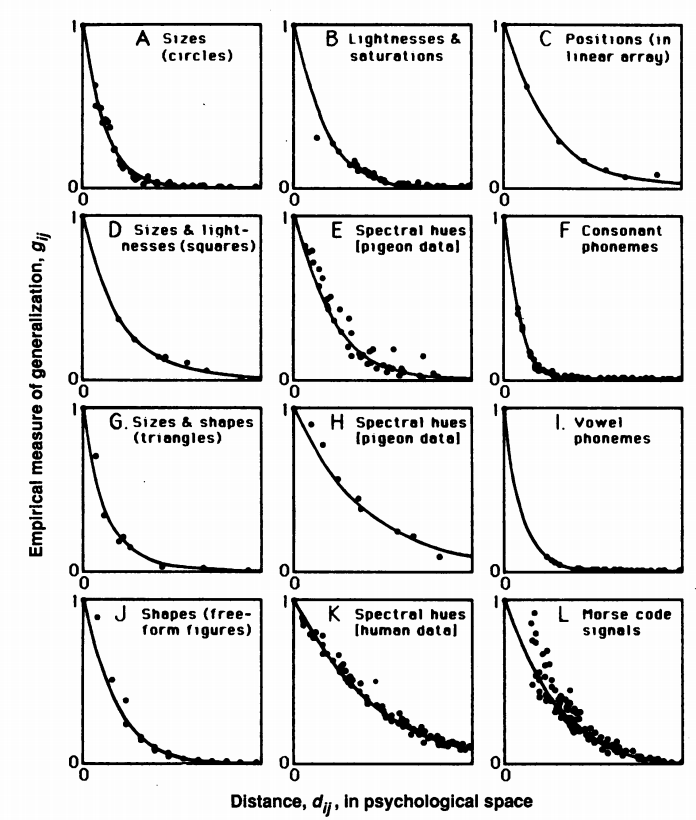}
    \caption{Perceptual generalization in different modalities and contexts. Distances along the horizontal axes are measured in the intrinsic metric of the perceptual space in which the points sit \cite{Shepard1987}. }
    \label{fig:wrapfig}
\end{figure}

Perceptual generalization can be thought of as a pseudo-prior to categorization. Before the system can split things into clean, discrete categories, it needs to build the objects of the perceptual space (the things to categorize). Once this is done, but still before discretization, the system has to understand the boundary of each perceptual object\footnote{This will be interpreted in chapter 3 and 4 as the boundary of a topological subspace of the perceptual space. This notion will allow us to give a more formal definition than the vague one given here and will also lead to a clean method of discretizing the perceptual space.}. 

\begin{definition}
	The \textbf{perceptual boundary} (sometimes shortened to just boundary) of a category is the collection of points which are \textit{extreme} in the category. That is, these are the points which are only present in the completion of a category.  
\end{definition}

One of the important themes of the research surrounding sensory systems is that of distinguishing the boundary of the perceptual space and the various percepts it contains \cite[Chapter 1, Section 2]{Harnad1987}. This may seem like an easy task, but in the abstract this is incredibly difficult. The difficulty lies in the lack of rigor behind the definition of a category. Is something an element of a particular category, its boundary, or something else? These questions will be answered in the case of olfaction in chapter 4. For now, we move away from general sensory systems and take a closer look at the olfactory system, its associated brain regions, and what we know about how the system builds and categorizes representations.   

\section{The Olfactory System} 
The olfactory system can be broken up (coarsely) into two main regions: the olfactory bulb and piriform cortex. We shall focus on the olfactory bulb as the piriform cortex is much less understood. As the following figure (Figure 1.2) shows, the olfactory bulb is divided into several layers. Each plays a key role in the transmutation of the physical stimulus to a usable perceptual object. As we still do not understand the full functionality of each of the layers individually, we will treat them as separate objects and present what we \textit{do} know about the different layers. 
\begin{figure}[!htb]
    \centering
    \includegraphics[width=\linewidth]{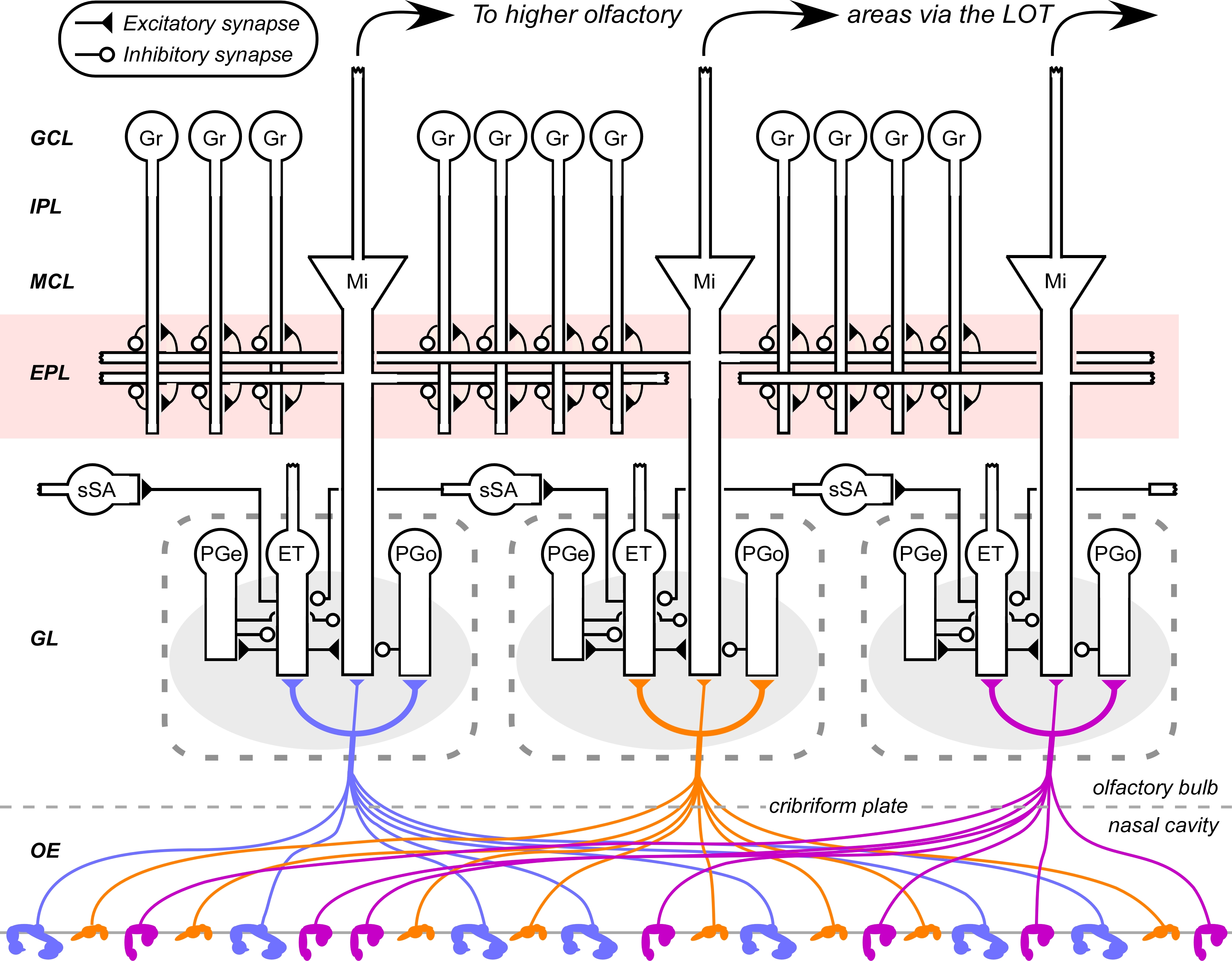}
    \caption{Schematic of the mamalian olfactory bulb microcircuitry. Layers are identified, and are arranged bottom to top, external to internal.}
    \label{fig:wrapfig}
\end{figure}

\subsection{(Pre-)Processing} 

In the same flavor as the previous section, we need to specify what the sensory neurons are. In Figure 1.2, the layer marked OE (olfactory epithelium) and the colored receptors are precisely the olfactory sensory neurons (abbreviated OSNs). These are chemical receptors and their level of activation (spike frequency) is directly proportional to the binding affinity of the odorant. This binding information is  processed at a variety of places before being sent off to piriform cortex and other higher-order brain regions. The main layers we concern ourselves with here are $GL,EPL,$ and $GCL.$   
In contrast with other modalities (such as vision) olfaction is intrinsically high-dimensional and this high-dimensionality is consistent across species. In humans there are roughly $350$ different types of OSNs whereas in mice there are upwards of $1000$ \cite{Cleland2014}. Each distinct type of OSN converges to a neuropil tangle which is roughly spherical in nature. We call these tangles $\textit{glomeruli}$ and the layer consisting of all of them is $GL.$ 

The largest cells protruding from the glomeruli are \textit{mitral cells}. These pyramidal cells embody the immediate connection of the OB to piriform cortex. In Figure 1.2 the mitral cells are drawn to be in one-one correspondence with the glomeruli and are seen to sample from only one glomerulus: this is false in general. It turns out that in most mammalian tetrapods the mitral cells do indeed sample from a single glomerulus. In \cite{Mori81a} and \cite{Mori81b}, it was shown that in some turtles and reptiles, the mitral cells can sample from a variety of glomeruli. The benefits of this cross sampling are still not very well understood. 

The final change of information in the OB is the modification by the granule cells. These are inhibitory synapses which delay the mitral cell action potential \cite{Cleland2014}. It is thought that these synapses  play a large role in the formation of the perception of an odorant, but no published research has looked at this yet. We do know however that learning is related to granule cell firing patterns. With repeated trials of an odorant, the number granule cells which fire decreases monotonically. Each consecutive trial leads to a more specific and more refined response. This fits in with Shepard's idea on generalization. This specialization implies however, that granule cells can become specified fairly quickly and thus the brain should ``run out" of possible specificity. That is is theory, the granule cells could become specified to only one odorant. This however is a poor allocation of energy and would then require the genesis of a hoard of new granule cells for each variant of the same odor. 
Some recent work by \cite{Moreno2009} has shown that granule cells \textit{do} exhibit adult neurogenesis which is incited from the piriform cortex. This neurogenesis is the reason for which it is thought that granule cells play an important role in building of perception and for which we can learn odors well into our adult years. It does not however remove the flawed idea that granule cells can become highly specialized.   

Now that we are aquatinted with the general form of the OB we need to discuss the general schema of processing. In $GL,$ the periglomerular cells (PGs) and superficial short-axon cells (sSA) are thought to be the cells which begin the construction of perceptual categories. The evidence for this comes from recent work by \cite{Borthakur2019} which shows that learning can occur at the glomerular layer and not just at the granule cell layer. This implies at a minimum, that the purpose of the early (exterior) layers of the OB are to normalize data and to reduce noise in the sensation. Further it increases contrast between similar odorants. A nice analogy to this is the existence of edges in visual perception. There is an enormous amount of cortical space allocated to the processing of edges. This helps build a better image and in the same way, contrast enhancement in bulb, ``builds a better odor."   

The common theme to keep in mind for this system, is that of sampling with noise. Every layer samples from the previous in order to get a more specific perceptual category at the end. With the introduction of some noise (variance) we can eliminate some of the theoretical aspects of the system. For instance, the granule cell hyper-specification from before can be formally disregarded as the inherent noise of the odorants will not allow for the accuracy necessary to determine "exactly" what the odorant is. What this does tell us however is that we can combine the notions of categorical perception and generalization for this system to arrive on what we shall call \textit{Categorical Generalization}. At a first pass, this idea is the construction of a generalized perceptual category which given some learning set, all contained in the same perceptual category, is the extension of the learning set via the rules of generalization set forth in $\cite{Shepard1987}.$  

\begin{definition}
	Let $O$ be an odorant and $\mathcal{O}$ the corresponding perceptual category generated by learning on $O$. Then the \textbf{Generalized Category} $G\mathcal{O}$ is the perceptual category which extends $\mathcal{O}$ by generalizing its boundary.    
\end{definition} 

To understand this idea further, we shall investigate computational models of the olfactory system and see how the introduction of this concept motivates our model constructed in chapter 4.

\subsection{Modeling of Olfaction} 
Models of the olfactory system come in two main types: anatomical and theoretical (perceptual) . The anatomical models focus on understanding the biochemistry and spike timing of the OSN and related cells of the OB, whereas perceptual models tend to be fantastical speculation on the "perceptual space" of the system \cite{Cleland_Linster2005} \cite{Ermentrout_Terman2010}. In chapter 4, we shall propose a model which has the advantage of being mixture of the two, with the advantage of being mathematically elegant. Before then, let us understand some of the current problems with modeling and what makes the olfactory system significantly different than the other sensory systems. 

Each flavor of OSN has a different receptive field and thus we can consider the "space of possible physical inputs" to be some collection of points in a \textit{$350$-dimensional space}, with each different "dimension" defined by a different OSN receptive field. Compared to the three dimensions of vision, this is monstrous. This aspect of the olfactory system makes studying it substantially different from other modalities. One feature for example, is that distances tend to increase with the dimension. What we mean by this is the following: consider the unit sphere in an even dimensional space. The volume of a cube of side length $2,$ centered at the origin, has volume $2^{2k}$ where $2k$ is the dimension. Whereas, the volume of the unit sphere sitting inside this box is $\frac{\pi^k}{k!}$. So as $k$ increases, the volume of the unit sphere actually decreases. What this tells us is that, proportionally in higher dimensions, more points lie outside the unit sphere than inside. The importance of the above observation cannot be understated. It implies that there are theoretically an incredibly large number of possible odorants detectable by the OSNs as well as decreases the probability that any two odorants which are chemically different will be identified as similar. We can go one step further and say that the physical marker of an odorant and the sensation thereof is a large determining factor in the construction of the perception of that odor. 

The many thousands of OSNs converge onto glomeruli, of which there are the exactly same number as the different receptor types ($\sim350$). The main interest in the glomerular layer is the possibility of pre-processing, and learning \cite{Cleland2014}. This idea is fairly recent and provides an interesting new direction for computational models such as \cite{Li_Cleland2013}. We shall not spend any time on this topic however as it will play a minimal role in the later chapters.     

\begin{remark} The remainder of this section will be dedicated to the modeling of mitral and granule cells. These two cell types occupy a majority of the mental theatre of researchers in this field as they are the most mysterious cells in the olfactory bulb. \end{remark}

We begin with mitral cells. As compared to the roughly 350 glomeruli, there are about 3500 mitral cells (in humans) and even more in some mammalian tetrapods. The key feature of mammalian tetrapods is the independent sampling of the mitral cells from a distinct glomerulus. As mentioned above, this is not always the case and due to this fact, modeling these cells is a delicate procedure. Most authors elect to simply ignore the potential cross-sampling. 

As with most modeling, the early approaches were through linear algebra (see chapter 2) and some form of calculus \cite{Ermentrout_Terman2010}. The type of modeling which makes use of calculus extensively is not particularly helpful for building understanding of the perceptual space as a geometric object. The use of linear algebra though is quite important in the construction of a perceptual space.
In \cite{Zaidi2013} and many others, mitral cells are modeled as vectors in a Euclidean geometry. The important part here is the type of geometry chosen. Euclidean geometries are inherently the most restrictive geometry as it assumes no curvature in the perceptual space.

\begin{example}
	To see why a Euclidean geometry is restrictive, consider two points on a piece of paper. Let $d$ be the distance separating the points. Now, given any transformation of the paper which retains the flatness (a rotation or reflection) the distance between the points will stay the same. Now, let us introduce a fold into the paper. This can bring the points closer together in the ambient three-dimensional space but their distance along the paper will not change. If instead of a fold we make it a smooth change, this is precisely the introduction of curvature. 
\end{example}

Nonetheless, this choice of model has been shown repeatedly to not be useful. Simply speaking, perceptual distances do not sit well inside a linear space. It is convenient however to have the mathematical ease of a Euclidean space. For this reason, current research (such as \cite{Olshausen2018}) has begun to try and understand \textit{manifolds} (see chapter 3 for a definition) and their applications to sensory processing. These are objects which "look like" Euclidean space on a local scale. The advantage of these spaces is that we can introduce curvature to the perceptual space, while still retaining the linear structure on the tangent space at every point. In fact, the problem with Euclidean space is not unique to it. Any space with constant curvature will have the same deficit. We recommend running through the example above but exchanging the piece of paper with a ball or a saddle. This will give the other two types of spaces of constant curvature. 
Even though the above approach is flawed, some interesting results have appeared in other modalities \cite{Maio_Rao2007} that imply we may want to consider vector-like mitral cells in olfactory system models. Furthermore, the use of some high-level algebra and differential geometry has led to the investigation of certain mathematical objects called \textit{Lie Groups} (see chapter 3 for a definition). These play an important role in mathematics and physics so it is no surprise that they have shown up in neuroscience as well. 

We now turn our attention to granule cells. One large mystery surrounding them is the aforementioned adult neurogenesis. It was shown in \cite{Moreno2009} that in order for the olfactory system to function at its current level of accuracy, adult neurogenesis is necessary. Some have argued however that all evidence of adult neurogenesis is actually remnants of embryonic stem-cell differentiation. We shall not contest either of these topics here as the data is inconclusive either way. 
On a different note, granule cells are believed to be the workhorses of olfactory learning \cite{Cleland2014}. These cells inhibit the action potentials of the far larger mitral cells and attribute to the variance in spike-timing seen across the bulb for different odors. It should also be noted that there are orders of magnitude more granule cells than mitral cells. The exact mechanism for mitral cell inhibition is up for debate, however it is clear that the piriform cortex plays some critical role in the excitation-inhibition loop. Surprisingly however, models tend to not deal with subtle intricacies of granule cell inhibition. One possible explanation for this is that granule cells only act locally, in contrast to mitral cells which can inhibit relatively far away neighbors. This local action is not readily dealt with in computer models, and combining it with the relatively global action of mitral cells (sometimes having to intertwine the two) has been a blockade for some time now. 

As one final question of this chapter, we want to define the perceptual categories in olfaction. Given an odorant, the generalized category associated to that odorant is the result of the generalization gradients above. In practice, one should think of this in the following way: suppose $O$ is the odorant (or combination thereof) corresponding to an orange. Then the generalized category of unlearned oranges may encompass all citrus fruits. This is clearly too broad to be of use when differentiating particular species of orange or even ripeness. Therefore, we know that there must be some mechanism (granule cell interactions) which restricts the size of the generalized categories so that they are of use for identification. In fact, as we shall see in chapter 4, we have proposed a way of generating some specific hierarchies from such general data given some non-zero amount of learning. Geometrically we can view this as constructing some rough approximation for the perceptual space which somehow encodes the differences between distinct classes of odorants.   

This completes the brief introduction to the computational neuroscience of olfaction.

\chapter{An Introduction to Algebra}
%\title{Chapter 2: Algebra}

\section{Preliminaries: Set Theory}

\setcounter{theorem}{-1}
Here we lay down the basics of set theory, its notation and how it is used in practice. We start with a definition 
\begin{definition}
	A \textbf{Set} $S,$ is any collection of elements (normally denoted with the corresponding small letter) with cardinality some ordinal. The \textbf{Order} (size/cardinality) of a set $S,$ is the number of elements in $S$ and denoted $|S|.$  
\end{definition}	 
We have the natural notion of a subset, denoted $T\subseteq S.$ If $T$ is strictly smaller than $S,$ then we write $T\subsetneq S.$ The collection of all subsets of a set $S$ is called the\textit{power set} and is denoted $\mathcal{P}(S).$ 
Some classic examples of sets are the natural numbers, denoted \[\N=\{0,1,...\}\] and the integers, denoted \[\Z=\{0,1,-1,2,-2,...\}.\]
Some more interesting sets are $\Q,\R,\C$
the sets of rational, real, and complex numbers respectively. Notice that $\N\subsetneq \Z\subsetneq \Q \subsetneq \R \subsetneq \C.$ For this reason, unless specified, we will use $\C$ in examples.  

Additionally, we can define intersections and unions of sets. If $S,T$ are two sets we define their intersection $S\cap T=\{x:x\in S \text{ and } x\in T\}$ and their union  $S\cup T=\{x:x\in S \text{ or } x\in T\}.$ Further, if $T\subseteq S,$ we can define the complement of $T$ is $S,$ to be $T^c=S-T=\{ s\in S:s\notin T\}.$  

\begin{definition}
	Let $X,Y$ be two sets. We define the \textbf{Cartesian Product},  denoted $X\times Y,$ as the set of all ordered pairs of elements in $X$ and $Y.$ That is \[ X\times Y=\{(x,y):x\in X,y\in Y\}\] 
\end{definition}

\begin{example}
	Let $X=\{1,2\}$ and $Y=\{a,b\}.$ Then \[X\times Y=\{ (1,a), (1,b) ,(2,a) ,(2,b)\}.\]
For finite sets, it is easy to see that $|X\times Y|=|X||Y|$ as for each element $x\in X$ we can look at the subset $\{x\}\times Y\subseteq X\times Y$ each of these sets has size $|Y|.$ As there are $|X|$ choices for $x,$ the claim follows.    
\end{example}

\begin{definition}
	A $\textbf{Function}$ $f:S\to T$ is a mapping between sets which assigns to each element $s$ in the source space $S,$ an element $f(s)=t\in T.$\footnote{The symbol $\in$ is to be read as "an element of." If we use the symbol $\notin$ the slash means "not". For example $-1\in \Z$ should be read as $-1$ is an element of the integers and $-1\notin \N$ should be read as $-1$ is not an element of the natural numbers.  Once comfortable with this notion, it is common practice to say $-1$ \textit{is} an integer.}  For this reason, we call $S$ the \textbf{domain} of $f,$ and $T$ the \textbf{codomain} of $f$. Denote by $f^{-1}(t)=\{s\in S:f(s)=t\}$ this is called the \textbf{Pre-Image} of $t$ under $f.$
\end{definition}	
We can compose functions assuming the codomain of the first is contained in the domain of the second. We can actually relax this requirement to be that the image, denoted $\im f,$ is contained in the domain of $g.$  

Notice that $f$ may not hit every element of $T:$ that is there may exists some $t\in T$ such that $t\neq f(s)$ for any $s\in S.$ The following sister definitions provide us with insight into this exact situation. 
\begin{definition}[Injective, Surjective, and Bijective]
Let $f:S\to T$ be a function. 
\begin{description}
	\item[1)] $f$ is called \textbf{injective} if whenever $f(s_1)=f(s_2)$ this implies (denoted $\implies$) that $s_1=s_2.$\\
	\item[2)] $f$ is called \textbf{surjective} if for all (denoted $\forall$) $t\in T,$ there exists at least one $s\in S,$ such that $f(s)=t.$ 
	\item[3)] A function which is both injective and surjective is called \textbf{bijective}. 
\end{description}
\end{definition}
\begin{example}
	Let $f:\Z\to \Z$ be defined by $f(n)=2n.$ Then $f$ is injective trivially. $f$ is not surjective as for any odd number $l=2k+1$ cannot be written as $2n$ for any $n\in \Z.$ For an example of a surjective map, consider the absolute value function \[ |\cdot |:\Z\to \N\]
	$f(z)=f(-z)=|z|.$ In more standard notation, one writes $z\mapsto |z|.$ 
\end{example}
\begin{proposition}
	Let $f:A\to B$ and $g:B\to C$ be injective (respectively surjective, bijective) functions. Then $g\comp f:A\to C$ is injective (resp. surjective, bijective).  
\end{proposition}
\begin{proof}\text{}\\
	(Injectivity) Suppose that $(g\comp f)(a)=(g\comp f)(a').$ As $g$ is injective, we know that $f(a)=f(a').$ Now, as $f$ is injective, we have that $a=a'.$ \\\\
	(Surjectivity) Let $c\in C.$ As $f$ is surjective, we know that the domain of $g$ is all of $B.$ Now, we know that $c=g(b)$ for some $b\in B.$ As $f$ is surjective, we have that $b=f(a)$ for some $a\in A.$ Thus, for all $c\in C,$ there exists at least one $a\in A$ such that $g\comp f(a)=c.$ As bijectivity is a combination of the previous two statements, this completes the proof.  
\end{proof}

\begin{theorem}
	Let $f:X\to Y$ be a bijective function. Then there exists a map $g:Y\to X$ such that $f\comp g=\id_Y$ and $g\comp f=\id_X.$
\end{theorem}
\begin{proof}
	Define $g:Y\to X$ as $g(y)=f^{-1}(y).$ This is well defined as $f$ is bijective so $y\in \im f$ and $\exists! x\in X$ such that $f(x)=y.$ Then \[g\comp f(x')=f^{-1}(f(x))=x\] by bijectivity of $f.$ Further, \[f\comp g(y)=f(f^{-1}(y))=f(x')=y\]
	by bijectivity. Hence, $g$ satisfies the properties and we are done.   
\end{proof}

\begin{definition}
	Let $X$ be a set. We say $E\subseteq X\times X$ is an \textbf{Equivalence Relation on $X$} if the following properties hold: \begin{enumerate}
			\item $(x,x)\in E$ for all $x\in X.$ 
			\item If $(x,y)\in E$ then $(y,x)\in E.$
			\item If $(x,y),(y,z)\in E$ then $(x,z)\in E.$ 
	\end{enumerate}
	We call these properties reflexivity, symmetry, and transitivity respectively. It is common practice to not write $E$ as a set of ordered pairs but rather write $x\sim y$ if $(x,y)\in E.$ We then say $\sim$ is an equivalence relation on $X.$ Further let $[x]$ (also denoted  $\bar{x}$ in some cases) be the set of all elements $y\in X$ such that $x\sim y.$ We call $[x]$ the \textbf{Equivalence Class} of $x.$ We denote the set of equivalence classes as $X/\sim.$  
\end{definition}

\begin{lemma}
	Let $\sim$ be an equivalence relation on a set $X.$ Then $\sim$ induces a partition of $X$ via equivalence classes. This is equivalent to saying for all elements $x,y\in X,$ either $[x]=[y]\in X/\sim$ or $[x]\cap [y]=\varnothing$ the empty set.   
\end{lemma}
\begin{proof}
	Suppose $[x]\neq [y]$ and $[x]\cap [y]\neq \varnothing.$ Let $w\in [x]\cap [y].$ Then $x\sim w$ and $y\sim w.$ Using the symmetry and transitive property of $\sim,$ we have that $x\sim y.$ Therefore $[x]=[y]$ a contradiction. Hence, either $[x]=[y]$ or $[x]\cap [y]$ for all $x,y\in X.$ 
\end{proof}

\begin{example}
	Let $\Z$ denote the set of integers as above. Fix some $n\geq 0.$ Define $a\sim b$ if $a-b=kn$ for some integer $k.$ The space $\Z/\sim\;:=\Z_n$ is called the set of integers modulo $n.$ Notice that  $\Z_n=\{0,1,2,...,n-1\}.$ Define the operation $(\cdot) \mod n:\Z\to \Z_n$ which sends $k\in \Z$ to $[k]$ which is equivalent to its remainder after dividing by $n.$   
\end{example}

\section{Group Theory}
We have opted to start this section with a few examples to introduce the idea of a group before giving the rigorous definition. 
\begin{example}
\textbf{}
	\begin{enumerate}
		\item Consider the set $\Z.$ We can define $+:\Z\times \Z\to \Z$ by $(a,b)\mapsto a+b.$ Clearly if $a\neq 0,$ then $-a$ exists and is different from $a.$ Further $a+(-a)=0.$ This makes $0$ the additive identity in $\Z.$ 
		\item Let $D_n$ denote the set of symmetries of the regular $n$-gon. Then it is left as an exercise to the reader, to prove that $|D_n|=2n.$ Note that we can compose two such symmetries. Take for example the case $n=4.$  Let the rotation by $90^\circ$ counterclockwise be denoted $r=R_{90^\circ}$ and the vertical reflection $s.$ Then $rs$ is the reflection along the primary diagonal. There is an identity element $r^0=R_{0^\circ}.$     
		\item Let $\C^\times$ denote the set of all non-zero complex numbers. Then we can define $\cdot:\C^\times\times \C^\times\to \C^\times$ by $(w,z)\mapsto w\cdot z=wz$ the standard complex multiplication. Here $1$ is the multiplicative identity.  
	\end{enumerate}
\end{example}
With these examples in mind, we can now define groups in more abstraction. In general, one can think of groups as symmetries of some object, be it an $n$-gon or some set. We will make this more precise. 
\begin{definition}
	Let $G$ be a set and define $\mu:G\times G\to G$ be a binary operation such that \begin{enumerate}
		\item For all $x,y,z\in G,$ $\mu(x,\mu(y,z))=\mu(\mu(x,y),z).$ 
		\item There exists $ e\in G$ such that $\mu(e,g)=g=\mu(g,e)$ for all $g\in G$
		\item For all $g\in G$ there exists $h\in G$ such that $\mu(g,h)=\mu(h,g)=e.$
	\end{enumerate}
	We commonly denote $\mu(g,h)$ as $gh$ when the operation is clear. Further, the last condition tells us that every element has an inverse and we denote $g^{-1}:=h$ from that condition. We call $G$ equipped with $\mu,$ a \textbf{Group} and denote it $(G,\mu).$ We say a group is \textbf{Abelian} if for all $g,h\in G,$ we have that $gh=hg.$       
\end{definition}
\begin{remark}
	Other common notations for groups are $(G,\cdot)$ and $(G,\star)$ where $\cdot$ and $\star$ denote the multiplication operations.  
\end{remark}
It should now be obvious that $(1)$ and $(3)$ in Example $2.2.1$ are example of groups (i.e. every integer has an inverse, namely its negative and every non-zero complex number is invertible. For $(2),$ notice that applying $r$ $n-$times, we get $e.$ Therefore $r^n=e$ and $r^{n-1}=r^{-1}.$ Further, $s^2=e.$ and so $s$ is its own inverse. 

Now we lay down some important non-examples. These, for various reasons, violate one or many of the group axioms. 

\begin{nexample} \text{}
	\begin{enumerate}
		\item Consider the $\Z,\Q,\R$ under standard multiplication. $\Z$ is not a group as all other elements than $\pm 1,$ are not invertible as $\frac{1}{n}$ is not an integer. Why do $\Q$ and $\R$ fail?  
		\item(Integers Modulo $n$) Let $\Z_n$ denote the set of integers $\{0,1,...,n-1\}$ together with multiplication modulo $n.$ Multiplying modulo $n,$ means that we first multiply the numbers using normal arithmetic and then "remove"  $n$ as many times as possible and the remaining number is their product. For an example let $n=5,$ then \[ 3\cdot 4\equiv 2\mod 5\] Another interpretation of this involves remainders. When long dividing, if the two objects do not divide one another, we are left with a remainder. For integers, $(\cdot) \;\text{mod } n$ precisely gives the remainder when dividing by $n.$ Under this multiplication operation not every element here has an inverse, namely $0.$ For $n\neq p$ a prime number, we can find other elements which are not invertible. Take for instance $n=6$ and the element $2.$ We leave it to the reader to check this. 
	\end{enumerate}
\end{nexample}

\begin{lemma}   
	For any group $(G,\cdot),$ inverses are unique. Further, the identity element is unique.
\end{lemma}
\begin{proof}
	Let $g\in G.$ Suppose there exist $h,h'\in G,$ $h'\neq h$ both inverses for $g.$ Then on one hand we have that \[ h'gh=(h'g)h=eh=h\]
	on the other hand we have that \[ h'gh=h'(gh)=h'e=h'\]
	Therefore $h'=h$ a contradiction. Hence, $h=h'$ is the unique element such that $gh=hg=e.$ To see that the identity element is unique, use the same process as above. This completes the proof. 
\end{proof}
\begin{remark}
	For the remainder of the text, we will refer to groups by the underlying set $(G,\cdot):=G$ when the multiplication is understood and there is no room for confusion. This is standard notation and in most cases the multiplication is well understood. We will specify the multiplication when we have a choice of operation.
\end{remark}
\begin{corollary}
	If $g,h\in G$ are any elements. Then $(gh)^{-1}=h^{-1}g^{-1}.$ 
\end{corollary}
\begin{corollary}
	Suppose $G$ is a group such that every non-identity element is an involution (that is $g^2=e$). Then $G$ is abelian. 
\end{corollary}
\begin{proof}
	This proof is left as an exercise to the reader. \textbf{Hint:} Using the fact that inverses are unique, realize that $x^2=e\implies x=x^{-1}$ for all non-identity elements.  	
\end{proof}

In practice, it can be hard to know if a given set is indeed a group.
The following theorem is integral in identifying groups from abstract sets. 
\begin{theorem}
	Let $G$ be a set equipped with an associative binary operation and suppose there exists $e\in G$ with the following properties \begin{enumerate}
		\item $ge=g$ for all $g\in G$
		\item For every $g\in G$ there exists $h\in G$ such that $gh=e.$
	\end{enumerate}
	Then $G$ is a group. 
\end{theorem}
\begin{proof}
	For $g\in G,$ pick $h\in G$ as in $(2).$ Then it suffices to show that $eg=g$ and $hg=e.$ Using $(2)$ again for $h,$ we can find an element $i\in G$ such that $hi=e.$ Then \[ g=ge=g(hi)=(gh)i=ei=i\]
Therefore, \[ hg=h(ei)=(he)i=hi=e\]
as desired. Now, \[g=ge=g(hg)=(gh)g=eg.\] This completes the proof.  	   
\end{proof}

This theorem gives us a criterion to check whether or not a set $X$ is actually a group. In practice, this is much more convenient to check than the entirety of the group axioms. An example of this is the set $\N.$ We have an identity element $0.$ However, we cannot find $n^{-1}=-n$ for any non-zero element. Therefore $\N$ is not a group. However, $\N$ is the prototypical example of a \textbf{Semi-Group}: a set which has an associative, unital binary operation where not every element has an inverse. These objects will play a role in chapter $6$ when discussing Toric Varieties. 

The following lemma is provided for ease with later proofs. It gives a criterion for a subset $H\subseteq G$ to be a subgroup. 
\begin{lemma}[Subgroup Criterion]
	Let $H\subseteq G$ be any subset. Let $x,y\in H.$ If $xy^{-1}\in H,$ for all $x,y\in H,$ then $H$ is a group and thus a subgroup of $G.$
\end{lemma}
\begin{proof}
	Assume $xy^{-1}\in H$ for all $x,y\in H.$ If $x=y,$ then $xx^{-1}=e\in H.$ Let associativity is clear as the multiplication is inherited from $G.$ To show every element is invertible, consider $x\in H$ and $e,$ then $ex^{-1}=x^{-1}\in H.$ Using this, consider $x$ and $y^{-1}.$ Then $x(y^{-1})^{-1}=xy\in H.$ Therefore $\cdot:H\times H\to H$ defines an associative, binary, unital, and invertible map. Hence, $H$ is a group and in fact a subgroup of $G.$ 
\end{proof}

\subsection{Group Homomorphisms}
Now that we have the basic objects of this section, we can consider maps between them. Note in the following definition, the maps act as you would expect: preserving the structure of both groups.
\begin{definition}
	Let $\varphi:(G,\cdot)\to (H,\star)$ be a map. $\varphi$ is a \textbf{Group Homomorphism} (or a Morphism of groups, see Ch. 3) if for all  $g,g'\in G,$ we have that \[ \varphi(g\cdot g')=\varphi(g)\star \varphi(g')\]
	That is $\varphi$ is \textbf{equivariant} with respect to the multiplication operations on $G$ and $H.$ A group homomorphism which is bijective is called an \textbf{Group Isomorphism}. If such a map exists, then the domain and codomain groups are said to be isomorphic and denoted $G\cong H.$  
\end{definition}
\begin{example}\text{}
	\begin{enumerate}
		\item Let $G=\{1,-1,i,-i\}$ where $i=\sqrt{-1}.$ Then define $f:\Z \to G$ by $f(m)=i^m.$ This makes $f$ a group homomorphism as \[f(m+n)=i^{m+n}=i^m\cdot i^n=f(m)\cdot f(n)\]
		\item Let $\R$ denote the set of all real numbers and $\R^\times$ denote the set of non-zero real numbers. Define $g:\R^\times \to \R^times$ by $x\mapsto x^2.$ Then $g$ is a homomorphism as $(xy)^2=x^2y^2$ for real numbers. We encourage the reader to investigate how changing the domain and/or range to $\R_{\geq 0}$ changes the properties of the homomorphism.   
	\end{enumerate}
\end{example}
We now lay down two definitions which are integral to the study of algebra and have analogs in all other branches of mathematics. 

\begin{definition}
	Let $H\subseteq G$ be a set contained in a group $G.$ We call $H$ a \textbf{Subgroup} if for all $h,h'\in H,$ $h\cdot h'\in H$ and $h^{-1}\in H.$ We denote subgroups using the notation $H\leq G.$ Denote by $gH=\{gh:h\in H\}$ for any $g\in G.$ Then we call a subgroup \textbf{Normal} if \[gHg^{-1}=H\] for all $g\in G$ and write $H\trianglelefteq G.$ Denote by \[Z(G)=\{g\in G:gh=hg, \forall h\in G\},\] the \textbf{Center} of $G. $ It should be obvious that $Z(G)$ is a normal subgroup of $G.$  
\end{definition}
\begin{definition}
	Let $\varphi:G\to H$ be a group homomorphism. Define the \textbf{Kernel} of the homomorphism $\varphi$ to be \[\ker \varphi=\{g\in G: \varphi(g)=e_H\}\]
	This is the set of all elements which are annihilated under the mapping $\varphi.$ 
\end{definition}
\begin{proposition}
	The set $\ker \varphi$ is a group. In particular, it is a normal subgroup of $G.$ 
\end{proposition}
\begin{proof}
We first show that $\ker \varphi$ is non-empty. Let $e\in G$ be the identity. We claim that $\varphi(e_G)=e_H.$ To see this, recognize that \[\varphi(gg^{-1})=\varphi(g)\star \varphi(g^{-1})=\varphi(g)\star \varphi(g)^{-1}=\varphi(e_G)=e_H\] for all $g\in G.$ Thus, $\ker \varphi$ is nonempty. 
	As $\ker \varphi\subseteq G,$ $\ker \varphi$ inherits multiplication from $G.$ Notice that for $x,y\in \ker \varphi,$ we have that \[\varphi(x\cdot y)=\varphi(x)\star \varphi(y)=e_h\]
	Therefore $\ker \varphi$ is closed under multiplication. Further, it is closed under inverses for the same reason. Hence, $\ker \varphi $ is a group and $\ker \varphi \leq G$. To check normality, notice that \[  \varphi(gxg^{-1})=\varphi(g)\star\varphi(x)\star\varphi(g)^{-1}=e\]
	for all $g\in G.$ Hence, $\ker \varphi \trianglelefteq G.$     
\end{proof}
As shown by the proof above, the homomorphism condition is quite restricting and powerful. We used the fact that $\varphi(g^{-1})=\varphi(g)^{-1}$ for group homomorphisms. it is left to the reader to check this fact.  Now, we have the following result which is important when proving other theorems. 
\begin{theorem}
	Let $\varphi:G\to H$ be a group homomorphism. Then $\varphi$ is injective if and only if $\ker \varphi=\{e\}.$ 
\end{theorem}
\begin{proof}
	$(\Rightarrow)$ Assume that $\varphi$ is injective. That is $\varphi(x)=\varphi(y)\implies x=y$ for all $x,y\in G.$ Then let $g\neq e\in \ker \varphi.$ By injectivity, \[\varphi(g)=\varphi(e)=e_H \implies g=e \]
	Hence $\ker \varphi=\{e\}.$\\\\
$(\Leftarrow)$ Assume now that $\ker \varphi=\{e\}. $ Then suppose $\varphi(g)=\varphi(h).$ This tells us that \[\varphi(gh^{-1})=\varphi(g)\varphi(h)^{-1}=e_H.\] Therefore $gh^{-1}\in \ker \varphi.$ As $\ker \varphi =\{e\},$ we know $gh^{-1}=e$ and hence, $g=h.$ This completes the proof. 
\end{proof}

\begin{example}
	Let $G$ be a \textbf{simple} group (that is the only normal subgroups are $\{e\}$ and $G$ itself). Then any map $f:G\to H$ is either injective or trivial. This follows from Proposition $2.2.15$ and Theorem $2.2.16$.  
\end{example} 

Just as with sets, we can build the Cartesian product of groups, $G$ and $H,$ denoted $G\times H.$ As a set it is precisely the set $G\times H,$ but now we endow this with a group structure taken  component-wise. That is \[(g_1,h_1) (g_2,h_2)=(g_1 g_2,h_1 \star h_2)\]
For a concrete example, consider the set $\Z\times \R^\times$ ($\R^\times$ is the group of all non-zero real numbers under multiplication). Here \[(k_1,r_1)(k_2,r_2)=(k_1+k_2,r_1r_2)\] as the multiplication in $\Z$ is addition.    

\subsection{Quotient Groups and an Isomorphism Theorem}
At this point, we have the ability to construct a group, transition between groups, and "multiply" groups to make new ones. Just as with high-school algebra, we can now consider dividing, or taking quotients of groups. 
\begin{definition}
	Let $G$ be a group and $H$ any subgroup. We denote by $G/H$(resp. $H\backslash G)$ the set of all left(resp. right) \textbf{Cosets} \[gH=\{gh:h\in H\}\]
	under the equivalence relation that $g\sim g'\iff g=gh'$ for some $h\in H.$
This is in general not a group as multiplication is not well defined. 
\end{definition}
Notice how the notation for this set of cosets is the same notation we use for equivalence relations on a set. The reason for this is that then we take left(right) cosets, we are essentially glueing $G$ along the orbits of the subgroup $H.$

The first question one can ask about this set is when does it  becomes a group? In other words, for what $H\leq G$ is $G/H$ a group. The following Theorem provides an answer. 
\begin{theorem}
	Let $G$ be a group and $N$ a subgroup. Then $G/N$ (read $G$ mod $N$) is a group under the operation $(gN)(hN)=(gh)N$ if and only if $N\trianglelefteq G.$ Further there is a canonical homomorphism $G\to G/N$ which sends $g\mapsto gN$ such that $\ker (G\to G/N)=N.$ 
\end{theorem}
\begin{proof}
	We first need to show that the proposed group operation is well defined. Suppose $xN=gN$ and $yN=hN.$ These two statements are equivalent to $x=gn$ and $y=hn'.$ Then \[(xy)N=(xN)(yN)=(gn)N(hn')N=g(nN)h(n'N)=(gN)(hN)=(gh)N\]
Thus the multiplication is well defined. \\

\noindent $(\Rightarrow)$ Now assume $N$ is normal in $G.$ Multiplication is associative by definition and the unit element is $eN.$ It remains to show that $gN$ has an inverse and that it is unique. Let $g^{-1}$ be the inverse of $g$ in $G.$ Then \[ (gN)(g^{-1}N)=(gg^{-1})N=eN\]
So $g^{-1}N$ is an inverse for $gN.$ Suppose there exists some $y\in G$ such that $(gN)(yN)=(yN)(gN)=N.$ Then starting from the middle: \[ yN=(eN)(yN)=(g^{-1}N)(gN)(yN)=(g^{-1}N)(eN)=g^{-1}N\]
Hence, $g^{-1}N=yN$ and $G/N$ is a group. 

We defer the other direction of the proof for a moment. Define $\varphi:G\to G/N$ by $\varphi(g)=gN.$ This is a homomorphism by the multiplication in $G/N.$ If $x\in \ker \varphi$ then $\varphi(x)=xN=N.$ Therefore $x\in N$ and $\ker\varphi \subseteq N.$ The reverse inclusion is obvious and thus \[\ker \varphi=N\]

\noindent $(\Leftarrow)$ Now suppose $G/N$ is a group. Consider the canonical projection $G\to G/N$. Then $\ker \varphi=N$ by above and by Proposition 2.2.15 we conclude that $N$ is normal.    
\end{proof}
\begin{corollary}
	Let $G$ be an abelian group. Then for every subgroup $H\leq G,$ $G/H$ is an abelian group. 
\end{corollary}
\begin{proof}
	The fact that $G$ is abelian tells us that $gH=Hg$ for all $g\in G.$ To see that $G/H$ is abelian, let $g,g'\in G.$ Then \[(gH)(g'H)=(gg')H=(g'g)H=(g'H)(gH)\]
\end{proof}

\begin{example}
	Recall the group $\Z_n$ from above. To formally define $\Z_n,$ we consider the group $\Z$ and the subgroup of multiples of $n$ denoted $n\Z.$ Then \[\Z_n:=\Z/n\Z\] 
	That is, we glue the integers along the multiples of $n.$ The group operation in $\Z$ is $+$ and therefore \[ x\equiv y \mod n \iff x+y=kn, k\in \Z\]
 The quotient is a group as $\Z$ is abelian. 	
\end{example}

Now that we have the idea of quotients, we can define one of the most useful theorems in algebra: the First Isomorphism Theorem. The proof of which will introduce one of the most fundamental objects in algebra: the commutative diagram. These will show up many times in the latter parts of this text and as such, we encourage the reader to try and prove the following theorem themselves before reading the proof. 
\begin{theorem}[First Isomorphism Theorem]\label{First Iso}
	Let $G,H$ be groups and $\varphi:G\to H$ be a group homomorphism. Then \[G/\ker \varphi\cong \varphi(G)\]	
\end{theorem}
\begin{proof}
	Consider the commutative diagram \[ \begin{tikzcd}
	G \arrow[r,"\varphi"] \arrow[d,swap,"q"] & \varphi(G)\\
	G/\ker \varphi \arrow[ur,swap,dashed,"\hat{\varphi}"]
	\end{tikzcd}\] 
	The top arrow is surjective by definition and the map $q$ is the canonical quotient. Denote the cosets in $G/\ker \varphi$ as $[g].$ We define the map $\hat{\varphi}([g])=\varphi(g).$ To show that $\hat{\varphi}$ is well defined, consider $[g]=[h]$ that is $g=ah$ where $a\in \ker \varphi.$ Then \[\varphi(g)=\varphi(ah)=\varphi(a)\varphi(h)=\varphi(h)\]
	Thus, $\hat{\varphi}$ is well defined. It is a homomorphism as \[  \hat{\varphi}([g][g'])=\hat{\varphi}([gg'])=\varphi(gg')=\varphi(g)\varphi(g')=\hat{\varphi}([g])\hat{\varphi}([g']) \]
By the commutativity of the diagram, $\hat{\varphi}$ is surjective. We compute \[\ker \hat{\varphi}=\{[g]\in G/\ker \varphi: \varphi(g)=e_H\}= \ker \varphi \] 	
As $\ker \varphi$ is the identity element in the quotient space, $\hat{\varphi}$ is injective. Hence, $\hat{\varphi}$ is an isomorphism.  
\end{proof}

\begin{corollary}
	If $\varphi:G\to H$ is a surjective homomorphism then \[G/\ker \varphi\cong H.\] 
\end{corollary}

The next tool we will discuss is fundamental to the study of algebra. 
\begin{definition}
	Consider a sequence of groups \[\begin{tikzcd}
	...\arrow[r]&G_{i-1} \arrow[r,"d_i"] & G_{i} \arrow[r,"d_{i+1}"] & G_{i+1}\arrow[r] &... 
	\end{tikzcd}\]
	we say that the sequence is \textbf{exact} at $G_i$ if $\ker d_{i+1}=\im d_i.$ If the sequence  is exact at every $G_i,$ we say the sequence is exact and we call it a \textbf{Long Exact Sequence}. If the sequence has the following form \[\{e\}\to G_1\to G_2\to G_3\to \{e\}\]
	we say the sequence is a \textbf{Short Exact Sequence}.
\end{definition}

\begin{example}
Let $G$ be a group and $N$ a normal subgroup. We can rephrase the quotient construction as the unique (up to isomorphism) group $H$ such that the following sequence is exact \[ \{e\}\to N\to G\to H\to \{e\}\]
Here, the arrow $N\to G$ is the inclusion. Exactness tells us that $N\to G$ is injective, and that $G\to H$ is surjective. Thus, by the first isomorphism theorem, $G/\ker (G\to H)\cong H.$ As $\ker (G\to H)=\im(N\to G)=N,$ we have our result. 
\end{example}

\subsection{Group Actions}
Let $G$ be a group. Just as with the dihedral groups $D_n,$ we can ask how a group may act on a set; that is, how does it permute the elements? The formalization of this, a group action, is essential when understanding the later topics in this section. We give the following two definitions
\begin{definition}
	Let $G$ be a group and $X$ be a set. A (left)\textbf{Group Action} on $X$ is a map $\cdot:G\times X\to X$ such that \begin{enumerate}
	\item $h\cdot(g\cdot x)=hg\cdot x$
	\item $\exists e\in G$ such that $e\cdot x=x$ for all $x\in X.$
	\end{enumerate}
\end{definition}
\begin{definition}
	Let $G$ be a group and $X$ a set as above. A (left)group action is a group homomorphism \[\varphi:G\to \text{Sym}(X)\]
	where $\text{Sym}(X)=\{f:X\to X: f$ is bijective$\}.$ This is a group under composition. Inversion is well defined as every map is bijective. This is called the permutation representation of the group $G$ on $X.$ 
\end{definition}
\begin{lemma}
	Definition $2.2.26$ and $2.2.27$ are equivalent. 
\end{lemma}
\begin{proof}
	It is clear that $2.2.27\implies 2.2.26$ as a group homomorphism gives the associativity and the identity element of the group gives the identity map.  
	
	Thus it suffices to show that $2.2.26\implies 2.2.27.$ Define $\psi_g$ to be the map on $X$ such that $\psi_g(x)=g\cdot x.$ We know that $\psi_g$ is invertible $(\psi_{g^{-1}}$) and thus $\psi_g\in \text{Sym}(X).$ Define a map $\varphi:G\to \text{Sym}(X)$ by $\varphi(g)=\psi_g.$ Then by the associative property of the action we get that \[\varphi(gg')=\psi_{gg'}=\psi_{g}\comp \psi_{g'}=\varphi(g)\comp \varphi(g')\]
	Hence, $\varphi$ is a group homomorphism and the definitions are equivalent.   
\end{proof}
We can think of group actions as shuffling the elements of the set they act on. The kernel of an action is precisely the kernel of the resulting homomorphism. We say an action is \textbf{faithful} if the associated permutation representation is injective. Further, we call an action \textbf{transitive} if the has precisely one orbit. That is, for every pair $(x,y)\in X\times X$, there exists $g\in G$ such that $g\cdot x=y.$  
\begin{remark}
	We have been careful to refer to left and right multiplication. If $G$ is non-abelian, these are different operations. When doing more advanced mathematics, one can consider multiplication or an action on both the left and the right. This has some major consequences but as we will not make use of the them, we have made the decision to omit such a discussion.
\end{remark}  
\begin{lemma}
	Let $G$ be a group and suppose $G$ acts on a set $X$.
	 \begin{enumerate}
		\item Let $\stab_G(x)=\{g\in G:g\cdot x=x\}$ and $\orb_G(x)=\{g\cdot x:g\in G\}$ denote the stabilize and orbit of the point $x\in X$ under the action of $G.$ Then $\stab_G(x)$ is a subgroup of $G.$
		\item If $X=G,$ then the action is transitive and faithful. Further, any subgroup acts faithfully. 
	\end{enumerate}
\end{lemma}
\begin{proof}
	\begin{enumerate}
		\item It is clear that $\stab_G(x)$ is a subset of $G.$ It carries the standard group multiplication and is non-empty as $e\in \stab_G(x).$ It suffices to  show that all non-identity elements have an inverse. Let $g\in \stab_G(x).$ Then \[  x=e\cdot x=(g^{-1}g)\cdot x=g^{-1}\cdot (g\cdot x)=g^{-1}\cdot x\]
Thus, $g^{-1}\in \stab_G(x)$ and by the subgroup criterion, $\stab_G(x)$ is a subgroup of $G$. 
		\item Let $G$ act on itself by left multiplication. To show the action is faithful, suppose $g\cdot h=g'\cdot h.$ Then \[ gh=g'h\iff ghh^{-1}=g'hh^{-1}\iff g=g'\]
		So the map $G\to \text{Sym}(G)$ is injective. To show it is transitive, let $h,i\in G,$ we need to show that there is an element $g\in G$ such that $gh=i.$ Pick $g=ih^{-1}.$ This is a group element and $gh=ih^{-1}h=i.$ Therefore every element is in the orbit of a single element, namely the identity element. Hence, the action is faithful and transitive. As $H\leq G,$ this faithful map restricts to any subgroup. \end{enumerate} \end{proof}

\begin{corollary}
	Let $G$ act on a set $X.$ If this action is transitive, then it is equivalent to the action of $G$ on $G/H$ by left multiplication for some $H\leq G.$  
\end{corollary}
\begin{proof}
	Let $x\in X$ and  consider $H=\stab_G(x).$ Transitivity gives us that for all $y\in X,$ $y=gx$ for some $g\in G.$ Suppose $gx=g'x$ then $(g')^{-1}g\in H.$ This makes the map  \[ f:G/H\to X\;\;\;\; gH\mapsto gx\] a bijection. It remains to show that this map is $G$-equivariant. Let $g\in G$ and $w\in X.$ We can write $w=g_1x$ for some $g_1.$ Then \[ f(gg_1H)=gg_1x=gw=gf(g_1H)\]
	Hence $f$ is $G$-equivariant and the actions are equivalent. 
\end{proof}
In this case, $G/H$ is called the orbit space of the action as no element is stabilized in the set. This will play an important role in the next chapter. 

\section{Vector Spaces and Linear Algebra}
Linear algebra is one of the oldest and core subjects to mathematics. It began as the study of solutions to linear systems of equations and has grown into the study of transformations on vector spaces. For example, given the following set of equations 
\begin{align*}
	3x+2y&=4\\
	x+6z&=13\\
	3x+6y-z&=0
\end{align*}
what values of $x,y,z$ satisfy them? There are a variety of ways to find solutions, but perhaps the simplest is to use matrices. 
\begin{definition}
	A \textbf{Matrix} is any rectangular array of numbers, symbols, operators, etc. arranged in rows and columns such that addition and multiplication are well defined. If $A$ is a matrix of finite size, it is convention to read the lengths of the sides as ``rows by colums". That is a matrix with $3$ rows and $4$ columns is a $3\times 4$ matrix. Addition is taken component-wise whereas multiplication is done as follows: let $A,B$ be  $m\times n$ and $n\times k$ matrices. Then \[ (AB)_{ij}=\sum_{l=1}^n A_{il}B_{lj} \]
where $A_{ij}$ is the element of $A$ in the $i^{th}$ row and $j^{th}$ column. An $n\times n$ matrix $A$ is $\textbf{invertible}$ if there exists an $n\times n$ matrix $B$ such that $AB=BA=I_n$ which is the matrix with $1s$ along the main diagonal and $0$ elsewhere. 
\end{definition}
We can turn the system of equations above into the single matrix equation \[ \Mat{3 & 2 & 0\\1 & 0 & 6\\ 3 & 6 & -1} \Mat{x\\y\\z}= \Mat{4\\13\\0}\]
We leave it to the reader to check that $x=1,y=\frac{1}{2},z=2$ is the solution. We shall spend no time talking about the various methods for solving linear systems of equations as they are no use to the latter parts of the text. Instead we shall spend a majority of this section on abstract vector spaces defined over a field (defined below).

\subsection{Field Theory, Briefly}

As we have just seen with groups, endowing a set with multiplication has some striking implications. In this section we consider a new algebraic object, a field. Broadly, this is a set equipped with two operations, addition and multiplication which are compatible. 
\begin{definition}
	Let $F$ be a set and suppose it is equipped with two operations $+,\cdot.$ Let $F^\times$ denote the set of non-zero elements of $F.$ Suppose $(F,+)$ and $(F^\times,\cdot)$ are abelian groups. If for all $a,b,c\in F,$ \[ a(b+c)=ab+ac=ba+ca=(b+c)a\]
	then $F$ is a \textbf{Field}. In a field, we denote the identity for the addition as $0=0_F$ and for multiplication as $1=1_F.$ For any field, we can define  the $\textbf{characteristic}$ of $F,$ $\Char F$ to be the minimal  $n\in \N$ such that  $n\cdot 1=0.$ If no such $n$ exists, we say that $\Char F=0.$
\end{definition} 

\begin{example}
	The quintessential example of a field is the real numbers $\R.$ One can then construct $\C$ the complex numbers as a field which contains $\R.$ These fields both have $\Char F=0.$ For an example of positive characteristic, consider $\Z_p$ where $p$ is prime. This is a field and has characteristic $p.$ A good exercise to test your understanding is to prove that $\Z_n,$ for $n\neq p^l$ for some $l\neq 0\in \N,$ fails to be a field.  
\end{example}

\begin{example}[Polynomials]\label{Polynomials}
	Let $F$ be a field and denote by $F[x]$ the set of all formal polynomials $\sum^n a_ix^i,$ with $a_i\in F.$  For any polynomial $f\in F[x]$ define the degree of $f,$ denoted $\deg f,$ to be $\deg f=\max\{i:a_i\neq 0\}.$  We define addition as \[ \sum^n a_ix^i+\sum^m b_ix^i=\sum^{\max\{n,m\}} (a_i+b_i)x^i\]
where $a_i$(resp.$b_i)$ is considered to be $0$ if $i>n(resp. i>m)$ and multiplication as \[ \left(\sum^n a_ix^i\right)\cdot \left( \sum^m b_ix^i  \right)=\sum^{n+m}\left( \sum_{j+k=i} a_jb_k\right) x^i\]
This makes $F[x]$ a group under addition. It is not a group under multiplication as the set of invertible elements is precisely the constant polynomials as $x^n$ does not have an inverse for $n\geq 1.$ Therefore $F[x]$ is not a field. As we will see later, $F[x]$ is a ring. (See Section 2.4) If $f\in F[x]$ cannot be written as $f=gh$ for $g,h\in F[x]$ and $\deg g,\deg h\neq 0$ then $f$ is said to be $\textbf{irreducible}.$

We sometimes adjoin numbers to a field in the same way we do with formal variables. Let $i=\sqrt{-1}.$ Then $\R[i]$ Then by the rules above this consists of all finite sums $\sum r_j i^j.$ However, $i^2=-1$ and we can reduce this set to be \[ \R[i]=\{ a+bi: a,b\in \R\}=\C\]
This is precisely the definition of the complex numbers.   
\end{example}

\begin{remark}
We will only concern ourselves with characteristic $0$ as positive characteristic is a bit technical and does not play a role in the later chapters of this text. 
\end{remark}

\begin{definition}
	Let $E$ be a field which contains $F$ as a subfield. Then we say that $E$ is an \textbf{extension} of $F$ and denote this $E/F.$ Further, the degree of the extension, denoted $[E:F],$ is the integer $n$ such that $E\cong F^n=\prod^n F.$   
\end{definition}
Similar to groups, we can define Field homomorphisms. 
\begin{definition}
	Let $F,E$ be field and $f:F\to E.$ If for all $a,b\in F$ we have that \begin{align*} f(a+b)=f(a)+f(b)&& f(ab)=f(a)f(b)
 \end{align*}
 the $f$ is a \textbf{Field Homomorphism}. A bijective field homomorphism is an isomorphism. 
\end{definition}
For groups, this was where the story ended. For fields, due to the added structure, we have the following lemma. 
\begin{lemma}
	Every non-zero field homomorphism is injective. 
\end{lemma}
\begin{proof}
	Let $F,E$ be fields and $x,y\in F.$ Suppose $\alpha:F\to E$ is a morphism. If $\alpha(x)=\alpha(y),$ then \[ \alpha(x)-\alpha(y)=0\iff \alpha(x-y)=0=\alpha(0)\]
	If $a-b\neq 0,$ then put $q=a-b.$ Using the multiplication, $\alpha(q)\alpha(q^{-1})=\alpha(1)=1.$ But $\alpha(q)=0.$ This is a contradiction and thus $a-b=0$ and $\alpha$ is injective.  
\end{proof}
\noindent In this proof we used the fact that for non-trivial field homomorphisms $\alpha(1)=1.$ We leave it to the reader to check this. 
\begin{definition}
	Let $F$ be a field and $S$ any subset. Denote by $F_S$ the subfield of $F$ containing $S.$ It is a fairly simple exercise to show that this field always exists. For the special case that $S=\{1\},$ we call $F_1=F'$ the prime subfield of $F$ as it is the field generated by $1.$ A less trivial exercise is to prove that if $F$ is finite with characteristic $p$ then $F'\cong \Z_p$ and if $F$ is infinite and $\Char F=0,$ then $F'\cong \Q.$
\end{definition}
\begin{definition}
	Let $L/K$ be a field extension. An element $a\in L$ is $\textbf{algebraic}$ over $K$ if there exists $f\in K[x]$ such that $f(a)=0.$ $L$ is called an \textbf{algebraic extension} if every element is algebraic. A field $L$ is called \textbf{algebraically closed} if any for all $f\in L[x],$ $f(x)=0\implies x\in L.$ 
\end{definition}	
\begin{example}\text{}
	\begin{enumerate}
		\item $\C/\R$ is an algebraic field extension as the degree of the extension is finite and by the Fundamental Theorem of Algebra, $\C$ is algebraically closed. 
		\item $\Q(\sqrt{2})/\Q$ is an algebraic extension. 
		\item $\R/\Q$ is \textbf{not} an algebraic extension. Consider the element $e=\lim_{n\to \infty} (1+\frac{1}{n})^n.$ This is known to be transcendental 
	\end{enumerate}
\end{example}

\begin{proposition}\label{minimal polynomial}
	Let $L/K$ be an algebraic extension. Then for every element $\alpha\in L,$ there exists a unique monic irreducible \footnote{Definition: A monic polynomial is a polynomial whose highest degree term has coefficient 1} $m_\alpha\in K[x]$ such that $m_{\alpha}(\alpha)=0$ and $\deg m_\alpha$ is minimal among polynomials which have $\alpha$ as a root.  
\end{proposition}
\noindent We shall omit the proof of this proposition as it does not add to the text. 

The last theorem we shall prove on fields tells us that every intermediate set, closed under addition and multiplication, of an algebraic extension is a field. More precisely, 
\begin{theorem}
	Let $L/K$ be an algebraic extension and $S$ a set such that $S$ is a group under addition and is closed under multiplication. If $L\supseteq S\supseteq K,$ then $S$ is a field.  
\end{theorem} 
\begin{proof}
	As $S\subseteq L,$ it is commutative and has a unit element. It suffices to show that for all $s\neq 0\in S$ that $s^{-1}$ exists and is contained in $S.$ Existence follows from the fact that $s\in L$ and is non-zero. To show it is contained in $S,$ we use Proposition \ref{minimal polynomial}. As $L/K$ is an algebraic extension, the minimal polynomial $m_s$ of $s$ over $K$ exists. Let \[ m_s=x^n+a_{n-1}x^{n-1}...+a_0\]    
with each $a_i\in  K.$ Evaluating at $x=s,$ we get \[  -a_0=s(s^{n-1}+...+a_1)\implies s(s^{n-1}+...+a_1)\left(\frac{-1}{a_0}\right)=1\]
By Lemma 2.15, we have that $s^{-1}=(s^{n-1}+...+a_1)\left(\frac{-1}{a_0}\right)\in S.$ Hence, $S$ is a field. 
\end{proof}
Just as with groups, we can talk about actions of fields on sets. This does not vary from the theory of groups however as $\text{Sym}(X)$ is not a field so defining the action in this way is uninteresting. We thus need a different object to study. 

\subsection{Vector Spaces}
Linear algebra has emerged from its concrete origins in system of equations to the beautiful abstract algebra it is today. Vector spaces comprise the main objects of study. These objects, as we will see, are incredibly well understood and intersect every area of mathematics. The main references for this section are $\cite{Cooperstein2015}$ and \cite{Knapp2006}. 

We begin with the definition.  
\begin{definition}
Let $V$ be a set, and $F$ a field. Equip $V$ with two operations \begin{align*}
	+:V&\times V\to V\\
	\cdot:F&\times V\to V
\end{align*}
which are compatible in the sense that for all $f\in F$ and $v,w\in V,$ we have that $f(v+w)=fv+fw=(v+w)f$ and $1v=v.$ If under these operations $V$ is an abelian group together with an action of $F,$ we say $V$ is an $F$-$\textbf{Vector Space},$ with elements $v\in V$ called \textbf{vectors} and elements $f\in F$ called \textbf{scalars}. The element $fv$ is a scaled vector. A subset $W\subseteq V,$ which is closed under the operations of addition and scalar multiplication is called an $F$-vector subspace. Typically we simply say subspace if the underlying field is understood. 
\end{definition}

\begin{example}
	We have already seen some examples of vector spaces and subspaces. 
	\begin{enumerate} 
	\item Let $F$ be a field and $E$ a finite field extension. It is clear that a field satisfies the definition of a vector space over itself. Now, by the finiteness condition on $E,$ we know that $E\cong F^n$ and therefore we can extend the action of $F$ to each component of $E.$ That is \[ f\cdot e=f\cdot (e_1,...,e_n)=(fe_1,fe_2,...,fe_n)\]
	This is given by the diagonal inclusion of $F\hookrightarrow F^n$ which sends \[ f\mapsto \underset{\underset{n-times}{\rotatebox[origin=c]{-90}{$\Bigg\}$}}}{(f,f,...,f)}\]  
	
	\item For a non-trivial example consider the space $F[x].$ 
	\end{enumerate}
\end{example}

\begin{definition}
An $F$-$\textbf{linear combination}$ of vectors is anything of the form $v=\sum a_iv_i$ for finitely many $i$ with each $a_i\in F.$ If $v_1,...,v_m$ is a collection of vectors in a vector space $V,$ denote by \[ \ip{v_1,...,v_m}\] the set consisting of all linear combinations of the $v_i.$ This is canonically a subspace of $V.$ We say that $v_1,...,v_m$ is a \textbf{spanning set} for a vector space $V$ if every $v\in V$ can be written as a linear combination of the $v_i.$ Given a set $B$ we will denote by \[ \text{Span}_F(B)\]
the minimal vector space generated by the elements of $B.$ We will omit $F$ if it is clear from the situation and or if the section is true regardless of the field chosen.    
\end{definition}

\begin{corollary}
	Every vector space admits a spanning set.
\end{corollary}
This follows immediately from the definition as $V$ is a spanning set for itself. A more interesting statement is that there exists a unique (up to conjugation) \textit{minimal} spanning set 

\begin{definition}
	Let $v_1,...,v_n$ be vectors in a vector space $V.$ We say these vectors are \textbf{linearly independent} if \[ \sum^n_{i=1} a_i v_i=0 \iff a_i=0\;\; \forall i\]
	We will commonly abuse the term \textit{linearly independent} and refer to sets as linearly independent if all of the finite subsets of elements are linearly independent.  
\end{definition}
\begin{example} \textbf{}
	\begin{enumerate}
		\item Let $V=\C$ treated as a real vector space via the inclusion of $\R\hookrightarrow \C.$ Its elements are written as $z=x+iy.$ Let $z_1,z_2,z_3$ be three, non-colinear ($z_i\neq a_jz_j \;\;\forall i,j\in \{1,2,3\}$) complex numbers. It can be shown that $z_3$ can be written uniquely as $a_1z_1+a_2z_2.$         
		\item For a more concrete example consider $V=\R^3.$ Let \[v_1=\Mat{1\\2\\0} \;\;\;\;\;\;\;\;v_2=\Mat{4\\-2\\9}\;\;\;\;\;\;\;\;v_3=\Mat{16\\2\\27}\]
		It should be easy to see that $v_3=4v_1+3v_2.$ Notice that if we change the third coordinate of $v_2$ to $0,$ we have that $v_3$ is no longer a linear combination of $v_1$ and $v_2.$     
	\end{enumerate}
\end{example}

\begin{definition}
	Let $\mathcal{B}=\{v_i\}_{i\in I}$ be a spanning set of the vector space $V$. We call $\mathcal{B}$ a \textbf{basis} if it is linearly independent. We denote elements of $V$ with respect to this basis as column vectors (tuples) $v=(k_1,...,k_n,...)^t$ which means $v=\sum_I k_1v_i.$ 
\end{definition}
It should be noted immediately that any basis $\mathcal{B}$ for a vector space is necessarily minimal among the sets with the above properties.  

\begin{theorem}\label{cardinality}
	Let $\mathcal{B}$ and $\mathcal{C}$ be two bases for the vector space $V.$ Then $|\mathcal{B}|= |\mathcal{C}|.$
\end{theorem}
\begin{proof}
	We shall prove this is two cases $|\mathcal{B}|$ is finite and $|\mathcal{B}|$ is infinite. 
	 Suppose first that $|\mathcal{B}|<\infty.$ We want to give bounds on the size of $\mathcal{C}.$ 
	\begin{lemma}
		Suppose that $|\mathcal{C}|>|\mathcal{B}|.$ Then $\mathcal{C}$ is linearly dependent. 
	\end{lemma}
	\begin{proof}
		As $\mathcal{B}$ is a basis, the set $\mathcal{B}\cup c_1$ must be linearly dependent. Therefore, up to reordering, we can assume that $b_n\in \text{Span}\{ c_1,b_1,...,b_{n-1}\}.$ This is now a linearly independent set. Notice that by assumption $\{c_j\}$ is linearly independent.  Therefore,  repeating the above process with $c_j$ for $2\leq j\leq n$ and reordering, we conclude that $\{ c_1,..., c_n\}$ is a linearly independent, spanning set. As $|\mathcal{C}|>n,$ we then conclude that $\mathcal{C}$ is linearly dependent.      
	\end{proof}	
\noindent From this lemma, we conclude that $|\mathcal{C}|\leq |\mathcal{B}|.$ The key step of the proof relied on the fact that $\mathcal{B}$ was a basis. We can similarly apply this logic to $\mathcal{C}$  and deduce then that $|\mathcal{B}|\leq |\mathcal{C}|.$ Hence, they must be equal. 
	
	Now assume $|\mathcal{B}|$ is infinite. The method above will not work as sets with infinite cardinality as adding an element does not give any information regarding linear dependence. We can rephrase this part of the proof however as \textit{there exists a bijection $f:\mathcal{B}\to \mathcal{C}$ }. We can construct such a function in the following way: let $\mathcal{B}=\{ b_i: i\in I\}$ and $\mathcal{C}=\{c_j:j\in J\}$ with $I,J$ some indexing sets of infinite cardinality. For an arbitrary element $c_j\in \mathcal{C},$ we know that $c_j\in \text{Span}(\mathcal{B}).$ In particular, we know that $c_j\in \text{Span}(B_j)$ a finite subset of $\mathcal{B}.$ Put \[ B=\bigcup_{j\in J} B_j\]
As $\mathcal{C}$ is a basis, it is in particular a spanning set. Therefore $B$ is also a spanning set. As $B\subseteq \mathcal{B}.$ we know that  $B=\mathcal{B}$ and therefore \[\mathcal{B}=\bigcup_{j\in J} B_j\]
As each $B_j$ is finite we know that $\left| \bigcup_{j\in J} B_j \right| \leq |J|=|\mathcal{C}|.$ Hence, $|\mathcal{B}|\leq |\mathcal{C}|$ and, by applying the same logic, we have that $|\mathcal{C}|\leq |\mathcal{B}|.$  The proof is completed by the following theorem, a proof for which can be found in \cite[Appendix A.6]{Knapp2006}.
\end{proof}
\begin{theorem}[Schroeder-Bernstein]
	If $A$ and $B$ are sets such that there exists an injective function $f:A\to B$ and and injective function $g:B\to A$ then $|A|=|B|.$ 
\end{theorem}

This now begs the question: "does every vector space admit a basis?" The next theorem will give an answer to this, but before giving a proof, we need the following famous lemma from Logic. 
\begin{lemma}[Zorn's Lemma]\label{Zorn's Lemma} 
	Let $P$ be a partially ordered set. Suppose that every totally ordered set has an upper bound. Then $P$ contains a maximal element.  
\end{lemma}
\begin{definition}
	A \textbf{partial order} on a set $X$ is a reflexive, antisymmetric, transitive, binary relation $\preceq$. A \textbf{total order} is a partial order such that for all pairs $(x,y)$ either $x\preceq y$ or $y\preceq x.$     
\end{definition}
The rest of the components of the lemma are self explanatory. The proof of this lemma will be omitted as it does not add to the text. Although it seems innocuous, this lemma provides the technical support for many proofs in algebra. For example: 

\begin{theorem}\label{basis}
	Let $V$ be a vector space defined over the field $F.$ Then: 
	\begin{enumerate}
		\item Every spanning set contains a basis. 
		\item Every linearly independent subset can be extended to a basis.
		\item $V$ has a basis. 
	\end{enumerate}
\end{theorem}	
We present the proof given in \cite{Knapp2006}.
\begin{proof}
	(b) Let $E$ be a linearly independent subset of $V.$ Let $\mathcal{S}$ be the collection of all linearly independent subsets of $V$ containing $E.$ Then $\mathcal{S}$ is a partially ordered set under inclusion and non-empty as $E\in \mathcal{S}$. Let $\mathcal{T}$ be a totally ordered subset of $\mathcal{S}$ and consider \[ A=\bigcup_{T\in \mathcal{T}} T\]
	We claim that $A\in \mathcal{S}.$ It clearly contains $E$ by construction. It remains to show it is linearly independent. To see this, suppose not. Then there exist $v_1,...,v_n\in A$ such that $c_1v_1+...+c_nv_n=0$ with not all $c_i=0.$ Let $A_j\in \mathcal{T}$ be an element which contains $v_j.$ Then as $\mathcal{T}$ is totally ordered. There exists some $A_n'$ such that $A_n'\supseteq A_j$ for all $j\leq n.$ As $A_n'$ is linearly independent, $c_i=0$ for all $i$, a contradiction. Hence, $A$ is linearly independent and an upper bound for $\mathcal{T}.$ Thus, all totally ordered sets have an upper bound and by Zorn's Lemma, there is a maximal element $B\in \mathcal{S}.$ it remains to be shown that $B$ is a spanning set. Let $v\in V$ be arbitrary. Suppose $v\notin \text{Span}_FB.$ Then $\{v\}\cup B$ is a linearly dependent set by the maximality of $B.$ Therefore, there exist constants $c,c_1,...,c_m$ and vectors $v_1,...,v_m$ such that \[ cv+c_1v_1+...+c_mv_m=0\] with not all $c,c_1,...,c_m=0.$ We know that $c\neq 0$ as $B$ is linearly independent. Therefore $v=-c^{-1}(v_1c_1+...+v_mc_m).$ Hence, $v\in \text{Span}_FB$ and $B$ is a spanning set.      
	
	(a) Now Let $E$ be a spanning set. Let $\mathcal{S}$ denote the partially ordered set of linearly independent subsets contained in $E$ ordered by inclusion. Let $\mathcal{T}$ be a totally ordered subset of $\mathcal{S}.$ Let $A$ be the union of all of the elements of $\mathcal{T}.$ Then it is clearly an upper bound by the argument in (b) above. By Zorn's Lemma $\mathcal{S}$ contains a maximal element $M$ and by an easy modification of the proof showing that $B$ was linearly independent in part (b), we conclude that $M$ is a spanning set and therefore $M$ is a basis. (c) now follows from (a) by taking $E=V$ and follows from (b) by taking $E=\varnothing.$    
\end{proof}

Now, by Theorems \ref{basis} and \ref{cardinality}, we know bases exist and that their cardinality is unique. Therefore it is an invariant of the vector space and motivates the following definition. 
\begin{definition}
	Let $V$ be an $F$-vector space and $\mathcal{B}$ a basis. By the \textbf{F-dimension} of $V$ we mean \[ \dim_F V=|\mathcal{B}|\] Here it is important to distinguish the field of definition. 
\end{definition}	 
\begin{example}
	Let $\mathbb{F}_p$ denote the field with $p$ elements. It is a fun exercise to prove that for any natural number $n\in \N,$ there is a field extension $\F_{p^n}.$ Each of these fields is a vector space of dimension $n$ over $\F_p$ given by adjoining a root of an irreducible polynomial of degree $n$ and thus is isomorphic to $\F_p^n.$ We can see this isomorphism explicitly after we develop the theory of rings in the next section.  
\end{example}  

\begin{example}
	We now give an interesting example of an infinite dimensional vector space. Consider $\R$ defined over $\Q.$ At first glance, this looks non-sensical as an infinite dimensional vector space as $\Q$ is dense in $\R.$ However, suppose $\R\cong \Q^n$ for some $n\in \N.$ Then we can pick a basis $\{x_1,...,x_n\}$ of $\R$ over $\Q.$ By Cantor's diagonalization argument, we know that $|\R|>|\Q|.$ In fact, $\Q$ is countably infinite and $\R$ is uncountably infinite. Using the basis we have picked, the claim $\R\cong \Q^n$ would imply that $\R$ is countably infinite as the finite product of countably infinite sets is necessarily countably infinite. This is a contradiction and thus \[ \dim_\Q \R\neq n\;\;\;\;\;\; \forall n\in \N\]
Another way to think about this is to look at all transcendental numbers, $t,$ over $\Q$ (numbers such as $\pi,e,\ln(2)$ etc.) If we look at $\Span[\Q]{t}\cong \Q\subsetneq \R$ we get disjoint one dimensional subspaces for each unique transcendental number. 
\begin{lemma} 
	There are only countably many algebraic numbers. 
\end{lemma}
\begin{proof}
	A real number, $r,$ is algebraic if there exists $f\in \Q[x]$ such that $f(r)=0.$ Therefore, we need a bound on the cardinality of $\Q[x]$ as this gives an upper bound on the cardinality of the algebraic numbers. Notice that $\{x^i\}_{i\in \N}$ is a basis for $\Q[x]$ as a $\Q$ vector space. This is a countable basis and therefore $\Q[x]$ is a countably infinite dimensional vector space. Hence, $\Q[x]$ is countably infinite as a set and therefore the cardinality of the algebraic numbers is at most countably infinite.   
\end{proof}
\begin{corollary}
	There are uncountably many transcendental numbers. 
\end{corollary}

\noindent Using the construction from above, we now know that sitting inside $\R$ are uncountably many copies of $\Q$, each having trivial intersection, and thus $\R$ is an infinite dimensional vector space over $\Q.$ 	   
\end{example}

\subsection{Linear Transformations and Quotients}
Now that we have the notions of basis and dimension, we can introduce the idea of linear maps between vector spaces. These play a massive role in modern mathematics as well as many applied areas. The reason, as will be shown shortly, is that linear maps are in some sense the ``easiest" functions to understand. Further, there is a natural association of a matrix to any linear map, regardless of dimension. This will give us a clear method to tackle problems like. Example 2.3.19(b) and after Definition 2.3.1. First, we introduce the notion of quotient for vector spaces. This treatment will mirror the treatment for groups above, but will elucidate the differences that vector spaces bring. 

Similar to the case of sets, we want to impose a notion of equivalence on a generic vector space $V.$ We do this by identifying an entire subspace, not just a subset. 
\begin{definition}
	Let $W\subseteq V$ be a subspace. We define the \textbf{quotient space} $V/W=V/\sim$ where $v\sim v'$ if $v-v'\in W.$ It is easy to check that this is an equivalence relation. As $V$ is an abelian group, we have that $V/W$ is also an abelian group under the operation $[x]+[y]=[x+y]$. We define scalar multiplication as $k[v]:=[kv].$ This turns $V/W$ into a vector space.     
\end{definition}	
We shall see some examples of these after Theorem 2.75 below. Before this, we give the first definition of linear maps and some first properties. 

\begin{definition}
	Let $K$ be a field and $V,W$ be two $K$-vector spaces. We say a function $f:V\to W$ is a \textbf{linear transformation} if for all $v,v'\in V$ and $k,k'\in K,$ \[ f(kv+k'v')=f(kv)+f(k'v')=kf(v)+k'f(v')\in W\] 
	The set of all $v\in V$ such that $f(v)=0$ is called the $\textbf{kernel}$ and is denoted $\ker f.$ Similarly, the image, denoted $\im f$ is defined as the set of $w\in W$ such that $w=f(v)$ for some $v.$ We retain the same definitions of isomorphism as for groups above. 
\end{definition}	 

\begin{lemma}
	The canonical map $q:V\to V/W$ is linear and surjective.
\end{lemma}
\begin{proof}
	By definition, $q(kx+y)=[kx+y]=[kx]+[y]=k[x]+[y]=kq(x)+q(y).$ Therefore, $q$ is a linear transformation. Now let $\mathcal{C}$ be a basis for $V/W.$. Let $C'$ be a choice of representatives for the elements of $\mathcal{C}$ in $V.$ Then $\mathcal{C}=q(C')$ and extending by linearity, we get that $V/W=\text{Span }\mathcal{C}=\text{Span }q(C').$ Hence, $q$ is surjective.
\end{proof}

\begin{dtheorem}\label{Basic_Linear}
	Let $f:V\to W$ be a linear transformation. Then: 
	\begin{enumerate}
		\item $\ker f$ and $\im f$ are vector subspaces of $V$ and $W$ respectively. We then call $\dim_K \ker f$ the \textbf{nullity} and $\dim_K \im f$ the \textbf{rank}. 
		\item $f$ is injective if and only if $\ker f=0.$
		\item(First Isomorphism Theorem) $V/\ker f\cong \im f.$
		\item If $\dim_KV=\dim_KW<\infty$ then the following are equivalent: \begin{enumerate}
			\item f is injective
			\item f is surjective
			\item f is an isomorphism
		\end{enumerate} 
	\end{enumerate} 
\end{dtheorem}
\begin{proof}
	(a), (b), and (c) follow from the fact that linear functions are additive group homomorphisms that also respect scalar multiplication. This implies that $\ker f$ and $\im f$ are additive abelian groups closed under scalars by the $K$-equivariance. What remains to be proven for (c) is that the following diagram of linear maps commutes \[ \begin{tikzcd}
		V \arrow[r,"f"] \arrow[d,swap, "q"] & \im f\\
		V/\ker f \arrow[ur, swap, "\hat{f}"] 
	\end{tikzcd}\]  Forgetting the $K$-equivariance momentarily, the diagram commutes on the level of abelian groups by the proof of Theorem $\ref{First Iso}.$ Therefore, we need to show that $K$-equivariance of $\hat{f}.$ If $k\in K,$ then \[ \hat{f}(k[v])=\hat{f}([kv])=f(kv)=kf(v)=k\hat{f}([v])\] 
	By the proof of Theorem \ref{First Iso}, we know that $\hat{f}$ is a bijective linear map and thus a. vector space isomorphism.  
	
	(d) If suffices to prove that $(i)\iff (ii)$ as $(iii)\implies (i),(ii)$ trivially and $(i) \implies  (ii)$ makes $f$ a bijective linear map, hence an isomorphism. \\\\($\Rightarrow$) If $f$ is injective, pick $B$ a basis for $V.$ Then $f(B)$ is linearly independent by linearity. Since $\dim W=\dim V,$ $f(B)$ is a basis for $W$ and $f$ is surjective.  \\\\ ($\Leftarrow$) If $f$ is surjective, again let $B$ be a basis for $V$ and $f(B)$ the corresponding basis of $W.$ Let $u\in \ker f.$ We need to show $u=0.$ As $B$ is a basis, let $u=k_1v_1+...+k_nv_n$ be the unique expansion of $u$ in the basis $B.$ By the linearity of $f,$ we know that $f(u)=k_1f(v_1)+...+k_nf(v_n)=0_W.$ However, $f(B)$ is a basis for $W$ and consequently $k_i=0$ for all $i.$ Thus $u=0.$ This completes the proof. 
\end{proof}
    
\begin{corollary}\label{Dimensions}
	If $V$ and $W$ are finite dimensional vector spaces such that $dim V=\dim W,$ then $V\cong W$
\end{corollary}
\begin{proof}
	Let $B$ be a basis for $V$ and $C$ a basis of $W$ let $f:V\to W$ be defined by \[ f(k_1b_1+...+k_nb_n)=k_1c_1+...+k_nc_n\]
	This is clearly injective and by Theorem \ref{Basic_Linear}(d), an isomorphism.  
\end{proof}

We will not provide a proof for the following theorem as it is more or less an exercise in Category theory which will be postponed until Chapter 3.
\begin{theorem}\label{UMP_VS}
	Let $\mathcal{B}$ be a basis for a vector space $V.$ Let $U$ be any other vector space. If $f:\mathcal{B}\to U$ is any function, then there exists a unique linear transformation $F:V\to U$ such that the following diagram commutes: \[ \begin{tikzcd}
		\mathcal{B} \arrow[r,"f"] \arrow[d,swap,  "\iota"] & U \\
		V \arrow[ur,swap, "F"] &
	\end{tikzcd}\]
\end{theorem}  

This is an example of a \textbf{universal mapping propery}. These types of theorems are abundant in algebra and will be seen to be parts of more general schema in Chapter 3. 

\begin{example}\label{Key_Definitions} We now give some examples of vector spaces that arise from the consideration of various linear maps. \\
(a)  \textbf{(Direct Sums and Direct Products)} Let $\{V_i\}_{i\in I}$ be a collection of vector spaces. We define two objects \begin{align*} \bigoplus_{i\in I} V_I=\{ (v_i)_{i\in I}: \text{all but finitely many } v_i=0\} && \prod_{i\in I} V_i=\{ (v_i)_{i\in I}\} \end{align*} the direct sum and direct product respectively of vector spaces. 
		These objects come with natural linear maps $\iota_j:V_j\to \bigoplus V_i$ and $\pi_j:\prod V_i \to V_j.$ For a finite indexing set, $\bigoplus V_i= \prod V_i$ and thus the symbols $\oplus$ and $\times$ will be used interchangeably. IN general however $\bigoplus V_i\hookrightarrow \prod V_i.$The key feature of $\oplus$ for finite indexing sets is that \[\dim\left(\bigoplus V_i\right) =\sum \dim V_i\] This follows from the fact that we can take individual bases in each coordinate space. As will be seen in the next chapter, $\oplus$  is a coproduct (or colimit) of vector spaces and $\times$ is a product (or limit) of vector spaces. \\\\
(b) \textbf{(Hom and Dual Spaces)} Let $\Hom_F(V,W)$ denote the set of all $F$-linear transformations $V\to W.$ This can be made into an $F$-vector space by defining addition and scalar multiplication point-wise. For the case of $W=F,$ we denote \[\Hom_F(V,F)=V^*\] the dual space to $V.$ If $\dim V<\infty,$ then there exist isomorphisms (non-canonically) of $V\cong V^*$ and (canonically) of $V\cong V^{**}$ the double dual. We call elements of $V^*$ \textbf{linear functionals } on $V.$ Let $T:V\to W$ be a linear transformation, and define $T^*=T^t$ the \textbf{transpose map} as $T^t(g)=g\comp T:W^*\to V^*.$ Below, we will show the motivation behind such a naming and its relation matrices. It can be seen that in general \[\dim\Hom(V,W)=\dim V\cdot \dim W.\]        
\end{example}
 To finish this subsection, we shall go back to the start and relate matrices to linear maps on vector spaces. 
\begin{theorem}\label{Matrix_of_Linear}
	Let $V$ and $W$ be finite dimensional vector spaces over the field $K$ and $f:V\to W$ a linear transformation. Then, there exists a matrix $A$ such that with respect to the bases on $V$ and $W, f(v)=Av.$ where the right side is taken to be matrix multiplication of the $\dim W\times \dim V$ matrix by the $\dim V\times 1$ vector in $V.$ Further more, given any matrix $M$ of size $\dim W\times \dim V,$ this corresponds to a linear map $g:V\to W.$ Moreover, this correspondence is bijective.   
\end{theorem}
\begin{proof}
	Put $\mathcal{M}(V,W)$ to be all matrices in the bases $\mathcal{B},\mathcal{C}$ over $V,W$ respectively. Notice that this is a vector space over $K$ and that the matrices $E_{ij}$, whose only non-zero entry is a $1$ in position $ (i,j)$, forms a basis. Define $f_{ij}:V\to W$ to be the unique linear extension (Theorem \ref{UMP_VS}) of the map on the bases which sends $b_i\mapsto c_j$. This gives an inclusion of $\mathcal{E}$, the basis of $\mathcal{M}(V,W),$ into $\Hom(V,W)$ via the map \[ \varphi:\mathcal{E}\to \Hom(V,W)\;\;\;\;\;\;\;\;\;\; \varphi(E_{ij})=f_{ij}\]  We claim that $\{f_{ij}\}$ is a linearly independent set. To see this, consider the unique extension $\hat{\varphi}: \mathcal{M}(V,W)\to \Hom(V,W)$ and the arbitrary sum \[ 0=\sum_{i,j} a_{ij}f_{ij}\]
	Evaluating this at one of the $b_i,$ we get that \[ 0=\sum_j a_{ij}w_j\] and by linear independence all $a_{ij}=0.$ Hence, $\{f_{ij}\}$ is linearly independent and as there are $\dim W \cdot \dim V$ many elements, we know that it is a basis for $\Hom(V,W)$ by Example \ref{Key_Definitions} and Corollary \ref{Dimensions}. Hence, \[ \hat{\varphi}:\mathcal{M}(V,W)\to \Hom(V,W)\] is a surjection and by Theorem \ref{Basic_Linear}, an isomorphism. This completes the proof. 
\end{proof}

What this tells us is that every matrix can be treated as a linear transformation and thus the transpose map $f^t:W^*\to V^*$ has a matrix representation as the transpose matrix. As we will see later, this correspondence between matrices and linear maps can be exploited to prove a variety of theorems. One of the main theorems will be on determinants, to be defined in section 2.5 which relates invertibility of a matrix (and of the corresponding linear map) to its determinant.

\section{Ring Theory}
We now enter the belly of the algebraic beast. Ring and module (section 2.5) theory generalizes both fields and vector spaces in a way which makes doing mathematics with them significantly more difficult. However, we are lucky in that for the main applications in Chapter 4 and 5, we only need sufficiently nice objects called local and/or noetherian rings. Modules over these rings are relatively controlled and thus are incredibly important for analyzing these objects. A majority of this section comes from \cite{Knapp2006},\cite{Rotman2015} and \cite{DummitFoote2004}. The material on commutative rings follows \cite{Matsamura1986} and \cite{AtiyahMacdonald1969}. Similar to the previous sections, we begin with some definitions: 

\begin{definition}
	Let $R$ be a set equipped with two associative binary operations ($+,\times$). We call $R$ a \textbf{ring} if the following hold: 
	\begin{enumerate}
		\item $R$ is an abelian group under $+.$ 
		\item $R$ is closed under $\times.$ That is for all $a,b\in R,$ $a\times b=ab\in R.$ 
		\item For all $a,b,c\in R,$ $a(b+c)=ab+ac$ and $(a+b)c=ac+bc.$
	\end{enumerate}
	If in addition there exists an element $1_R$ such that $x\times 1_R=x,$ for all $x\in R$ then we say that $R$ is \textbf{unital}. We call $R$ \textbf{commutative} if $a\times b=b\times a$ for all $a,b\in R.$ A \textbf{ring homomorphism} is a function $f:R\to S$ such that for all $s,t\in R,$ \begin{align*} f(s+t)=f(s)+f(t) && f(st)=f(s)f(t) \end{align*}
	If $R$ and $S$ are unital, then we also impose the condition that $f(1_R)=1_S.$ The set of units (multiplicative invertible elements) is denoted $R^\times.$
\end{definition}
\begin{remark}
	\color{red} It is common practice to assume that all rings are unital. This makes one's job much easier when considering homomorphisms and related objects. We shall follow this convention for the remainder of the text and note the instances when an object does not contain a unit. \color{black}
\end{remark}
\begin{lemma}\label{Units}
	Let $R$ be a ring. Then the set $R^\times$ is a group under multiplication. 
\end{lemma}

\begin{example}
	Rings play a key role in the later parts of this text and therefore it is imperative that we have a wealth of examples to draw from. 
	\begin{enumerate}
		\item Let $F$ be a field, then $F$ is a commutative, (unital) ring, where every non-zero element has an inverse. Therefore $F^\times=F-0.$ 
		\item All of the sets $\Z,\Q,\R,\C$ are rings with additional and multiplication defined as usual. In fact, $\Z$ is the prototypical example of a commutative ring which is not a field. For $\Q,\R,\C$ their group of units is the set of non-zero elements. For $\Z,$ its easy to see that $\Z^\times=\{\pm 1\}\cong \Z/2\Z.$ 
		\item Let $V$ be a finite dimensional $K$-vector spaces of $\dim V=n,$ then \[M_{n}(K):=\mathcal{M}(V,V)\] the set of $n\times n$ matrices is a ring with identity element $I_n=\text{diag}(1,...,1)$ the matrix with 1s along the main diagonal. Further, the group of units is special and gets its own symbol \[ GL_n(K):=M_n(K)^\times\] 
		\item Consider the vector subspace of $M_2(\C)$ with basis \begin{align*}
			\textbf{1}=\Mat{1 & 0\\ 0& 1} &&  \textbf{i}=\Mat{i & 0\\ 0& -i} && \textbf{j}=\Mat{0 & 1\\ -1& 0} && \textbf{k}=\Mat{0 & i\\ i& 0}   
		\end{align*}
		We denote this space as $\mathbb{H}=\text{Span}_\R\{\textbf{1},\textbf{i},\textbf{j},\textbf{k}\}.$ These are the hamiltonian quaternions and are an example of a \textbf{division ring}, one where every non-zero element has a multiplicative inverse.  
		\item Let $V$ be a vector space over a field of $\text{char } F\neq 2$ and equip $V$ with a \text{bilinear} map \[ [-,-]:V\ds V\to V\] which satisfies the following conditions for all $x,y,z\in V$ \begin{enumerate}
			\item $[x,y]=-[y,x] $ (Anti-commutativity)
			\item $[x,[y,z]]+[z,[x,y]]+[y,[z,x]]=0$ (Jacobi Identity)
		\end{enumerate}	
		Then $V$ is a $\textbf{Lie Algebra}.$ Every lie algebra is a non-commutative, non-associative, non-unital ring. These will play a part in the theory developed in Chapter 3. 
		\item Consider $R[x]$ the polynomial ring with coefficients in a ring $R.$. This is a ring as discussed in Example \ref{Polynomials}. The group of units is necessarily $R^\times$ as these are the only elements with formal inverses. 
	\end{enumerate}
\end{example}

\begin{proposition}\label{Ring_Z_module}
	Let $R$ be a ring. Then there exists a unique ring homomorphism $\varphi:\Z\to R$ 
\end{proposition}
\begin{proof}
	Fix $r\in R$ and define $\varphi_r(n)=nr=r+r+...+r$ ($n$-times) in $R.$ For all $r\neq 1_R,$ this is a homomorphism of non-unital rings. Notice that each of these is determined completely and uniquely by where it sends $1.$ Hence, put $\varphi=\varphi_1:\Z\to R.$ This sends $1\mapsto 1_R$ and therefore is the desired homomorphism.  
\end{proof}
We define the kernel of a ring homomorphism in direct analog to vector spaces and group homomorphisms. The following theorem is the ring version of Theorem \ref{Basic_Linear}.  
We leave the proof as an exercise to the reader as it follows with slight modification from the proof of Theorem \ref{Basic_Linear}.  

\begin{theorem}
	Let $S$ and $R$ be rings and $\varphi:R\to S$ a ring homomorphism. Then: 
	\begin{enumerate}
		\item $\ker \varphi$ is a subring with no unit and $\varphi(R)$ is a ring. 
		\item $\varphi$ is injective if and only if $\ker \varphi=0.$
	\end{enumerate}	
\end{theorem}

Let us focus on $\ker \varphi$ for a moment. It is a special example of an \textbf{Ideal} of $R.$ 
\begin{definition}\label{Prime_maximal}
	Let $R$ be a ring. A \textbf{left ideal} $I\leq (R,+)$ is a subgroup of the additive group of $R$ such that \[RI=\{ri:r\in R, i\in I\}\subseteq I\]
	Similarly a right ideal is a subset $J\subseteq R$ such that $JR\subseteq J.$ We call an ideal, $\lie{m},$ \textbf{maximal} if there are no other ideals of $R$ which properly contain $\lie{m}.$ We call an ideal, $\lie{p},$ $\textbf{prime}\footnote{The word prime here comes from the notion of prime integer. Normally a number is prime if its only factors are 1 and itself. An equivalent condition is that $p\in \Z$ is prime if and only if when $p$ divides the product $ab$ for some $a,b\in \Z,$ then either $p$ divides $a$ or $p$ divides $b.$}$ if $ab\in \lie{p}$ then either $a\in \lie{p}$ or $b\in \lie{p}.$   
\end{definition}
\begin{remark}\label{Commutative_Ideal}
	Notice that in a commutative ring $R,$ every left ideal is also a right ideal. An ideal which is a left and right ideal, is called $\textbf{two-sided}.$ Further, over a commutative ring we can think of ideals in the same way we thought about vector spaces. The main difference however is that we cannot normally pick a generating set for $I$ as non-trivial ideals exist for rings which do not exist for fields. 
\end{remark}

\begin{example}\label{Principal_ideal_example}Lets consider some ideals in the rings given above. 
	\begin{enumerate}
		\item Every field has no proper non-zero ideals. This follows from the fact that an ideal $I$ is necessarily a vector space over $F$ and therefore has a basis. If $I$ is non-zero the basis has to be the element $1_F.$ 
		\item For any ring $R,$ let $S$ be a subset. We can form $\ip{S}$ the \textbf{ideal generated by S} by taking \[ \ip{S}=\bigcap_{S \subset I \subseteq R} I  \]
		where $I$ is an ideal. We leave it to the reader to check that the intersection of ideals is necessarily an ideal. We call an ideal $\textbf{principal}$ if $I=\ip{r}$ for some element $r\in R.$  
		\item $\Z$ is an example of a \textbf{Principal ideal domain}. This means that every ideal is principal and thus $m\Z$ are all possible ideals. Because of this, any ideal which contains $1\in \Z$ must be $\Z$ itself. In fact in any ring, an ideal which contains $1_R$ must be the entire ring.
	\end{enumerate}
\end{example}
\begin{lemma}
	Let $\varphi:B\to A$ be a ring homomorphism. Then if $\lie{p}$ is a prime ideal in $A,$ $\varphi^{-1}(\lie{p})$ is a prime ideal in $B.$ This does not hold true for maximal ideals. 
\end{lemma}
\begin{proof}
	Let $ab\in \varphi^{-1}(\lie{p}).$ Then $\varphi(ab)=\varphi(a)\varphi(b)\in \lie{p}$ and thus one of $\varphi(a)$ or $\varphi(b)$ is an element of $\lie{p}.$ Hence, either $a$ or $b$ is an element of $\varphi^{-1}(\lie{p})$ and it is a prime ideal of $B.$ 
	
	For a counterexample in the maximal case, consider the canonical inclusion $\iota: \Z\hookrightarrow \Q.$ As $\Q$ is a field, its only ideal is $0$ but $ \iota^{-1}(0)=0,$ which is not maximal in $\Z.$  
\end{proof}

Similar to vector spaces and groups, we can take quotients of rings by two-sided ideals. To see why these are the natural choice for quotients, consider that we want to have the quotient $R/I$ become a ring again. To do this, let $I$ be an arbitrary subgroup of $(R,+).$ A coset of $I$ in $R$ will be denoted $r+I$ for $r\in R.$ If we define addition and multiplication in the obvious way \[ (r+I)+(s+I)=(r+s)+I\]
\[ (r+I)(s+I)=(rs)+I\]
As $(R,+)$ is an abelian group, we know that the as groups $R/I$ is well defined under $+.$ We need to make sure it is well defined under $\times.$ That is for any $r,s\in R$ and $\alpha,\beta\in I$ we should have that \[ (r+\alpha)(s+\beta)+I=rs+I\]
If we set $r=s=0$ we see that $I$ must be closed under multiplication. Thus, $I$ is a subring (without unit) of $R$. By setting $s=0$ we see that $r\beta\in I$ for all $r\in R$ and $\beta\in I.$ Therefore $I$ is closed under multiplication on the left by $R.$ Setting $r=0$ and letting $s$ vary we also see that $I$ must be closed under multiplication by $R$ on the right. Conversely, if $I$ is closed under left and right multiplication by $R$ then the relation above must be satisfied. Hence, being a two-sided ideal is a necessary and sufficient condition for $R/I$ to be a ring. What we have just shown is the following technical lemma: 

\begin{lemma}\label{Ideals}
	Let $I\subseteq (R,+)$ be a subgroup. A necessary and sufficient condition for $R/I$ to have the structure of a ring is that $I$ is a two-sided ideal of $R.$
\end{lemma}

\begin{proposition}[First Isomorphism Theorem for Rings]\label{First_Iso_Ring}
	Let $\varphi:R\to S$ be a ring homomorphism. Then:
	\begin{enumerate} \item $\ker \varphi$ is a two-sided ideal of $R$ and  $R/\ker \varphi \cong \varphi(R).$
	\item If $I$ is any two-sided ideal of $R$, the map \[ \pi:R\to R/I\;\;\;\;\;\;\;\;\;\;\;\;\;\; r\mapsto r+I\]
	is a surjective ring homomorphism with kernel $I.$ Hence, every two-sided ideal of $R$ can be realized as the kernel of some homomorphism.  
	\end{enumerate}
\end{proposition}
\begin{proof}
	(a) The majority of this proof mirrors that of Theorem \ref{First Iso}. What remains to be proven is that the map $\hat{\varphi}(r+I)=\varphi(r)$ is a bijection between $R/I$ and $\varphi(R).$ This follow immediately from the definitions of a ring homomorphism. \\
	
	(b) We know that $R/I$ is a ring from the discussion before the statement of the proposition. In particular $R$ and $R/I$ are abelian groups and therefore $\pi:R\to R/I$ is a group homomorphism. To see it is a ring homomorphism, consider two elements $r,s\in R.$ Then \[ \pi(rs)=rs+I=(r+I)(s+I)=\pi(r)\pi(s)\]
	Further, $\pi(1)=1+I=1_{R/I}.$ Hence, $\pi$ is a ring homomorphism. 
\end{proof}
The final proof of this subsection is possibly the most useful and interesting isomorphism theorem. 

\begin{theorem}[Fourth Isomorphism Theorem(s)]\label{Fourth_Iso_Rings} \text{}
	\begin{enumerate}
		\item Let $G$ be a group and $N\trianglelefteq G.$ Then the subgroups of $G/N$ are in one-to-one correspondence with the subgroups of $G$ containing $N.$ 
		\item Let $R$ be a ring and $I$ a two-sided ideal. Then the subrings of $R/I$ are in one-to-one correspondence with the subrings of $R$ containing $I.$ 
	\end{enumerate}
\end{theorem}
\begin{proof}
	We shall prove (b) and (a) will follow immediately by the same argument. Let $\pi:R\to R/I$ be the canonical projection map and $S$ be a subring of $R$ containing $I.$ Then $I$ is a two-sided ideal in $S$ and thus $S/I$ is a ring contained in $R/I.$ Now assume that $P\subseteq R/I$ is a subring. Then $\pi^{-1}(P)=\{r\in R: r+I\in P\}.$ We first check that this is a ring. If $a,b,c\in \pi^{-1}(P)$ then \[\pi(ab+c)=(ab+c)+I=(ab+I)+(c+I)=(a+I)(b+I)+(c+I)\] This is an element of $P$ by the definition of a ring. Thus, $\pi^{-1}(P)$ is a ring and $\pi^{-1}(P)/I=P$ so Lemma \ref{Ideals} tells us that $I$ is an ideal in $\pi^{-1}(P).$ This completes the proof.     
\end{proof}

An immediate corollary of this theorem is that 
\begin{corollary}
	Let $R$ be a commutative ring and $\lie{m}$ an ideal. Then $\lie{m}$ is maximal if and only if $R/\lie{m}$ is a field. 
\end{corollary}
This single corollary will play a large role in the formulation of certain categorical and algebraic objects later in the text. 
\begin{proposition}\label{Maximal_ideal_exist}
	Every commutative ring has a maximal ideal. 
\end{proposition}
\begin{proof}
	Let $\mathcal{P}$ be the set of proper ideals of $R$ ordered by inclusion. Every chain $\mathcal{C}$ in $\mathcal{P}$ has an upper bound, namely $\bigcup_{C\in \mathcal{C}} C.$ This is easily seen to be an ideal. Applying Zorn's lemma, $\mathcal{P}$ has a maximal element, $\lie{m}.$ By definition, $\lie{m}$ is a maximal ideal.  
\end{proof}

\subsection{Commutative Algebra}
This will take up the remainder of this section on rings. The theory developed later in this text relies on the notion of local and noetherian rings. These play a huge role in algebraic geometry and the theory of smooth manifolds. Specifically they form the basis for which sheaves (see Chapter 3) can be built. It is known that understanding sheaves on a space is equivalent to understanding the space itself. Therefore,  to get a better grasp on the geometry later, we will need to understand sheaves. To do so, we start with commutative algebra and build our way up.

\begin{remark} For the remainder of this chapter, all rings are assumed to be commutative and unital. Ideals are two-sided (by Remark \ref{Commutative_Ideal}) and thus $R/I$ can be given the structure of a ring always. 
\end{remark} 
We start with Principal Ideal Domains (defined in Example \ref{Principal_ideal_example}) and their generalization: Unique Factorization Domains and integral domains. 

\begin{definition}
	Let $R$ be a ring. A $\textbf{zero-divisor}$ in $R$ is an element $a\in R$ such that either $ab=0$ or $ba=0$ for some $b\in R.$ A ring with no non-zero zero-divisors is called an \textbf{integral domain}. This amounts to being able to cancel the element from expressions such as \[ ab=cb\implies a=c\] 
	In an integral domain, an element $r$ which is nonzero and not a unit is called \textbf{irreducible} if whenever $r=ab$ then one of $a$ or $b$ is a unit. Otherwise $r$ is \textit{reducible}. An element is called $\textbf{prime}$ if the ideal $\ip{p}$ is a prime ideal in the sense of Definition \ref{Prime_maximal}. An integral domain is a principal ideal domain if every ideal is principal.    
\end{definition}      
Prime elements and irreducible elements are closely related. In most of the examples we have presented, they are in fact the same! The following lemma asserts this  

\begin{lemma}
	In an integral domain, $R$, every prime element is irreducible. If we assume further that $R$ is a principal ideal domain (P.I.D.) , then an element is prime if and only if it is irreducible. 
\end{lemma}
\begin{proof}
	Assume $p=ab.$ Then by definition $p$ divides $a$ or $p$ divides $b.$ Assume without a loss of generality that $p$ divides $a.$ Then $a=px$ for some $x\in R.$ Then \[ p=pxb\]
	As $R$ is an integral domain it has no zero-divisors and thus $xb=1$ and $b$ is a unit. Hence, $p$ is irreducible. 
	
	Now assume further that $R$ is a $P.I.D.$ we need to show that irreducible elements are prime. Let $r$ be an irreducible element. Suppose that $r\in M$ an ideal of $R.$ By the hypothesis, $M=\ip{m}$ and $r\in \ip{m}\implies r=mx$ for some $x\in R.$ By irreducibility, either $m$ or $x$ is a unit.  Thus either $\ip{m}=\ip{r}$ or $\ip{1}.$ Hence, $\ip{r}$ is a maximal ideal and all maximal. ideals are prime.      
\end{proof}

\begin{definition}
	An integral domain $R$ is called a $\textbf{Unique Factorization Domain}$ if every non-unit element has a unique (up to units) factorization into irreducible elements. 
\end{definition}
Unique factorization is a topic that should be familiar to everyone. It is a standard result in high-school level mathematics that every integer can be written as a product of prime numbers. As $\Z$ is a P.I.D. this agrees with the definition above. The following result puts all of these rings into context with what we have done prior. We shall not prove it as it does not add to the theory. 

\begin{theorem}\label{Inclusion_Rings}
	The following inclusion of integral domains holds: \[ \text{Fields}\subsetneq \text{Principal Ideal Domains} \subsetneq \text{Unique Factorization Domains} \]
\end{theorem} 

\begin{example}
	Let $F$ be a field and consider $F[x]$ the polynomial ring. It is a well known fact that $F[x]$ is a P.I.D. In fact, we can relax the restriction that $F$ is a field and consider $R[x]$ the polynomial ring with coefficients in a ring $R.$ There is a nice result \cite[Theorem 7, Chapter 9.3]{DummitFoote2004} that says that if $R$ is a U.F.D. then so is $R[x].$ Shortly, we shall see that there is another theorem of this variety, Hilbert's Basis Theorem, which asserts that if a ring is Noetherian then so is $R[x].$  
\end{example}

\begin{definition}
	Let $R$ be a ring and $I$ an ideal. We say that $I$ is \textbf{finitely generated} if there exists a finite set $S$ such that $I=\ip{S}.$ We call the ring $R$ \textbf{Noetherian} if every ideal is finitely generated.   
\end{definition}

\begin{theorem}[Hilbert's Basis Theorem] \label{Hilbert_basis_theorem}
	Let $R$ be a noetherian ring. Then $R[x_1,..,x_n]$ is noetherian for any $n\in \N.$ 
\end{theorem}
\noindent The proof of this theorem is moderately technical and will be omitted. 

The main reason we consider Noetherian rings is that we have the following proposition which characterizes chains of ideals in $R.$

\begin{proposition}
	A ring $R$ is Noetherian if and only if every ascending chain $I_1\subseteq I_2 \subseteq ... $ of ideals stabilizes: that is there exists $n^*\in \N$ such that for all $n\geq n^*$. we have that $I_n^*=I_n.$  
\end{proposition}
\begin{proof}
	$(\Rightarrow)$ Let $I_1\subseteq I_2 \subseteq...$ be an ascending chain of ideals. Then by assumption all of these are finitely. generated. Consider $I=\bigcup_{n\in \N} I_n.$ This is also an ideal as any two elements can be taken to be in some higher indexed ideal and thus addition is well defined. Multiplication by $R$ also follows immediately. Thus, $I$ is an ideal and finitely generated by assumption. Let $\{a_1,...,a_n\}$ be a generating set. Then each of these is contained in some $I_j.$ Take $j^*=\max\{j:a_j\in I_j\}$ then all of these elements lie in $I_j^*$ and the chain stabilizes. 
	
	$(\Leftarrow)$ Let $I$ be an ideal and consider the set of all finitely generated ideals contained in $I.$ This has a maximal element $\lie{m}$ by Zorn's Lemma (Lemma \ref{Zorn's Lemma}). We assert that $\lie{m}=I.$ If not, then there would be an ascending chain of ideals which did not stabilize, namely take the generating set $N$ of $I.$ Then we can pick $x_i\in N$ such that $Rx_1\subsetneq Rx_1+Rx_2\subseteq ...  .$ This is an infinite chain of ideals and it does not stabilize. Hence, $I=\lie{m}$ and $I$ is finitely generated. This completes the proof.      
\end{proof}

We will see early on in the next chapter that we can define the notion of Noetherian for topological spaces. We use this notion to relate noetherian rings to a certain topology on $\Spec{R}$ called the Zariski topology. 

Another key class of rings are rings that have a single maximal ideal. 

\begin{definition}
	Let $R$ be a ring. We say that $R$ is a $\textbf{local ring}$ if there exists a single maximal ideal in $R.$ It is customary to denote local rings as $(R,\lie{m})$ or $(R,\lie{m},k)$ where $k=R/\lie{m}$ is called the residue field of $R.$ It is easy to show that $\lie{m}$ is precisely the set of all non-units in $R.$   
 \end{definition}

We shall end this section with a discussion of localization and local rings. This will be related to some geometry in the next chapter. We shall push off giving examples of local rings until the next chapter when they arise quite naturally in the theory of manifolds and schemes.  

\begin{dtheorem}
	A subset $S\subseteq R$ is called \textbf{multiplicative} if $x,y\in S\implies xy\in S.$ We define the \textbf{localization} with respect to the multiplicative set $S$ as the set of symbols \[ S^{-1}R=\left\{ \frac{r}{s} :r\in R,s\in S\right\}/\sim\]
	where $\frac{r}{s}\sim \frac{a}{b}$ if there exists $t\in R$ such that $t(rb-sa)=0.$ 
	If $S=R-\lie{p}$ for some prime ideal $\lie{p}$ then $S^{-1}R=R_\lie{p}$ is a local ring.    
\end{dtheorem}
\begin{proof}
	Using the standard definitions of addition and multiplication of fractions, we see that $S^{-1}R$ is indeed a ring. We now need to show that $R_{\lie{p}}$ is a local ring. We claim that $\lie{p}R_{\lie{p}}$ is a maximal ideal in $R_\lie{p}.$ It is easily shown to be an ideal, and thus it suffices to show maximality. Suppose not, then $\exists I\subseteq R$ an ideal such that $\lie{m}\subsetneq I$ by Proposition $\ref{Maximal_ideal_exist}.$ As $\lie{p}R_\lie{P}$ consists of all non-unit elements, it follows that $I$ must contain a unit. Therefore $I$ contains $1$ and must be $R$ $R_\lie{p}$ itself. Hence, $\lie{p}R_\lie{p}$ is a maximal ideal. Uniqueness follows from the same reason: any other ideal must contain a unit and therefore is the entire ring. Hence, $R_\lie{p}$ is a local ring.    
\end{proof}
This theorem gives the motivation for calling the operation Localization. As will be seen later, prime ideals will be the most important ideals in a ring as they precisely give a bijection between certain bits of geometry and algebra. For now, we shall move on to module theory. 

\section{Modules and Multilinear Algebra}
The final topic of this section is Module theory. This is a generalization of vector spaces over a field as we now allow for the ground space to be a ring. It is far more common to come across modules over rings than vector spaces. For this reason, there is an entire theory of modules and their generalizations to categories which is used extensively in the next chapter. This section draws from $\cite{DummitFoote2004}, \cite{Matsamura1986}, \cite{Rotman2015}, \cite{Lang2002},$ and $\cite{Knapp2006}.$ 
\begin{definition}
	Let $R$ be a ring. An abelian group $M$ is called an \textbf{$R$-module} if there exists an action map $R\times M\to M$ which is associative. We also impose that $1m=m$ for all $m\in M.$ A module is said to be \textit{finitely generated} if there exists a finite set $S$ such that $M=\text{Span}_RS.$ Here we adopt the same notion of span as for vector spaces. A \textbf{module homomorphism} is a linear map which is $R$-equivariant. A module is an \textbf{$R$-algebra} if $M$ also has the structure of a ring. That is, a map $\mu: M\times M\to M$ which is a multiplication and satisfies the axioms of multiplication in a ring. 
\end{definition}

\noindent Let us look at some examples of modules and submodules. 
\begin{example}
	\begin{enumerate}
		\item Let $M=R.$ Then $R$ caries the structure of an $R$-module trivially. Further, every ideal $I$ can be considered as an $R$-submodule. Moreover, we can define $R^n=R\ds R\ds .... \ds R$ as a module by multiplication in each coordinate.    
		\item Let $\varphi:A\to B$ be a ring homomorphism. Then $B$ can be given the structure of an $A$-module by defining $a\cdot b=\varphi(a)b.$ The properties of a ring homomorphism guarantee that this is indeed an action and satisfies the axioms of a module. 
		\item Let $F$ be a field and $V$ a vector space over $F.$ Then $V$ is an $F$ module. In particular, all $F$-modules are vector spaces. This is not true for modules. 
		\item Let $R=\Z$ and let $M=\Z\ds \Z/m\Z.$ Then $M$ is an $R$-module by the multiplication defined by \[ n\cdot (a,[b])=(na,[nb])\] In fact, \textit{any} abelian group $G$ has a natural structure of a $\Z$-module. This comes from the identification of $n\cdot g=g^n$ or $ng$ depending on whether one uses multiplicative or additive notation. 
		\item The polynomial ring $R[x_1,...,x_n]$ is an $R$-module in the obvious way $r\cdot f=rf.$    
	\end{enumerate}
\end{example}
\noindent The main difference between vector spaces and modules is that modules do not always have bases (in fact they rarely have them).  For example, over most rings, the existence of quotient rings implies that modules can actually be very from from just $R^n.$ 

Let us consider some operations on modules. 
\begin{definition}
	Let $M,N$ be $R$-modules. We define \[ M\ds N:=\{ (m,n): m\in M, n\in N \} \]
	to be the \textbf{external direct sum} of modules. For some collection of modules $\{M_i\}$ we can take their direct sum. If the indexing set is infinite we define $\bigoplus_I M_i$ as the tuples with finitely many non-zero elements. The direct sum comes equipped with natural morphisms $M,N\hookrightarrow M\ds N.$ We can also build the  \textbf{internal direct sum} for two submodules of a larger module $P$. In this case we denote the internal direct sum as \[ M+N:=\{m+n:m\in M,n\in N\}\]
	It is an easy exercise to show that these are isomorphic if $M\cap N=0.$       
\end{definition}
The definition here is exactly the same as the one given in Example $\ref{Key_Definitions}.$ 

\subsection{Quotients, Morphisms, and Free Modules} 
For $M$ an $R$-module and $N\subseteq M,$ we can define the quotient module $M/N$ in the same way we defined $V/W$ for vector spaces $W\subseteq V.$ Furthermore, we can realize every submodule as the kernel of some module homomorphism via the short exact sequence \[ 0\to N\to M\to M/N\to 0\]
In fact, short exact sequences play a key role in the theory of modules. For example: 

\begin{dproposition}
	We say a short exact sequence $0\to M\to N\to P\to 0$ \textbf{splits} (on the left) if there exists a morphism $N\to M$ which when composed with the inclusion is the identity. It splits on the right if there exists a morphism $P\to N$ which composes with the projection to be the identity. If a short exact sequence splits on the right, then $N\cong M\ds P.$ This gives a characterization of short split short exact sequences.  
\end{dproposition}
 \begin{proof}
 	Let $i:M\to N$ and $p:N\to P$ be the arrows in the above exact sequence. As we assume the sequences is split, we know $\exists j:P\to N$ such that $p\comp j=1_P.$ We will show that $N=\im i\ds \im j$. Let $n\in N,$ then $n-(j\comp p)(n)\in \ker p$ as $p(n-jpn)=p(n)-p(n)=0.$ By exactness, there exists some $m\in M$ such that $i(m)=n-jpn.$ It follows then that $N$ is the internal direct sum \[ N=\im i +\im j\]
 	We need to prove that $\im i\cap \im j=0.$ Let $a\in M$ and $b\in P$ such that $i(a)=x=j(b).$ Applying $p$ to both sides, we get that $j(b)=x=0.$ Hence, $N\cong M\ds P.$  
 \end{proof}
 
 \begin{corollary}
 If $0\to U\to V\to W\to 0$ is a short exact sequence of vector spaces, then the following conditions are equivalent: \begin{enumerate}
 	\item The sequence splits on the left. 
	\item The sequence splits on the right. 
	\item $V\cong U\ds W.$
 \end{enumerate}
  \end{corollary}
	
 \begin{corollary}[Rank-Nullity Theorem]
 	For a linear map $f:V\to W$ between vector spaces, we have that \[\text{rank}(f)+\text{nullity}(f)=\dim V.\] 
 \end{corollary}
 \begin{proof}
 	Set up the following exact sequence \[ 0\to \ker f \to V\to \im(f)\to 0\]
 \end{proof} 
\noindent Furthermore, combining this result with the First Isomorphism Theorem (\ref{Basic_Linear}), we see that every short exact sequence of finite dimensional vector spaces splits! We would like an analogous theorem for modules, but this cannot happen by the following example: 

\begin{example}
	Consider the exact sequence of $\Z$-modules \[ 0\to \Z\to \Q\to \Q/\Z \to 0\]
	This is exact because it takes the form of $0\to \ker f\to M\to f(M)\to 0$ for some module homomorphism. To see it does not split, we need to show that \[ \Q\not\cong \Z\ds \Q/\Z\]
	An easy calculation shows that $\Q/Z$ has elements of arbitrary finite order and therefore cannot be a direct summand of $\Q.$ 
\end{example}
\noindent In fact, it is a fairly standard exercise to show that an exact sequence of $R$-modules splits if and only if the middle term is a direct sum of the first and third terms. 

We now move on to some theorems for modules which we have seen before. 
 
 \begin{theorem}[First and Fourth Isomorphism Theorems]\label{Module_Isos} \text{}
 	\begin{enumerate}
		\item Let $\varphi:M\to N$ be an $R$-module homomorphism. Then $M/\ker \varphi \cong \varphi(M).$
		\item Let $N\subseteq M.$ Then the submodules of $M/N$ are in one-to-one correspondence with the submodules of $M$ containing $N.$ 
	\end{enumerate}
 \end{theorem}
 \begin{proof}
 	The proof of this follows immediately from the proof for the ring case: Theorems \ref{First_Iso_Ring} and \ref{Fourth_Iso_Rings}.
 \end{proof}
It should be no surprise at this point that this theorem is true. After all, we can regard modules as abelian groups and the result held true there. The only thing needing to be checked is $R$-equivariance, but this follows immediately from the definitions. 

\begin{definition}
	We call an $R$-module \textbf{free} if $M\cong R^n=R\ds R\ds ...\ds R$ for some $n\in \N.$ We say that $M$ is \textbf{finitely presented} if there exists a short exact sequence \[ 0\to K\to F\to M\to 0\] such that $K,F$ are free, finitely generated $R$-modules. For any set $S$ we can build the free $R$-module $R\ip{S}$ with basis $S.$ The following theorem gives a universal property for such modules.  
\end{definition}
\begin{theorem}[Universal Property of Free Modules]
	Let $S$ be a set and $M$ an $R-$module such that there exists a map $\varphi:S\to M.$ Then there exists a unique $R$-module homomorphism $\hat{\varphi}:R\ip{S}\to M$ such that the following diagram commutes: \[ \begin{tikzcd}
		S \arrow[r,"\varphi"] \arrow[d,swap,"i"] & M \\
		R\ip{S} \arrow[ur,swap,"\hat{\varphi}"] &
	\end{tikzcd}\]
\end{theorem}

\noindent The idea of finitely presented modules becomes important for the theory of sheaves which will be developed at the end of the next chapter. For now, we have a fundamental result on modules over a principal ideal domain. 

\begin{theorem}[Fundamental Theorem of Finitely Generated Modules over a P.I.D.] \label{Modules_over_PIDs} 
	Let $R$ be a P.I.D. and $M$ a finitely generated $R-$module. Then \[ M\cong R^k\ds R/\ip{r_1}\ds R/\ip{r_2}\ds...\ds R/\ip{r_n} \] for some non-unit elements $\{r_i\}. $
\end{theorem} 
\begin{proof}
	We will show existence of such a decomposition. Let $v_1,...,v_n$ be a generating set for $M$ and consider $R\ip{x_1,...,x_n}$ the free module on the same number of generators. It is obvious that there is a homomorphism $\varphi:R^n\to M$ which sends $x_i\mapsto v_i.$ This homomorphism is surjective by construction. Therefore, by Theorem \ref{Module_Isos}, we have that $M\cong R\ip{S}/\ker \varphi.$ We know then that $r_1x_1,...,r_mx_m$ is a generating set for $\ker \varphi$ and therefore $\ker \varphi=\bigoplus_{i\leq m} Rr_ix_i.$ Taking the quotient, we get that \[ M\cong \bigoplus_{i\leq n} Rx_i /\bigoplus_{i\leq m}  Rr_ix_i=\bigoplus \left( Rx_i./R r_ix_i \right) \ds R^{n-m} \]
	The terms of the direct sum become $R/\ip{r_i}$ under the natural identification. Hence, \[ M\cong R^{n-m}\ds R/\ip{r_1}\ds R/\ip{r_2}\ds...\ds R/\ip{r_n}\] 
\end{proof}
As $\Z$ is a principal ideal domain, this applies to abelian groups as well. As will be seen in the next chapter, this theorem becomes incredibly important in homology theory for finite cell complexes. 
	
\subsection{Multilinear algebra} 
The final part of this chapter will concern multilinear algebra. In this section we introduce an other operation on modules which gives another way of building new modules from old ones, the tensor product, and discuss related topics such as exterior powers of modules and vector spaces. 

For this section, let $R$ be a ring and $L, M, N$ $R$-modules. Further, let $S\supseteq R$ be a ring containing $R.$  
\begin{definition}
	We call a function $\theta:M\ds N\to L$ \textbf{bilinear} if it is linear in each argument. In general, we have multilinear functions $f:\bigoplus M_i\to L$ which are linear in each argument.  
\end{definition} 
It should not come as surprising that multilinear functions are a bit more difficult to deal with than linear functions. There is a way to convert between the two, but this involves a new module. 
\begin{definition}[Definition/Construction]
	Let $F(M\times N)$ denote the free $R$-module generated by $M\times N.$ Consider the submodule $G$ which is generated by the relations \begin{align*}
		(a,b)&\sim (a,b)\\
		(a+a',b)&\sim(a,b)+(a',b)\\
		(a,b+b')&\sim(a,b)+(a,b')\\
		(ar,b)&\sim (a,rb)\;\;\;\; r\in R
	\end{align*}
We define $M\tensor_RN=F(M\times N)/G.$ As $R$ is commutative, we have that $M\tensor_RN$ is an $R-$module with multiplication defined by the final relation of $G.$	
There is a canonical map \[ \tensor:V\times W\to V\tensor W\] which sends $(m,n)\mapsto m\tensor n.$ Elements of $M\tensor_RN$ are sums of the formal symbols $m\tensor n$ which is called a $\textbf{simple tensor}$ 
\end{definition}

\begin{theorem}[Universal Property]
	For every bilinear map $\varphi:M\times N\to L$ there exists a unique linear map $\hat{\varphi}:M\tensor N\to L$ such that $\varphi=\hat{\varphi}\comp \tensor.$
\end{theorem}
This theorem is sometimes given as the definition of the tensor product as it implies the tensor product is unique up to isomorphism. The nice part of this theorem is it gives a bijection \[ \text{Bil}(M,N,L)\cong \Hom_R(M\tensor N,L) \]
where $\text{BIl}(M,N,L)$ is the set of bilinear maps $M\times N\to L.$ Therefore, we can turn multilinear functions into linear ones by using the appropriate number of tensors. This is also true for all arbitrary collections of modules and multilinear maps.   
 
Let us look at some immediate applications of tensor products. 
\begin{proposition}
	Let $V,W$ be finite dimensional vector spaces (or $R$-modules) over $K$ a field. Then there is a non-canonical isomorphism \[  V^*\tensor W\cong \Hom_K(V,W)\]
	which sends $\Pi(\varphi\tensor w)=\varphi(-)w.$ 
\end{proposition}
\begin{lemma}
	For $V,W$ as above, $\dim V\tensor W=\dim V\cdot \dim W.$ 
\end{lemma}
\begin{proof}
	Let $B,C$ be bases for $V$ and $W$ respectively. Then we can pick as a basis for $V\tensor W,$ the set of simple tensors $b_i\tensor c_j$ for $b_i\in B$ and $c_j\in C.$ There are $|B|\times |C|$ of these. 
\end{proof}

\begin{proof}[Proof of Proposition 2.2.15]
	It suffices to prove that this map is injective as the dimension of the spaces are the same, we can apply Theorem \ref{Basic_Linear}(d). So, let $\phi\tensor w\in \ker \Pi.$ Then $\Pi(\phi\tensor w)(v)=\phi(v)w=0$ for all $v\in V.$ If $w=0$ we are done. So assume $w\neq 0.$ Then we know that $\phi(v)=0$ for all $v\in V.$ Therefore, by the uniqueness of $0\in V^*,$ we have that $\phi=0.$ This completes the proof.   
 \end{proof}
\begin{remark}
	Notice that the map itself is canonical, but the choice of basis is not. In general, for simple tensors $\varphi\tensor w\in V^*\tensor W$ the map is canonical, we need bases to extend this map to the entire tensor product.  
\end{remark}

Sometimes it is useful to consider the module $M$ as an $S$-module instead of an $R$-module.
 The following lemma gives a way to do such a thing. 
\begin{lemma}
	We can \textbf{extend scalars} from $R$ to $S$ by taking \[ M\mapsto M\tensor_R S\]
	Where the module structure on $S$ is given by the inclusion map. This is an $S$ module.   
\end{lemma}
\begin{proof}
	Let $s\in S.$ We need to define $s(m\tensor t)$ and then extend by linearity. Well, simply define $s(m\tensor t)=m\tensor st.$ As $R$ and $S$ are commutative, this is a valid action. 
\end{proof}
\begin{lemma}
	There is a canonical isomorphism of $R\tensor_R M\cong M$ for any $R$-module $M.$ 
\end{lemma}
\begin{proof}
	Let $\varphi:M\tensor_R R\to M$ be given by $\varphi(\sum m_i\tensor r_i)=\sum r_im_i$ We claim this is an isomorphism. Consider the map $m\mapsto m\tensor 1.$ This is an inverse for $\varphi$ on both the left and right. Hence, $\varphi$ is an isomorphism.  \end{proof}

Now let us investigate the module $M$ and its tensor powers: $M^{\tensor n}=\bigotimes^n M.$ These spaces parametrize, in some sense, the multilinear maps of $\prod^nM \to M.$ 
We can build an algebra out of these modules by taking a large direct sum. 
\begin{definition}
	Let $V$ be an $R$-module. The \textbf{Tensor Algebra} of $M$ is the $R$-algebra \[ T^\bullet(M)=\bigoplus_{n\in \N} M^{\tensor n} \] The algebra structure on $T^\bullet(M)\footnote{Some authors simply write $T(M)$ or $T^*(M)$ for the tensor algebra, we do not use these as it will become difficult to distinguish $T(M)$ and $TM$ in the next chapter. }$ is given by concatenation $v\in M^{\tensor n}$ and $w\in M^{\tensor m}$ then $v\tensor w\in M^{\tensor m+n}.$    
\end{definition}
We have the following universal property of the tensor algebra, 
\begin{proposition}[Universal Mapping Property of the Tensor Algebra]
	Let $A$ be an $R$-algebra and $f:M\to A$ an $R$-module homomorphism. Then there exists a unique $R$-algebra homomorphism $\hat{f}:T^\bullet(M)\to A$ extending $f$ such that the following diagram commutes: \[ \begin{tikzcd}
		M \arrow[r,"f"] \arrow[d,swap,"i"] & A \\
		T^\bullet(M) \arrow[ur,swap,"\hat{f}"] &
	\end{tikzcd}\]
\end{proposition}
The proof of this is the same flavor as for the other universal mapping properties and thus will not be produced here. What we will concern ourselves with however is a certain ideal of $T^\bullet(M).$ 
\begin{definition}
	A tensor $v\in T^\bullet(M)$ is called $\textbf{alternating}$ if $v$ has the following form: \[ v= m_1\tensor ... m_i \tensor ...\tensor  m_i \tensor ...\tensor m_n\]
	The repeating element is the focus. Let $\lie{J}$ be the ideal of $T^\bullet(M)$ generated by all such alternating elements. Sometimes we say that $\lie{J}=\ip{v\tensor v}$ for $v\in T^\bullet(M).$   
\end{definition}
\begin{lemma}
	$\lie{J}$ coincides with the ideal \[\lie{L}=\ip{x\tensor y+y\tensor x-(x+y)\tensor(x+y)+x\tensor x+y\tensor y}\] only if $\text{Char} (R)\neq 2.$ 
\end{lemma}
This is an easy manipulation of the defining relations for tensors. 
What this lets us build is the final object of this chapter: the Exterior Algebra. 
\begin{dtheorem}
	Let $R$ be a ring with $\text{Char} (R)\neq 2$ and put $\bigwedge^\bullet(M)=T^\bullet(M)/\lie{J}.$ This is called the $\textbf{exterior algebra}$ of $M$ and comes with the following universal property: Given any $R$-algebra $A$ and a map $\phi:M\to A$ such that $\varphi(m)^2=0,$ there exists a unique algebra homomorphism $\bigwedge^\bullet(M)\to A$ which makes the associated diagram commute.  	
\end{dtheorem}
\begin{proof}
	The universal property for the tensor algebra gives us a map, $\Psi,$ to $A.$ Taking the kernel of this map, we see that it is precisely when $\Psi(m\tensor m)=0.$ Hence, $\Psi$ descends to a map on $\bigwedge^\bullet(M).$ This completes the proof.  
\end{proof}
It is common practice to denote elements of $\bigwedge^\bullet(M)$ with $\wedge$ instead of $\tensor.$ In this way, we get immediately that $v\wedge w=-w\wedge v.$ This is equivalent to the condition, $v\wedge v=0.$
  
\begin{remark}
	We shall end this section with some nice properties of the exterior algebra so that we can use them in the next chapter readily. 
	\begin{enumerate}
		\item We can build $\bigwedge^k(M)$ in a similar way to building $\bigwedge^\bullet(M)$ we simply quotient $T^k(M)=\bigoplus^k M^{\tensor n}.$ In this vein, $\bigwedge^k(M) \wedge \bigwedge^{l}(M)\subseteq \bigwedge^{l+k} (M)$ which gives $\bigwedge^\bullet(M)$ an algebra structure. 
		\item If $V$ is a finite dimensional vector space of dimension $n.$ Then it can be shown that $\dim \bigwedge^k(V)=\binom{n}{k}.$ Therefore $\bigwedge^\bullet(M)$ is a finite dimensional algebra. 
		\item Recall the definition of a lie algebra from above. A different way to say the conditions of a lie algebra are that $V$ is a vector space equipped with a map \[[-,-]:{\bigwedge}^2(V)\to V\] satisfying the Jacobi identity. 
		\item As we will see in the next section, we can equivalently consider $\bigwedge^k(V)$ the vector space of differential $k$-forms on $V.$ This allows us to do calculus on these spaces and is a bridge between the theory of manifolds (chapter 3) and algebra, among others. 
		\item(Determinants) Let $V$ have dimension $n$ and consider the top exterior power $\bigwedge^n(V).$ This is a $1$-dimensional space by (2) above. Consider any $T\in \Hom(V,V):=\End(V)$ and define the extension \[T:\bigwedge^n(V)\to \bigwedge^n(V) \;\;\;\;\;\;\;\;\;\;\;\;\;\;\;\;\;\;\; T(v_1\wedge...\wedge v_n)=Tv_1\wedge ...\wedge Tv_n \]
		As this is an endomorphism of a $1$-dimensional space, it must be given by $Tv=\lambda v$ for some $\lambda\in K.$ Therefore we define the \textbf{determinant} of $T$ to be the unique number $\lambda$ such that \[ Tv_1\wedge...\wedge Tv_n=(\det T) (v_1\wedge...\wedge v_n) \]
		It then follows from the definition that for $S,T\in \End(V),$ we get $\det ST=\det S\cdot \det T.$ Those readers familiar with the determinant formula of a matrix should notice this as the standard property of the determinant. Furthermore, we have the following lemma 
		\begin{lemma}
			A matrix $M$ is invertible if and only if $\det M\neq 0.$ 
		\end{lemma} 
		\begin{proof}
			Abusing notation, by Theorem \ref{Matrix_of_Linear} we consider the linear transformation associated to the matrix $M.$ Then $\det M\neq 0$ implies that $M:\bigwedge^n(V)\to \bigwedge^n(V)$ is an isomorphism. Hence, $M$ has an inverse as a linear transformation and thus as a matrix.  
		\end{proof}
		This gives a nice way to think about determinants as the \textit{volume} of the parallelepiped spanned by the basis vectors $Tv_1,...,Tv_n.$ 
	\end{enumerate}
\end{remark}
This completes the chapter.

\chapter{Topology and Geometry: From Spaces to Sheaves} 
This section will run through the basics of category theory, (point-set) topology, differential geometry, and sheaf theory. The main goal is to define and give important properties of \textit{manifolds}. To mathematicians, these are generalizations of Euclidean space and provide a natural context to do calculus on non-flat spaces (more on this in Section 3.3).  There is some ambiguity on the definition of a manifold for psychologists which causes some technical problems when comparing computational models which claim to rely on the 
"manifold" structure. We shall give the formal, mathematical constructions of these objects and in Chapter 4, use this to construct a perceptual space which encodes the generalized perceptual categories of Chapter 1. Before then, we want to bridge the gap from the previous chapter to this one by exploring category theory. 

\section{Category Theory} 
Category theory began as an observation that many of the well known results of algebra (such as the First isomorphism theorems above) seemed to be linked. We now know that the reason this is true follows from general facts about what are known as \textit{Additive} and \textit{Abelian} categories. Although this theory is beautiful to those who fully understand the concepts, it can be seen as esoteric and impenetrable by some beginners. As we are assuming little to no familiarity with these topics,  we shall go into a bit more detail for most of the proofs in this section and provide several examples for each definition and theorem. For references, we make extensive use of \cite{MacLane1971}, \cite{Knapp2006},\cite{Knapp2007},\cite{Rotman2009}, and \cite{Lee2012}.     

\subsection{Categories and Functors} 
Before giving the definition of a category, we want to understand, more precisely, the language used in the previous chapter. The main goal will be to understand the relationship between morphisms of groups, rings, and modules. Category theory provides a setting in which these are all intimately related.  

\begin{example}
	\item Let $G,H$ be groups (not necessarily abelian). Denote by $\Hom(G,H)$ the set of all group homomorphisms. If $G$ and $H$ are assumed to be abelian, then $\Hom(G,H)$ can be endowed with the structure of an abelian group in a natural way: for any $f\in \Hom(G,H)$ define \[ n\cdot f(g)=f(ng)=nf(g)\in H\]
	Notice that for $H$ non-abelian we can still define a $\Z$-module structure on $\Hom(G,H)$ by $n\cdot f(g)=f(ng).$ We can similarly define a $\Z$-module structure if $G$ is non-abelian and $H$ is abelian. Thinking about $\Hom$ as a \textit{function on the set of all groups}\footnote{We are intentionally being sloppy here. As will be seen shortly $\Hom(-,-)$ is a functor $\textbf{Grp}\to \textbf{Set}.$} we can ask if it preserves group homomorphisms. To check this, let $\varphi:G\to G'$ be a morphism of groups. Define \[ \varphi^*:\Hom(G',H)\to \Hom(G,H)\;\;\;\;\;\;\;\;\;\;\;\; f\mapsto f\comp \varphi\]
	If instead we had a morphism $\psi:H\to H',$ then there is a canonical map \[ \psi_*:\Hom(G,H)\to \Hom(G,H')\]
	defined as you would imagine. Therefore, $\Hom$ can somehow detect which argument a morphism was taken in. If it is the first argument then the order is reversed, whereas the second argument preserves the order. 
	
	\item If we generalize the above example to rings and ring homomorphisms, we get the exact same result. Let $R,R',S,S'$ be rings and $\varphi:R\to R'$, $\psi:S\to S'$ be ring homomorphisms. Then $\varphi^*$ and $\psi_*$ are defined according to the definitions above.
	\item The same story for rings works with modules as well. This should not be surprising however as every abelian group is a $\Z$-module and we know how $\Hom$ works for abelian groups. 
\end{example}
This undercuts the original conclusion about $\Hom$; it can detect which argument is being manipulated but cannot (without some poking) detect group, ring, or module structures. What we do know is that it also plays suitably nice with morphisms for the correct objects. It is precisely this notion which categories and functors generalize. 

\begin{definition}
	A (small) \textbf{category} is a triple $\script{C}=(\operatorname{Obj}(\script{C}), \Hom_\script{C}(-,-),\circ)$ with $\operatorname{Obj}(\script{C})$ a set, an assignment for any two objects $A,B\in \operatorname{Obj}(\script{C})$ a set $\Hom_\script{C}(A,B)$ of \textbf{morphisms} between $A$ and $B$, and a function $\circ$ such that for all $A,B,C\in \operatorname{Obj}(\script{C}),$ \[\circ: \Hom_{\script{C}}(B,C)\times \Hom_{\script{C}}(A,B) \to \Hom_{\script{C}}(A,C) \]
	These are subject to the following axioms
	\begin{enumerate}
		\item $\Hom$ sets are disjoint (that is every element has a unique domain and codomain).
		\item There exists $1_A\in \Hom_{\script{C}}(A,A)$ for all $A\in \operatorname{Obj}(\script{C})$ such that $1_A\circ f =f.$ and $g\circ 1_A=g.$
		\item The map $\circ$ is associative.  
	\end{enumerate}  
	If it is clear from the context, we shall simply write $\Hom(A,B)$ for the set of morphisms. A \textbf{subcategory} of $\script{C}$ is a triple $\script{D}=(\operatorname{Obj}(\script{D}),\Hom(-,-),\comp)$ where $\operatorname{Obj}(\script{D})$ is a subset of $\Obj(\script{C})$ and $\Hom_\script{D}(A,B)\subseteq \Hom_{\script{C}}(A,B).$ Composition is taken as in $\script{C}.$    
\end{definition}
Notice that this definition does not require the objects themselves to be sets. This distinction is what makes proving things in category theory particularly frustrating: one cannot reference elements of an object when defining a morphism. 
\begin{example}\label{Categories_examples} 
\begin{enumerate}
	\item Consider the following graph \[
\begin{tikzcd}
                                                               & \bullet \arrow[ld] \arrow[loop, distance=2em, in=125, out=55] &                                                               \\
\bullet \arrow[rr] \arrow[loop, distance=2em, in=215, out=145] &                                                               & \bullet \arrow[lu] \arrow[loop, distance=2em, in=35, out=325]
\end{tikzcd}\]
	Define a category $\script{C}$ whose objects are the vertices of the above graph, the morphisms are the arrows, and composition is concatenation of paths. Notice that the objects of this category have no notion of element (i.e. they are not sets) and therefore if we wish to prove something about this category, we have to rely on "arrow theoretic" proof. That is to say we need to understand the morphisms in the category instead of the objects. 
	
	\item We now return to the algebraic objects of the previous chapter. For your favorite object in the previous chapter, it should be obvious that they form a category. We denote the categories as such: \begin{enumerate}
		\item $\textbf{Grp}$: the category of groups. 
		\item $\textbf{Ring}$: the category of rings. 
		\item $\textbf{Field}$: the category of fields. 
		\item $R-\textbf{Mod}$: the category of $R$-modules for a fixed ring $R.$ 
		\item $\textbf{Vect}_K$: the category of $K$-vector spaces. 
		\item $\textbf{Ab}$: the category of abelian groups. 
	\end{enumerate} 
	Notice that $\textbf{Ab}$ is a subcategory of $\textbf{Grp}.$ In fact, every category above can be realized as a subcategory of $\textbf{Grp}!$ 
	
	\item The "category" of Sets is denoted $\textbf{Set}.$ The quotations here are for caution: the "collection of all sets" is not itself a set (try to prove this!) but instead a proper class. We are going to ignore almost all set theoretic problems that may arise. Nonetheless, this is an honest category (once you fix your model of set theory) and it is quite important. A majority of what will come up when we discuss functors can be realized as some generalization of something involving sets.   
\end{enumerate}
\end{example}
\begin{remark}
	For the remainder of this thesis, we shall denote categories by calligraphic or script letters $\mathcal{C},\script{C}$ if we are in a general setting, or a corresponding bold-face name such as $\textbf{Grp}$ for the category of groups. 
\end{remark}

\begin{definition}
	Let $\script{C}$ and $\script{D}$ be two categories. We define the product category $\script{C}\times \script{D}$ as the category whose objects are pairs $(C,D)$ and whose morphisms are pairs $(f,g).$ 
\end{definition}

Now that we have the notion of a category, we may ask if there are any "special" morphisms in this category. What we mean by special here will become apparent shortly. Consider the category $\textbf{Set}.$ The following lemma gives a different characterization of injective and surjective functions which is easily generalizable. 

\begin{lemma}\label{Concrete_set} 
	Let $f:A\to B$ and $g:A'\to B'$ be two functions. Then $f$ is injective if and only if for any two arrows $i_1,i_2:C\to A,$ the equality \[f\comp i_1=f\comp i_2\implies i_1=i_2.\] 
	Similarly, $g$ is surjective if and only if for any two arrows $s_1,s_2:B'\to C',$ the equality \[s_1\comp g=s_2\comp g\implies s_1=s_2.\] This means that injective maps are $\textit{left}$ cancellable and surjective maps are $\textit{right}$ cancellable. 
\end{lemma}
\begin{proof}
	We shall prove the injective case and leave the surjective case to the reader. ($\Rightarrow$) Assume that $f$ is left cancellable. For any $a,a'\in A,$ let $\varphi_a:\{*\}\to A$ be the function which picks out the element $a.$ Then if $f$ is left cancelable and \[ f(\varphi_a)=f(\varphi_{a'})\implies \varphi_a=\varphi_{a'}\implies a=a'\] 
	Hence $f$ is injective. The other direction is obvious from the definition of injective. This completes the proof. 
\end{proof}
Notice that we can re-write the injectivity condition on the level of diagrams as \[ \begin{tikzcd}
C \arrow[r, "i_2"', shift right] \arrow[r, "i_1", shift left] & A \arrow[r, "f"] & B
\end{tikzcd}\]
More generally, we can think of arrows in arbitrary categories which have the left (resp. right) cancellable property. 
\begin{definition}
	Let $\mathcal{C}$ be a category and $f:A\to B$ be a morphism. We say that $f$ is \textbf{monic} when for any pair of morphisms $g,h:C\rightrightarrows A,$ the equality. $f\comp g=f\comp h$ implies $g=h.$ We say that $f$ is \textbf{epic} when for any pair of morphisms $p,q:B\rightrightarrows D,$ the equality. $p\comp f=q\comp f$ implies $p=q.$ We call $f$ an \textbf{isomorphism} if there exists $r:B\to A$ such that $fr=1_B$ and $rf=1_A.$ Further, we denote isomorphisms by either $A\cong B$ or $A\overset{\sim}{\to} B.$  
\end{definition}

	In all \textit{concrete} categories (ones which can be realized as subcategories of \textbf{Set}) monic maps are injective and epic maps are surjective. This mirrors the result of Lemma \ref{Concrete_set}. In fact, this is precisely the definitions of isomorphism coincide with the categorical one for all of the algebraic objects in Chapter 2! In general, the converse is not true. Let $R,S$ be two rings and $UR,US$ their underlying sets. Then an injective function $f:UR\to US$ need not be a ring homomorphism. For an easy example, consider $R=S=\Z.$ Then the map \[ 2: \Z\to \Z\;\;\;\;\;\;\;\;\;\;\;\;\;\;\;\;x\mapsto 2x\] is a perfectly well defined injective function but is definitely not a ring homomorphism as $1$ cannot be written as $2z$ for some $z\in \Z.$ 
	
Something else which needs generalization is the equivalence in \textbf{Set} between isomorphisms and bijections. In general, every isomorphism is necessarily monic and epic. The converse may not be true (take for example the ring example above but change where the identity is sent). We want to deal with categories where this is true. 
	 
\begin{definition}
	A category $\script{B}$ is called $\textbf{balanced}$ if all monic, epic morphisms are isomorphisms. 
\end{definition}
It should be clear that all concrete categories are balanced. More often than not, this is something which needs to be proven but is not too hard. 

Before moving forward,  it is important to label some distinguished objects of certain categories. 
\begin{definition}
	An object $T\in \script{C}$ is a \textbf{terminal} object if for all objects $A\in \script{C}, $ there exists a unique (denoted $\exists!$) $A\to T.$ An object $I\in \script{C}$ is \textbf{initial} if for all objects $A\in \script{C}, $ there exists a unique $I\to A.$ A \textbf{zero} object is an object which is both terminal and initial.  
\end{definition}
\begin{proposition}\label{Inital_terminal_zero_unique} 
	Initial, terminal, and zero objects are unique up to unique isomorphism. 
\end{proposition}
\begin{proof}
	The proof for initial, terminal, and zero objects is exactly the same. For this reason, we shall only prove the initial case. Let $I_1,I_2$ be two initial objects. By definition there exist unique morphisms $\iota_1:I_1\to I_2$ and $\iota_2:I_2\to I_1.$ It suffices to show that $\iota_2\comp  \iota_1=1_{I_1}$ and $\iota_1\comp \iota_2=1_{I_2}.$ As the objects are initial, the set $\Hom(I_i,I_i)$ contains a single element, namely $1_{I_i}.$ As the composition $\iota_1\comp \iota_2\in \Hom(I_2,I_2)$ it must be $1_{I_2}.$ By the same reasoning we have that $\iota_2\comp\ \iota_1=1_{I_1}.$ Hence, $I_1\cong I_2$ and this isomorphism is unique.       
\end{proof}

\begin{example}
	Zero, initial, and. terminal objects are incredibly important in the theory of \textit{abelian} categories (section 3.1.4). For this reason, we give the following exmaples:
	\begin{enumerate}
		\item In \textbf{Grp} the zero object is the trivial group $G=\{1\}.$ 
		\item In \textbf{Ring} the initial object is $\Z$ while there is no terminal object.
		\item In R-\textbf{Mod} the zero object is the $0$ module.    
	\end{enumerate} 
\end{example}

\subsubsection{Functors} 
Now that we have the notion of a category, we want to define morphisms of categories. Similar to the restrictions of a ring homomorphism, we want a morphism of categories to preserve both the objects and the morphisms. 
\begin{definition}
	Let $\script{C},\script{D}$ be two categories. A (covariant) \textbf{functor} $F:\script{C}\to \script{D}$ subject to the following: 
	\begin{enumerate}
		\item For all $A\in \operatorname{Obj}(\script{C})$, $F(A)\in \operatorname{Obj}(\script{D})$ and similarly for morphisms.  
		\item If $A\overset{f}{\rightarrow} B \overset{g}{\rightarrow} C$ is a sequence of morphisms in $\script{C},$ then $F(g\comp f)=F(g)\comp F(f)$ is a morphism in $\script{D}.$
		\item $F(1_A)=1_{F(A)}.$    
	\end{enumerate}
	Dually, we have the notion of $\textbf{contravariant}$ functors for which $F(g\comp f)=F(f)\comp F(g).$ It is common practice to write $FX$ for an object as opposed to $F(X).$ We shall use these notations interchangeably.  
\end{definition} 	

Functors play a core role in the rest of the theory presented in this thesis. Specifically, they will form an important class of objects called sheaves (see section 3.3.2 below) which will ease the technical burden of understanding the geometry. of perceptual spaces.  

\begin{lemma}
	Let $F:\script{C}\to \script{D}$ be a functor. Then if $\varphi:A\to B$ is an isomorphism in $\script{C},$ then $F(\varphi)$ is an isomorphism in $\script{D}.$ 
\end{lemma}	
\begin{proof}
	Let $\psi$ be $\varphi^{-1}$ in $\script{C}.$ Computing $F(\varphi\comp \psi)$ and $F(\psi\comp \varphi),$ we see that by property (b) of the definition of a functor, we have that \begin{align*} 1_{F(A)}=F(1_A)=F(\psi)\comp F(\varphi)  &&  1_{F(B)}=F(1_B)=F(\varphi)\comp F(\psi)\end{align*} 
	Hence, $F(\varphi)$ is an isomorphism.  
\end{proof}

The following examples of functors will play an exceptional role in section 3.3 below. 
\begin{example}\text{}
	\begin{enumerate}
		\item Let $(-)^{op}:\textbf{Cat}\to \textbf{Cat}$ be an endofunctor of the category of categories (this morphisms in this category are functors). This sends a category $\script{C}$ to the $\textbf{opposite category}$ $\script{C}^{op}.$ The objects of this category are the objects of $\script{C}$ but the morphisms have their target and source flipped. That is, if $f:A\to B$ is a morphism in $\script{C}$ then $f^{op}:B\to A$ is a morphism in $\script{C}^{op}.$ This allows us to redefine contravariant functors as covariant functors from the opposite category. As an added fact, $(\script{C}^{op})^{op}=\script{C}.$ 
		\item Consider $\Hom_\script{C}(-,-):\script{C}^{op}\times \script{C}\to \textbf{Set}.$ This is a \textit{bifunctor} and is covariant in the first argument and covariant in the second argument.  
		\item In $R$-\textbf{Mod}, $-\tensor_R-$ is a bifunctor, covariant in both arguments. As we assume $R$ is commutative, $\tensor$ makes $R$-\textbf{Mod} into a symmetric monoidal category \footnote{We shall not define this here, but instead suggest \cite[Chapter XI]{Kassel1995}. Kassel uses the term \textit{tensor cateogry} which is equivalent to ``monoidal cateogry."}.  Algebras are monoid objects in this category. 
		\item Let $U:\textbf{Grp}\to \textbf{Set}$ be the \textit{forgetful functor} which sends a group to its underlying set. In fact, in any concrete category we have a forgetful functor to$\textbf{Set}.$ 
	\end{enumerate}
\end{example}

If $\script{C}$ and $\script{D}$ are categories, then denote by \[ \operatorname{Fun}(\script{C},\script{D}):=\{ F:\script{C}\to \script{D}\}\]
We want to turn this into a category. In order to do this, we need to introduce the idea of a \textit{morphism of functors}. 
\begin{definition}
	Let $F,G:\script{C}\to \script{D}$ be two functors of the same variance. A \textbf{natural transformation} is a family of morphisms $\{\tau_X\}$ which intertwine the functors as the following diagram shows \[ \begin{tikzcd}
F(X) \arrow[r, "\tau_X"] \arrow[d, "F(f)"'] & G(X) \arrow[d, "G(f)"] \\
F(Y) \arrow[r, "\tau_Y"']                   & G(Y)                  
\end{tikzcd}   \] In this case we write $\tau:F\to G.$  
\end{definition} 
These define the morphisms in $\operatorname{Fun}(\script{C},\script{D})$ and make it a category. Isomorphisms are natural transformations for which every $\tau_X$ is an isomorphism in $\script{D}.$ In this case, we say that two functors are \textit{naturally equivalent}.  The following lemma gives a description of Natural transformations involving the $\Hom(A,-)$ functor. 
\begin{lemma}[Yoneda Lemma]
	Let $G:\script{C}\to \textbf{Set}$ be a functor and $A$ an object in $\script{C}.$ Then there is a bijection \[ y:\operatorname{Nat}(\Hom(A,-),G)\to G(A)\] 
\end{lemma} 
\begin{proof}
	Define $y(\tau)=\tau_A(1_A).$ To show this is injective, suppose $y(\tau)=\tau_A(1_A)=\sigma_A(1_A)=y(\sigma).$ For any object $B\in \script{C},$ and $\varphi\in \Hom(A,B),$ we have the following commutative diagram \[  \begin{tikzcd}
\Hom(A,A)\arrow[r, "\tau_A"] \arrow[d, "\varphi_*"'] & G(A) \arrow[d, "G\varphi"] \\
\Hom(A,B) \arrow[r, "\tau_B"']                   & G(B)                  
\end{tikzcd}  \]
So that $\tau_B(\varphi)=G\varphi \tau_A(1_A)=G\varphi \sigma_A(1_A)=\sigma_B(\varphi).$ Hence, $\tau_B=\sigma_B$ for all $B\in \script{C}$ and thus $\tau=\sigma.$ So $y$ is injective. 

To show it is surjective, let $x\in G(A).$ For every object $B\in \script{C}$ and $\psi\in \Hom(A,B),$ define $\tau_B(\psi)=(G\psi)(x).$ We claim then that $\tau$ is a natural transformation. Indeed, for any $\theta\in \Hom(B,C),$ then commuting square \[   \begin{tikzcd}
\Hom(A,B)\arrow[r, "\tau_B"] \arrow[d, "\theta_*"'] & G(B) \arrow[d, "G\theta"] \\
\Hom(A,C) \arrow[r, "\tau_C"']                   & G(C)                  
\end{tikzcd}  \]   
Then going clockwise we get that $G\theta\tau_B(\psi)=G\theta G\psi(x).$ Going counter-clockwise we have that $\tau_C(\theta_*\psi)=\tau_C(\theta\psi)=G\theta\psi(x).$ As $G$ is a functor, these are equal. Thus, $\tau$ is a natural transformation and $\tau_A(1_A)=G1_A(x)=x.$ Hence $y$ is bijective. This completes the proof.  
\end{proof}

Now let $F:\script{C}\to \script{D}$ be a functor and $X,Y\in \script{C}.$ Then $F$ induces a function on $\Hom$-sets \[ F_{X,Y}:\Hom_{\script{C}}(X,Y)\to \Hom_{\script{D}}(FX,FY)\]
which takes a function $f$ to $F(f).$ 
\begin{definition}
	We say that $F$ is: \begin{enumerate}
		\item \textbf{Full} if $F_{X,Y}$ is surjective for all $X,Y.$ 
		\item \textbf{Faithful} if $F_{X,Y}$ is injective for all $X,Y.$ 
		\item \textbf{Fully-Faithful} if $F_{X,Y}$ is bijective for all $X,Y.$ 
	\end{enumerate}
\end{definition}

Therefore, concrete categories are those which admit a faithful functor into \textbf{Set}. In general, fully-faithful functors play the same role as bijective functions on sets. In \textbf{Cat} isomorphisms are necessarily fully-faithful. In general, a bijection on the level of $\Hom$-sets is incredibly important.  

\subsection{Adjoints} 
We now explore the final claim of the previous part. Let $F:\script{C}\rightleftarrows \script{D}:G$ be functors such that there exists a natural transformation $\eta:1_C\to GF.$ Then we want to understand the induced morphism \[ \Hom_\script{D}(FX,Y)\to \Hom_{\script{C}}(X,GY)\]
\begin{definition}
	Let $F:\script{C}\rightleftarrows \script{D}:G.$ We say that $(F,G)$ are an \textbf{adjoint pair} if \[  \Hom_\script{D}(FX,Y)\overset{\sim}{\longrightarrow} \Hom_{\script{C}}(X,GY)\]
	for all $X\in \script{C}$, $Y\in \script{D}.$ Further, the bijection is natural in $X$ and $Y.$ In this case, we say that $F$ is \textit{left adjoint} to $G$ and $G$ is \textit{right adjoint} to $F.$ We denote this by $F\dashv  G.$ 
\end{definition}
\begin{theorem}
	An adjoint pair $(F,G)$ induces two natural transformations $\eta:1_\script{C}\to GF$ and $\varepsilon:FG\to 1_{\script{D}}$ such that the compositions \begin{align*}
		F\overset{F\eta}{\longrightarrow} FGF \overset{\varepsilon F}{\longrightarrow} F && G\overset{\eta G}{\longrightarrow} GFG \overset{G \varepsilon }{\longrightarrow} G
	\end{align*} 
	are the identity morphisms. 
\end{theorem}
\begin{proof}
	Let $\varphi_{X,Y}:\Hom_\script{D}(FX,Y)\overset{\sim}{\longrightarrow} \Hom_{\script{C}}(X,GY)$ be the bijection for the adjoint pair. Then if $Y=FX,$ the element $1_{FX}\in \Hom_{\script{D}}(FX,FX)$ induces a morphism \[ \eta_X:=\varphi_{FX,FX}(1_{FX}): X\to GFX\] 
	Define $\eta:1_\script{C}\to GF$ by $\eta_X.$ We need to show that $\eta$ is natural in $X.$ Consider the following diagram: 
	  \[   \begin{tikzcd}
X\arrow[r, "\eta_X"] \arrow[d, "f"'] & GFX \arrow[d, "GF(f)"] \\
Y \arrow[r, "\eta_Y"']                   & GFY                 
\end{tikzcd}  \] 
It commutes by the fact that $\varphi$ is natural in both $X,Y.$ Similarly, we define $\varepsilon_Y:=\varphi_{GY,Y}^{-1}(1_{GY}).$ Its naturality is checked in a similar manner. Now, \[ 1_{GY}=\varphi_{GY,Y}(\varepsilon_Y)=G(\varepsilon_Y)\comp \eta_{GY}\]
again by the naturality of $\varphi$. We have the respective statement for $1_{FX}.$ This completes the proof.  
\end{proof}
\begin{remark}
	The natural transformations $\eta:1_{\script{C}}\to GF$ and $\varepsilon:GF\to 1_{\script{D}}$ are called the \textbf{Unit} and $\textbf{Counit}$ of the adjunction. We then denote an adjunction as a quadruple $(F,G,\eta,\varepsilon).$ 
\end{remark}	

\begin{corollary}
	If $(F,G,\eta,\varepsilon)$ and $(F',G,\eta',\varepsilon')$ are adjoint pairs, then $F$ and $F'$ are naturally isomorphic. 
\end{corollary}
\begin{proof}
	$\eta$ and $\eta'$ are universal arrows for each $x.$ Therefore, there exists a unique isomorphism $\theta_X:FX\to F'X$ for all $X\in \script{C}.$ This family of isomorphisms is natural in $X$ by the universality of the units. Hence, $F\cong_{\operatorname{Nat}} F'.$ 
\end{proof}

Adjoint functors play a large role in understanding isomorphisms of categories. In fact, two categories are \textit{equivalent} if there exists an adjoint pair $(F,G,\eta,\varepsilon)$ such that $\eta$ and $\varepsilon$ are natural isomorphisms. To build up some intuition, here are some examples of adjoint functors. 

\begin{example}
	\begin{enumerate}
		\item Let $\ip{}:\textbf{Set}\to \textbf{Grp}$ be the \textit{free group} functor and $U$ the forgetful functor. This sends a set $X$ to the group $\ip{X}$ which is the group generated by all words in the elements of $X.$ It is characterized by the property that for any function $f:X\to G$ a group, there exists a unique group homomorphism $\hat{f}:\ip{X}\to G.$ We claim this makes $\ip{}\dashv U.$ In fact, the universal property gives a bijection \[ \Hom_{\textbf{Grp}} (\ip{X},G) \leftarrow \Hom_{\textbf{Set}}(X,UG)\]
		In fact, for any concrete algebraic object we get an adjunction between the free functor and the forgetful functor in the same way. 
		\item Consider $\Hom(M,-)$ and $-\tensor_RM$ as covariant endofunctors of $R$-\textbf{Mod}. Then for any objects $A,B\in R$-\textbf{Mod}, there is a bijection \begin{align*}
			\Hom(A\tensor M,B)\to \Hom(A,\Hom(M,B)) && f\mapsto \tilde{f}
		\end{align*}
		where $ \tilde{f}(a)(m)=f(a\tensor m) .$ In this case, we have some additional facts that come from the adjunction. The two most interesting (and important) ones are: \begin{align*} \Hom(M,\prod A_i)=\prod \Hom(M,A_i) && M\tensor \bigoplus A_i=\bigoplus (M\tensor A_i)
		\end{align*} for arbitrary indexing sets. We will see shortly that this is more generally a property of adjoint functors between \textit{abelian categories}.
	\end{enumerate}
\end{example}

\subsubsection{Limits and Colimits} 
We now want to generalize the last example and understand products and coproducts in generic categories. These manifest as \textit{limits} and \textit{colimits} respectively. Recall that a product of two objects $A,B$ is an object $A\times B$ together with two maps $A\times B\to A$ and $A\times B\to B.$ To be more precise, this is somehow the universal object such that for any other object with maps $C\to A$ and $C\to B,$ there exists a unique map $C\to A\times B$ such that the following diragram commutes \[ \begin{tikzcd}
C \arrow[rd, "\exists!", dotted] \arrow[rrd, bend left] \arrow[rdd, bend right] &                               &   \\
                                                                                & A\times B \arrow[d] \arrow[r] & A \\
                                                                                & B                             &  
\end{tikzcd}\]

Let us now generalize this. 
\begin{definition}
	An \textbf{inverse system} in a category $\script{C}$ is a collection indexed by a partially ordered set $I$, $\{A_i,\varphi^j_i:A_j\to A_i\}_{i\preceq j}$ such that $\varphi_{k}^j\varphi_j^i=\varphi^i_k$ for all $ i\preceq j\preceq k.$ Equivalently, an inverse system is a functor $A:I^{op}\to \script{C}$ such that $A(i)=A_i$ and $A(i\to j)=\varphi_{i}^j.$ Therefore, $A\in \script{C}^{I^{op}}=\operatorname{Fun}(I^{op},\script{C}).$
\end{definition}
An inverse system is thus a \textit{diagram} in the category $\script{C}$ of shape $I^{op}.$ 
\begin{example}\label{limits} 
	\begin{enumerate}
		\item Let $I=\{1,2,3\}$ with the partial order $1\preceq 2$ and $1\preceq 3.$ Then diagrams of shape $I^{op}$ look like \[ \begin{tikzcd}
			\text{} & A \arrow[d]\\
			B\arrow[r] & C
		\end{tikzcd}\]
		\item If $I$ is discrete (that is the only partial order is equality) then a diagram of shape $I^{op}$ is an indexed family of objects. This is the case for products as above.  
		\item Let $M$ be a concrete object. Then the subsets of $M$ are ordered under inclusion and thus give a diagram of shape $M^{op}.$  
	\end{enumerate}
\end{example}
\begin{definition}
	Let $A\in \script{C}^{I^{op}}$ be an inverse system. Then we define the \textbf{inverse limit} (projective limit or limit) as the universal object $\varprojlim A_i$ together with morphisms $\alpha_j:\varprojlim A_i\to A_j$ for all $j$ satisfying the following compatibility conditions: 
	\begin{enumerate}
		\item $\varphi_i^j(\alpha_j)=\alpha_i$ for $i\preceq j.$ 
		\item If $C$ is an object of $\script{C}$ together with morphisms $\{\beta_i\}$ which are compatible with $A,$ then there exists a unique morphism so that the following diagram commutes for all $i\preceq j$: \[ \begin{tikzcd}
C \arrow[rd, "\exists!", dotted] \arrow[rrd, bend left,"\beta_j"] \arrow[rdd, bend right,"\beta_i"] &                               &   \\
                                                                                & \varprojlim A_i \arrow[d,swap,"\alpha_j"] \arrow[r,"\alpha_i"] & A_i \\
                                                                                & A_j     \arrow[ur,swap,"\varphi^j_i"]                        &  
\end{tikzcd}\] 
	\end{enumerate}
\end{definition}
These objects are complicated to look at but are so useful that it's worth the technicalities. The following examples tie together some previous topics which at first so not seem necessarily related but are all examples of limits.  
\begin{example}
	\begin{enumerate}
		\item	Consider the following diagram $D$ in $R$-\textbf{Mod} \[ \begin{tikzcd}
			\text{} & A \arrow[d,"f"]\\
			0\arrow[r] & C
		\end{tikzcd}\]
		Then $\varprojlim D=\ker f.$ In this case, we see that the limit must have the following set representation \[ \varprojlim D=\{ (x,y)\in 0\times A: 0=f(y)\}\]
		In fact, arbitrary limits exist in $R$-\textbf{Mod} by a simple argument considering sets like those above. 
		\item Clearly, products as above are now limits. over the discrete set $I=\{1,2\}.$ 
		\item We define the \textbf{pullback} of a diagram of the form Example \ref{limits} (a), to be their limit. Almost always, these have a set representation as in example (a) here. In this case, we denote $\varprojlim D=A\times_C B.$  
		\item If we want to define intersections without using elements, we can do it using limits. Let $A\to C$ and $B\to C$ be monic morphisms (they are subobjects). Taking the limit of this diagram we get \[  \begin{tikzcd}
			A\cap B \arrow[r] \arrow[d]& A \arrow[d,"i"]\\
			B\arrow[r,"j"] & C
		\end{tikzcd}\]
		The resulting morphisms are clearly monic. 
	\end{enumerate}
\end{example}

We have the dual notion to the above construction. 
\begin{definition}
	A \textbf{direct system} in a category $\script{C}$ is a collection indexed by a partially ordered set $I$, $\{A_i,\varphi^i_j:A_i\to A_j\}_{i\preceq j}$ such that $\varphi_{k}^j\varphi_j^i=\varphi^i_k$ for all $ i\preceq j\preceq k.$ Equivalently, an direct system is a functor $A:I \to \script{C}$ such that $A(i)=A_i$ and $A(i\to j)=\varphi_{j}^i.$ Therefore, $A\in \script{C}^{I}=\operatorname{Fun}(I,\script{C}).$
\end{definition}

\begin{example}\label{colimits} 
	\begin{enumerate}
		\item Let $I=\{1,2,3\}$ with the partial order $1\preceq 2$ and $1\preceq 3.$ Then diagrams of shape $I^{op}$ look like \[ \begin{tikzcd}
			 A \arrow[d] \arrow[r] & B\\
			C & 
		\end{tikzcd}\]
		\item If $I$ is discrete (that is the only partial order is equality) then a diagram of shape $I$ is an indexed family of objects. This is the case for products as above.  
		\item Let $M$ be a concrete object. Then the subsets of $M$ are ordered under inclusion and thus give a diagram of shape $M.$   
	\end{enumerate}
\end{example}
\begin{definition}
	Let $A\in \script{C}^{I}$ be an direct system. Then we define the \textbf{direct limit} (inductive limit or colimit) as the universal object $\varinjlim A_i$ together with morphisms $\alpha_j: A_j\to \varinjlim A_i$ for all $j$ satisfying the following compatibility conditions: 
	\begin{enumerate}
		\item $\alpha_j\varphi_j^i=\alpha_i$ for $i\preceq j.$ 
		\item If $C$ is an object of $\script{C}$ together with morphisms $\{\beta_i\}$ which are compatible with $A,$ then there exists a unique morphism so that the following diagram commutes for all $i\preceq j$: \[ \begin{tikzcd}
                                                              & A_i \arrow[rdd, "\beta_i", bend left] \arrow[d, "\alpha_i"] \arrow[ld, "\varphi^i_j"'] &   \\
A_j \arrow[rrd, "\beta_j", bend right] \arrow[r, "\alpha_j"'] & \varinjlim A \arrow[rd, "\exists!", dotted]                                            &   \\
                                                              &                                                                                        & C
\end{tikzcd}
		\]
	\end{enumerate}
\end{definition}

\begin{example}
	\begin{enumerate}
		\item	Consider the following diagram $D$ in $R$-\textbf{Mod} 
		 \[ \begin{tikzcd}
			 A \arrow[d] \arrow[r,"f"] & B\\
			0 & 
		\end{tikzcd}\]
		Then $\varinjlim D=\coker f.$ In this case, we see that the limit must have the following set representation \[ \varinjlim D=(B\ds 0)/\{ (f(x),0)\in 0\ds A: x\in A\}\]
		In fact, arbitrary colimits exist in $R$-\textbf{Mod} by a simple argument considering sets like those above. 
		\item Clearly, coproducts as above are now colimits. over the discrete set $I=\{1,2\}.$ 
		\item We define the \textbf{pushout} of a diagram of the form Example \ref{colimits} (a), to be their colimit. Almost always, these have a set representation as in example (a) here. In this case, we denote $\varprojlim D=A\ds_C B.$  
		\item If we want to define internal sums without using elements, we can do it using colimits. Let $A,B$ be two objects. Then $A\cap B\to A$ and $A\cap B\to B$ are monic morphisms (they are subobjects). Taking the colimit of this diagram we get \[  \begin{tikzcd}
			A\cap B \arrow[r,"j"] \arrow[d,swap,"i"]& A \arrow[d] \\
			B \arrow[r]& A+B
		\end{tikzcd}\]
		The resulting morphisms are clearly monic. 
	\end{enumerate}
\end{example}

The following proposition gives motivation for thinking of limits and colimits as functors. 

\begin{proposition}
	Let $I$ be a partially ordered set. Then all limits and colimits exist in $R$-\textbf{Mod}.  
\end{proposition}
\begin{proof}
	We prove the case of limits. The case of colimits is then formally dual and left as a fun exercise. Consider $L\subseteq \prod_{i\in I} A_i$ the submodule of threads \[ L=\{ (a_i): \varphi^j_i(a_j)=a_i\}\]
	By construction this comes with compatible maps $\alpha_i:L \to A_i.$ 
	
	Now let $X$ be any module with compatible maps $\{\beta_i\}$. Define $\theta:X\to \prod A_i$ by \[ \theta(x)=(\beta_i(x))\]
	Then $\im \theta\subseteq L$. Further $\alpha_i\theta: x\mapsto (f_i(x))\mapsto f_i(x).$ Hence, the limit diagram commutes. To show that $\theta$ is unique, let $\pi:X\to L$ be another such morphism. Then $\pi(x)=(a_i)$ and $\alpha_i\pi(x)=a_i.$ Thus if $\alpha_i\pi(x)=f_i(x),$ we have that $\pi=\theta$ and thus \[ L\cong \varprojlim A_i\] This completes the proof.  
\end{proof}
This proposition says that $R$-\textbf{Mod} is complete and cocomplete (meaning that all limits and colimits exist). So clearly, $\varinjlim:R$-$\textbf{Mod}^I\to R$-\textbf{Mod} is functorial. We would like to show this in general. This is not true however. 
\begin{example}
	Let $\textbf{Ring}$ be the category of rings. Then if $\{R_i\}$ is an indexed family of objects, \[ \varinjlim R_i\not\in \textbf{Ring}\]
	Why is this? Well, the unit element is necessarily $(1,1,...).$ But this is non-zero in every entry and thus cannot be an element of the colimit (in this case it is the infinite direct sum).  In fact, most categories are not complete or cocomplete. When they are, it is obvious that $\varinjlim$ and $\varprojlim$ are functors. For more information, see \cite{Hilton1997}.    
\end{example}

\subsection{Abelian Categories and $R$-\textbf{Mod}}
We now move into the final subsection. Here we are interested in categories which generalize the category of $R$-modules or abelian groups. The defining charactersitics of these categories is that we can: 
\begin{itemize}
	\item Always take kernels and cokernels
	\item Have an object $0.$
	\item Can take arbitrary products and coproducts. 
	\item $\Hom(A,B)$ is an abelian group (or $R$-module). 
\end{itemize}
What of these properties is necessary in generalizing? This section will give an answer to this. At the end, we will introduce some homological algebra. This will allow us to associate invariants to modules. We start with additive categories. 

\begin{definition}
	A category $\script{A}$ is \textbf{additive} if the following are true: 
	\begin{enumerate}
		\item $\Hom(A,B)$ is an abelian group for all $A,B\in \script{A}.$
		\item There exists a zero object $0.$
		\item Composition is distributive. That is $f(g+h)=fg+fh$ and $(g+h)i=gi+hi.$ 
		\item Finite products and coproducts exist.   
	\end{enumerate}
	A functor $F:\script{A}\to \script{B}$ is \textbf{additive} if $F(f+g)=F(f)+F(g).$ That is the morphism \[ F_{X,Y}:\Hom(X,Y)\to \Hom(FX,FY)\] is a group homomorphism. 
\end{definition}
The following proposition gives some properties of additive categories and additive functors. 
\begin{proposition}
	Let $\script{A},\script{B}$ be additive categories. Then finite products and coproducts are isomorphic. Moreover if $T$ is an additive functor, then $T(A\ds B)=T(A)\ds T(B).$ 
\end{proposition}
For a proof of this statement see \cite{Rotman2009}. 

Now, using the constructions of $\ker$ and $\coker$ from above, we can prove 
\begin{lemma}
	Let $f\in \Hom_{\script{A}}(A,B)$ be a morphism in an additive category. \begin{enumerate}
		\item If $\ker f$ exists, then $f$ is monic if and only if $\ker f=0.$
		\item If $\coker f$ exists, then $f$ is epic if and only if $\coker f=0.$ 
	\end{enumerate}	
\end{lemma}
\begin{proof}
	Let $\iota:\ker f\to A$ be the morphism from the diagramatic definition above. If $\iota=0,$ and if $g:X\to A$ satisfies $fg=0,$ then by the universal property of limits, there exists a morphism $\theta:X\to \ker f$ with $g=\iota \theta=0.$ Hence, $f$ is monic. 
	
	For for the opposite direction consider the diagram \[ K\underset{0}{\overset{\iota}{\rightrightarrows}}A\overset{f}{\to} B \]
	Since $f\iota=0=f0,$ we have that $\iota=0.$ The proof for cokernels is dual to this one. 
\end{proof}

\begin{definition}
	An additive category $\script{A}$ is \textbf{Abelian} if \begin{enumerate}
		\item Every morphism has a kernel and cokernel.
		\item Every monomorphism is a kernel and every epimorphism is a cokernel. 
	\end{enumerate}
\end{definition}
\begin{example}
	In $R$-\textbf{Mod}, we have that every submodule $S\subseteq M$ can be realized as a kernel via the map $M\to M/S.$ Cokernels are then the projections as given by the first isomorphism theorem (Theorem \ref{Module_Isos}).  Therefore, the requirements of an abelian category make it look strikingly like $R$-\textbf{Mod}.
\end{example}

We are now able to form the same definitions as in Chapter 2, but now in the context of abelian categories. 
\begin{definition}
	A sequence of morphisms $A\overset{f}{\to} B\overset{g}{\to} C$ in $\script{A}$ is called \textbf{exact} if $\ker g=\im f$ as subobjects in $\script{A}.$ Now let $0\to A\to B\to C\to 0$ be exact in $\script{A}$ We say an additive functor $F:\script{A}\to \script{B}$ between abelian categories is \begin{enumerate}
		\item \textbf{Left Exact} if $0\to FA\to FB\to FC$ is exact.
		\item \textbf{Right Exact} if $FA\to FB\to FC\to 0$ is exact.
		\item \textbf{Half Exact} if $FA\to FB\to FC$ is exact. 
		\item \textbf{Exact} if $0\to FA\to FB\to FC\to 0$ is exact. 
	\end{enumerate}
\end{definition}

\begin{lemma}[Snake Lemma] \label{Snake_Lemma}
	Consider the following commuting diagram in an abelian category  \[ 
\begin{tikzcd}
\text{} & A' \arrow[d, "\psi"'] \arrow[r, "\alpha_1"] & A \arrow[d, "\varphi "] \arrow[r, two heads, "\alpha_2"] & A'' \arrow[d, "\theta"] \arrow[r] & 0\\
0\arrow[r] & B' \arrow[r,hook, "\beta_1"']                    & B \arrow[r, "\beta_2"']                       & B''  &                  
\end{tikzcd}\] 
If the rows are exact, then there exists a morphism $\partial:\ker \theta\to \coker \psi$ making the following sequence exact \[ \ker \psi\to \ker \varphi\to \ker \theta\to \coker \psi \to \coker \varphi\to \coker \theta\]
\end{lemma}
\begin{proof}
	Extend the above diagram to include $\ker \theta$ and $\coker \psi.$ Now form the pull-back and pushout accordingly: 
		\[ 
\begin{tikzcd}
\text{} & & A\times_{A''} \ker \theta \arrow[r,two heads] \arrow[d] & \ker \theta \arrow[d,hook] \arrow[r]& 0 \\
\text{} & A' \arrow[d, "\psi"'] \arrow[r, "\alpha_1"] & A \arrow[d, "\varphi "] \arrow[r, two heads, "\alpha_2"] & A'' \arrow[d, "\theta"] \arrow[r] & 0\\
0\arrow[r] & B' \arrow[d,two heads]  \arrow[r,hook, "\beta_1"']                    & B \arrow[r, "\beta_2"']   \arrow[d]                    & B''  &    \\
0 \arrow[r] & \coker \psi \arrow[r,hook] & \coker \psi \ds_{B'} B       & &         
\end{tikzcd}\] 
From this, we immediately see that the sequence \[ 0\to A'\to A\times_{A''} \ker \theta\to \ker \theta\to 0\]
and the dual statement with the cokernel are exact. Label the morphisms $\sigma:=(A\to A\times_{A''} \ker \theta), \gamma:=(\coker \psi \ds_{B'} B \to B''),$ and the composite morphism $\epsilon:=(A\times_{A''} \ker \theta\to  \coker \psi \ds_{B'} B).$ From the exactness of the rows in the above diagram, we get that $\gamma \epsilon=0$ and $\epsilon \sigma=0$ and thus $\epsilon$ factors through the cokernel of $\sigma$ and the kernel of $\gamma.$ As these two objects are $\ker \theta$ and $\coker \psi,$ define \[ \delta:\ker \theta\to \coker\psi\] as this morphism. 

This yields a sequence of morphisms \[ \ker \psi \to \ker \varphi \to \ker \theta \overset{\delta}{\longrightarrow} \coker \psi \to \coker \varphi\to \coker \theta\]
For all pairs of morphisms not involving $\delta,$ exactness follows immediately. For the remaining morphisms, note that it suffices to show that $\ker \varphi \to \ker \theta \to \coker \psi$ is exact as we can then dualize the argument to get the same result for the dual sequence. To show this, let $S\in \mathcal{A}$ and $\pi:S\to \ker\theta$ any morphism such that $\delta \pi=0.$ Form the pullback and adjoin it to the diagram as follows \[ 
\begin{tikzcd}
                                             & S_1 \arrow[d] \arrow[r, two heads]                                                      & S \arrow[d, "\pi"]                             \\
                                             & A\times_{A''}\ker \theta \arrow[d] \arrow[r, two heads] \arrow[ldd, dashed, shift left] & \ker \theta \arrow[llddd, "\delta", bend left] \\
A' \arrow[d, "\psi"'] \arrow[r, "\alpha_1"]  & A \arrow[d, "\varphi"]                                                                  &                                                \\
B' \arrow[r, "\beta_1"] \arrow[d, two heads] & B                                                                                       &                                                \\
\coker \psi                                  &                                                                                         &                                               
\end{tikzcd}\]
where the dashed morphism, call this $f$, exists by the fact that $A\times_{A''}\ker \theta \to B\to B''$ is the zero morphism. Now, the composition $S_1\to \coker \psi$ is $0$ and thus, we can find an epic morphism $S_0\twoheadrightarrow S_1$ such that the composition $S_0\to B'$ factors through $A'$. Denote by $g$ the morphism $S_0\to A'$.  Define the composite morphism $\lambda:S_0\to A$ and then consider \[ \lambda- f\comp k:S_0\to A\]
This must factor through $\ker \varphi$ by the commutativity of the diagram above. Hence, we get a commuting square \[ \begin{tikzcd}
	S_0\arrow[r,two heads] \arrow[d] & S_1 \arrow[d] \\
	\ker \varphi \arrow[r] \arrow[d] & \ker \theta \arrow[d]\\
	A\arrow[r,"\alpha_2"] & A''
\end{tikzcd}\]
The existence of this commuting diagram is equivalent to the exactness of the sequence $\ker \varphi \to \ker \theta \to \coker \psi.$ Dualizing this argument we get the exactness of the other morphisms. This completes the proof. 
\end{proof}
Now we can tie together adjoints and abelian categories. 
\begin{theorem}
	Let $F:\script{A}\rightleftarrows \script{B}:G$ be adjoint functors with $F\dashv G.$ Then $F$ is right exact and $G$ is left exact. Further $F(\varinjlim A_i)=\varinjlim F(A_i)$ and $G(\varprojlim A_i)=\varprojlim G(A_i).$ 
\end{theorem}
The proof relies on the Yoneda Embedding \cite{MacLane1971} which we will not cover. This theorem thus implies a stronger result than we stated before about $\Hom$ and $\tensor.$ 
\begin{corollary}
	$\Hom$ is left exact in both arguments and $\tensor$ is right exact in both arguments. 
\end{corollary}  

Therefore, given a short exact sequence of $R$-modules, the resulting sequences \[ 0\to \Hom(Y,A)\to \Hom(Y,B)\to \Hom(Y,C)\]\[ Y\tensor A\to Y\tensor B\to Y\tensor C\to 0\]
are exact. 
What we would like to understand is when $\Hom$ and $\tensor$ are exact everywhere. Thus, for the rest of this chapter, we shall assume we are working in $R$-modules. This may seem at first like we are becoming too specific to be of any use for category theory. The following theorem tells us that this is not correct. 
\begin{theorem}[Mitchell]
	Let $\script{A}$ be a small abelian category. Then there exists an exact, fully-faithful functor $\script{A}\to R$-\textbf{Mod} for some ring $R.$ 
\end{theorem}
See \cite{Rotman2009} for details.

\subsubsection{Projective, Injective, and Flat modules} 
\begin{definition}
	An $R$-module $P$ is \textbf{projective} if for every surjective map $M\to N$ and any map $P\to N$ there exists a map $M\to P$ making the following diagram commute \[\begin{tikzcd}
                                                      & P \arrow[d] \arrow[dl,swap, "\exists", dotted] &   \\
M \arrow[r, two heads]  & N \arrow[r] & 0
\end{tikzcd}\]

	Dually an $R$-module $I$ is \textbf{injective} if for every injective map $0\to L\to M$ and any morphism $L\to I$ there exists a  morphism making the following diagram commute: 
	\[\begin{tikzcd}
            & I                           &                                   \\
0 \arrow[r] & L \arrow[u] \arrow[r, hook] & M \arrow[lu, "\exists"', dotted]
\end{tikzcd}\]
\end{definition}
These definitions seem obtuse and out of nowhere. The following lemma makes them seem less so arbitrary. 
\begin{lemma}
	The functor $\Hom(P,-)$ is exact if and only if $P$ is projective. Also, the functor $\Hom(-,I)$ is exact if and only if $I$ is injective. 
\end{lemma}
\begin{proof}
	We prove the injective case and leave the projective one to the reader as it is the same argument. 
	$(\Rightarrow)$ Assume first that $\Hom(-,I)$ is exact. Then for any exact sequence of modules \[ 0\to A\to B\to C\to 0\]
	the sequence \[0\to \Hom(C,I)\to \Hom(B,I)\to \Hom(A,I)\to 0\]
	is exact. In particular, the map $ \Hom(B,I)\to \Hom(A,I)$ is surjective. Being surjective means that for any morphism $\varphi:A\to I, $ there exists a morphism $\hat{\varphi}:B\to I$ which makes the diagram above commute. This is precisely the definition of $I$ being injective. 
	
	Now assume $I$ is injective. Then we have a surjective map $\pi: \Hom(B,I)\to \Hom(A,I)$ by definition. For any $f\in \Hom(A,I)$ the definition tells us that $f=i^*(g)$ for some $g\in \Hom(B,I).$ Hence, $\pi=i^*$ and $\Hom(-,I)$ is exact. This completes the proof. 
\end{proof}

\begin{definition}
	A module is called $\textbf{flat}$ if $-\tensor_R M$ is exact. Moreover, every projective module is flat \cite{Rotman2009}.  
\end{definition}

For a given module $M,$ we want to understand how far $M$ is from being projective, injective, or flat. Clearly the functors $\Hom$ and $\tensor$ will not tell us this information. What they imply is that $M$ is simply not flat (projective, injective). To remedy this, we will find a \textit{free resolution} of $M$ which is quasi-isomorphic to $M$ so that we can measure how far $M$ is from being one of the special modules above. 
\begin{definition}
	A \textbf{free resolution} of an $R$-module $M$ is an exact sequence $F_\bullet\to M\to 0.$ That is, a collection of free modules $F_i$ and morphisms $\alpha_i$ so that \[ ...\to F_1\to F_0\to M\to 0\] is exact (hereby realizing $M$ as the cokernel of the map $F_1\to F_0$). If every $F_i$ is projective (resp. flat) then $F_\bullet$ is a projective (resp. flat) resolution of $M.$ As injective modules are dual to projective ones, we have that an \textbf{injective resolution} of $M$ is an exact sequence $0\to M\to I^\bullet.$    
\end{definition}  
We care about these resolutions because if we look at the quotients $\ker \alpha_i/\im \alpha_{i+1}=0$ for all $i>0.$ If we truncate the sequence and only consider up to $F_1.$ Then the cokernel of $\alpha_1=M.$ Therefore, this sequence is in some sense no different from $M$ itself. The next part goes into more detail about this. 

\subsubsection{Derived Functors} 
For a general abelian category, we have the notion of short exact sequences. In addition to this, we have the notion of \textit{(co)chain complexes}.  These will be the central objects we want to consider when answering the questions posed in the previous section. 

\begin{definition}
	Let $(C_\bullet,d_\bullet)$ be a collection of objects in an abelian category $\script{A}$ together with a morphism $d_n:C_n\to C_{n-1}.$ We call $(C_\bullet,d)$ a \textbf{chain complex} if $d_{n-1} \comp d_n=0.$ If instead we have an object $(C^\bullet,\partial^\bullet)$ such that $\partial^n:C^n\to C^{n+1}$ such that $\partial^{n+1} \comp \partial^n=0$ then we say the pair is a \textbf{cochain complex}. It is common practice to drop the index on the differential $d_\bullet$ or $\partial^\bullet$ and simply denote them $d$ and $\partial.$ We shall adopt this convention.      
\end{definition}    

A morphism of (co)chain complexes $(C_\bullet,d)$ and $(D_\bullet,d')$ is a \textbf{chain map} $f_\bullet$ (resp. $f^\bullet$), that is a collection of maps $f_i$ so that the following diagram commutes for all $n,$ \[ 
\begin{tikzcd}
	C_n\arrow[r,"d"]  \arrow[d,swap, "f_{n}"] & C_{n-1} \arrow[d,"f_{n-1}"] \\ D_n \arrow[r,swap, "d'"] & D_{n-1} 
\end{tikzcd}
\]       

With this notion of morphism, we can build a new category $(\textbf{c})\textbf{Ch}(\script{A})$ of (co)chain complexes. Notice that because of the condition $d^2=0,$ we have that $\im d_n\subseteq \ker d_{n-1}.$ 
\begin{definition}
	Let $(C_\bullet,d)$ be a chain complex. Define the \textbf{$n-$th homology groups} of $C_\bullet$ as \[ H_n(C_\bullet)=\ker d_n/\im d_{n+1}\]
These are in fact groups as shown in \cite{Rotman2009}. \\

	Two chain complexes are \textbf{quasi-isomorphic} if there exists a chain map $f_\bullet:C_\bullet\to D_\bullet$ such that  $(f_i)_*:H_i(C_\bullet)\overset{\sim}{\to} H_i(D_\bullet)$ where $(f_i)_*$ is defined as $[\alpha]\mapsto [f_i \comp \alpha].$ This is well-defined by the definition of a chain map. Further $f_i\comp \alpha\in \ker d_{n}'.$ We have completely analogously the definition of \textbf{cohomology groups} $H^i(C^\bullet).$ We call a (co)chain complex is \textit{exact} if all of the (co)homology groups are identiically $0.$        
\end{definition}

We now return to the content of the previous section. Let $M$ be an $R$-module and $P_\bullet$ a projective resolution of $M.$ 

It then follows from the discussion above that $P_\bullet\to 0$ (truncating the free resolution of $M$ at $P_0$) and $M$ are quasi-isomorphic as chain complexes (here $M$ is considered as the trivial chain complex with differential $0$ everywhere). We can use this to our advantage. 
 For any $R$-module $A,$ consider $\Hom(-,A).$ The resulting cochain complex \[ \Hom(P_0,A)\to \Hom(P_1,A)\to \Hom(P_2,A)\to ...\]
is no longer exact.  
\begin{definition}
	The \textbf{n-th cohomology groups} or $n$-th $\Ext$ groups of $M$ and are denoted \[ \Ext_R^n(M,A):=H^i(\Hom(P_n,A))\] \end{definition}
\begin{remark}
	It can be shown \cite{Hilton1997} that these groups do not depend on the resolution taken. In fact, it does not even matter if we resolve $A$ or $M.$ There is a dual construction of $\Ext^n(M,A)$ where instead of a projective resolution of $M,$ we take an injective resolution of $A.$  
\end{remark}
For $\tensor,$ we have the corresponding construction but now we only use projective resolutions as $\tensor$ is covariant in both arguments. 
\begin{definition}
	The \textbf{n-th homology groups} or $n$-th $\Tor$ groups of $M$ are \[ \Tor_n^R(M,A):=H_i(P_n\tensor A) \]
\end{definition}
We now generalize to arbitrary abelian categories. 
\begin{definition}
	An abelian category is said to \textbf{have enough projectives} if every element has a projective resolution (respectively, enough injectives and enough flats)
\end{definition}
Let $\script{A}$ be an abelian category with enough projectives and $F:\script{A}\to \script{B}$ be a right exact functor. Then for any projective resolution of an object $M,$ we can repeat the operation above to define the \textit{derived functors} of $F.$ To be more specific, let $P_\bullet$ be a projective resolution of $M.$ 
\begin{definition}
	The functors \[ L_iF(M)=\ker (FP_n\to FP_{n-1})/\im (FP_{n+1}\to FP_n) \] are called the \textbf{left derived functors} of $F.$ Dually if $G$ is left exact and $I^\bullet$ is an injective resolution, we can define $R^iG$ as the \textbf{right derived functors} for $G.$ 
\end{definition}    
One may ask why we do not consider the left derived functors for a left exact functor. The answer to this is that these are all zero, or at least un interesting. They tell you nothing about exactness as $0$s appear in the sequences. 
\begin{proposition}
	If $F$ is exact then $R^iF$ and $L_iF$ are $0$ for all $i>0.$ 
\end{proposition}
\begin{proof}
	As $F$ is exact, the resulting long sequences are exact. Hence, the quotient groups are $0$ and $R^iF$ (resp. $L_iF$) is $0.$
\end{proof}
\begin{remark}
	The derived functors measure the extent to which $M$ is not projective, injective, or flat. More generally, they measure how far $F$ is from being exact. If $R^iF$ is non-zero for only very large $i,$ then $F$ is very close to being exact. Whereas if $R^2F$ is non-zero, then $F$ is nowhere close to being exact.
\end{remark} 
The final theorem we present in this section is the most useful for computing these functors. 
\begin{theorem}
	Let $0\to A\to B\to C\to 0$ be exact in $\script{A}$ and $F:\script{A}\to \script{B}$ be a right exact functor. Then there is a long exact sequence \[ ... L_iF(A)\to L_iF(B)\to L_iF(C)\to L_{i-1}F(A) \to ...\]
	in the derived functors. The same is true for left exact functors. 
\end{theorem}
The proof of this is immediate from the Snake Lemma \ref{Snake_Lemma}.  The reason it is so important is because if we know that either $A,B,$ or $C$ is $F$-acyclic (that is $L_iF(C)=0$) then we get isomorphisms of the remaining groups! This single fact underlies most of homological algebra and will be integral in section 3.3.2. 
This completes this brisk tour of category theory. 
\newpage

\section{Topology}
We shall depart from category theory for the time being and return to it in section 3.2.2. For the meantime, we shall introduce the second major topic of this chapter: topological spaces. The purpose of these objects is to formalize the somewhat colloquial notions of connectedness, compactness, and other concepts. The culmination of all of this will be to define and give some basic properties of singular homology groups for a topological space. This concept will prove contentious in chapter 4 as some researchers have recently proposed using homology to discover geometric properties of the perceptual space.  

\subsection{Topological Spaces and Continuous Maps} 
The story of topology starts with the definition of a topological space. Before we give this though, we want to motivate the study of such objects by looking at the familiar case of $\R^n$ and in particular $\R.$ In high-school algebra, we call sets of the form $(a,b)$ open and $[a,b]$ closed. Similarly, sets of the form $B_r(p)=\{ x\in \R^n: |x-p|<r\}$ are open in $\R^n$ and if we change $<$ to $\leq,$ we get closed sets. In fact, we can have arbitrary open sets in $\R^n$ but all of them are built out of sets of the form above. We want to generalize all of this and formalize what we mean by \textit{open} and \textit{closed}. Some good references for this section are \cite{Lee2011}, \cite{Munkres2000}, and \cite{Fuchs2016}. The last of which is a fairly recent and thorough treatment of the material in Section 3.2.2. 

\begin{definition}
	Let $X$ be a set. A \textbf{topology} on $X$ is a collection of subsets $\mathcal{T}\subseteq \mathcal{P}(X)$ the power set, subject to the following conditions: 
	\begin{enumerate}
		\item $\varnothing,X \in \mathcal{T}.$
		\item $\mathcal{T}$ is closed under arbitrary union. That is if $\{U_i\}_{i\in I}$ is a collection of elements of $\mathcal{T}$ with $|I|$ arbitatry, then \[ \bigcup_{i\in I} U_i\in \mathcal{T}\]
		\item $\mathcal{T}$ is closed under finite intersections. That is if $\{U_i\}_{i\in I}$ is a collection of elements of $\mathcal{T}$ with $|I|<\infty$, then \[ \bigcap_{i\in I} U_i\in \mathcal{T}\]
	\end{enumerate}   
	Elements of the topology are called \textbf{open} sets. A subset $V\subseteq X$ is called \textbf{closed} if $X-V\in \mathcal{T}.$
	A set equipped with a topology is called a \textbf{topological space}.   
\end{definition} 
Notice that open and closed are not mutually exclusive: $X$ is always closed and open (sometimes abbreviated to \textit{clopen}) and some sets, such as $[0,1)$ in $\R$ are both not closed and not open. Further, simply because a set is not open does not imply closure. 

\begin{example}\label{topologies}
	\item For any set, we can give it the \textbf{discrete topology} where every subset is declared open. Dually, we can define the trivial topology in which only $\varnothing$ and $X$ are open. 
	\item For any subset $A\subseteq (X,\mathcal{T}),$ we can topologize $A$ by taking the open sets to be \[A\cap \mathcal{T}:=\{ A\cap U : U\in \mathcal{T}\}\]
	This is called the \textbf{subspace topology}. 
	\item The topology generated by the open balls in $\R^n$ above is called the \textbf{standard topology} on $\R^n.$
\end{example} 

We want to formalize the final example above. That is we want to answer the question: \textit{what does it mean to generate a topology?}
Similar to a basis for a vector space, we want to define an analogous object for a topology. 

\begin{definition}
	Let $X$ be a topological space with topology $\mathcal{T}.$ Then a collection of subsets, $\script{B}$, of $X$ is called a basis for the topology $\mathcal{T}$ if the following conditions are satisfied: 
	\begin{enumerate}
		\item Every $B\in \script{B}$ is open in $X.$
		\item Every open set $U\in \mathcal{T}$ can be written as a union of some collection of elements of $\script{B}.$ 
	\end{enumerate}
\end{definition}

It should now be clear that the standard. topology on $\R^n$ is the topology with basis consisting of the open balls. Now that we have this definition, we want to understand when it is applicable. Further, what conditions on a collection of subsets of a topological space make it a basis? The following proposition answers this in full. 

\begin{proposition}\label{Basis} 
	Let $X$ be a set and $\script{B}$ a collection of subsets. Then $\script{B}$ is a basis of a topology on $X$ if and only if the following conditions are satisfied: \begin{enumerate} 
		\item $\bigcup_{B\in \script{B}} B=X$
		\item For every $B_1,B_2\in \script{B}$, $B_1\cap B_2\in \script{B}$ and if $B_1\cap B_2\neq \varnothing,$ there exists $B_3\in \mathscr{B}$ such that $B_3\subseteq B_1\cap B_2.$ 
	\end{enumerate}
	In fact, this topology is the unique topology generated by $\script{B}.$ 
\end{proposition}
\begin{proof}
	Suppose $\script{B}$ is a basis. Then (a) is satisfied immediately as every open set is a union of basis elements and $X$ is open in any topology. For (b), as. $B_1$ and $B_2$ are open, $B_1\cap B_2$ is open. Therefore we can write \[ B_1\cap B_2=\bigcup B_i\]
	where $B_i\in \script{B}$ are basis elements. Pick any of the $B_i$ to satisfy (b). 
	
	For the reverse. direction, we need to show that the conditions above imply that $\mathcal{T}_\script{B}$ is indeed a topology on $X.$ By the (a), $X,\varnothing \in \mathcal{T}_\script{B}.$ Let $\{U_i\}$ be an arbitrary collection of open sets.  Then each $U_i=\bigcup_{j\in J_i} B^i_j.$ with each $B^i_j\in \script{B}.$ Then \[ \bigcup U_i=\bigcup_{I} \bigcup_{J_i} B^i_{j} \]
	So $\mathcal{T}_\script{B}$ is closed under arbitrary unions. To show it is. closed under finite intersection, let $U_1,U_2\in \mathcal{T}_{\script{B}}.$ Then for every $x\in U_1\cap U_2,$ there exists some $B_1\subseteq U_1$ and $B_2\subseteq U_2$ such that $x\in B_1\cap B_2.$ By condition (b), we know there exists some $B_3$ such that $x\in B_3\subseteq B_1\cap B_2\subseteq U_1\cap U_2.$ 
	Then $U_1\cap U_2$ is a union of each of these basis elements as $x$ varies and hence is open. Therefore $\mathcal{T}_\script{B}$ is closed under pairwise intersection and by induction, all finite intersections. Hence $\mathcal{T}_\script{B}$ is a topology on $X.$ 
	Uniqueness follows immediately from the definition of a basis. This completes the proof. 
\end{proof} 

This proposition says that it suffices to define a topology by giving a basis. In Section 3.3, we will use this to topologize manifolds in a unique way so that they are sufficiently nice. 

We need to step back a bit and think about how we topologize $\R^n.$ We have given a basis for some topology on $\R^n$ above. What if we want to build a topology on $\R^n$ out of the topologies on $\R.$ To answer this, we shall generalize to the notion of product topology. 
\begin{definition}
	Let $\{X_\alpha\}_{\alpha\in J}$ be a $J$-indexed family of topological spaces. As a basis for some topology on the product space $\prod_J X_\alpha,$ we have the sets of the form \[ \prod U_\beta \]
	where $U_\beta$ is open in $X_\beta$ and $U_\beta=X_\beta$ for all but finitely many $\beta\in J.$ This topology is called the \textbf{product topology}.  
\end{definition}
There is a naive topology on the product which removes the final condition that $U_\beta=X_\beta$ for all but finitely many $\beta$. This is called the \textit{box topology}. In the case of $J$ finite, these are equivalent. It is generally less useful than the product topology as it is too fine; that is too many sets are open. For this reason, whenever we have a product space, we assume it has the product topology.

In $\R^n,$ it is relatively easy to distinguish whether or not a point lies within a given set. For a general topological space, this is daunting as the topology may be particularly bad. We need to generalize the above notion arbitrary spaces so that we can speak of boundaries of sets. To be more formal, let $X\subseteq Y$ be topological spaces.  We say that $x\in \operatorname{Int}(X)$ the \textbf{interior} of $X$ if there exists an open set $U\subsetneq X$ such that $x\in U.$ The boundary of $X$, denoted $\partial X$ is the collection of points $\{y\}$ such that for any open set $P$ containing $y,$ $P\cap X$ is non-trivial. We define the \textbf{closure} of $X$ to be \[ \overline{X}=\operatorname{Int}(X)\cup \partial X\]
It should be noted that this only makes sense for topological subspaces. More generally it makes sense in the context of embeddings (see Example  \ref{Continuous_Examples} below).  

\begin{proposition}
	Let $X$ be a topological space and $A$ a subspace. Then $\operatorname{Int}(A)$ is open, $\partial{A}$ is closed, and $\overline{A}$ is closed.
\end{proposition}
\begin{proof} 
	For each point $x\in \operatorname{Int}(A)$ let $U_x\subseteq \operatorname{Int}(A)$ be an open set containing $x.$ Then $\operatorname{Int}(A)$ is. the union of these $U_x$ and is thus open. Consider $X- \partial A$ we wish to show that this is open. From the definition, \[ X-\partial A=\operatorname{Int}(A)\cup (X-\overline{A}) \]
	Therefore, it suffices to show that $X-\overline{A}$ is open. Let $p\in  X-\overline{A}.$ As $p\notin \overline{A},$ there exists some $V\subseteq X$ open such that $V\cap \overline{A}=\varnothing$ and $p\in V.$ As $X-\overline{A}$ is a union of these open sets, it is open. This completes the proof. 
\end{proof} 
It should be clear now that a set is open (resp. closed) if and only if $A=\operatorname{Int}(A)$ (resp. $A=\overline{A}$).  

Now that we have the notions of topologies and bases, we can give a general definition of continuity. 
\begin{definition}
	Let $f:X\to Y$ be a function between topological spaces. We call $f$ \textbf{continuous} if for all $V\subseteq Y$ open, $f^{-1}(V)$ is open in $X.$ We call a map \textbf{open} if for all $U\subseteq X$ open, $f(U)$ is open in $Y.$ 
\end{definition}

Together with continuous maps, topological spaces define a category denoted $\textbf{Top}.$ If we add an additional stipulation that every space be given a distinguished point, then we can define the category $\textbf{Top}_*$ of pointed topological spaces and base-point preserving maps. 

The following examples of continuous maps are all fun exercises to the reader. They are incredibly important for later parts of this chapter. 

\begin{example}\label{Continuous_Examples}\text{}
\begin{enumerate}
	\item Let $(X,x_0)$ and $(Y,y_0)$ be pointed topological spaces. Then the constant map $x\mapsto y_0$ is continuous. 
	\item Let $f:X\to Y$ be a continuous map. Then for any subspace $A\subseteq X,$ the restriction map $f|_A:A\to Y$ is also continuous. 
	\item Let $f:X\to Y$ be a continuous map, and denote the image by $f(X).$ Then for any subspace $Z\subseteq Y$ with $f(X)\subseteq Z,$ the map $f^Z:X\to Z$ is continuous. 
	\item The composition of continuous maps is continuous. 
	\item Any inclusion map is continuous. That is, if $A\subseteq X$ then there exists a map $A\hookrightarrow X$ and this map is continuous. In general an injective continuous map is called a \textbf{topological embedding} if it is a homeomorphism onto its image. 
\end{enumerate}
\end{example}

Notice that the definition of continuity pays no mind to closed subsets. Could we possibly get a different definition if we replace open with closed in the definition? The following lemma gives a negative answer. 
\begin{lemma}
	A function $f:X\to Y$  is continuous if and only if for all closed subsets $V\subseteq Y$, $f^{-1}(V)$ is closed in $X.$ 
\end{lemma}
\begin{proof}
	$(\Leftarrow)$ Let $B$ be a closed set in $Y$ and $C$ its complement. By definition it is open. We want to show that $f^{-1}(C)$ is open in $X.$ Consider \[f^{-1}(B)=f^{-1}(Y)-f^{-1}(C)=X-f^{-1}(C)\]
	As $f^{-1}(B)$ is closed in $X,$ we conclude that $f^{-1}(C)$ is open. \\\\
	$(\Rightarrow)$ Let $B$ be closed in $Y.$ We need to show that $f^{-1}(B)$ is closed. We need to show that $\overline{f^{-1}(B) }=f^{-1}(B).$ Let $x\in \overline{f^{-1}(B) }.$ Then \[ f(x)\in f(\overline{f^{-1}(B)})\subseteq \bar{B}=B\]
	where the inclusion follows from continuity. 
	Therefore, $x\in f^{-1}(B)$ and $\overline{f^{-1}(B) }\subseteq f^{-1}(B).$ Hence, $f^{-1}(B)$ is closed. 
\end{proof}
\noindent Therefore, defining continuity in terms of closed sets is equivalent to defining it in terms of open sets. 

We now want to define quotient objects in $\textbf{Top}.$ Let $A$ be a subspace of a topological space $X.$ Then we define an equivalence relation on $X$ as $x\sim y$ if $x,y\in A.$ Then we have the quotient space $X/\sim$ which is also written $X/A$. We want to topolgize $X/A$ in a way which makes the canonical map $X\to X/A$ continuous. 
\begin{definition}
	The \textbf{quotient topology} is defined as the coarsest topology for which the canonical morphism $\pi:X\to X/A$ is continuous. Equivalently, $P\subseteq X/A$ is open if and only if $\pi^{-1}(P)$ is open in $X.$ 
\end{definition}

\begin{remark} This will allow us to give a topological structure to the Generalized Categories from Chapter 1 and give a coarse categorization from the perceptual space.  
\end{remark}

The quotient topology can be particularly opaque as it depends entirely on $X$ and $A.$ To give some idea of how it can manifest, lets give some examples of quotient spaces: 

\begin{example}
	\begin{enumerate}
		\item Let $S^1:=\{x\in \C: |x|=1\}.$ Then consider the subspace $\{-1,1\}.$ It then turns out that $S^1/\{-1,1\}$ is equivalent to two circles which touch at a single point. The topology of this space is then inherited from its embedding into $\C.$ Therefore, the quotient topology in this case is easy to see. 
		\item Consider $\Z\hookrightarrow \R.$ Then $\R/\Z$ is equivalent to the interval $[0,1]$ with the identification of $0\sim 1.$ Hence, the quotient space is $S^1.$ 
\end{enumerate}
\end{example}

What do we mean here by "equivalent?" We claimed above that $\textbf{Top}$ is a category and thus equivalent should mean an isomorphism. What are the isomorphisms in this category? 

\begin{definition}
	Let $f:X\to Y$ be a continuous map. We call $f$ a \textbf{homeomorphism} if there exists $g:Y\to X$ such that $g\comp f=1_X$ and $f\comp g=1_Y.$ Notice that every homeomorphism is necessarily a bijection. 
\end{definition}

If we consider "spaces up to homeomorphism" this is an equivalence relation. That is, we can think of isomorphism classes of topological spaces. This is a large area of research for say curves and surfaces. Before we move on to other general topological properties, we shall give some generic properties of homeomorphisms. 

\begin{theorem}
	Let $f:X\to Y$ be a bijective function between topological spaces. Then $f$ is a homeomorphism if and only if $f(\mathcal{T}_X)=\mathcal{T}_Y$. Further if $f$ is a homeomorphism, then $f$ is an open map.   
\end{theorem} 
\begin{proof}
	Notice that the second statement follows immediately from the first.\\
	$(\Rightarrow)$ Let $U\in \mathcal{T}_X.$ Then \[f(U)=(f^{-1})^{-1}(U)\] As $f$ is a homeomorphism, $f^{-1}$ is continuous and so $f(U)$ is open in $Y$ and thus $f(U)\in \mathcal{T}_Y.$ Therefore, we have an injection $f(\mathcal{T}_X)\hookrightarrow \mathcal{T}_Y.$ This map is surjective as $f$ is continuous. Thus, $f(\mathcal{T}_X)=\mathcal{T}_Y.$ \\\\
	$(\Leftarrow)$ Assume now that $f(\mathcal{T}_X)=\mathcal{T}_Y.$ $f$ is continuous as for any $V\in \mathcal{T}_Y,$ \[ f\comp f^{-1}(V)=V\] and therefore $f^{-1}(V)\in \mathcal{T}_X.$ Similarly, $f^{-1}$ is thus continuous. This completes the proof.   
\end{proof}
\begin{example}
	Some classic examples of homeomorphisms are translations and dilations of $\R^n.$ These are maps of the from $f(x)=x+\lambda$ and $f(x)=cx$ for some $\lambda\in \R^n$ and $c\in \R.$ More importantly, let $V,W$ be finite dimensional vector spaces. Then any linear map $V\to W$ is necessarily continuous. In fact, as we will see in the next section, these maps are smooth!
\end{example}

\begin{example}
	We end this subsection with an interesting example of topological spaces. Let $G$ be a group. Then we call $G$ a \textbf{topological group} if multiplication and inversion are continuous maps. A morphism of topological groups is a continuous group homomorphism.  
\end{example}

\subsubsection{Connectedness, Compactness, and Hausdorff} 
Now we give some characterizations of certain topological spaces. The properties are important for many mathematical applications and will be intrinsically important for for the next section and chapter 4. We shall do them all in one pass and then go into some detail about their relationships to each other.  

\begin{definition} Let $X$ be a topological space.
\begin{enumerate}
	\item $X$ is \textbf{connected} if there do not exist open sets $U_1,U_2$ such that $U_1\cap U_2=\varnothing$ and $U_1\cup U_2=X.$  
	\item $X$ is \textbf{compact} if for every open cover $\mathcal{U}$ of $X$ there exists a finite subcover. An open cover of a topological space is a collection of open sets $\mathcal{U}=\{U_i\}$ such that $X\subseteq\bigcup U_i.$  
	
	\item If $X$ is non-empty an contains at least two elements, then $X$ is \textbf{Hausdorff} if for any two distinct\footnote{Distinct here means that there exists some open set about $x$ which does not contain $y$. Spaces with this property are sometimes called Kolmogorov}  points, $x,y\in X,$ there exists open sets $U_x,U_y\subseteq X,$ such that $x\in U_x,$ $y\in U_y$, and  $U_x\cap U_y=\varnothing.$  
\end{enumerate} 
\end{definition}
These notions depend highly on the topology of $X.$ For instance, every space is connected (resp. compact) if equipped with the trivial topology and every space is disconnected (resp. non-compact) if it is equipped with the discrete topology. 
In general, a space is not connected but can be broken up into connected components. This partitions the set into distinct subsets which can be of great use. There is another notion of connectedness which is slightly stronger.  

\begin{definition}
	A topological space $X$ is path-connected if for each pair of points $a,b\in X,$ there exists a continuous path $\gamma:[0,1]\to X$ such that $\gamma(0)=a$ and $\gamma(1)=b.$ 
\end{definition}
\begin{proposition}
	If $X$ is path-connected, then $X$ is connected. 
\end{proposition}
\begin{proof}
	Assume for the sake of contradiction that $X$ is disconnected. Let $X=U\cup V$ with $U\cap V=\varnothing.$ Let $a\in U,$ $b\in V,$ and $\gamma$ a path between them. Then $\gamma^{-1}(X)=\gamma^{-1}(U)\cup \gamma^{-1}(V).$ This implies that $[0,1]$ is. disconnected which is a contradiction. Hence, $X$ is connected.  
\end{proof} 
This proposition proves our assertion from before that path-connectedness is a stronger condition that connectedness. In fact, there are some highly non-trivial examples where the converse is not true.  

\begin{example}
	Let $X$ be the space of lines in $\R^2$ connecting the origin to the points $(1,\frac{1}{n})$, together with the point $(1,0)$ (note this does not include the line segment $(0,0)\to (1,0)$.)  Then $X$ is connected, but not path-connected. See Figure 3.1.   
\end{example} 

Similar to connected components, we can define \textit{path}-connected components. For a topological space $X,$ denote the set of path-connected components by $\pi_0(X).$ 

\begin{figure}[!htb]
    \centering
    \includegraphics[width=.8 \linewidth]{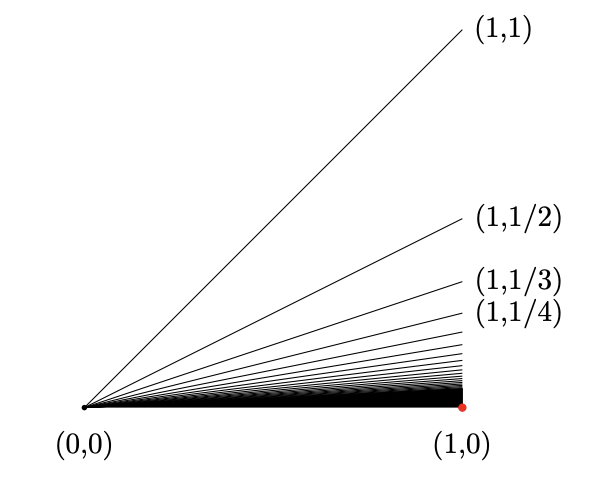}
    \caption{The Witches Broom. An example of a connected but not path-connected topological space. It is the union of all line segments $[(0,0),(1,\frac{1}{n})] \cup \{(1,0)\}.$ }
    \label{fig:wrapfig}
\end{figure}

\noindent We want to understand how each of these notions interacts with (1) each other and (2) continuous maps.  

Let us investigate (2) first. 

\begin{theorem}
	Let $f:X\to Y$ be a continuous function. If $X$ is connected (resp. compact), then so is $f(X).$  
\end{theorem}   
\begin{proof}
	Let $X$ be connected. Assume for the sake of contradiction that $f(X)$ is disconnected. Then, let $f(X)=A\cup B.$ Each of these is open in $Y$ and thus $f^{-1}(A)$ and $f^{-1}(B)$ is open in $X.$ Further, $f^{-1}(A)\cup f^{-1}(B)=X.$ This contradicts the connectedness of $X$. Hence, $f(X)$ is connected. 
	
	Now assume $X$ is compact. Let $\mathcal{V}$ be an open cover for $f(X).$ Then \[ X=\bigcup_{V_i\in \mathcal{V}} f^{-1}(V_i) \]
	is an open cover. As $X$ is compact, there exist finitely many $V_i$ such that $X=\bigcup^n f^{-1}(V_i).$ Therefore $V_1,...,V_n$ are is a finite open subcover of $\mathcal{V}.$ Hence, $f(X)$ is compact. This completes the proof.     
\end{proof}

This theorem is highly important to \textit{any} field of mathematics that concerns itself with topologies of any kind. As it turns out, many theorems only work for compact spaces. So, knowing that compactness is preserved under continuous maps is crucial. Let's understand compact sets a bit better.

\begin{proposition} Let $X$  be a compact space. 
	\begin{enumerate}
		\item If $A\subseteq X$ is closed, then $A$ is compact. 
		\item If $X\subseteq Y$ a Hausdorff space, then $X$ is closed in $Y.$
	\end{enumerate}
\end{proposition} 
\begin{proof}
	(a) Let $\mathcal{A}$ be an open cover of $A.$ As $A$ is closed, $A^c=X- A$ is open in $X$ and $\mathcal{A}\cup A^c$ is an open cover for $X.$ As $X$ is compact, there exists a finite subcover. If this resulting subcover contains $A^c,$ discard it. Else, this is a finite cover of $A.$ This proves (a). \\
	
	(b) Let $y\in X^c.$ We want to construct an open set $V$ containing $y$ such that $V\cap X=\varnothing.$ As $Y$ is Hausdorff, for every $x\in X,$ there exists disjoint open sets $U_x$ and $\widetilde{U_x}$ such that $x\in U_x$ and $y\in \widetilde{U_x}.$ Then $\bigcup_{x\in X} U_x$ is an open cover of $X$. By compactness, there is a finite collection of points $\{x_i\}$ such that $X=\bigcup_{i} U_{x_i}.$ Put \[ V=\bigcap_{i} \widetilde{U_{x_i}}\]
	This is open and disjoint from $X$ by construction. Hence, $X^c$ is open and thus $X$ is closed. 
\end{proof}

In a similar theme to topologies, we would like to know how connectedness, compactness, and Hausdorff-ness interact with products. 
 
\begin{proposition}
	Let $\{X_i\}$ be a family of connected (resp. Hausdorff) spaces. Then $\prod X_i$ is connected (resp. Hausdorff). 
\end{proposition}
\noindent We leave the proof of this proposition as an exercise to the reader as it follows entirely from the definitions. 

For compactness, there are two results and both are surprising. 
\begin{theorem}[Heine-Borel]
	A subset of $\R^n$ is compact if and only if it is closed and bounded. 
\end{theorem}

\begin{theorem}[Tychonoff]
	Let $\{X_i\}$ be an arbitrary collection of compact spaces. Then $\prod X_i$ is compact. 
\end{theorem}
\noindent Although we shall not prove this, it is interesting to know that this theorem is equivalent to the axiom of choice as its proofs rely entirely on Zorn's Lemma. This is arguably the most important theorem in all of point-set topology. For a proof of both theorems see \cite{Munkres2000}. 

\subsubsection{Metric Spaces}
We now give a brief introduction to metric spaces which will allow us to formally discuss "perceptual metrics" in chapter 4. 
\begin{definition}
Let $X$ be a set. A \textbf{metric} on $X$ is a function $d:X\times X\to \R_{\geq 0}\cup \{\infty\} $ such that \begin{enumerate}
	\item For $x,y\in X,$ $d(x,y)=d(y,x). $
	\item For all $x,y\in X,$ $d(x,y)=0\iff x=y.$
	\item For all $x,y,z\in X,$ $d(x,z)\leq d(x,y)+d(y,z).$ 
\end{enumerate}  
The final condition is called the triangle inequality and is the defining characteristic of metrics. The set $(X,d)$ is called a metric space. A function $f:(X,d)\to (Y,g)$ between metric spaces is called a \textbf{metric map} if \[ g(f(x),f(y))\leq d(x,y)\] If equality holds for all $x,y,$ then $f$ is called an \textbf{isometry}. The collection of all metric spaces and all metric maps forms a category denoted $\textbf{Met}.$ 
\end{definition}

\begin{theorem}
	Let $(X,d)$ be a metric space. Then $d$ induces a topology on $X$ (called the \textbf{metric topology}). This gives a faithful functor \[ \textbf{Met}\hookrightarrow \textbf{Top}\]
	The image is the category of \textit{metrizable spaces} (those which are homeomorphic to metric spaces).    
\end{theorem}
\begin{proof} 
	Let $x\in X$ and put  \[ B_r(x):=\{y\in X: d(x,y)<r\}\]
	Let $\script{B}$ be the collection of all such balls for all points $x\in X.$ We claim that $\script{B}$ is a basis. It suffices to check the conditions of Proposition \ref{Basis}. Clearly, $X=\bigcup_{B\in \script{B}} B.$ Let $B_r(x)$ and $B_{r'}(x')$ be two elements of $\script{B}$ such that $B_r(x)\cap B_{r'}(x')\neq \varnothing.$ By the triangle inequality, for any $y\in B_r(x)\cap B_{r'}(x')$ we can find $\delta_1<r$ and $\delta_2<r'$ such that  $B_{\delta_1}(y)\subseteq  B_r(x)$ and $B_{\delta_2}(y)\subseteq B_{r'}(x').$ Pick $\delta=\min\{ \delta_1,\delta_2\}.$ Then $B_\delta(y)$ is contained in the intersection. Hence, $\script{B}$ is a basis for a topology on $X.$ 
	
	The functor $\textbf{Met}\to \textbf{Top}$ is precisely the forgetful functor which sends $(X,d,\mathcal{T})\to (X,\mathcal{T}).$   
\end{proof}

Now consider $\{x_n\}$ a sequence of points in a metric space. We say that $\{x_n\}$ converges to a point $x$ if $d(x,x_n)\to 0$ as $n\to \infty.$ A \textbf{Cauchy Sequence} is a sequence $\{x_n\}$ such that there exists $n^*$ and for all $m,n>n^*,$ $d(x_m,x_n)<\epsilon$ for any $\epsilon>0.$ 
\begin{definition}
	We call a metric space \textbf{complete} if every Cauchy sequence converges.   
\end{definition} 
\begin{theorem}
	Let $(X,d)$ be a metric space. Then there exists a metric space $(\widehat{X},d)$ such that $\widehat{X}$ is complete with respect to $d$ and there is a map $X\to \widehat{X}.$ which is an isometry with dense\footnote{Given a topological space $X$ and a subspace $A,$ we call $A$ \textit{dense} in $X$ if $\overline{A}=X.$} image.  
\end{theorem}
\noindent See \cite{Knapp2005a} for a full proof of this statement. 

With this theorem in mind, we want to give the definitions/exmaples of some complete metric spaces and how they arise. 

\begin{dexample}
	Let $V$ be a $\C$-vector space. A \textbf{norm} on $V$ is a map $V\to \C$ such that \begin{enumerate}
		\item $||x||\geq 0$ with equality if and only if $x=0.$
		\item $||ax||=|a|\cdot ||x||$ for all $a\in \C.$
		\item $||x+y||\leq ||x||+||y||$ for all $x,y\in V.$ 
	\end{enumerate} 
	Clearly, a norm induces a metric $d(x,y)=||x-y||$ on $V.$ We call $V$ a \textbf{Banach} space if $(V,d)$ is complete. \\
	
	Similarly, we can define a $\textbf{hermitian inner product}$ on $V$ as a sesquilinear (one and a half linear) map $\ip{,}:V\to \C$ such that \[ \ip{x,y}=\overline{\ip{y,x}}\] 
	where $\overline{\ip{,}}$ denotes the complex conjugate. 
	
	This defines a norm and hence a metric on $V.$ If $(V,d)$ is complete with respect to this metric, then $V$ is a \textbf{Hilbert Space}. These are some of the most important spaces for Harmonic analysis and representation theory. We shall use Banach spaces and tensor products to understand manifolds better in Section 3.3. 
\end{dexample}
	
\subsection{Basic Algebraic Topology} 
In this section, we shall introduce a different approach to topology which considers a weaker form of equivalence but focuses on algebraic invariants attached to topological spaces. The main references for this subsection are \cite{Hatcher2001}, \cite{Rotman1988}, and \cite{Fuchs2016}. We start with the notion of homotopy. 

\begin{definition}
	Let $X \overset{f}{\underset{g}{\rightrightarrows}} Y$ be continuous maps. A \textbf{homotopy} between $f$ and $g$ is a continuous function \[ H: [0,1]\times X\to Y\]
	such that $H(0,x)=f(x)$ and $H(1,x)=g(x).$ If such a homotopy exists we say that $f$ and $g$ are homotopic, and denote this $f\simeq g.$ Two spaces are said to be \textbf{homotopy equivalent} if there exist function $f,g$ such that $f\comp g\simeq 1_X$ and $g\comp f\simeq 1_Y.$   
\end{definition}  
Notice that considering "spaces up to homotopy equivalence" is a weaker condition than "spaces up to homeomorphism". In fact, spaces which are homeomorphic are necessarily homotopy equivalent. In fact if we consider pointed topological spaces, then there is a category $\textbf{Htpy}$ where the morphisms are homotopy classes of maps. In this category, we consider the morphisms with source $S^1$ and a fixed target $X$. More generally we consider morphisms with source $S^n.$ 
\begin{definition}
	The space of maps \[ \pi_n(X):= \Hom_{\textbf{Htpy}}(S^n,X)\]
	are called the $\textbf{$n^{th}$ homotopy groups}.$ The group law is defined by concatenation in each coordinate. A topological space is called \textbf{simply connected} if $\pi_1(X)=0.$ Further, if $\pi_n(X)=0$ for all $n\geq 1,$ then $X$ is contractible. Equivalently, $X$ is homotopy equivalent to a point.  
\end{definition}    
For $n\geq 2,$ these are abelian groups. These are algebraic invariants for the space $X.$ By this we mean that if $X\simeq Y,$ then $\pi_n(X)\cong \pi_n(Y)$ for all $n$ \cite{Hatcher2001}. The problem with these homotopy groups is that they are almost always not computable, and even if they are it is incredibly difficult. For this reason, we want to consider a better algebraic invariant: homology and cohomology. These in some sense classify the number of holes of each dimension in a space. 
\begin{example}\label{Torus}
	Let $T^2=S^1\times S^1$ be the torus depicted below. \[\begin{tikzpicture}[yscale=cos(70)]
    \draw[double distance=5mm] (0:1) arc (0:180:1);
    \draw[double distance=5mm] (180:1) arc (180:360:1);
  \end{tikzpicture}\]
  It is clear that this has two loops which cannot be continuously deformed into one another: one goes around the large center hole and the other around the thickness of the torus. Are there any 2-dimensional holes? Before we give the answer, consider that topological tori are hollow. Therefore, there is some inner volume contained in a torus which stops certain loops from being contractible. 
  
  The answer to the above question is \textit{yes} and there is only 1. There are no higher-dimensional holes. We shall see that a formal way to answer these questions is by computing the homology groups for $T^2,$ which given the statements above should be \[ H_n(X,\Z)=\begin{cases}
  		\{*\} & n\geq 3\\
		\Z & n=0,2\\
		\Z^2 & n=1
  \end{cases} \]
\end{example}

\subsubsection{Simplicial Complexes}
In a way, simplicial complexes are the most basic topological objects for which to define homology and cohomology. As is such, we give a brief introduction to them here. 
\begin{definition}
	A \textbf{simplex} $\Delta^k$ is the convex hull of $n+1$ points embedded in $\R^n.$ A \textbf{simplicial complex} is a union of copies of $\Delta^i$ such that $\Delta^i\cap \Delta^j=\Delta^{k}$ with $k\leq j.$  
	
\end{definition}
To define homology one needs the language of chains
\begin{definition}[Chains]\label{Chains} 
	 Let $K$ be a simplicial complex and denote by \[ C_n^\Delta(K)=\left\{ \sum_i m_i\Delta^n|m_i\in \Z\right\}  \]
	 the free abelian group generated by $n-$simplices. If $\Delta_i^n=[v_0,...,v_n]$ then define the boundary map \[\partial_n:C_n^\Delta(K)\to C_{n-1}^\Delta(K)  \] 
	 in the following manner: \[\partial_n(\Delta^n)=\sum_{i=0}^{n}(-1)^i[v_0,...,\widehat{v_i},...,v_n]  \]
	 This makes $\partial_n$ a group homomorphism. This set is called the set of \textbf{simplicial n-chains}
	
\end{definition}
This yields the following sequence for any given $K,$ \[\begin{tikzcd}
...\arrow[r,"\partial_{n+1}"] &C_n^\Delta\arrow[r,"\partial_n"]&C_{n-1}^\Delta \arrow[r,"\partial_{n -1}"]&C_{n-2}^\Delta \arrow[r,"\partial_{n-2}"]&...
\end{tikzcd}  \]

\begin{lemma}
	$\partial_n\comp \partial_{n+1}=0.$
\end{lemma}
\begin{remark}
	We will drop the superscript $\Delta$ when it is clear that the chain complex is constructed from a simplicial complex.
\end{remark}
\begin{proof}
	We apply the definition twice to the generators of $C_{n+1}.$ \begin{align*}
		\partial_n\partial_{n+1}(\Delta^n)&=\partial_n\left( \partial_{n+1}[v_0,...,v_n]\right)\\
		&=\partial_n\left( \sum_{i=0}^{n+1}(-1)^i[v_0,...,\widehat{v_i},...,v_n]  \right)\\
		&=\sum_{i=0}^n\left( \sum_{j=0}^{i-1} (-1)^j[v_0,...,\widehat{v_i},...,\widehat{v_j},...,v_{n+1}]+\sum_{j=i+1}^{n+1} (-1)^{j-1}[v_0,...,\widehat{v_i},...,\widehat{v_j},...,v_{n+1}]    \right)   \\
		&=0
	\end{align*}
\end{proof}

So it is clear that $\im(\partial_n)\subseteq \ker \partial_{n-1}$ for all $n.$ 
\begin{definition}
	For all $n,$ put \[ Z_n^\Delta(K)=\ker \partial_n   \]
	consisting of \textbf{cycles} and 
	\[B_n^{\Delta}(K)=\im \partial_{n+1}   \]
	consisting of \textbf{boundaries}. 
	Define the \textbf{n-th homology group } \[ H_n^\Delta(K)=Z_n^\Delta(K)/B_n^\Delta(K)   \]
\end{definition}
Let $f$ be a simplicial map (that is to say that $f(\sum t_i\Delta^i)=\sum t_i f(v_i)$) between two complexes $K,L.$ Then $f$ induces a map on the chain complexes \[ f_\sharp:C_n^\Delta(K)\to C_n^\Delta(L), \;\; f_\sharp(\Delta^n)=f\comp \Delta^n   \]
and thus a map on the homology  groups \[ f_*:H_n^\Delta (K)\to H_n^\Delta(L),\;\; f_*([z])=[f\comp z]   \]
This gives us the following commutative square \[ \begin{tikzcd}
Z_n^\Delta(K)\arrow[r,"\partial_n"] \arrow[d,"f_\sharp"] & Z_{n-1}^\Delta(K) \arrow[d,"f_\sharp"]\\
Z_n^\Delta(L)\arrow[r,"\partial_n"]& Z_{n-1}^\Delta(L)
\end{tikzcd}   \]

\begin{lemma}
	The boundary map preserves equivalence classes of the homology groups, that is \[\partial_n (z+\partial_{n+1} c)=\partial_n(z) \]
\end{lemma}
\begin{proof}
	We use the homomorphism property to get that \[ \partial_n(z+\partial c)=\partial_n(z)+\partial_n(\partial_{n+1} c)=\partial_n(z)  \]
	as desired.
\end{proof}
Consider the torus from Example \ref{Torus}. To give the torus a simplicial structure, we want to realize it as a quotient space. That is $T^2\cong \R^2/\Z^2.$ For this reason we can view the torus as a square with opposite sides glued together. Now the simplicial structure should be more or less obvious.

In general, a the sequence of maps in Definition \ref{Chains} is called a chain complex. Further we denote \[C_\bullet^\Delta=\begin{tikzcd}
...\arrow[r,"\partial_{n+1}"] &C_n^\Delta\arrow[r,"\partial_n"]&C_{n-1}^\Delta \arrow[r,"\partial_{n -1}"]&C_{n-2}^\Delta \arrow[r,"\partial_{n-2}"]&...
\end{tikzcd}  \]

Notice that for two simplicial complexes, and a simplicial map $f:K\to L$ we get the following map of complexes \[ f_\sharp:C_\bullet^\Delta(K)\to C_\bullet^\Delta(L)   \]
\begin{lemma}
	Consider two chain complexes $C_\bullet,D_\bullet$ and a chain map $g=(g_n)_{n\geq 0}.$ Then there exists a family of induced homomorphisms \[g_{n,*}:H_n^\Delta(C_\bullet)\to H_n^\Delta(D_\bullet)   \] 
\end{lemma}

\begin{proof}
	Let $\xi\in H_n(C_\bullet).$ Let $z\in Z_n(C_\bullet)$ such that $\xi=[z].$ Consider $g_n(z).$ Then \[ \partial_n g_n(z)=g_{n-1}(0)=0  \]
	So $g_n(z)\in Z_n(D_\bullet).$ That is to say $[g_n(z)]=\eta\in H_n(D_\bullet).$ 
	Let $z'\in C_n$ such that $z\sim z'.$ Suppose $z'=z=\partial_{n+1}(c),c\in C_{n+1}.$ Then \begin{align*} g_n(z')&=g_n(z)+g_n\partial_{n+1}(c)\\
	&=g_n(z)+\partial_{n+1}g_{n+1} (c)\\
	&=g_n(z)+b, b\in B_n(D_\bullet)   
	\end{align*}
	so $g_n(z')\sim g_n(z)$ and $g_{n,*}$ is well defined. 
\end{proof}
\subsubsection{Singular Complexes}

\begin{definition}
	Define the set of \textbf{singluar n-chains} as \[ C_n(X)=\Z[\{\sigma^n:\Delta^n\to X  \}/\sim ] \]
	where $\sim$ is homotopy equivalence. Each element of this group can be written as 
	 \[c=n_1\sigma^n_1+...+n_k\sigma^n_k \]
	This is the set of all possible embeddings of an $n-$simplex into $X.$
	Further, we have a \textbf{boundary operator} \[\partial_n:C_n(X)\to C_{n-1}(X), \partial(\sigma^n)=\sum^n_{i=1} (-1)^i \sigma^n|_{[v_0,...,\widehat{v_i},...,v_n]}   \]
	satisfying the same relations as with simplicial boundary maps. 
 \end{definition}

Further, for every pair of maps $f:X\to Y,g:Y\to W$ we get the induced maps on chain complexes \[ g_{n,\sharp}\comp f_{n,\sharp}=(g\comp f)_\sharp  \]
This induces a functor $H_n:\textbf{Top}\to \textbf{Ab}$ given by $X\mapsto H_n(C_\bullet(X)).$ 
\begin{definition}[Chain Homotopy]
	If $C_\bullet$ and $D_\bullet$ are chain complexes and $f,g$ chain maps then a \textbf{Chain Homotopy} $E=(E_n)_{n\geq 0} $ is a collection of homomorphisms \[E_n:C_n\to D_{n+1}  \]
	such that \[\partial _{n+1}E_n+E_{n-1}\partial _n=g_n-f_n.\]
\end{definition}
\begin{lemma}
	If $f$ and $g$ are chain homotopic, then \[f_{n,*}=g_{n,*}:H_n(C_\bullet)\to H_n(D_\bullet)  \]
\end{lemma}
\begin{proof}
	Let $z\in Z_n(C_\bullet).$ Put $\xi=[z]\in H_n(C_\bullet).$ Then, \[g_n(z)=f_n(z)+\partial_{n+1}E_n(z)  \]
	Hence, $g_n(z)\sim f_n(z)$ so $g_*\xi=f_*\xi.$  
\end{proof}

\begin{theorem}[Homotopy Invariance]
	If $f\simeq g:X\to Y$ then on singluar homology, \[f_*=g_*:H_n(X)\to H_n(Y)  \]
\end{theorem}
\begin{proof}
	See \cite{Hatcher2001}.
\end{proof}
This theorem shows us that singluar homology is invariant under homotopy. This will become hugely important when classifying topological spaces. 

Let $A\subseteq X.$ We want to explore the computability of $H_n(X).$ As we know it, $\pi_n(X)$ is difficult to compute. It turns out, that $H_n(X)$ is relatively easy to compute (in most cases) and therefore is used substantially more by various areas of mathematics as a way of providing invariants to spaces. The following definition hints at one possible way of computing these groups explicitly.  \begin{definition}[Relative Chain Groups]
	Let $A\subseteq X$ be a subspace. Define the \textbf{relative Chain Groups} \[C_n(X,A)=C_n(X)/C_n(A)  \]
\end{definition}  
\begin{remark}
	Notice that chains in $X$ descend to chains relative to $A.$ That is, the following diagram exists and the top square commutes: \[ \begin{tikzcd}
	C_n(A) \arrow[r,"\partial_n"] \arrow[d] & C_{n-1}(A)\arrow[d]\\
	C_n(X) \arrow[r,"\partial_n"] \arrow[d,"q"]& C_{n-1}(X)\arrow[d,"q"]\\
	C_n(X,A) \arrow[r,dashed,"\partial_n"] & C_{n-1}(X,A)
	\end{tikzcd}   \]
\end{remark}
\begin{definition}[Relative Homology groups]
	In light of the previous remark, let $A\subseteq X$ and $C_n(X,A)$ the relative chain groups. We define the relative homology groups \[ H_n(X,A)=Z_n(X,A)/B_n(X,A)   \]
	where $Z_n(X,A)=\partial_n^{-1}C_{n-1}(A)/C_n(A)$ and $B_n(X,A)=[B_n(X)+ C_n(A)]/C_n(A).$ We can re-write $H_n(X,A)$ as \[ H_n(X,A)=\frac{\partial_n^{-1}C_{n-1}(A)}{B_n(X)+C_n(A)}  \]
\end{definition}
Now the question stands: how do we restrict an element of $H_n(X)$ to an element of $H_N(X,A)?$ Simple. Pass it to the quotient. This will still be a cycle. 

Suppose $\eta\in H_n(X,A)$ and $\bar{z}\in Z_n(X,A)$ such that $\eta=[\bar{z}].$ Then $\partial_n z\in Z_{n-1}(A).$ This gives us a map \[ \partial:H_n(X,A)\to H_{n-1}(A)    \]
which sends $z+\partial_{n+1} c+d\mapsto \partial_n d\in B_{n-1}(A)$ where $d\in C_n(A).$
Piecing all of this information together gives us the following Theorem. 
\begin{theorem}
	For any subspace $A\subseteq X,$ we get the following exact sequence of homology groups \[ \begin{tikzcd}
	...\arrow[r]&H_n(A) \arrow[r,"i_*"] &H_n(X)\arrow[r,"j_*"] & H_n(X,A) \arrow[lld,swap,"\partial"]&\\
	& H_{n-1}(A)\arrow[r,swap,"i_*"]& H_{n-1}(X) \arrow[r] &...
	\end{tikzcd}   \]
\end{theorem} 
\begin{proof}
This follows immediately from the Snake Lemma. 
\end{proof}

Consider the triple $B\subseteq A\subseteq X$ and the associated exact sequence \[ 0\to C_n(A,B)\to C_n(X,B)\to C_n(X,A)\to 0 \]

\begin{theorem}[Excision Theorem]
	Let $X$ be a topological space and $A$ a subset. Suppose that $Z\subseteq \close{Z}\subseteq \operatorname{Int}(A).$ Then \[H_n(X,A)\cong H_n(X\setminus Z,A\setminus Z)   \]
\end{theorem}

We shall not prove the Excision theorem, but instead note its usefulness. It tells us that $H_n(X,A)$ is computable. If we consider $H_n(X,x_0)$ for some point $x_0\in X,$ then this is $H_n(X)$ and therefore the homology groups are computable. It can be shown further that $H_n(X)$ is finitely generated if $X$ has finitely many $n$-cells. Furthermore, under some mild conditions, we have an isomorphism \[ H_n(X,A)\cong H_n^\Delta(X,A)\]
This should not be surprising as giving a simplicial structure to a topological space is precisely dictating the embeddings of simplices of various dimensions. \\

The final topic we shall introduce in this section is \textit{Cohomology}.  This is somehow formally dual to the notion of homology. It will seem a bit contrived in this setting, but in some fields (like algebraic geometry and differential geometry) cohomology is the most natural algebraic invariant on a space. 

Consider $C_\bullet(X)$ the singular chain complex for $X.$ The functor $\Hom(-,G)$ for any abelian group $G$ (we could have chosen a ring $R$ here if we only consider it as an abelian group). Then we get a new chain complex \[ C^\bullet(X):=\begin{tikzcd}
...\arrow[r,"d^{n-1}"] &C^{n-1}(X;G) \arrow[r,"d^n"]&C^n(X;G) \arrow[r,"d^{n+1}"]&C^{n+1}(X;G) \arrow[r,"d^{n+2}"]&...
\end{tikzcd}  
\]
where $C^n(X;G):=\Hom_\Z(C_n(X),G)$ and $d^n:=\Hom(\partial_{n},G)=\partial_{n}^*.$ This is a chain complex as $\Hom(-,G)$ is a functor.   
\begin{definition}
	The \textbf{cohomology groups} of a topological space $X$ are the groups \[ H^n_{sing}(X;G):=\ker d^{n+1}/\im d^n\]
\end{definition}	
The one immediate advantage of cohomology is that there is a canonical ring structure on \[ H^*(X;G)=\bigoplus_{n\in N} H^n(X;G) \]
call the \textbf{cup product} defined as follows: let $\varphi\in C^l(X;G)$ and $\psi\in C^k(X;G).$ Then \[ \varphi\smile \psi ( [v_0,...,v_{k+l}])=\varphi(\sigma|_{[v_0,...,v_k]})\psi(\sigma|_{[v_k,...,v_{k+l}]})\]
This cup product induces a map $H^l(X;G)\times H^k(X;G)\to H^{k+l}(X;G)$ which is compatible with the quotients. 

The only theorem we will present here is the Universal Coefficient Theorem. Naively, one would assume that somehow $H_n(X)$ and $H^n(X;G)$ are related (possibly by $\Hom$). This is not necessarily true. What is true however is that there exists a short exact sequence involving these two, as the following theorem dictates: 

\begin{theorem}[Universal Coefficient Theorem]
	Let $X$ be a topological space, $C_\bullet(X)$ its singular chain complex and $C^\bullet(X;G)$ an associated cochain complex. Then the cohomology groups are determined by the split exact sequence \[ 0\to \Ext^1_\Z(H_{n-1}(X),G) \to H^n(X;G) \to \Hom(H_n(X),G)\to 0\]
\end{theorem} 
\noindent For a proof, see $\cite{Rotman1988}.$ This ends the section on algebraic topology.

\section{Differentiable Manifolds and Vector Bundles}
Manifolds pop up in every area of mathematics and play the starring role in the model we develop in Chapter 4. They are generalizations of Euclidean space ($\R^n$) and allows for a variety of new geometry to occur. All together they form a category $\textbf{Man}_\infty$ which gives a concrete example of a suitably bad category whose objects are easy to understand. This section will run through the basic theory of manifolds, vector bundles, and sheaves. We conclude with a discussion of de Rham theory which ties together the topological information on a manifold. Good references for the first two sections are \cite{Wedhorn2016}, \cite{Lee2012}, \cite{Tu2011}, and \cite{Guillemin_Pollack1974}.     

\subsection{Smooth Maps and the category $\textbf{Man}_\infty$}
There are two approaches to smooth manifolds which are commonly used: analytic and algebraic. We shall focus on the algebraic theory as it more closely ties in the later sections here. We will not entirely neglect the analytic theory as we need the notion of differentiation which is purely analytic. We start with the definition of an atlas:

\begin{definition}
	Let $M$ be a topological space and $\alpha\in \N\cup \{\infty\}$. A \textbf{chart} at $p\in M$ is a pair $(\varphi,U)$ with $p\in U\subseteq M$ open, and $\varphi:U\to \varphi(U)\subseteq \R^n$ (for $n$ not depending on $p$) a homeomorphism. A collection $\script{A}$ of charts is called a $C^\alpha$-$\textbf{atlas}$ on $M$ if for all $p,q\in M,$ there are charts $(\varphi_p,U_p)$ and $(\varphi_q,U_q)$ which are \textit{compatible}: the transition map \[ \varphi_p\comp \varphi_q^{-1}:\varphi_q(U_p\cap U_q)\to \varphi_p(U_p\cap U_q)\] is a $C^\alpha$-homeomorphism (each partial derivative is $\alpha$-times differentiable in each coordinate). If $\alpha=\infty,$ then we call the chart maps and the transition maps smooth. In this case, the atlas is called \textbf{smooth}.     
\end{definition}

\begin{example}
	Let $S^n=\{ x\in \R^{n+1} : |x|=1\}$ where $|x|=\sqrt{x_1^2+...+x_{n+1}^2}$ be the $n$-sphere equipped with the subspace topology. To construct a smooth atlas on $S^n,$ we need to give charts. Consider the open subsets \[ U_i^{\pm}:= \{ x\in S^n: x_i>0 (\operatorname{resp}. <0)\} \]
	and the function $f:\mathbb{D}^n\to \R^n$ by $f(u)=\sqrt{1-|u|^2}.$ Then $U_i^{+}\cap S^n$ is the graph of this function and $U_i^-\cap S^n $ is the graph of $-f.$ Each $x_i\in U_i^+\cap S^n$ can then be written as \[x_i=f(x_1,...,\widehat{x_i}, ....,x_{n+1})\]
  Define the maps $\varphi_i^{\pm}:U_i\to \R^n$ by  $\varphi_i^{\pm}(x_1,...,x_{n+1})=(x_1,...,\widehat{x_i},...,x_{n+1}).$
  There are seen to be smooth. Further they are compatible trivially. Hence, $\script{A}=\{ (\varphi_i^{\pm},U_i^{\pm})\}$ is a smooth atlas on $S^n.$  
\end{example}

\begin{definition}\label{Manifold} 
	Let $M$ be a topological space equipped with a $C^\alpha$-atlas $\script{A}.$ We call $M$ an $n$-dimensional $C^\alpha$-$\textbf{manifold}$ if $M$ is Hausdorff and there exists a countable basis for the topology.  
\end{definition}
\begin{remark}
	Normally, the requirement of an atlas is stated as $M$ is \textit{locally Euclidean}. This is the key property of manifolds over normal Euclidean space. They do not need to be $\R^n$ or even $C^\alpha$-homeomorphic to $\R^n,$ only locally.  
\end{remark}

We shall study only smooth manifolds here. The non-smooth cases are important, however not for this thesis. For smooth manifolds, we would like to know that $M$ does not depend on the atlas. 
\begin{proposition}
	Let $M$ be a smooth manifold with atlas $\script{A}.$ Then there exists some $\script{A}^\sharp$ a unique atlas which is maximal and contains all atlases on $M.$ 
\end{proposition} 
\begin{proof}
	Define $\script{A}^\sharp$ as the set of all charts which are smoothly compatible with the charts in $\script{A}.$ Let $(\varphi,U)$ and $(\psi,V)$ be charts in  $\script{A}^\sharp.$ Put $x=\varphi(p)\in \varphi(U\cap V).$ Then as $\script{A}$ is an atlas, there exists a chart $(\theta,W)$ such that $p\in W.$ As $p\in U\cap V\cap W$ the intersection is non-empty. Therefore by construction the map \[ (\psi\comp \theta^{-1})\comp (\theta\comp \varphi^{-1}):\varphi(U\cap V\cap W)\to \psi(U\cap V\cap W) \]
	is smooth and therefore $\psi\comp \varphi^{-1}$ is smooth. Hence, $\script{A}^\sharp$ is an atlas on $M$ containing $\script{A}.$ 
	
	To show it is unique, let $\script{B}$ be another such atlas. Then in particular, each of its charts is smoothly compatible with charts in $\script{A}.$ Hence, $\script{B}\subseteq \script{A}^\sharp$ and by maximality they are equal. 	
\end{proof}
\begin{example}
	The following examples of manifolds show up everywhere and thus should be well understood. 
	\begin{enumerate}
		\item The unit sphere $S^n$ from above was shown to exhibit a smooth atlas. The fact that it is Hausdorff and second countable follows from being a compact subset of $\R^{n+1}.$ 
		\item Consider the action of $\R$ on $\R^n$ by $r(x_1,...,x_n)=(rx_1,...,rx_n).$ Then the quotient space $(\R^n-\{0\})/\R$ is called the $\textit{real projective space}$ and is denoted $\mathbb{P}^{n-1}(\R)$ or just $\mathbb{P}^n$ if the field is understood. We denote elements here as equivalence classes $[x_1,...,x_n].$ These are equivalence classes of lines in $\R^n$ which go through the origin. To give charts on $\mathbb{P}^n,$ we consider maps of the form \[\varphi_i[x_1,...,x_n]=\left(\frac{x_1}{x_i},...,\frac{x_{i-1}}{x_i},\frac{x_{i+1}}{x_i},...,\frac{x_n}{x_i}\right)\in \R^{n-1}.\] Then an easy check shows that these are smooth and are compatible. Hence, $\mathbb{P}^n$ is a smooth manifold. Moreover, it is compact!      
		\item Let $M$ and $N$ be two smooth manifolds. Then $M\times N$ has the structure of a smooth manifold given by charts of the form $(\varphi\times \psi,U_M\times V_N).$ 
		\item Let $M(m,n,\R)$ be the $m\times n$ matrices with real entries. This is a smooth manifold by the diffeomorphism $M(m,n,\R)\to \R^{mn}.$ If $m=n$ we denote this by $M(n,\R)$ or $M_n(\R).$ Notice that for $m=n,$ $M_n(\R)$ comes equipped with a ring structure given by matrix multiplication. In this case, there are many distinguished open submanifolds. The most important is $GL_n(\R)$ the group of invertible linear transformations. We will return to this example later as it is the principal example of a \textit{Lie Group}. These will turn out to be group objects in the category of manifolds.  
	\end{enumerate}	
\end{example}

Due to the above proposition, we will assume without a loss of generality that $M$ is equipped with its maximal atlas. Now, we can define morphisms of smooth manifolds. 
\begin{definition}
	Let $M$ and $N$ be two smooth manifolds. Then a function $F:M\to N$ is a \textbf{smooth map} if for all $(\varphi,U)\in \script{A}_M$ and $(\psi,V)\in \script{A}_N$ such that $F(U)\cap V\neq \varnothing$ the map $\psi\comp F\comp \varphi^{-1}:\varphi(U)\to \psi(V)$ is smooth. As a diagram: \[\begin{tikzcd}
U \arrow[d, "\varphi"'] \arrow[rr, "F"]              & & V \arrow[d, "\psi"] \\
\varphi(U) \arrow[rr, "\psi\circ F\circ\varphi^{-1}"] & & \psi(V)            
\end{tikzcd}\]
A bijective smooth map whose inverse is smooth is a diffeomorphism. 
\end{definition}

The composition of smooth maps is smooth by an extended version of the diagram above. Therefore, we have defined a category $\textbf{Man}_\infty$ of smooth manifolds with morphisms as smooth maps. Using this, we can now define the functor: \[ C^\infty:\textbf{Man}_{\infty}^{op} \to \textbf{Alg}_\R \]
where $C^\infty(M)=\Hom_{\textbf{Man}_\infty}(M,\R)$ and we define the operations point-wise. It is contravariant by the following: 
let $F:M\to N$ be a morphism. Then \begin{align*} F^*:C^\infty(N)\to C^\infty(M)   && F^*(s)=s\comp F       \end{align*}
Further, if $M\to N\to P$ is a sequence, then \[ (F\comp G)^*(d)=d\comp (F\comp G)=(d\comp F)\comp G=G^*\comp F^*\]

A derivation of this ring at $p\in $M is a function $d:C^{\infty}(M)\to \R$ such that $d$ is linear and $d(fg)=f(p)(dg)+(df)g(p).$ 
\begin{remark}
	In fact, $C^\infty(M)$ is a Banach space. This changes the situation is quite a subtle way. If $M\times N$ is a product manifold, one would expect the smooth functions to be $C^\infty(M)\tensor_\R C^\infty(N).$ However, this cannot be true as trigonometric functions exist. Therefore,  we need to take come metric completion of this tensor product. For more information see \cite{Ryan2002}.  
\end{remark}
\begin{definition}
	Let $M$ be a smooth manifold. If $p\in M,$ we define the $\textbf{tangent space}$ to $M$ at $p$ to be \[ T_pM=\{ (f:C^{\infty}(M)\to \R): f \text{ is a derivation at } p \}\]
	This is clearly an $\R$-vector space. In fact, it is finite dimensional. 
\end{definition} 
\noindent Elements of the tangent space should be thought of as vectors which are tangent to $M$ at the point $p.$ We put $\dim_p M=\dim_\R T_pM.$

The \textit{germ} of a function $f:M\to \R$ at the point $p$ is an equivalence class $[f]$ where two functions $f,g$ are equivalent at $p$ if there exists an open neighbourhood $W$ of $p$ such that $f=g$ on $W.$ Denote by $C^\infty_{M,p}$ the set of all germs at $p.$ This is a local ring with maximal ideal $\lie{m}_p$ all functions which are non-zero at $p.$   

The following Theorem gives equivalent formulations of the tangent space. 
\begin{theorem}\label{Tangent_Spaces} 
	Let $M$ be a smooth manifold and $T_pM$ its tangent space at $p.$ The space of germs $C^\infty_{M,p}$ is a local ring with maximal ideal $\lie{m}_p.$ The following are equivalent formulations of the tangent space: 
	\begin{enumerate}
		\item Let $D_pM=\operatorname{Der}(C^\infty_{M,p},\R)$\footnote{A derivation is a linear map $f$ such that it satisfies the Leibniz rule for multiplication: $f(xy)=f(x)y+xf(y).$ $\operatorname{Der}$  is the space of all such functions. For more details, see \cite{Lee2012} or \cite{Wedhorn2016}.}. Then $D_pM\cong T_pM$ by the map which sends $[f]\mapsto f.$ 
		\item Let $\gamma:(-1,1)\to M$ be a smooth curve with $\gamma(0)=p$. Then \[C_pM=\{\gamma'(0) : (\gamma:[0,1]\to M) \text{ a smooth curve and }\gamma(0)=p  \}/\sim\] where $\gamma\sim \delta$ if for all germs $f\in C^\infty_{M,p}$ we have $(f\comp \gamma)'(0)=(f\comp \delta)'(0).$ Then $C_pM\cong T_pM.$ 
		\item $\left(\lie{m}_p/(\lie{m}_p)^2\right)^*\cong T_pM.$ 
	\end{enumerate}
\end{theorem}
See \cite{Wedhorn2016} for a proof of this. The hardest one to prove is $(c)$ and it relies heavily on the fact that $M$ is $C^\infty.$ If $M$ were say $C^n,$ then this would not be true and $\dim \lie{m}/\lie{m}^2$ is infinite. 

The operation of passing to the tangent space is functorial in $M.$ That if if $F:M\to N$ is a morphism, then \[ T_pF:T_pM\to T_{F(p)}N\] is a linear map defined by $d\mapsto d\comp F^*.$ Therefore, $T_p(G\comp F)=d\comp (G\comp F)^*=d\comp F^*\comp G^*=T_{F(p)}G\comp T_pF.$ 

\begin{definition}
	Let $M$ be a smooth manifold. We define the \textbf{tangent bundle} of $M$ as the disjoint union\footnote{For a definition see \cite{Lee2012}. It is the coproduct in the category of sets.}  \[ TM=\coprod_{p\in M} T_pM=\{(p,v): p\in M, v\in T_pM\}\]
	The \textbf{cotangent bundle} is $T^*M=\coprod (T_pM)^*$ defined analogously.  
\end{definition}
There is a canonical projection $\pi_M:TM\to M$ given by $(p,v)\mapsto p.$ Then $\pi_M^{-1}(p)=T_pM.$ Picking a basis of $T_pM$ so that we may identify it with $\R^n,$ we find that in a neighbourhood of $p$, $\pi_M^{-1}(U)\cong U\times \R^n.$ This property of the tangent bundle is called \textit{local trivialization}. Further, using this identification we get that $TM$ (and thus $T*M$) are smooth manifolds of dimension $2\dim M.$  This definition makes $T:\textbf{Man}_\infty\to \textbf{Man}_\infty$ into an endofunctor \cite{Lee2012}. 

Using the tangent bundle, we can now study certain $C^\infty(M)$ modules which arise naturally. Let $s:M\to TM$ be a smooth map such that $\pi_M\comp s=id_M.$ We call $s$ a $\textbf{section}$ of $TM$ and denote the space of all sections as \[ \Gamma(M,TM):=\{ (s:M\to TM): \pi_m\comp s=id_M\}\]
This is an $\R$-vector space under point-wise addition. Moreover it can be given the structure of a $C^\infty(M)$-module. If $s$ be a section and $f\in C^\infty(M).$ Then we define $(f\cdot s)(p)= (p,f(p)s(p)).$ If $U\subseteq M$ is a submanifold, we define $\Gamma(U,TM)$ as the sections of the bundle over $U.$ 
\begin{definition}
	A smooth section $s\in\Gamma(M,TM)$ is called a smooth $\textbf{vector field}$ on $M.$ It associated to each point in $M$ a tangent vector $v\in T_pM.$ We call $M$ \textbf{parallelizable} if there exists vector fields $\{V_1,...,V_n\}$ such that $\{V_1(p),...,V_n(p)\}$ is a basis for $T_pM$ for all $p.$  
\end{definition} 
\begin{proposition}
	Let $M$ be parallelizable, then $TM=M\times \R^n.$ 
\end{proposition}
\begin{proof}
	Let $\{V_1,...,V_n\}$ be a parallelization of $M.$ Then the map $\varphi:TM\to M\times \R^n$ given by \[ \varphi(p,\sum a_iV_i(p))=(p,\sum a_i e_i)\]
	is smooth trivially. Further, as $T_pM\cong \R^n$ via the isomorphism $V_i(p)\mapsto e_i,$ we get that this map is a diffeomorphism. Hence, $TM\cong M\times \R^n$ is trivial.   
\end{proof}

One important operation that vector fields admit is the Lie Derivative. Given two vector fields $V$ and $W,$ we define \[ \script{L}_V(W)=[V,W]=VW-WV\]
Here $Xf(p)=X_p(f)$ is a derivation of $C^\infty(M).$ It is readily checked that $[X,Y]$ is again a vector field. Hence $\Gamma(M,TM)$ admits the structure of a Lie algebra.

\subsubsection{Immersions and Submersions}
Given the discussion above, we can now formulate some special morphisms in $\textbf{Man}_\infty.$ 
\begin{definition}
	Let $F:M\to N$ be a morphism in $\textbf{Man}_\infty.$ We define the $\textbf{rank}$ of $F$ at the point $p\in M,$  as \[ \operatorname{rk}(T_pF:T_pM\to T_{F(p)}N)\] If we pick bases for $T_pM$ and $T_{F(p)}N$ respectively, then by choosing bases, we can get a smooth map \begin{align*} M\to M(m,n,\R) && p\mapsto T_pF \end{align*}  
	Further, the matrix of $T_pF$ is (up to a choice of basis) \[ \Mat{I_{rk(F)} & 0\\0 & 0}\]
	 
\end{definition} 

From this definition it follows that if $r=\operatorname{rk}(F)$ at $p$ and $(\varphi,U)$ is a chart at $p,$ then there exists a smooth function \[ g:\varphi(U)\to \R^{n-r}\]
which sends $0\to 0$ and $T_0g=0.$ We now have the immediate corollary
\begin{corollary}
	For every $p\in M,$ there exists an open neighbourhood $U$ such that $rk_p(F)\leq rk_{q}(F)$ for all $q\in U.$ 
\end{corollary}
This tells us that the rank of a smooth function can only stay the same or increase in a neighbourhood of a point. If equality holds, we say $F$ has \textbf{constant rank} at $p.$

\begin{corollary}\label{Constant Rank} 
	If $F$ has constant rank at $p,$ then there exist charts $(\varphi,U)$ and $(\psi,V)$ of $p$ and $F(p)$ respectively, such that \[ \psi\comp F\comp \varphi^{-1}(x_1,...,x_m)=(x_1,...,x_r,0,...,0)\]
\end{corollary} 
This corollary is incredibly important to the study of manifolds as it gives a local representation of $F$ in such a way that we can disregard a significant number of variables. There are two extreme cases of the above corollary. 

\begin{definition}
	$F:M\to N$ is called a: \begin{enumerate}
		\item \textbf{Immersion} if $T_pF$ is injective for all $p.$ 
		\item \textbf{Submersion} if $T_pF$ is surjective for all $p.$
	\end{enumerate}
	A smooth immersion which is also a topological embedding is called a \textbf{smooth embedding}. 
\end{definition} 

Embeddings are particularly useful as they exhibit manifolds and sitting inside others. Immersions are also incredibly important. The following example is of an object which cannot be embedded into $\R^3$ but instead can be immersed. 
\begin{example}
	Let $I^2$ be the product of $[0,1]$ with itself. We are going to build $K$ the Klein bottle. 
	Consider the relation $(x,0)\sim (x,1)$ and $(0,y)\sim (1,1-y).$ Then we get an object which cannot be embedded into $\R^3$ but can be immersed. This is the glueing of two mobius bands together to get a 1-sided object with no edges. To show it is a manifold is not particularly difficult as we have a representation of it above.  
	\begin{figure}[!htb]
    \centering
    \includegraphics[width=.8 \linewidth]{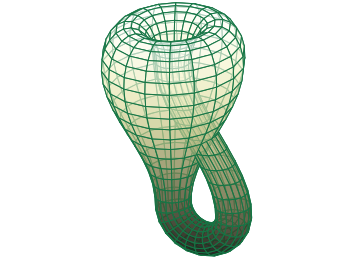}
    \caption{The Klein bottle immersed in $\R^3.$ It can be embedded in $\R^4.$ }
    \label{fig:wrapfig}
\end{figure}
\end{example}
We want to understand how immersions, submersions, and embeddings interact with surjective, injective, and bijective maps. 

\begin{theorem}[Global Rank Theorem]
	Let $F:M\to N$ be a smooth map of smooth manifolds with constant rank. Then 
	\begin{enumerate}
		\item If $F$ is injective, then $F$ is an immersion. 
		\item If $F$ is surjective, then $F$ is a submersion.
		\item If $F$ is bijective, then $F$ is a diffeomorphism. 
	\end{enumerate}
\end{theorem}
The proof of this relies on a strong theorem from functional analysis. As we do not develop this theory here, the proof will be omitted. For a full treatment, see \cite{Lee2012} and \cite{Knapp2005b}. This theorem gives a sufficient condition for a smooth map to be an immersion (resp. submersion) and it is much more easily checked than the normal immersion (submersion) condition.

\subsubsection{Vector Bundles}
We now want to understand some generalizations of the (co)tangent bundle from above. 
\begin{definition}
	Let $M$ be a smooth manifold. We call a triple $(E,\pi, V)$ consisting of a smooth manifold, a projection map, and a real vector space \textbf{real vector bundle}  of rank $\dim V$ over $M$ if: 
	\begin{enumerate}
		\item $\pi:E\to M$ is surjective and a local diffeomorphism. 
		\item For each $p\in M,$ the fibre $\pi^{-1}(p)\cong {p}\times V\cong V$ is endowed with the structure of a $\dim V$-dimensional real vector space.  
		\item For each $p\in M$, there exists a neighbourhood $U$ of $p$ and a homeomorphism $\Phi:\pi^{-1}(U)\to U\times V$ satisfying: 
			\begin{enumerate}
				\item $\pi_U\comp \Phi=\pi$ (where $\pi_U:U\times V\to U$ is the projection)
				\item For each $q\in U,$ the restriction $\Phi_q:E_q\to \{q\}\times V$ is a vector space isomorphism. 
			\end{enumerate}
	\end{enumerate}
\end{definition}
Similarly, we could have defined vector bundles as $E=\coprod_{p\in M} V_p$ where $V_p=\{p\}\times V.$ In this sense, we see that $TM$ and $T^*M$ are vector bundles. Similar to those, $\Gamma(M,E)$ is a $C^\infty(E)$-module. The main purpose of this section is to understand transformations on bundles and transformations between them. 
\begin{definition}
	Let $(E,\pi)$ and $(E',\pi')$ be vector bundles over $M$ and $M'$ respectively. Then a \textbf{bundle homomorphism} is a map $F:E\to E'$ which is linear on each fibre, such that there exists a map $f:M\to M'$ and the following diagram commutes: 
	\[
	\begin{tikzcd}
E \arrow[d, "\pi"'] \arrow[r, "F"] & E' \arrow[d, "\pi'"] \\
M \arrow[r, "f"']                    & M'                      
\end{tikzcd}
	\]
\end{definition} 
	
\begin{proposition}
	If $F$ is smooth, then $f$ is smooth.  
\end{proposition}	
\begin{proof}
	$f=\pi'_M\comp F\comp \zeta$ where $\zeta$ is the zero section. This is a composition of smooth maps and therefore smooth. 
\end{proof}
This lets us define a category $\textbf{Bun}(M)$ whose objects are vector bundles over $M$ and where morphisms are bundle homomorphisms. The forgetful functor \[U:\textbf{Bun}(M)\to \textbf{Man}_\infty\] is faithful. In general, it is not full as there exist smooth maps $E\to E'$ which do not commute with the projection maps. We will denote by $\textbf{Bun}(M)^{<\infty}$ the category of finite rank vector bundles. This category will become interesting in the next section when we relate it to categories of certain sheaves.    

\begin{example} We now construct some interesting bundles over various manifolds. 	
\begin{enumerate}
	\item Let $M=S^1.$ Define an equivalence relation on $\R^2$ by $(x,y)\sim (x',y')$ if $(x',y')=(x+n,(-1)^ny).$ Put $E=\R^2/\sim.$ We claim $E$ is a non-trivial bundle over $S^1.$ First, let $q:\R^2\to E$ be the quotient map. Consider the following diagram \[\begin{tikzcd}
\R^2 \arrow[d, "\pi_1"'] \arrow[r, "q"] & E \arrow[d, "\pi", dashed] \\
\R \arrow[r, "\varepsilon"']            & S^1                       
\end{tikzcd}\]
where $\varepsilon(x)=e^{2\pi ix}.$ Then $\pi$ is determined as the map which makes this diagram commute. This makes $(E,\pi)$ a real line bundle on $S^1$ which is non-trivial (by the twist of $(-1)^n$).  This is the chief example of how local information can be deceptive when trying to understand something globally. 

	\item Let $M$ be a manifold and $V$ any vector space. Then $M\times V$ has the canonical structure of a vector bundle on $M.$  
	\item Let $E,E'$ be vector bundles over $M.$ Then $E\ds E'$ is a vector bundle whose fibres are $V\ds V'.$ This is called the \textbf{Whitney Sum} of vector bundles. 
\end{enumerate}
\end{example}

If $E$ and $E'$ are vector bundles on a smooth manifold $M,$ denote their space of smooth sections by $\Gamma(E)$ and $\Gamma(E').$ If $F:E\to E'$ is a bundle homomorphism, it induces a map \[ \widetilde{F}:\Gamma(E)\to \Gamma(E')\]
given by \[ \widetilde{F}(\sigma)(p)=F(\sigma(p))\]
Because a bundle homomorphism is linear on fibres, $\widetilde{F}$ is $\R$-linear on sections. In fact, it is even $C^\infty(M)$-linear. We can characterize all $C^\infty(M)$-linear maps $\Gamma(E)\to \Gamma(E')$ by the following Theorem. 

\begin{theorem}
	Let $E,E'$ be vector bundles on $M$ and $\script{F}:\Gamma(E)\to \Gamma(E')$ a map. Then $\script{F}$ is $C^\infty(M)$-linear if and only if $\script{F}=\widetilde{F}$ for some $F:E\to E'.$ 
\end{theorem}
The proof of this goes beyond the scope of this text. See \cite{Lee2012} for details. What this theorem tells us is that for vector bundles, we have a bijective correspondence between \[ \Hom_{C^\infty(M)}(\Gamma(E),\Gamma(E'))\overset{\sim}{\longrightarrow} \Hom_{\textbf{Bun}(M)}(E,E')\]

\begin{remark}
	The key to vector bundles is that they somehow encode both global and local information of the manifold. Further, understanding the category $\textbf{Bun}(M)$ is in some sense equivalent to understanding the slice category (see \cite{MacLane1971} for a definition) $\textbf{Man}_\infty/M.$ Overall, we shall use these objects to transfer information from the physical space of a sensory system to the perceptual space. In fact, this will be how we build the perceptual space. 
\end{remark}

We could have equivalently defined \textit{fibre bundles} and gone through this section in more generality. These are similar to vector bundles but we do not require that the fibres be vector spaces. The story of these objects is largely mysterious as they are nearly too general to say anything interesting about. Importantly though, they still have the property that all fibres are isomorphic. 

This concludes the section on manifolds.

\subsection{Sheaves} 
The story of sheaves begins where we just finished: fibre bundles. Notice that fibre bundles are characterized by the fact that the fibres over every point are necessarily isomorphic. Sheaves seek to generalize this idea by removing the restriction of constant fibres. Sheaves are key in nearly every area of mathematics as they encode geometric information which is otherwise difficult to access. In the late 1960s, Alexander Grothendieck first developed the idea that understanding sheaves on a space is equivalent (and in some sense better) than understanding the space itself. In this section we will give the first properties of (pre)sheaves, define ringed spaces, and construct the category $\mathcal{O}_X$-\textbf{Mod}. We conclude the section with a brief introduction to sheaf cohomology, which in the same style as cohomology in the previous section, will provide rich invariants to the associated manifolds. Most of the material of this section comes from \cite{Iversen1986}, \cite{Hartshorne1977} , \cite{Wedhorn2016}, \cite{Eisenbud_Harris2000}, and \cite{Bredon1997}. As this forms the most technical material of this thesis, we shall only prove those statements which are fundamental to the reader's understanding and will point to the appropriate reference otherwise.  
\begin{remark}
	Due to the technical stress of this section, we encourage the reader to skip a majority of the proofs of the statements presented here. The proofs of some of a majority of these theorems can be opaque on a first pass and thus should be revisited only if a deeper understanding is desired. 
\end{remark}

Before we give the formal definitions of sheaves, recall some of the facts we proved about $C^\infty(-)$ as a functor $\textbf{Man}_\infty\to \textbf{Ring}.$ Fix $M\in \textbf{Man}_\infty.$ We know that $C^\infty_M(U)$ is a ring for any open submanifold $U\subseteq M.$ Additionally, $C^\infty_{M,x}$ is a local ring for each $x\in M.$ Further, we showed that given an open cover $\mathcal{U}$ of $M$ and smooth functions defined on each $U_i$ such that $f_i|_{U_i\cap U_j}=f_j|_{U_i\cap U_j}$ then there exists a unique global smooth function $g$ with the property that $g|_{U_i}=f_i.$ What we will see is that $C^\infty_M$ is the \textit{structure sheaf} of $M$. For now, let's start, as always, with some definitions.

\begin{definition}
    Let $(X,\mathcal{T})$ be a topological space and $\script{C}$ a category. A \textbf{presheaf} on $X$ is a functor \[\mathcal{F}:\mathcal{T}^{op}\to \script{C}.\] $\mathcal{T}^{op}$ is the category whose objects are open subsets of $X$ and whose morphisms are one point sets if $V\subseteq U$ and empty otherwise. Morphisms of presheaves are natural transformations of functors.  
\end{definition}
\begin{remark}
    Notice that for $V\subseteq U,$ there is a unique morphism denoted \[\res^U_V:\mathcal{F}(U)\to \mathcal{F}(V).\] We sometimes call $\mathcal{F}(U)$ the set of sections of $\mathcal{F}$ over $U,$ and denote this $\Gamma(U,\mathcal{F}).$ Additionally, instead of writing $\res^U_V(s)$ for the image of $s$ in $\mathcal{F}(V)$, we instead write $s|_V.$   
\end{remark}
Every presheaf is the same as a contravariant functor. We use the term presheaf when we want to discuss some gluing conditions which we will see later. Some classical examples of presheaves are \[C^\alpha_M=\{f:M\to \R:f \operatorname{ is } \alpha \operatorname{ times differentiable}\}\] for a real $C^\alpha$-manifold $M$ and $\alpha\in \N \cup \{\infty\}.$ 

\begin{definition}
    Let $X$ be a topological space and $\mathcal{F}$ a presheaf on $X.$ $\mathcal{F}$ is a \textbf{sheaf} if the following condition is satisfied \begin{enumerate}
        \item[(\textbf{Sh})] If $U\subseteq X$ is an open set and $\{U_i\}_{i\in I}$ is an open cover of $U$ such that for all $i$ there exists $f_i\in \mathcal{F}(U_i)$ and for all $i\neq j\in I$ $f_i|_{U_i\cap U_j}=f_j|_{U_i\cap U_j}$ then there exists a unique $f\in \mathcal{F}(U)$ such that $f|_{U_i}=f_i.$  
    \end{enumerate}
\end{definition}
\begin{remark}
    This definition can be generalized to general categories. To do this correctly however one needs the language of sites. We will not cover these but refer the reader to \cite{Metzler2003},\cite{Kashiwara2006}, and \cite{Carchedi2011} for an in depth treatment. 
\end{remark}
\begin{example}
    We have already seen an example of a sheaf, namely $C^\alpha.$ It is easy to check the gluing condition (\textbf{Sh}). Other common examples are $\Omega^p_M$ the set sheaf of differential forms of degree $p$ and $L$ the sheaf of locally constant functions on a space. 
\end{example}
Sheaves allow for local information to be glued together to make global information. What we mean by local here is up to some interpretation. We can either mean open neighbourhoods of points or the points themselves. As points are almost never open (except for discrete sets) we need to figure out how to define $\mathcal{F}(x).$ The following definition gives an answer in a category which admits colimits.  

\begin{definition}\label{Stalk} 
    Let $X$ be a topological space and $U(x)=\{U\in \operatorname{Open}(X):x\in U\}.$ Suppose $\mathcal{F}$ is a (pre)sheaf on $X.$ We define the $\textbf{stalk}$ of $\mathcal{F}$ to be \[ \mathcal{F}_x=\varinjlim_{U(x)} \mathcal{F}(U)\] 
\end{definition}
Here, we interpret the colimit as being taken over successively smaller sets containing $x.$ In fact, if there exists some minimal $U_x$ contained in all neighbourhoods of $x,$ then $\mathcal{F}_x=\mathcal{F}(U_x).$ We now want to understand how morphisms of sheaves interact with the stalks.   

\begin{remark}
	For the remainder of this text, we shall consider only sheaves of rings or more generally $R$-modules for some ring $R.$ This simplifies the situation and also turns out to be the situation for most spaces.  
\end{remark}	

\begin{proposition}
    A morphism of sheaves on a space $X,$ $\varphi:\mathcal{F}\to \mathcal{G}$ is an isomorphism if and only if it the induced map on stalks $\varphi:\mathcal{F}_x\to \mathcal{G}_x$ is an isomorphism.
\end{proposition}
\begin{proof}
	($\Rightarrow$) Let $x\in X$ and $U(x)$ as in Definition \ref{Stalk}. Consider \[\varinjlim: \textbf{Ring}^{U(x)} \to \textbf{Ring}\]where we consider $U(x)$ as a partially ordered set. As $\varphi$ is a natural transformation it gives two direct systems \begin{align*} 
	\{ \mathcal{F}(U),Res_V^U\}_{U(x)} && \{ \mathcal{G}(U),Res_V^U\}_{U(x)} 
\end{align*}
As $\varphi$ is an isomorphism, $\varphi_U$ is an isomorphism for all $U\in U(x).$ Therefore \[ \varinjlim_{U(x)} \{ \varphi_U:\mathcal{F}(U)\overset{\sim}{\longrightarrow} \mathcal{G}(U) \} =\varphi_x:\mathcal{F}_x\overset{\sim}{\longrightarrow} \mathcal{G}_x\]
is an isomorphism. As $x$ was arbitrary, we see that $\varphi_x$ is an isomorphism on all stalks.  

($\Leftarrow$) Now assume that $\varphi_x$ is an isomorphism for all $x\in X.$ We shall show that $\varphi_U$ is a bijection for all $U$ and thus taking $\psi_U=\varphi_U^{-1}$ makes $\varphi$ an isomorphism of sheaves. Let us first show that $\varphi_U$ is injective. If $s\in \mathcal{F}(U)$ is such that $\varphi_U(s)=0,$ then on all stalks $\varphi_x(s)=0.$ As $\varphi_x$ is an isomorphism, we see that $s_x=0$ for all $x\in U.$ Therefore, there exists some $W_x\subseteq U$ such that $s|_W=0$ with $x\in W_x.$ As $\bigcup W_x$ is a cover for $U,$ by the sheaf condition there exists a unique $s^*\in \mathcal{F}(U)$ such that $s^*|_W=s_W=0.$ By uniqueness, $s^*=s=0$ and $\varphi_U$ is injective. 

To show it is surjective, let $t\in \mathcal{G}(U).$ Let $x\in U$ and $t_x\in \mathcal{G}_x$ be the germ of $t$ at $x.$ As $\varphi_x$ is surjective, there exists $s_x\in \mathcal{F}_x$ such that $\varphi_x(s_x)=t_x.$ Pick a representative section $s(x)\in \mathcal{F}(V_x)$ such that $s(x)=s_x.$ Then $\varphi_{V_x}(s(x))$ and $t|_{V_x}$ have the same germ in $\mathcal{G}_x.$ Possibly replacing $V_x$ by a smaller open set, we may assume that $\varphi_{V_x}(s(x))=t|_{V_x}.$ The collection $\{V_x\}$ forms an open cover of $U$ and on each $V_x$ we have a section $s(x).$ Let $p,q\in X$ be distinct points. Then $s(p)|_{V_p\cap V_q} $ and $s(q)|_{V_p\cap V_q}$ are two sections in $\mathcal{F}(V_p\cap V_q)$ which are sent by $\varphi$ to $t|_{V_p\cap V_q}.$ As $\varphi_U$ is injective, we conclude that \[s(p)|_{V_p\cap V_q}=  s(q)|_{V_p\cap V_q}\]
By the sheaf condition there exists $s\in \mathcal{F}(U)$ such that $s|_{V_p}=s(p).$ Lastly, we need to check that $\varphi_U(s)=t.$ By construction $\varphi_{V_x}(s)=t|_{V_x}$ for all $x\in U.$ Now, applying the sheaf condition again to $\varphi_U(s)-t$ we see that this must be $0$ and hence $\varphi_U(s)=t$ and $\varphi$ is surjective. This completes the proof.   
\end{proof}

The collection of all $\script{C}$-valued sheaves on a topological space for a category denoted $\textbf{Sh}(X,\script{C})$ (Presheaves also form a category). Per the remark above, we shall denote \textbf{Sh}$(X,R$-\textbf{Mod}$):=$\textbf{Sh}$(X)$ when $R$ is well understood. Notice that for a morphism of sheaves the kernel presheaf defines a sheaf but the cokernel presheaf does not. Further, we would like for quotients to exist in this category. To remedy this, we come to the following definition. 

\begin{dproposition}
    For any presheaf $\mathcal{F}$ there is a sheaf $\widetilde{\mathcal{F}}$ and a natural morphism $\theta:\mathcal{F}\to \widetilde{\mathcal{F}}$ with the following universal property: for any sheaf $\mathcal{G}$ and morphism of presheaves $\varphi:\mathcal{F}\to \mathcal{G},$ there exsits a unique morphism of sheaves $\widehat{\varphi}:\widetilde{\mathcal{F}}\to \mathcal{G}$ with $\widehat{\varphi}\comp \theta=\varphi.$ That is, the following diagram commutes \[ \begin{tikzcd}
        \mathcal{F} \arrow[r,"\theta"] \arrow[d,swap,"\varphi"] & \widetilde{\mathcal{F}} \arrow[dl,dashed," \widehat{\varphi}"]\\
        \mathcal{G}&
    \end{tikzcd}
    \]
    The sheaf $\widetilde{\mathcal{F}}$ is called the $\textbf{sheafification}$ of $\mathcal{F}.$ One can prove that sheafification is functorial in presheaves. In fact, it is left adjoint to the forgetful functor $\textbf{Sh}(X)\to \textbf{PSh}(X).$   
\end{dproposition}
\begin{lemma}
	The canonical map $\theta:F\to \widetilde{F}$ induces an isomorphism on stalks. 
\end{lemma}
\begin{proof}
	Consider the construction of the sheafification of $\mathcal{F}$ as \[ \widetilde{F}(U)=\left\{ (s_x)\in \prod_{x\in U} \mathcal{F}_x: \forall x\in U, \exists W\subseteq U \operatorname{ and } \exists t\in \mathcal{F}(W) \operatorname{ such \;that } \forall w\in W, t|_{W}=s|_W  \right\}\]
The restriction maps are given by the restriction on the products. Now, by definition $\theta_x$ is necessarily the identity.  	
\end{proof}

\begin{remark}
    There is another way to build the sheaf associated to a presheaf. Given a presheaf $\mathcal{F}$ on $X,$ we can construct a sheaf $\textbf{Sp\'e}(\mathcal{F})=\bigsqcup_{p\in X} \mathcal{F}_p.$ This has a natural projection $\pi:\textbf{Sp\'e}(\mathcal{F})\to X$ which projects each stalk onto the point it is over. We topologize this space by endowing it with the strongest topology such that the sections $s\in \mathcal{F}(U)$ are continuous. It can be shown that these definitions agree.
\end{remark}
\noindent The sheafification operation allows us to define cokernels, quotients, and constant sheaves. All of this together tells us that if $\script{A}$ is an abelian category, then $\operatorname{Sh}(X,\script{A})$ is also an abelian category \cite{Iversen1986}. Specifically, \textbf{Sh}$(X)$ is an abelian category.
\begin{example}\text{}
	\begin{enumerate}
		\item Let $A$ be a ring. Then $A$ defines a presheaf $A_X$ by $A_X(U)=A$ for all $U$ open. If $X$ is connected, then this is a sheaf. If $X$ is disconnected this is not true. Suppose $X=X_0\sqcup X_1.$ Then if $a\neq b\in A,$ then defining  $a\in A_X(X_0) $ and $b\in A_X(X_1),$ they agree trivially on the empty intersection yet there is no element $c$ such that $c=a$ and $c=b.$ Hence, $A_X$ is not a sheaf in general. Therefore, we consider $\widetilde{A_X}$ which is the sheaf of locally constant functions on $X$ with values in $A.$ 
		\item We define $i_{x,*}(A)$ to be the $\textbf{skyscraper sheaf}$ which is defined by \[ i_{x,*}(A)(U)=\begin{cases}
			A & x\in U \\
			0 & x\notin U
		\end{cases}\]
		This is a sheaf on \textit{any} topological space and plays a key role in the theory as it provides good counter examples to many conjectural relationships.  
		
		\item Let $\mathcal{F}$ be a sheaf and $\mathcal{G}$ a subsheaf on $X,$ Then the functor $U\mapsto \mathcal{F}(U)/\mathcal{G}(U)$ is a presheaf. It is not a sheaf in general however. Therefore we can take the sheafification to get $(\mathcal{F}/\mathcal{G})(U)$ as a sheaf on $X.$ In general, $(\mathcal{F}/\mathcal{G})(U)$ does not agree with $\mathcal{F}(U)/\mathcal{G}(U)$.
	\end{enumerate}
\end{example}

As $\textbf{Sh}(X)$ is an abelian category, we can consider exact sequences of sheaves. 
\begin{definition}
	A sequence of sheaves on a space $X$ is a sequence \[0\to \mathcal{F}\hookrightarrow \mathcal{G}\twoheadrightarrow \mathcal{H}\to 0\]
	is \textbf{exact} if $\im  [\mathcal{F}\to \mathcal{G}]\cong \ker [ \mathcal{G}\to \mathcal{H}].$ Equivalently, this sequence is exact if the corresponding sequence \[ 0\to \mathcal{F}_x\to \mathcal{G}_x\to \mathcal{H}_x\to 0\] is exact on the level of stalks for each $x\in X$.   
\end{definition}	 
It follows from an identical argument for $\Hom,$ that $\Gamma(X,-)$ is a left exact functor \textbf{Sh}$(X)\to R$-\textbf{Mod}. For this reason we define \[ H^i(X,\mathcal{F}):=R^i\Gamma(X,\mathcal{F})\]
as the \textbf{Sheaf cohomology groups} of $\mathcal{F}.$ These will tie together the entire chapter in Section 3.2.4 via Theorem \ref{DeRham}. Before then however, we want to consider how sheaves perform under maps between spaces.

Up until this point, we have considered a fixed space $X.$ If we have a morphism of topological spaces $f:X\to Y,$ we want to build a sheaf on $Y$ which comes from $f$ in some way. 
\begin{definition}
    Let $f:X\to Y$ be a map of topological spaces. Suppose $\mathcal{F}$ is a sheaf on $X.$
    The \textbf{direct Image} (or \textbf{pushforward}) sheaf on $Y$ with respect to $f$ is the sheaf \[ f_* \mathcal{F}(V):=\mathcal{F}(f^{-1}(V))  \] 
    Further, we define the $\textbf{inverse image}$ sheaf on $X$ of a sheaf on $Y$ as \[ f^{-1}\mathcal{G}(U)=\varinjlim_{f(U)\subset V} \mathcal{G}(V)\]
\end{definition}
\begin{remark}
    In the previous definition, one may want to give a naive definition of the inverse image sheaf in the the style of the pushforward, that is $f^{-1}\mathcal{G}(U)=\mathcal{G}(f(U))$. This fails immediately however as we are not guaranteed that $f(U)$ is open. 
\end{remark}
Sometimes, topological spaces come naturally equipped with sheaves. Examples of this situation are smooth manifolds. to every real topological manifold $M,$ we have $C^0_M$ the sheaf of continuous functions $M\to \R$. 

\begin{definition}
    A \textbf{ringed space} is a topological space $X$ equipped with a sheaf of rings $\mathcal{O}_X$ called the structure sheaf of $X.$ A morphism of ringed spaces is a pair $(f,f^\sharp)$ with $f:X\to Y$ a continuous map and $f^\sharp:\mathcal{O}_Y\to f_*\mathcal{O}_X$ a map of sheaves. We call $(X,\mathcal{O}_X)$ a \textbf{locally ringed space} if the stalks $\mathcal{O}_{X,p}$ are local rings for all $p\in X.$ A morphism of locally ringed spaces is a pair where the map on sheaves is a local homomorphism of local rings (on stalks it sends the maximal ideal at $f(p)$ to the maximal ideal at $p$ surjectively). We call $\mathcal{O}_X$ the \textit{structure sheaf} of $X.$      
\end{definition}

\begin{proposition}
	Let $(M,\mathcal{O}_M)$ be a locally ringed space. Then $M$ is a smooth manifold in the sense of Definition \ref{Manifold} if and only if there exists an open cover $M=\bigcup U_i$ such that for each $U_i$ there exists $Y\subseteq \R^n$ open such that there is an isomorphism of locally ringed spaces $(U_i,\mathcal{O}_M(U_i))\overset{\sim}{\longrightarrow} (Y,C^\infty_{\R^n}(Y)).$ 
\end{proposition}	
\begin{proof}
	($\Leftarrow$) This direct in obvious by defining the charts of the atlas to be the projection onto the first coordinate of the morphisms $(f_i,f_i^\#)$ of ringed spaces. Then the sheaf condition guarantees the glueing axiom holds. 
	
	($\Rightarrow$) This direction is a bit more subtle. Let $M$ be a smooth manifold with atlas $\script{A}.$ Let $M=\bigcup U_i$ and $V\subset M$ an open subset. Then define \[\mathcal{O}_M(V)=\{ f:V\to \R : f|_{U_i\cap V}\comp \varphi_i^{-1}: \varphi_i(U_i\cap V)\to \R \operatorname{ is } C^\infty\}\]
	This makes $(M,\mathcal{O}_M)$ a ringed space. Further, it follows immediately that the induced morphisms $(U_i,\mathcal{O}_M|_{U_i}) \longrightarrow (Y,C^\infty_{\R^n}|_Y)$ are isomorphisms of ringed spaces. As the target is locally ringed, so is $(M,\mathcal{O}_M).$ 
\end{proof}

\begin{corollary}
	Let $M$ be a smooth manifold with smooth atlas $\script{A}.$ Then $(M,C^{\infty}_M)$ is a locally ringed space. 
\end{corollary}
	
In some sense, locally ringed spaces are the correct setting to study everything we have seen already. Manifolds and all of their analytic properties can be re-phrased in terms of operation on the sheaf $C^\infty(M).$ The only object which we have seen so far that needs some further discussion is vector bundles. We first discuss a generalization. 
\begin{definition}
	Let $(X,\mathcal{O}_X)$ be a ringed space. An $\mathcal{O}_X$\textbf{-Module} is a sheaf $\mathcal{F}$ on $X$ such that for each $U\subseteq X$ open, there is a map \[ \mathcal{O}_X(U)\times \mathcal{F}(U)\to \mathcal{F}(U)\]
	which turns $\mathcal{F}(U)$ into an $\mathcal{O}_X(U)$-module. A morphism of $\mathcal{O}_X$-modules is a morphism of sheaves which is $\mathcal{O}_X$-equivariant. 
\end{definition}  

In direct analogy with $R$-modules, we can consider some operations on $\mathcal{O}_X$-modules. 
\begin{example} For this set of examples, let $\mathcal{F}$ and $\mathcal{G}$ be $\mathcal{O}_X$-modules.
	\begin{enumerate}
		\item \textbf{(Direct Sums)} We can define $\mathcal{F}\ds \mathcal{G}(U)$ by $\mathcal{F}(U)\ds \mathcal{G}(U).$ It is nearly immediate that this is a sheaf. Therefore if $I$ is a finite indexing set, we can define the direct sum for over this set and this will be a sheaf. In the finite case, this does not hold true and therefore one must sheafifiy. 
		\item \textbf{(Tensor Products)}Consider the presheaf $T:U\mapsto \mathcal{F}(U)\tensor_{\mathcal{O}_X(U)} \mathcal{G}(U).$ This is not a sheaf in general (this takes some work to find an example). Therefore, we define $\mathcal{F}\tensor_{\mathcal{O}_X} \mathcal{G}=\widetilde{T}.$ 
		\item \textbf{(Hom)} We can consider the presheaf $U\mapsto \Hom_{\mathcal{O}_X|_U}(\mathcal{F}|_U,\mathcal{G}|_U).$ This is actually a sheaf and is denoted as \[ \mathcal{H}om_{\mathcal{O}_X}(\mathcal{F},\mathcal{G})\]
		It also turns out that $\mathcal{H}om$ and $\tensor_{\mathcal{O}_X}$ are adjoint endofunctors.
		\item \textbf{(Duals)} We define \[ \mathcal{F}^*:=\mathcal{H}om_{\mathcal{O}_X}(\mathcal{F}, \mathcal{O}_X)\]
		This is a sheaf on $X.$ There is a canonical morphism $\mathcal{F}\to \mathcal{F}^{**}$ the double dual given on stalks by \[ s_x\mapsto ev_{s_x}:\mathcal{F}_x\to \mathcal{O}_{X.x} \]
		the evaluation at $s_x$ map. Further, this gives another construction of the tangent and cotangent bundles. 
	\end{enumerate}
\end{example}
\noindent Now we want to define "free" $\mathcal{O}_X$-modules. 
\begin{remark} 
	For the remainder of this text, we shall write $\mathcal{O}_U$ for the restriction of the structure sheaf to $U\subseteq X$ open. 
\end{remark}

\begin{definition}
	We call an $\mathcal{O}_X$-module $\mathcal{F}$ \textbf{finite locally free} if there exists an open cover $\mathcal{U}=\{U_i\}_{i\in I}$ such that $\mathcal{F}|_{U_i}$ is isomorphic (as sheaves) to $\mathcal{O}_{U_i}^n$ for some $n\in \N.$ In this case, $\mathcal{F}_x$ is a free $\mathcal{O}_{X,x}$-module. Define $\operatorname{rk}_x(\mathcal{F}):=\operatorname{rk}_{\mathcal{O}_{X,x}}(\mathcal{F}_x).$ This defines a locally constant function \begin{align*}  X\to \N&& x\mapsto \operatorname{rk}_x(\mathcal{F})  \end{align*}
	called the \textit{rank} of $\mathcal{F}.$ 
\end{definition}

We can build a category $FLF(X)$ of all finite locally free sheaves on $X.$ It turns out that for FLF sheaves, the canonical morphism $j:\mathcal{F}\to \mathcal{F}^{**}$ is an isomorphism. These look surprisingly close to a generalization of vector bundles, and the following theorem explains why. 

\begin{theorem}\label{Bun_Sheaf} 
	There is an equivalence of categories $\operatorname{Bun}(X)^{<\infty}\leftrightarrows FLF(X)$ for any ringed space $X.$   
\end{theorem}
For a proof, see \cite{Wedhorn2016}.  What this theorem tells us is that we can assign to each vector bundle a finite locally free sheaf and vice-versa. Therefore, as the tangent and cotangent bundles are finite rank vector bundles on a manifold $M,$ we get corresponding sheaves $\mathcal{T}_M$ and $\Omega_M^1.$ It turns out that we can define the cotangent bundle using the sheaf $\mathcal{H}om$ from above \[ \Omega_M^1=\mathcal{T}_M^*=\mathcal{H}om(    \mathcal{T}_M,C^\infty_M)\]
Furthermore, as $C^\infty_{M,x}$ is a local ring (the maximal ideal $\lie{m}_x$ is all non-zero functions at $x$) we can define $T_xM=(\lie{m}_x/\lie{m}_x^2)^*$ and then \[ \mathcal{T}_{M,x}=(\lie{m}_x/\lie{m}_x^2)^*\]
This gives an explicit description of the stalks of $\mathcal{T}_M.$  

We now turn to some homological methods to end this subsection. 
Together with morphisms, $\mathcal{O}_X$-\textbf{Mod}$\hookrightarrow$ \textbf{Sh}$(X)$ is a full subcategory which can be shown to have enough injectives \cite{Iversen1986}. Injective $\mathcal{O}_X$-modules are defined analogously to $R$-modules. For this reason, given $\mathcal{F}$ in $\mathcal{O}_X$-\textbf{Mod}, we can find an injective resolution $\mathcal{J}^\bullet$ and thus a quasi-isomorphism \[ \mathcal{F}\overset{qis}{\longrightarrow} \mathcal{J}^\bullet\]
This gives us a way of computing $H^i(X,\mathcal{F}).$ In a similar manner to $R$-modules, \[ H^i(X,\mathcal{F})\cong H^i( \Gamma(X,\mathcal{J}^\bullet))\]

The following remarkable theorem gives yet another way to compute sheaf cohomology for constant sheaves corresponding to a ring $R$. 
\begin{theorem}\label{sheaf_singular} 
	Let $R$ be a ring and $\widetilde{R}_M$ the constant sheaf on $(M,C^\infty_M).$ Then there is an isomorphism \[ H^i(M,\widetilde{R}_M)\cong  H^i_{sing}(M;R)\]
\end{theorem}
The proof is quite technical but relies on \textit{sheafifying} the singular cochain complex. Once this is done it follows nearly immediately. For a full proof, see \cite{Wedhorn2016}. 
This concludes the section on sheaves. 

\begin{remark}
	The main point of sheaves is to facilitate the transfer of local information to global information via glueing. Notice that the axioms for sheaves and thus everything else in this section, was designed so that, under the right conditions , the sections glued to global ones. This action of taking local information to global information is precisely what needs to happen in the olfactory system. We have local actions of granule cells on mitral cells and these "glue" together to form an action of the entire GC layer. As you can guess, the notion of sheaves will show up to help with the mathematical formulation of this property.   
\end{remark}

\subsection{de Rham Cohomology} 
We end this chapter (and therefore all of the background material) with a short disucssion of de Rham theory for manifolds. This centers on the construction of \textit{differential forms} on a manifold and the exterior derivative. The main theorem we will prove is de Rham's theorem which gives an isomorphism of sheaf cohomology with so-called \textit{de Rham cohomology}.  This combined with Theorem \ref{sheaf_singular} gives the grand conclusion that singular cohomology on manifolds can be computed via the de Rham complex. The main references here are \cite{Lee2012} and \cite{Wedhorn2016}.  

This story begins with the construction of differential $k$-forms on a manifold. Before we can do this though, we need to define and study \textit{smooth functors}. These allow us to transform vector bundles and will extend to endofunctors of $\textbf{Bun}(M)^{<\infty}.$ 
\begin{definition}
	Let $F:\textbf{Vect}_\R\to \textbf{Vect}_\R$ be a functor (we will assume covariant but this is not neccessary). We say that $F$ is \textbf{smooth} if the induced map \[ F^\flat:\Hom(V,W)\to \Hom(F(V),F(W)) \]
	is smooth as a map of smooth manifolds. 
\end{definition}
\begin{example}Some common smooth functors which play a key role in the theory of smooth manifolds are presented below. 	\begin{enumerate}
		\item The functor $(-)^{\tensor k}$ is a smooth functor via the Hom-tensor adjunction. Indeed even taking $T^\bullet(-)$ is smooth. In general, most of the operations on vector spaces are smooth functors. Some care needs to be taken in the case of infinite indexing sets, but we shall ignore these cases.
		\item The functor $\bigwedge^k(-)$ is smooth. This will form the basis for all of de Rham theory. In general, if $F$ arises as a quotient of $\tensor^k$ by some homogeneous ideal (its generated by elements of the same degree) then $F$ is smooth.  
		\item The functor $(-)^*:=\Hom(-,\R)$ is smooth. This follows from the previous example.  
		\item If we fix a vector space $W,$ then $\Hom(W,-)$ and $\Hom(-,W)$ are smooth functors. This follows from the first and third example for the case of finite dimensional spaces. For infinite dimensional spaces this is more subtle and less useful. 
	\end{enumerate}
\end{example}

Now let $M$ be a smooth manifold and $\pi:E\to M$ be a vector bundle. If $F$ is a smooth functor, then $F$ admits an extension 
\[ \widehat{F}:\textbf{Bun}(M)^{<\infty}\to \textbf{Bun}(M) \]
by sending $E\mapsto \widehat{F}(E)$ where $\widehat{F}(E)_p=F(E_p).$ 
If $F$ takes finite dimensional vector spaces to finite dimensional vector space, then $\widehat{F}$ lands in $\textbf{Bun}(M)^{<\infty}.$ 
\begin{example}
	Consider the cotangent bundle from before $T^*M=\coprod_{p\in M} T^*_pM.$ Then we realize this as \[ T^*M=\widehat{(TM)^*}\] To construct differential forms, we need to consider $\bigwedge^k(T^*M).$ This is a smooth vector bundle on $M$ of rank $\binom{\dim M}{k}.$ By Theorem \ref{Bun_Sheaf}, we can associate a finite locally free $C^\infty_M$-module to $T^*M.$ What we would like to show is that this associated sheaf is $\Omega_M^1$ from before.   
\end{example}
We now give a second construction of $\Omega_M^1.$  
\begin{definition}
	Let $A$ be an $R$-algebra and $B$ an $A$-module. Then the module of derivations \[ \Omega_{B/A}=\{ db: b\in B\}/\sim\]
	where $\sim$ is defined by the relations for derivations as above. For $(X,\mathcal{O}_X)$ a ringed space, we can define \[\Omega_X^1(U):=\Omega_{\mathcal{O}_X(U)/R} \]
	where $\mathcal{O}_X$ is a sheaf of $R$-modules. For a manifold $(M,C^\infty_M),$ we have \[ \Omega_M^1(U)=\Omega_{C^\infty_M(U)/\R}\]
	This is \textbf{Cotangent sheaf} of $M$ and the tangent sheaf is its dual as a $C^\infty_M$-module. We define differential $k$-forms again, now as sections of the sheaf $\bigwedge^k \Omega_M^1.$ This is the locally free sheaf associated the $k$-th exterior power of the cotangent bundle on $M.$.  
\end{definition}
\begin{remark}
	In general if $(X,\mathcal{O}_X)$ is a locally ringed space, we cannot define $p$-forms as above. This is because $\Omega_X^1$ need not be a FLF sheaf. To remedy this, we use the canonical morphism \[ \Omega_X^1\to \Omega_X^{**}\]
	and take exterior powers. 
\end{remark}

\begin{proposition}
	The two constructions of $\Omega_M^1$ are equivalent. 
\end{proposition}
\begin{proof}
	This follows from Theorem \ref{Bun_Sheaf} and Proposition \ref{Tangent_Spaces}.   
\end{proof}

Now, consider $T^*M$ as the vector bundle associated to $\Omega_M^1.$ Then considering that $\bigwedge^kT^*M$ is a vector bundle as above we can sheafify it. As is expected, \[ \bigwedge^kT^*M\mapsto \Omega_M^k:=\bigwedge^k\Omega_M^1\]

\begin{definition}
Using the constructions above, the module of \textbf{differential k-forms} is 
the $C^\infty_M(M)$-module \[\Omega^k(M):=\Gamma(M,\bigwedge^kT^*M)\] 
\end{definition}
These modules come with a a differential $d^k:\Omega^k(M)\to \Omega^{k+1}(M)$ called the \textbf{exterior} \textbf{derivatives}. In the greatest generality, if $\omega\in\Omega^k(M)$ and $V_0,...,V_k$ are smooth vector fields on $M,$ then \[ d^k\omega(V_0,...,V_k)=\sum (-1)^i \omega(V_0,...,\widehat{V_i}, ...,V_k) +\sum (-1)^{ij} \omega([V_i,V_j],V_0,...,\widehat{V_i},...,\widehat{V_j},...,V_k)\]
 where $\widehat{V_i}$ means omission.

\begin{lemma}
	$d^{k+1}\comp d^k=0.$ 
\end{lemma}
The proof of this is identical to the proof for the singular chain complex. This complex is called the \textbf{De Rham Complex} associated to $M.$ 

\begin{definition}
A differential $k$-form $\omega$ is called \textbf{closed} if $d\omega=0.$ It is called \textbf{exact} if $\omega=d\eta$ for some $(k-1)$-form $\eta.$ The above lemma tells us that every exact form is closed.
 \end{definition}
 
Therefore, we can define a cohomology theory for $M$ via this complex as \[ H^i_{DR}(M):=\ker d_i/\im d_{i-1} \]
It then follows immediately that $H^0(M)\cong \R^{\pi_0(M)}.$ 
 
 \begin{example}
 	For $\R^n,$ the differential $1$-forms are generated by the formal symbols $dx_i$ where $\{x_i\}$ is a basis for $\R^n.$ For higher degrees, we have then that \[ \omega=\sum \alpha_{i_1,...,i_\ell} dx_{i_1}\wedge ...\wedge dx_{i_\ell} \]
	with $\alpha_{i_1,...,i_\ell}\in \R.$ Further, every $k$-form for $k\geq 1$ is necessarily closed. Hence, $H^i_{DR}(\R^n)=0$ for $i\geq 1$ and. $H^0(\R^m)=\R.$ 
 \end{example}
 
 Now that we have the notion of de Rham cohomology, we want to know what its relation is to sheaf cohomology with the corresponding complex of sheaves constructed in the examples above. 
	
\begin{theorem}\label{DeRham} 
	Let $M$ be a $C^\infty$-manifold. Then, we have the following isomorphism: 
		\[ H_{DR}^i(M)\cong H^i(M,\tilde{\R})\]
		where $\tilde{\R}$ is the constant sheaf on $M.$ 
\end{theorem}
\noindent This follows from the constructions above. For more details see \cite{Wedhorn2016}.  

The reason we care about this theorem is that it gives an analytic interpretation of singular cohomology. By Theorem \ref{sheaf_singular}, we have that de Rham cohomology is isomorphic to singular cohomology. Therefore, de Rham cohomology is encoding topological information about the manifold. Further, this isomorphism gives another way to compute sheaf cohomology. 

This ends the Chapter as well as the background material. We encourage the motivated reader to spend some time understanding the final sections here as they are both technical and widely applicable. They will be useful in understanding Chapter 4, as well as some recent claims of computational neuroscientists on the construction of geometric frameworks for perceptual spaces via homology and cohomology.

\chapter{A Geometric Framework for Olfactory Learning and Processing}
\subsection*{Abstract}
We present a generalized theoretical framework for olfactory representation, learning, and perception using the theory of smooth manifolds and sheaves. This framework enables the simultaneous depiction of sampling-based physical similarity and learning-dependent perceptual similarity, including related perceptual phenomena such as generalization gradients, hierarchical categorical perception, and the speed-accuracy tradeoff. Beginning with the space of all possible instantaneous afferent inputs to the olfactory system, we develop a dynamic model for perceptual learning that culminates in a perceptual space in which qualitatively discrete \textit{odor} representations are hierarchically constructed, exhibiting statistically appropriate consequential regions ("boundaries") and clear relationships between the broader and narrower identities to which a given stimulus might be assigned.  Individual training and experience generates correspondingly more sophisticated odor identification capabilities. Critically, because these idiosyncratic hierarchies are constructed from experience, geometries that fix curvature are insufficient to describe the capabilities of the system.  In particular, the use of a hyperbolic geometry to map or describe odor spaces is contraindicated.  

\section{Introduction}
The task of sensory systems is to provide organisms with reliable, actionable information about their environments.  However, such information is not readily available; the environmental features that are ecologically relevant to an organism are rarely directly evident in primary receptor activation patterns.  Rather, these representations of interest must be \textit{constructed} from the combined signals of populations of sensory receptors.  This construction process is mediated by sophisticated networks of neural circuitry that draw out different aspects of potentially important information from the raw input patterns.  We previously have proposed that these interactions and transformations can be most effectively modeled as a \textit{cascade of successive representations} \cite{Cleland2014}, in which each neuronal ensemble constructs its representation by sampling the activity of its antecedents.

The representational cascade that underlies odor recognition and identification is impressively powerful and compact.  Olfactory bulb circuits impose an internally generated temporal structure on afferent inputs \cite{Li_Cleland2017,Li_Cleland2013,Kashiwadani1999,Bathellier2006} while also regulating contrast \cite{Cleland2006}, normalizing neuronal activity levels \cite{Cleland2011,Banerjee2015,Cleland2020Engin}, and managing patterns of synaptic and structural plasticity \cite{Chatterjee2016,Strowbridge2009,Gao2009}. Transient periods of synchronization with postbulbar networks such as piriform cortex are likely to govern interareal communication \cite{Fries2015,Frederick2016,Kay2014}, including feedback effects on bulbar plasticity \cite{Strowbridge2009,Gao2009}.  The resulting perceptual system learns rapidly and is conspicuously resistant to retroactive and compound interference \cite{Herz1996, Stevenson2007}. Odors of interest also can be readily identified despite direct interference from simultaneously encountered competing odorants; this is a major unsolved problem in olfactory neuroscience, as competition for receptor binding sites by multiple odorant species profoundly degrades the odorant-specific receptor activity profiles on which odor recognition ostensibly depends.  We have constructed olfactory circuit models that learn rapidly, resist retroactive interference, and exhibit robust recall under high Bernoulli-Gaussian noise (which models a combination of sampling uncertainty, innate stimulus variance, and high levels of unpredictable competitive interference from other ambient odors) using a strategy of successive recurrent representations shaped by prior learning \cite{Imam2020}.  The success of this approach accentuates the implications of the profound plasticity of the early olfactory system:  odor representations, and the basic function of olfactory perception itself, are fundamentally and critically dependent on learning \cite{Wilson2003,WilsonBook2006,Royet2013}.  Odors and their implications -- excepting a few species-specific innately recognizable odors -- must be learned through individual experience.  Indeed, there is abundant evidence for the perceptual learning of meaningful odor representations, from their generalization properties \cite{Cleland2009} to the mechanisms of odor learning and memory \cite{Tong2014,Vinera2015,Wilson2003,Mandairon2011,Kermen2010}, to the association of odors with meaning and context even in peripheral networks \cite{Doucette2008,Nunez2014,Ramirez2018,WilsonBook2006,Mandairon2014,Herz2005,Aqrabawi2018a,Aqrabawi2020,Levinson2020}.  What we lack is a common theoretical framework in which all of these phenomena can be usefully embedded. 

\subsection{Perceptual frameworks}
Theoretical frameworks for understanding sensory systems include \textit{perceptual spaces} and \textit{hierarchical structures}.  Both are founded on metrics of similarity \cite{Zaidi2013,Edelman2012,Shepard1987,Clapper2019}, though the former presumes an essentially continuous space of some dimensionality into which individual stimulus representations are deployed, whereas the latter presumes some degree of qualitative category membership for each such representation, with intercategory similarities potentially being embedded in the hierarchical proximities among categories.  Perceptual spaces can be defined  using a variety of metrics, including both physical metrics such as wavelength (color) or frequency (pitch) and perceptual metrics such as those revealed by generalization gradients \cite{Shepard1987,Cleland2009,Cleland2002} or by ratings on continuous scales by test subjects.  Indeed, study of the transformations between physical and perceptual metric spaces is foundational to understanding sensory systems from this perspective \cite{Zaidi2013,Meister2015,Victor2017}.  In contrast, hierarchical structures arise from perceptual categorization processes, though relationships among the resulting categories still may respect underlying similarities in the physical properties of stimuli (see \textit{Discussion}). Critically, it is categories that are generally considered to be embedded with associative meaning (\textit{categorical perception}) \cite{Harnad1987,Goldstone2010,Aschauer2018}; a useful theoretical framework must concern itself with the construction of these categories with respect to the physical similarity spaces that are sampled during sensory activity.  That is, along their representational cascades, sensory systems can be effectively considered to transition from a physical similarity space metric to a perceptually modified space, arising from perceptual learning and within which hierarchical categorical representations can be constructed. 

Interestingly, the olfactory modality lacks a clear, organism-independent physical metric such as wavelength or pitch along which the receptive fields of different sensory neuron populations can be deployed (and against which the nonuniform sampling properties of the sensory system can be measured) \cite{Cleland2014}.  However, olfaction does provide an objective basis for an \textit{organism-dependent} physical similarity space.  In this framework, the activity of each odorant receptor type -- e.g., each of the $\sim400$ different odorant receptors of the human nose or the $>1000$ different odorant receptors of the rodent nose -- comprises a unit dimension.  Specifically, the instantaneous activation level of the convergent population of each receptor type provides a value from zero to one (maximum activation), such that any possible odorant stimulus can be defined as a unit vector embedded in a physical metric space with dimensionality equal to the number of receptor types.  Critically, in this framework, (1) the dimensions of this receptor-based metric space ($R$-space; see below) are linearly independent of one another, and (2) every possible profile of receptor activation, including any occluding effects of multiple agonists and antagonists competing for common receptors, is interpretable.  

Linear independence among the dimensions of $R$-space is important for analytical purposes, but their orthogonality is irrelevant \cite{Cooperstein2015}.  This is a vital distinction, not least because orthogonality depends on the statistics of the chemosensory input space and hence cannot be uniquely defined as a property of the olfactory system per se.  In principle, each receptor type should have regions of its receptive field that distinguish it from any other single receptor type, such that activation of a given receptor need not always imply activation of a particular different receptor (that is, no two dimensions will be identical).  However, within any given sensory world, as defined by a finite set of odorant stimuli with established probabilities of encounter, there will be reliable activity correlations among many pairs of receptor types that can support substantial dimensionality reduction \cite{Haddad2010}. Critically, however, these reduced dimensionalities are not characteristic of the olfactory system per se, as they are strongly reflective of the statistics of the stimulus set used and its particular interactions with the deployed complement of receptors. 

Such dimensionality-reduction efforts also have been applied to olfactory perceptual data \cite{Koulakov2011,Castro2013,Zarzo2006}.  These results also are limited by the statistics of stimulus sets, but additionally engage the problem of just what the olfactory system constructs from the space of its underlying physical inputs.  It is reasonably clear (even axiomatic) that the  sampling of physical odorant spaces is not uniform \cite{Castro2013,Koulakov2011} -- that is, odors are \textit{signal sparse} \cite{Berke2017} -- but, perhaps more importantly, the process of odor learning itself directly affects perceived olfactory similarity relationships within a context of learned generalization gradients \cite{Cleland2009,Cleland2011}.  A general framework for olfactory perception must reflect all of these phenomena, embedding physical and perceptual similarity spaces into a common geometrical framework that admits the construction of experience-dependent perceptual categories.  

\section{The geometries of olfaction}
In addition to dimensionality, the second fundamental property of a sensory space is its intrinsic geometry \cite{Zaidi2013}.  Establishing a geometry provides access to theorems by which representational structures can be formally defined and manipulated. However, it is not necessary to restrict the topology of a sensory space to a single geometry with fixed curvature -- in fact, as we indicate below, this is neither advisable nor ultimately possible for the olfactory modality.  Specifically, we here show that a mature olfactory perceptual space cannot be simply characterized as "Euclidean", "hyperbolic", or otherwise, as its dependence on plasticity instead produces a space comprising discrete regions that exhibit the properties of different geometries.  

We here present a generalized geometric framework for the construction of odor representations. The framework is based on the molecular/physiological encoding capacities of the input layer and the transformation of these physiological \textit{odorant} representations by perceptual learning into meaningful, cognitive \textit{odor} representations to which meaning can be ascribed.  Key features include the simultaneous depiction of sampling-based physical similarity and learning-dependent perceptual similarity within the perceptual space, a basis for the speed-accuracy tradeoff \cite{Frederick2017,Rinberg2006,Zariwala2013,Abraham2004}, and a soft hierarchical categorization process that culminates in a perceptual space in which qualitatively discrete \textit{odor} representations are hierarchically constructed through experience, exhibiting statistically appropriate consequential regions with probabilistic boundaries that reflect learned generalization gradients \cite{Cleland2009,Cleland2011,Shepard1987}.  Critically, individual training and experience generates progressively more sophisticated hierarchies and concomitantly superior odor identification capabilities \cite{Royet2013}.  

A simplified illustration of the analytical framework is depicted in Diagram \ref{Diagram1}. Briefly, the space of instantaneous physical inputs to an olfactory receptor activation space ($R$-space) comprising $N$ receptor types can be depicted as an $N$-dimensional unit cube.  Transformations arising primarily from initial post-sampling computations generate a modified receptor space termed $R'$; this space inherits the dimensionality of $R$-space but respects the nonuniform likelihoods of different state points within that space.  The subsequent transformation from $R'$ to $S$-space ("scent space") reflects the perceptual and categorical learning processes that construct perceptual representations of meaningful \textit{odors}.  

\vspace{0.5cm}
\begin{equation}\label{Diagram1} 
\begin{tikzcd}
R\arrow[r,"B"] &R' \arrow[dr,"\xi"] \arrow[d,swap,"\Delta"] \\
&S& M \arrow[l,"C^\infty(\R^m)"] 
\end{tikzcd} 
\end{equation}
\text{}

Formally, $R$ is a unit parallelepiped defined by primary olfactory receptor activation levels. $R'$ denotes a subspace of normalized points, following glomerular processing, and is the image $B(R).$  $M$ is a vector bundle over $R'$ of rank $\#$mitral cells and is generated by the output of mitral cells. $\xi$ denotes the input to mitral cells following glomerular processing, comprising a sparsened, statistically conservative manifold; it is a section of the vector bundle $M$. $S$ denotes the perceptual space, and is realized as a transformation of $R'$-space that embeds odor learning.

Importantly, this theoretical model is broadly independent of precisely where in the olfactory representational cascade these computations take place.  However, we consider that the map $B$ from $R$-space to $R'$-space reflects signal conditioning computations performed within the glomerular layer of the olfactory bulb \cite{Cleland2014,Cleland2020Engin}, whereas the subsequent transformation into $S$-space is mediated by computations within the olfactory bulb external plexiform layer network \cite{Imam2020}, inclusive of its reciprocal interactions with deeper olfactory cortices.  Briefly, we propose that the construction of categorical odor representations through statistical experience arises from learning-dependent weight changes between mitral cell principal neurons and granule cell interneurons in the external plexiform layer of the olfactory bulb.  In this theory, plastic interactions between these two populations construct meaningful, categorical \textit{}{odor} representations from the continuous, physical \textit{odorant} representations of $R'$-space based upon individual experience.  To construct this theoretical $S$-space, and attribute to it the capacities of generalization, speed-accuracy tradeoff, and experience-dependent hierarchical categorization, we first build a transitional space $M$ based on mitral cell activity representations, inclusive of the actions performed on these representations via their interactions with granule cell interneurons (Diagram 1).  This resulting $S$-space does not, indeed cannot, admit a single geometry, because of the essential requirement for locally adaptable curvature.  We describe this generative process in detail below. 

\setcounter{theorem}{0}
\subsection{$R$-Space}
The first representational structure in olfaction is directly derived from the ligands of the physical odorant stimulus interacting with the set of chemoreceptive fields presented by the animal's primary odorant receptor complement.  Both vertebrate and arthropod olfactory systems are based on large numbers of receptor neurons, each of which expresses one primary odorant receptor out of a family of tens (in \textit{Drosophila}) to over 1000 (in mice, rats, and dogs).  The axons of primary sensory neurons expressing the same receptor converge together to form discrete \textit{glomeruli} across the surface of the olfactory bulb (in vertebrates; the arthropod analogue is the antennal lobe), enabling second-order projection neurons (mitral cells) to sample selectively from one or a few receptor types.  The response of each receptor type
to an odor stimulus constitutes a unit vector that can range in magnitude from nonresponsive (0) to maximally activated (1).  A complete representational space for instantaneous samples of this input stream consequently has a dimensionality equal to the number of odorant receptor types $N$.  That is, in a species with three odorant receptors, the space containing all possible instantaneous input signals would be a three-dimensional unit parallelepiped (depending on the original placement of the vectors in 3-space), whereas the $R$-space of a mouse expressing 1000 receptor types would comprise a 1000-dimensional unit space.  As noted above, it is not necessary that these vectors be  orthogonal, only that they be linearly independent \cite{Cooperstein2015}; indeed, the orthogonality of these vectors cannot even be defined without reference to the statistics of the particular physical environment in which they are deployed.  

Formally, $R$-space is defined as the space of linear combinations of these vectors with coefficients in $(0,1).$ Consider the space of all possible odorant stimuli in a species expressing $N$ odorant receptor classes. Each odorant stimulus $s^*$ corresponds to a unique instantaneous glomerular response profile that can be represented as a vector $s^*\in \R^N$.  Normalizing the activation in each glomerulus enables us to consider $s^*\in \prod^n(0,1)$, the unit cube in $N$ dimensions. Denote this receptor activation-based representational space $R.$ Because the tangent space at all points is $T_xR\cong \R^N$, $R$ has dimension $N$ as a manifold.  

By considering a product of spaces, we are assuming that the responses of different glomeruli are orthogonal. In the greatest generality, we would need to consider points on a unit parallelepiped generated by the glomeruli. We can apply an invertible linear transformation (namely the matrix generated by the Gram-Schmidt process) to this parallelepiped to generate a cube (and vice-versa); this is a mathematical formalism and does not affect the particulars of this situation. Consequently, for  the  remaining sections, we can assume without a loss  of  generality that $R=\prod^n(0,1).$  

\subsection{Glomerular-layer computations, $R'$}
The first computational layer of the olfactory bulb -- the glomerular layer -- computes a number of transformations important for the integrity and utility of odor representations, including contrast enhancement \cite{Cleland2006}, global normalization \cite{Cleland2011,Banerjee2015}, and potentially other effects \cite{Borthakur2019}.  These processes substantially alter the respective probabilities of the points in $R$-space; for example, global feedback normalization in the deep glomerular layer ensures that the points at which most or all of the vectors have very high values will be improbable.  The outcome of this transformation is represented as $R'$, essentially a manifold embedded in $R$-space. 

In addition to the systematically unlikely points in $R$ that are omitted from the manifold $R'$, it is also the case that, under natural circumstances, most of the possible sensory stimuli $s^*$ that could be encountered in $R'$ actually never will be encountered in an organism's lifetime.  That is, odor representations within $R$-space are \textit{signal sparse} \cite{Berke2017}. Moreover, we argue that \textit{odor sources} $s^*$ are discrete, but inclusive of variance in quality and concentration, and hence constitute \textit{volumes} (manifolds) within $R'$. To account for this, we denote this variance by $s^*=(x,U_x)$, where $x\in R'$ and $U_x$ denotes an $n$-tuple of variances (i.e., one variance for each dimension of freedom in $R'$). That is to say, \[U_x=(\sigma_1^2,...,\sigma_n^2)\]
From this we arrive at the following definition:

\begin{definition}
    A pair $(x,U_x)$ constitutes an \textbf{odor source volume} in $R'$ if $U_x\neq 0$ and $(x,U_x)=s^*$ for some odorant $s^*.$  
\end{definition}

\noindent
That is, an odor source volume corresponds to a manifold within $R'$ that comprises the population of odorant stimulus vectors arising from the range of variance in receptor activation patterns exhibited by a particular, potentially meaningful, odor source.  This includes variance arising from nonlinearities in concentration tolerance mechanisms that cannot be completely avoided \cite{Cleland2011} as well as genuine quality variance across different examples of a source.  For example, the odors of \textit{oranges} vary across cultivars and degrees of ripeness; the odors of \textit{red wines} vary across grape cultivars, terroir, and production methods.  The source representation in $R'$ thereby corresponds to an odor source (e.g., orange, red wine), inclusive of its variance, and delineates the consequential region of the corresponding \textit{odor} category that will be developed via perceptual learning.  Critically, it is not important at this stage to specify multiple levels of organization within odor sources (e.g., red wine, resolved into Malbec, Cabernet, Montepulciano, etc., then resolved further by producer and season); it is the process of odor learning itself that will progressively construct this hierarchy of representations at a level of sophistication corresponding to individual training and experience. 

\subsection{$M$-Space}
The transformation from $R'$ to $S$-space depicted in Diagram 1 is mediated by the interactions of mitral and granule cells.  In this framework, mitral cells directly inherit afferent glomerular activity from $R'$ (Diagram 1, $\Delta$), but their activity also is modified substantially by patterns of granule cell inhibition that, via experience-dependent plasticity, effectively modify mitral cell receptive fields to also incorporate higher-order statistical dependencies sourced from the entire multiglomerular field.  (A simplified computational implementation of this constructive plasticity is presented in the learning rules of Imam and Cleland, 2020).  This is depicted in Diagram 1 as an effect $C^\infty(\R^m)$ of a mitral cell product space $M$ which contributes to the construction of $S$, in order to highlight the smooth deformations of $R'$ into $S$ via passage to $M.$    

This effects of mitral cell interactions, arising from experience, are modeled locally as a product space $M$ based on the principle that each glomerulus -- corresponding to a receptor type in $R'$ -- directly contributes to the activity of some number of distinct mitral cells.  In the mammalian architecture (shared by some insects, including honeybees), mitral cells receive direct afferent input from only a single glomerulus, such that the afferent activity in each mitral cell (or group of sister mitral cells) corresponds directly to a single receptor type.  In this "naive" case, $M$-space is globally a product.  To formalize this, we label the glomeruli $g_1,...,g_q.$ To each, we associate the number of mitral cells to which it projects; denoted $m_i\in \Z.$ Let $k=\sum^q m_i.$ Then, the naive space constructed from these data is \[R'\times \R^k=\{(r,v):r\in R',v\in \R^k\}     \]  
The interpretation of this space is as follows:  to each point in $R'$, we can associate a vector that is an identifier for how subsequent mitral-granule cell interactions in the olfactory bulb will transform the input in service to identifying it as a known percept.  The manifolds associated with particular odor source volumes in $R'$ will, owing to experience-dependent plasticity, come to exhibit related vectors that, in concert, manifest source-associated consequential regions.  These regions reflect categorical perceptual representations and are measurable as odor generalization gradients.  Simplified computational implementations have depicted these acquired representations as fixed-point attractors, tolerant of background interference and sampling error but lacking explicit consequential regions \cite{Imam2020}.

We refer to this space as \textit{naive} because it is globally a product space only for the mammalian architecture, in which the dimensionality of mitral cell output $m$ (the number of distinct mitral cells, grouping sister mitral cells together) is identical to that of glomerular output $k$.  However, this network architecture is not general; in nonmammalian tetrapods, for example, individual mitral cells may sample from more than one glomerulus \cite{Mori81a,Mori81b}.  This introduces a twist into the product space and ruins the naive structure, as $m$ now can be less than $k$.  In this general case where $m \leq k$, the mitral cell space becomes a rank $m$ vector bundle \[\R^m\hookrightarrow M \overset{\pi}{\rightarrow} R'\] over $R'$.  Nevertheless, it can be depicted locally as a product space because vector bundles are locally trivializable. Given any odor source volume $(x,U_x)$ we know that there exists either a subset $U'\subset U_x$ such that $\pi^{-1}(U')\cong U'\times \R^m$ or $U'\supseteq U_x$, then we can look exclusively that $U_x\times \R^m\cong \pi^{-1}(U')$ and this is a trivial bundle over the base.   

For simplicity, we here analyze the mammalian architecture case. In this architecture, the vector bundle is trivial because $m=k$; no mitral cells innervate multiple glomeruli, and there is no possible twisting of the fibers. Therefore, in mammals, $M$ is globally a product space,  \[M=R'\times \R^m\] 
rendering $M$ a smooth manifold with the convenient property that to every input $x\in R'$ we associate a point $(x,v)$, where $v$ is a vector whose $i^{th}$ component is the value of the output of the $i^{th}$ mitral cell.  Formally, we say that $M$ is a (trivial) vector bundle over $R'$ with fibre $\R^m.$ Then, the smooth maps which send $x\mapsto (x,v)$ such that composition with projection onto the first coordinate is the identity are called global smooth sections of the bundle, and the set of these is denoted $\Gamma(R',M).$ To any smooth manifold $P$, we can associate the ring of smooth functions \[C^\infty(P)=\{f:P\to \R:f \text{ is smooth}\}\] To any open subset, we have a restriction map $\res^P_U:C^\infty(P)\to C^\infty(U).$ In general, if $U\subseteq P$ is open, then $\Gamma(U,E)$ is a $C^\infty(U)-$module for any bundle $\pi:E\to P.$ $C^\infty(-)$ makes $P$ into a locally ringed space and $\Gamma(-,E)$ is a sheaf of $C^\infty(-)$-modules.

\subsection{$S$-Space}
%Now that we have the base spaces, $R'$ and $M,$ as well as the action $GC,$ we can construct our space of interest.  Goal being that physical and perceptual metrics can interplay nicely.  

$S$-space, or scent space, is a constructed perceptual space tasked with preserving physical relationships among odorants while also embedding the transformations arising from perceptual learning, specifically including those forming incipient categorical \textit{odors}.  To do this, we embed $R'$ into a higher-dimensional space (with dimension $N+1$).  Under this embedding, we represent perceptual learning in $S$ by growing $U_x$ in the positive $N+1th$ direction around odor source volumes in $R'$, which does not affect distance relationships in ${\R^N}$ (Figure~\ref{Landscapes}A).  (Discrimination training also can grow $U_x$ in the negative $N+1th$ direction).  To quantify this transformation, we construct two distance metrics, $d_{phys}$ and $d^{per}$ on $S$. 

\begin{definition}
	Let $x,y\in S$ be two points.  We define the \textit{physical metric} between the two points as the Euclidean distance between them in $R.$ In notation, \[d_{phys}(x,y)=|\pi_{\R^N}(x)-\pi_{\R^N}(y)|  \] 
\end{definition}
\noindent This metric reflects the physical similarities of the objects in the receptor space, which are not affected by perceptual learning (i.e., distension in $N+1$).  
\begin{definition}
	Let $x,y\in S.$ Consider $x$ and $y$ as vectors in $\R^{N+1}$. Then, let $\gamma:[0,1]\to S$ be the curve defined by $\gamma(0)=x$, $\gamma(1)=y$ and $\pi_{\R^{N}}(\gamma'(t))=w\cdot\left[\pi_{\R^N} \left( \gamma(1)-\gamma(0)\right) \right]$ with $w$ some real number dependent on $t$. The \textit{perceptual metric}, \[d^{per}(x,y)=\int_0^1||\gamma'(t)||dt  \] is the arc-length along the surface of $S$ between the points $x$ and $y$ (Figure~\ref{Landscapes}A). Notice that $\pi_{\R^N}(\gamma')$ is well defined as $S\hookrightarrow \R^{N+1}$ and thus the tangent space $T_{\gamma(t)}S\subseteq T_{\gamma(t)}\R^{N+1}=\R^{N+1}.$ 
\end{definition}  

% Warning received when using [h] instead of [ht]:  The float specifier 'h' is too strict of a demand for LaTeX to place your float in a nice way here. Try relaxing it by using 'ht', or even 'htbp' if necessary. If you want to try keep the float here anyway, check out the float package.
\begin{figure}[h!]
  \centering
  \includegraphics[width=\linewidth]{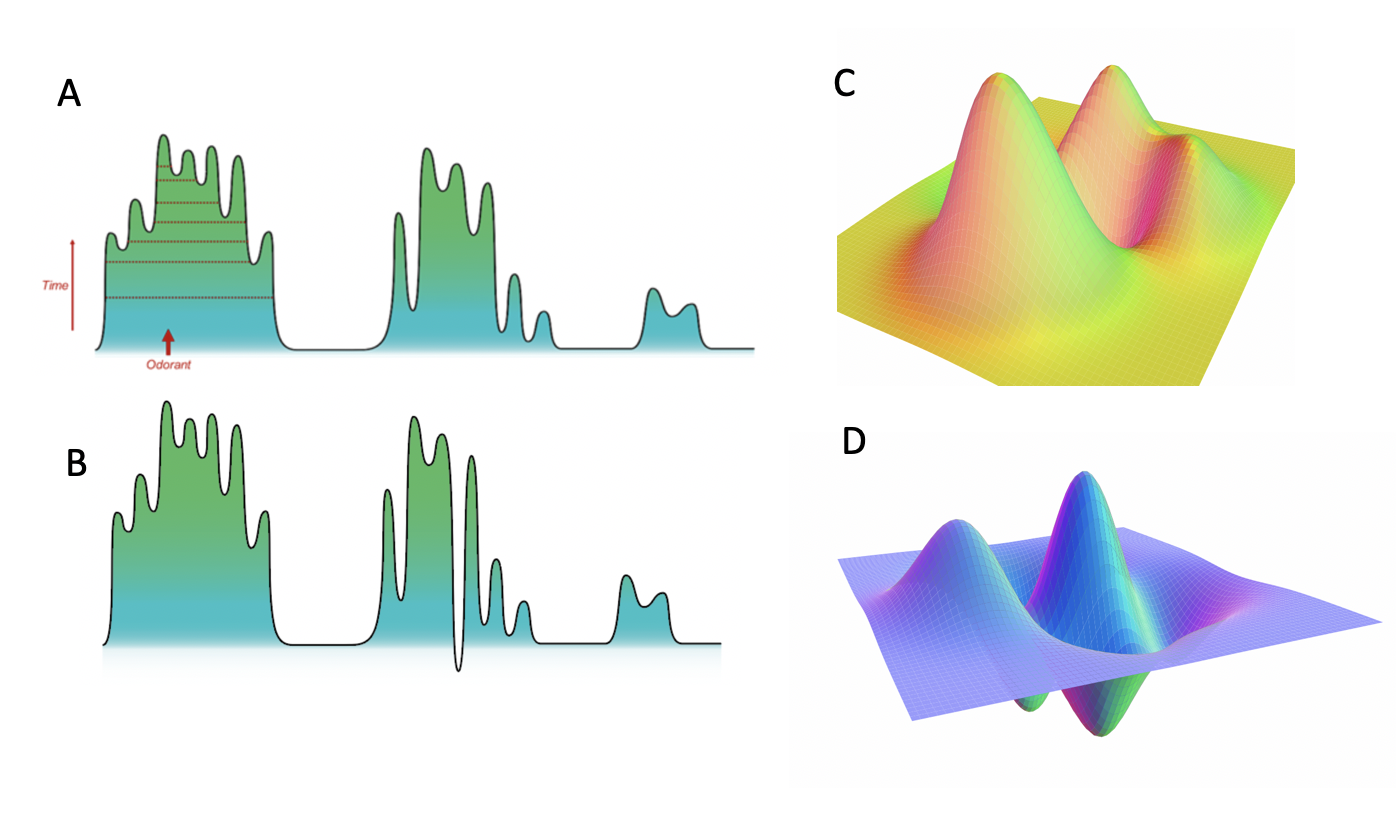}
  % PDF does not seem to work with '6in' scaling, but preserves resolution better. 
  % \rule[-.5cm]{4cm}{4cm} \rule[-.5cm]{4cm}{0cm}
  \caption{\textbf{Depictions of $S$-space in the cases of $N=1,2.$} (A) Three distinct odors in $S$-space in the case of $N=1.$ Going left to right, the first odor is highly learned with many distinct \textit{sub}-odors. Further, it is decorated with a distinction of a specific odor and the time axis. Per the discussion below, each red dotted line represents the formation of equivalence classes of odors at a given time. As time increases, specificity increases and this is reflected in the diagram. The second odor is overall less learned than the first, yet the first two sub-odors are known to be distinct as shown by the large valley between them. The third odor is poorly learned. (B) After learning has occurred, a valley has been created between the two sub-odor classes in the second odor. As the valley extends below the original line, we know that these two sub-odors are perceptually very different. (C)-(D) Depictions of a part of $S$-space for the case $N=2.$ Various amounts of learning have generated the landscapes presented. }  
  \vspace{0.5cm}
  \hrule
  \label{Landscapes}
\end{figure}

The relationship between these two metrics tracks the changes in $S$ induced by the construction of odor representations; specifically, $d^{per}$ reflects experience-dependent changes in the perceptual distance between $x,y\in S$ that are excluded from the $d_{phys}$ metric (Figure~\ref{Landscapes}A). Learning about an odor source $(x,U_x)$ progressively distends the volume (in $\R^N$) in the $N+1$ direction; over time, the shape of this distension will reflect the odor source volume in $R'$.  That is, over time, the breadths (in each of the $N$ dimensions) of the distension into the additional ($N+1th$) dimension will come to reflect the actual variances $U_x$ of the odor source $s^*=(x,U_x)$ as naturally encountered.  The quasi-discrete distensions formed in the additional dimension correspond to incipient categories -- i.e., categorically perceived \textit{odors} -- and their breadths and gradients can be measured behaviorally as generalization gradients \cite{Cleland2009,Cleland2011}. Importantly, the variance for each dimension of freedom of $U_x=(\sigma_1^2,...,\sigma_n^2)$ in $R'$ is independent; that is, different samples of a given natural odor source may vary substantially in some aspects of quality but not others, where an aspect of quality refers to the relative levels of activation of a given odorant receptor type (Figure~\ref{Landscapes}B). 

Formally, to construct the perceptual space $S$ in such a way that there exists a perceptual metric $d^{per}$ that interacts with the natural physical metric $d_{phys}$ of $R'$, we consider the embedding $R'\hookrightarrow \R^{N+1}.$ The open neighborhoods for each odor source volume define open sets in the subspace topology. If we embed $R'$ by the canonical inclusion $\R^N\to \R^{N+1},$ then $R'$ is flat in $\R^{N+1}$ because the final coordinate of its elements is $0.$ Therefore, we can consider transformations of $R'$ that smoothly vary the final coordinate. For each transformation $f$, denote the resulting space as $S:=S(f).$ This constitutes the evolving perceptual space. Define the map $\Delta:R'\to S$ as the distension of $R'$ in $N+1$ (Diagram 1). This map arises from considering $M$ and $R'$ simultaneously, and is a diffeomorphism trivially.

 To better understand the map $\Delta,$ we here construct it as the composition of maps among the spaces already described, specifically showing how the (acquired) properties of $M$  govern the mapping of $R'$ to $S$. The map $B:R\to R'$ reflects glomerular-layer transformations as described above. For a fixed smooth section $\xi:R'\to M$ (which always exists by the triviality of $M$), we generate Diagram 2 (an elaboration of Diagram 1),

\begin{equation}\label{Diagram2} \begin{tikzcd}
R \arrow[r, "B"]  & R' \arrow[d,swap, "\Delta(f)"] \arrow[dr,"\xi"]&\\
           & S  &M \arrow[l,"\text{id}_{R'}\times f"] 
\end{tikzcd} \end{equation}
where $\Delta(f)$ is defined to be a map that makes the diagram commute. 
Note that $\Delta$ depends on $f,$ and, therefore, so does $S$. That is, $S$ depends on the functions $\R^m\to \R$ from $M,$ which are smooth. To allow for ongoing plasticity, it is more correct to denote the perceptual space as $S:=S(f);$ however, as it will always be clear from context whether or not $f$ is fixed, we will simply refer to it as $S.$  The map $id\times f$ reflects the fact that $R'\subseteq \R^{N+1},$ and by construction $x_{N+1}=0$ for all $x\in R'.$ As $M=R'\times \R^m$, it follows that a dense set of maps $M\to \R^{N+1}$ which are the identity on $R'$ can be split as maps $i:R'\to \R^N$ and $f:\R^m\to \R.$ Therefore, because $\text{id}_{R'}\times C^\infty(\R^m)=C^\infty(\R^m),$ we abbreviate the collection of all maps $M\to S$ as $C^\infty(\R^m)$, as depicted in Diagram 1.  

The outcome of these transformations is a formal definition for the construction of categorical \textit{odor representations} in $S$:  

\begin{definition}
    Let $(x,U_x)$ be an odor source volume in $R'.$ We denote the image of this volume in $S$ as $(x,\widetilde{U_x}).$ This image denotes an $\textit{odor representation}$, also referred to as an $\textit{odor percept}$, or simply an  $\textit{odor}.$
\end{definition}

%FOOTNOTE HOWTO:  ...the collection of all maps $M\to S$ as $C^\infty(\R^m).$\footnote{This is consistent with our notation above as $\text{id}_{R'}\times C^\infty(\R^m)=C^\infty(\R^m).$}

\subsection{Forms and timescales of odor learning}

The construction of odor representations $(x,\widetilde{U_x})$ in $S$ enables the depiction of learning as a geometric object, naturally encompassing the transition between the physical and perceptual space depictions of the olfactory landscape and illustrating the construction of meaningful categorical odor representations based on individual experience. As we describe below, these odor representations admit hierarchy and exhibit the advantages of categorical perception. However, they remain continuous in $S$, with consequential regions that are not discretely delimited; i.e., olfactory perceptual categorization is ultimately heuristic.  This affords some powerful advantages.  For example, it provides a natural basis for behaviorally observed odor generalization gradients \cite{Linster1999,Cleland2002,Cleland2009,Cleland2011}, and enables incipient proto-categories to merge once the variance structure $U_x$ of the odor source indicates that different samples fall within a common, relatively broad, distribution with shared implications \cite{Cleland2011}.  As discussed below, this framework also admits the ongoing construction of hierarchies within odor representations, while retaining a natural basis by which to fall back to more general levels of the hierarchy when signal quality is low.   

Importantly, odor learning as depicted here incorporates only the progressive learning of categorical odor percepts that enable the subsequent association of olfactory sensory information with its broader implications -- that is, what we argue is the bulbar component of odor learning.  Specifically, in our present framework, we allow for, but do not directly describe, the multitude of these implications and their diverse effects upon perceptual learning.  For example, the perceptual odor representations that may arise from unrewarded experience are likely to differ from those that arise from reinforcement learning. Learned response generalization from punishment-associated odor stimuli is broader than that arising from reward-associated odor stimuli.  Different olfactory conditioning paradigms may promote either generalization or discrimination among different odorant stimuli, corresponding to the construction of different perceptual distances $d^{per}$ between them.  Each of these distinct and specialized modes of learning is considered to transform the plasticity-dependent distensions into dimension $N+1$ in specifically appropriate ways, here proposed to be governed by ascending inputs into the olfactory bulb from higher association cortices such as the AON and piriform cortex that regulate bulbar plasticity via reciprocal interactions, thereby modifying granule cell feedback effects on mitral cell activity as described below\color{black}.  

Finally, the robustness of odor memories is a factor that the present framework incorporates but does not specifically address.  It is established that some odor learning -- particularly unrewarded exposures yielding odor-specific habituation -- persists for mere seconds or minutes \cite{Linster_Menon_Singh_Wilson_2009}, whereas other odor learning persists for days \cite{Tong_Peace_Cleland_2014} or a lifetime \cite{Saghatelyan_deChevigny_Schachner_Lledo_2004}. Mechanistically, long-term odor memories are associated with protein synthesis in olfactory bulb \cite{Tong_Peace_Cleland_2014} and also with the incorporation of new adult-generated interneurons into the bulbar network, which appears to occur only after a certain amount of time spent learning.  The reliability of a given stimulus over time, with consistent associated implications, is of course a critical factor in animal learning.  In the present framework, we consider that individual learned distentions in dimension $N+1$ of $S$-space will be variously persistent, either fading back towards flatness with a given time constant or enduring indefinitely, according to learning-dependent temporal tags that are not explicitly discussed herein.  

\subsection*{The geometry of local plasticity}
Plasticity in neural systems in general, and in the olfactory bulb in particular, is locally governed.  Changes in cellular and synaptic functional properties rely substantially on the synaptic interactions of directly connected neurons and the locally regulated release of neurochemicals.  These local effects, coordinated by sophisticated network interactions, collectively generate global systemic performance at the network level.  The present odor learning framework also arises from localized plasticity: distensions into the additional ($N+1th$) dimension of $S$ arise from learning the activity profiles of individual sensory inputs, and are not globally governed (specifically, we argue that this arises from learned patterns of granule cell feedback onto mitral cells in olfactory bulb; for a simplified computational implementation of this process, see \cite{Imam2020}).  However, to characterize the functionality of the olfactory system as a whole, it is necessary to formally glue such local plasticity operations together, along with any relevant global processes, within a single analytical framework.  To do this, we employ the theory of sheaves \cite{Wedhorn2016}.

\subsubsection{Sheaves enable localized learning}
We formally consider the local actions of granule cells onto mitral cells, and their concomitant modification of mitral cell output, as follows, considering that these actions may rely both on afferent sensory information and on additional inputs delivered onto granule cells by piriform cortex and other association cortices \cite{Isaacson_Strowbridge_1998}.  Recall from the previous section that for any vector bundle $\pi:E\to P,$ we generate $(C^\infty(-),\Gamma(-,E))$ a pair of sheaves on $P$ such that $\Gamma(U,E)$ comes equipped with an action of $C^\infty(U)$ for all open $U\subseteq P.$
We here define an analogous pairing of sheaves to describe the modification by granule cells of afferent information contained in the mitral cell ensemble. The first step in this definition is to define a functor \[ \mu:\mathcal{T} \to \textbf{R}^m\]
where $\mathcal{T}$ is the category defined by the topology on $R',$ and $\textbf{R}^m$ is the set whose objects are linear subspaces of $\R^m$ and morphisms \[ \text{Mor}_{\textbf{R}^m}(U,V)=\begin{cases} \varnothing & U\not\subseteq V\\
\{*\} & U\subseteq V
\end{cases}\]  
To describe what this functor does, we need to turn to the anatomy of bulb. For a given odorant, the induced signal passed from glomeruli to mitral cells may not excite some mitral cells. This corresponds to the situation where $\xi(s)=(s,v)$ and $v$ has some coordinates equal to $0.$ These non-zero coordinates form a basis for some subspace of $\R^m.$ Let $n(\xi,s)$ be the number of non-zero coordinates in $v.$ Let $O$ be any open subset of $R'.$ Then \[\mu(O)=\R^\ell\]
where $\ell=\max\{n(\xi,p): p\in O\}.$ Composing $\tau$ and $C^\infty$ and using the sheaf condition of $C^\infty$ we conclude that $C^\infty(\mu(-))\in \textbf{Sh}(R').$ Now, we define $G(-)$ as a flabby (flasque) sheaf of rings on $R'$ which act on $C^\infty(\mu(-)).$
This action is precisely the interaction of local inhibition on mitral cells, and in particular on those mitral cells that are activated by a given odorant stimulus. This makes $C^\infty(\mu))$ a $G$-module (as sheaves).   

\subsubsection{Localized discrimination learning}
Learning about an odor is generally modeled as growing a distension into the additional ($N+1th$) dimension of $S$, with the breadths of the distension across its $N$ dimensions ultimately reflecting the physical profile of quality variance $U_x$ associated with the corresponding odor source $s^*=(x,U_x)$. This category-construction framework can be considered common to diverse forms of odor learning (e.g., nonassociative, reinforcement), despite their differences in other properties as noted above.  However, explicit \textit{discrimination learning} -- in which animals are rewarded for distinguishing physically similar odorants from one another by associating them with different outcomes -- requires that these distensions into the additional dimension also be locally retractable, so as to reduce or eliminate the similarity-based categorical overlap that may exist between the odor source volumes \textit{a priori}.  This is particularly important given the remarkable olfactory discrimination capabilities exhibited by appropriately trained animals \cite{Moser_Bizo_Brown_2019}. 

Consider two physically similar odorants $s^*=(x,\widetilde{U_x})$ and $t^*=(y,\widetilde{U_y})$ in $S.$  Because the early stages of odor learning are characterized by broadened generalization gradients \cite{Cleland2009}, presumably reflecting sampling uncertainty, odor representations (distensions in $S$) at this stage are likely to overlap:  $\widetilde{U_x}\cap \widetilde{U_y}\neq \varnothing.$  This is appropriate, given the likelihood (prior to discrimination training) that two highly similar odor stimuli, sampled in close succession, simply constitute two samples from the same odor source volume. However, discrimination training is capable of rapidly and strongly separating highly similar odors, and the \textit{between-category separation} principle of category learning \cite{PGJ-Harnad2019} indicates that we need to move them \textit{further} apart than they would be prior to learning.  Hence, discrimination learning needs to be able to not only retract distensions to zero, but to expand them in the negative direction if need be (see Figure 4.1B).  

To do this, we construct a map that decreases only those values of $f$ which are sufficiently close (within some small $\varepsilon>0$) to a distance-minimizing path $\gamma$  connecting $x$ and $y$. Its existence follows from the existence of smooth bump functions on $M.$  Fix $f\in C^{\infty}(\R^m)$ so that $S=S(f).$ We consider functions $\alpha\in C^\infty(\R).$ Then, by defining the learning operation as $S\mapsto S(\alpha \comp f)$ we have a realization of this transformation of learning two odors apart. In fact, what we have done here is defined a $\widetilde{C^\infty(\R)}$-module structure on $C^\infty(\mu).$ Therefore, by considering only the interaction of $\alpha$ and $f$ over $\gamma,$ we have reduced the problem of discrimination learning to a 1-dimensional problem depicted in Figure 4.1A-B. 
The map resulting from discrimination learning lengthens the perceptual metric $d^{per}$ between two similar odor source volumes, partitioning and expanding the previously shared space between the two representations so as to arbitrarily increase their perceptual dissimilarity, all without altering the physical distance $d_{phys}$ between their centers.

Importantly, discrimination learning inherently depends on at least two odor sources, so can be targeted even more specifically between them.  In high-dimensional space, can separate two such sources nearly arbitrarily without affecting similarity relationships among other nearby odor representations.  This cannot be depicted in our lower-dimensional plots as the number of dimensions is too small for all of the odors to essentially be independent. 

\begin{remark}
 	Based on the construction above, we can take $\widetilde{C^\infty(\R)}$ to be a rough approximation of $G$ as a sheaf. We cannot conclude that they are precisely equal as this would need more analysis which we have not presented here.  
\end{remark}
Putting all of this together, we arrive at the final (for now) version of the model. 
We now have, $R,R',M,G,S$ and can complete the picture of the model (reference Diagram \ref{Diagram1}).  The appearance of $G$ and $C^\infty(\mu)$ encodes the local-to-global transformations of granule cells and their interaction with the maps $M\to S$ which preserve $R'.$ 
\begin{equation}\label{Diagram3} \begin{tikzcd}
R \arrow[r, "B"] & {(R',G,C^\infty(\mu))} \arrow[rd, "\xi"] \arrow[d, "\Delta(f) "', dotted] &\\
& S  & M=R'\times \mathbb{R}^m \arrow[l, "\text{Id}_{R'} \times f"]
\end{tikzcd}\end{equation}
All together, this diagram encodes everything which we have constructed above and the relations of the various spaces.

\subsection{The construction of hierarchical odor categories}
The last original part of this section is the construction of hierarchical categories from the continuous spaces we have built above. The surprising advantage of the process above is that it gives a geometric interpretation of the speed-accuracy tradeoff for identifying odors in the wild.   

Suppose now that we need to identify a given odor. For example, a fox in the wild may be hunting an animal and tracking it by scent or a human trying to discern a specific spice in a dish while at a restaurant. What is the mathematical interpretation of such a situation and how does the model deal with this interpretation.   
We first view each peak as a continuous categorization for that stimulus (This is the image of a fully learned system). For instance we may have a peak defined for ``oranges". As we move up the peak we refine the categorization. Here refinement means entering a subcategory. From the discussion above we know that the peak will be parsed into a variety of sub-peaks which correspond to physically similar but perceptually different types of orange. Pictured below is a complex of categories, ordered by inclusion \[ \text{Citrus Fruit }\supseteq \text{Oranges} \supseteq \text{Ripe Oranges }  \supseteq \text{Ripe Valencia Oranges} \]
Although this example is linearly ordered, there is no need for there to be only one chain of inclusions. Every peak can break up into at most finitely many distinct subpeaks and thus the decomposition can become arbitrarily complicated.   

Now we shall construct the categorization by successively taking intersections with an affine hyperplane (see Figure 4.1(A) for an illustration in the case $N=1$) Suppose $P$ is a peak, determined by some odorant pair $(x,U_x)$, with several subpeaks $\{P_i\}_{i\in I}$. Then as each sub-peak has a boundary, we can define the minimum value attained in $P_i$. Let $P_i^*\subseteq P_i$ be the subset consisting of all points of $P_i$ with minimal $x_{n+1}$ value. Let $H_0=\{x\in \R^{n+1}:x_{n+1}=0\}$ be a hyperplane in $\R^{n+1}$ and define $H_t=H_0+(0,0,...0,t)$. This is an affine transformation of $H_0$ and geometrically is the translation of $H_0$ in the ${n+1}^th$ direction.  
\begin{lemma}
    $P_i^*=P_i\cap H_t$ for some $t>0.$ Further if $n\geq 2$, $P_i^*$ is connected. 
\end{lemma}
\begin{proof}
    Let $t^*$ be the ${n+1}^{th}$ coordinate of all elements in $P_i^*.$ Then by construction \[P_i^*\subseteq P_i\cap H_{t^*}\]
    For the reverse inclusion let $y\in P_i\cap H_{t^*}.$ Then $y\in P_i$ and $y_{n+1}=t^*$ and therefore $y\in P_i^*.$ Hence, $P_i^*=P_i\cap H_{t^*}.$
    The connectedness of $P_i^*$ follows immediately from the fact that $P_i^*=\partial P_i$ the boundary, and that $P_i$ is homeomorphic to $D^n$ the $n$-dimensional disk. For $n\geq 2$ $\partial D^n=S^{n-1}$ and is thus connected.  
\end{proof}

\noindent Using this lemma, we can now define a coarse categorization of $S.$ Let $t\in (0,1)$ be arbitrary. By the previous lemma, we know that if we consider $H_t\cap S$ we will get disjoint connected subsets of $S.$ So, consider the closed half space \[H_t^*=\{ x\in \R^{n+1}: x_{n+1}\geq t\}\]
Then $\partial H_t^*=H_t$ and $H_t^*\cap S$ is also a collection of disjoint connected subsets of $S.$ Let $\{S_i^t\}_{i\in I_t}$ be an enumeration of these subsets by the set $I_t$. Now let $\mathcal{P}$ be a partition of $(0,1).$ Then for each $p_j\in P$ we have the associated collection $\{S_i^{p_j}\}$ of subsets. We know by construction that for $j<j'$ that $\{S_i^{p_j}\}\supseteq\{S_i^{p_{j'}}\}.$ Therefore, we have built a method to break $S$ into discrete categories and given in the local structure of a tree. Using this, we arrive immediately at a hierarchical categorization of odors which is solely dependent on the amount of information learned about a class of odors.    

\subsection{"Olfactory space" is not hyperbolic}
The method we have built above prioritizes the construction of a coarse categorization (partial order) from a geometric structure. One may ask if it is possible to proceed in the other direction, that is build a geometric structure out of some form of categorization. This approach has been attempted by many researchers and in every case, there is a fundamental assumption made which makes the model unhelpful and in some cases, invalid. In \cite{Zhou2018} they make the claim that the human perceptual odor space (the analogue of $S$) is three dimensional and hyperbolic (constant negative curvature $-1$). This was based off of a calculation and subsequent averaging of the ranks of certain homology groups for simplicial complexes built from certain adjacency matrices. This computation was done on local data, more specifically on single odorants. After doing their analysis on the four odors tested, the researchers conclude that the each class of points best fit in a three dimensional hyperbolic space, even though there existed non-zero homology in higher degrees. From this they conclude that the entire space is hyperbolic. This is plainly false. From what they have shown, we have some evidence that the data locally looks hyperbolic but we cannot conclude any information about the global structure. Take for instance $TS^2$ where $S^2$ is the unit sphere in $\R^3$ and $TS^2$ is the tangent bundle. The Hairy Ball Theorem \cite{HairyBall} tells us that $TS^2$ is not a trivial bundle, and yet we can always locally trivialize a vector bundle. Therefore, the local structure tells you little about the global structure. This is one of the reasons their conclusion was flawed. Their claim also hinged on the computation of some homology groups for certain simplicial complexes generated by  "similarity matrices" and showing that the distributions of the rank of these groups closely matches simulation estimates for hyperbolic space. This would have worked, had they not  stopped computing the homology in degree 3. It does not take much thinking to concoct a graph (and thus a simplicial complex) whose homology groups are zero for n=1,2,3 and are non-trivial for some higher degree (for example, the iterated suspension of two points will yield simplicial complexes which are homotopy equivalent to spheres). This implies that the structure which they are trying to detect will have some higher dimensional pieces. Simply not considering these (possibly because of the method used in \cite{Giusti2015}) leads to a false conclusion. Hence, the conclusion that the olfactory perceptual space is hyperbolic is simply unfounded. More interestingly, should the perceptual space be related to the physical odorant space at all, there is no possible way to have constant curvature! In this situation, when learning occurred, it would be impossible to preserve the physical metric and the perceptual metric simultaneously.

\chapter{Future Directions}  
This chapter will serve to present those ideas which we have not incorporated into the model but believe are useful. Most of these topics are central to any field of mathematics and thus we should expect them to show up here too. Additionally, we close with a conjectural method to deal with noisy input odors and show its relation to some of the topics introduced in the first few chapters. 

\section{Lie Groups and Lie Algebras}

The representation theory of Lie groups is a fundamental field of mathematics. So fundamental in fact that one would be strained to find an area of mathematics which does not appear in the usual course of study. In this short chapter, we shall study one of theorems which lies in the intersection of complex analysis and representation theory: the Borel-Weil theorem. 

\begin{theorem}[Borel-Weil]\label{Borel-Weil}
	Let $K$ be a compact, connected Lie group and $T\subseteq K$ be a maximal torus. Let $G=K_\C$ be the complexification and $B=MA\overline{N}$ a Borel subgroup. Then the irreducible finite dimensional representations of $K$ stand in one-to-one correspondence with the dominant, analytically integral weights $\lambda\in \lie{t}^*$ with the correspondence given by \[ \lambda\mapsto \Gamma_H(K/T,L_{\lambda})\cong \script{F}_{B,\chi_{\lambda}}^{Hol}\]
	where $\Gamma_{H}(K/T,L_{\lambda})$ denotes the set of holomorphic sections of the bundle and \[\script{F}_{B,\chi_{\lambda}}^{Hol}=\left \{f:G\to \C \;\vline\; f(gb)=\chi_{\lambda}(b)^{-1}f(g), f \text{ holomorphic}  \right\}\]
	with $\chi_{\lambda}$ the character of $B$ associated to the analytically integral weight $\lambda.$ 
\end{theorem}

\noindent This was proven independently by Borel and Weil in \cite{Serre1954} and then by Harish-Chandra in \cite{HarishChandra1956}. The proof we shall present in section 5.5 is a combination of those presented in \cite{Knapp88}, \cite{Knapp86}, and \cite{Helgason2008}. As will be seen later, this theorem gives a geometric realization of a purely algebraic object and vice versa. Therefore, we may be able to apply similar methods to the model above and arrive at some striking consequences.

Recall that a smooth manifold is a second-countable, Hausdorff, topological space equipped with an atlas of (smooth) $C^\infty$-charts $\{\varphi_U:U\to \R^n\}$ which are injective. Morphisms of smooth manifolds are smooth maps which are compatible with the atlases. Putting these two together, we get the category $\mathcal{M}$ of smooth manifolds.
There is a functor $C^\infty(-):\mathcal{M}\to \R$-\textbf{Alg} which assigns to any smooth manifold $M$ an $\R$-algebra $C^\infty(M):=\{ f:M\to \R| f \operatorname{ smooth} \}$ with addition and multiplication defined point-wise. For each $p\in M$ we can define \[C^\infty_{M,p}:=\varinjlim_{U\ni p} C^\infty(U).\]

Let $M$ be a manifold and $p\in M$ a point. We define the tangent space at $p$ to be $T_pM:=\Der(C^\infty_{M,p}).$ This becomes an $\R$-vector space is we equip it with addition.The collection of all tangent spaces is called the tangent bundle and is denoted $TM.$ This admits a smooth structure and becomes a smooth manifold with dimension $2\dim M.$ The elements of $TM$ can be given as pairs $(p,v)$ where $p\in M$ and $v\in T_pM.$ There is a canonical projection $\pi_M:TM\to M$ which is a local diffeomorphism onto. A manifold is called $\textit{parallelizable}$ if $TM=M\times \R^n$ for $n=\dim M.$  

A \textit{section} of the canonical projection is a smooth map $f:M\to TM$ such that $\pi_M\comp f=1_M.$ The set of all smooth sections is denoted $\Gamma(M,TM)$ or $\lie{X}(M)$ and can be identified with the collection of smooth vector fields on $M.$ This has a natural structure as a $C^\infty(M)$-module. 

\begin{definition}
	A \textbf{Lie Group} is a group object in the category $\mathcal{M}.$ More explicitly, it is a smooth manifold $G$ equipped with two operations: \textbf{multiplication} $G\times G\to G$ which is smooth and \textit{inversion} $(-)^{-1}:G\to G$ which is also smooth. A \textbf{Lie group homomorphism} is a smooth map which respects the group structure.   
\end{definition} 

Let $G$ be a Lie group and $x\in G.$ Then $x$ defines a smooth automorphism $L_x:G\to G$ such that $L_x(y)=xy.$ An element $q\in \Gamma(G,TG)$ is called \textbf{left-invariant} if for all $x,y\in G$ we have \[T_yL_x(q_y)=q_{xy}\] where $T_yL_x : T_yG\to T_{xy} G$ is the tangent map. The space of all left-invariant vector fields on $G$ will be denoted as $\lie{X}_L(G).$ 
\begin{definition}
	A \textbf{Lie Algebra} is a vector space $V$ equipped with an alternating, bilinear form $ [-,-]:V\times V\to V$ satisfying the Jacobi Identity \[ [X,[Y,Z]]+[Z,[X,Y]]+[Y,[Z,X]]=0\] called the \textbf{Lie Bracket}. A \textbf{Lie algebra homomorphism} is a linear map $T:\lie{g}\to \lie{h}$ such that \[ T([X,Y])=[T(X),T(Y)] \]
	where the first bracket is in $\lie{g}$ and the second is taken in $\lie{h}.$  
\end{definition}

\begin{lemma}
	The map $(-)_1:\lie{X}_L(G)\to T_1(G)$ is a vector space isomorphism. Further, if we endow $\lie{X}_L(G)$ with the operation $[X,Y]=XY-YX.$ This makes $\lie{X}_L(G)$ a Lie algebra over $\R$. Further $(-)_1$ respects the bracket operation and gives $T_1(G)$ the structure of a Lie algebra. 
\end{lemma} 
\begin{proof}
	The map has an inverse given by $Xf(x)=X_1(L_{x^{-1}} f)$ where $L_{x^{-1}}f(y)=f(xy).$ The fact that this map respects the Lie bracket is obvious. 
\end{proof}

\begin{corollary}
	$TG\cong G\times T_1(G).$ 
\end{corollary}
\begin{proof}
	Every basis of $T_1(G)$ consists of global left-invariant vector fields and hence $G$ is parallelizable. 
\end{proof}
For the remainder of this section we shall denote Lie algebras by the corresponding lower-case gothic letters. That is if $G$ is a Lie group, then its Lie algebra is $\lie{g}.$ 

\begin{example}
	\begin{enumerate}
		\item Let $G=\R^n$ together with addition. This is a Lie group with Lie algebra $\R^n.$ In general, any finite-dimensional real vector space is non-canonically isomorphic to $\R^n$ for some $n$ and therefore carries a smooth manifold structure and therefore a Lie groups structure. 
		\item Recall that a matrix is invertible if $\det X\neq 0.$ Then $GL_n(\R)$ (resp. $\C$) is the collection of all invertible $n\times n$ matrices with entries in $\R$ (resp. $\C$). It is called the $\textbf{General Linear group}$. This is an open subset of $M_n(\R)$ (resp. $\C$) and therefore carries an obvious manifold structure. In fact, matrix multiplication and matrix inversion are smooth operations. This makes $GL_n(\R)$ (resp. $\C$) a real Lie group of dimension $n^2$ (resp. $2n^2$). Its Lie algebra is $\lie{gl}_n(\R)=M_n(\R)$ (resp. $\C$).  
		\item Define the operation $-^*:M_n(\C)\to M_n(\C)$ by $X\mapsto \overline{X}^T.$ The matrix $X^*$ is called the adjoint matrix to $X.$ Let $U(n)\subseteq GL_n(\C)$ to be the set of matrices such that $X^*X=I_n$ the $n\times n$ identity matrix. This is the \textbf{Unitary group} and is a closed subgroup of $GL_n(\C)$ and thus inherits a Lie groups structure. To find its dimension we pass to the Lie algebra $\lie{u}(n).$ An easy computation shows that $\lie{u}(n)$ consists of all skew-hermitian matrices ($X^*=-X$) and thus $\dim \lie{u}(n)=n^2.$ Further, $U(n)$ is a \textit{real} Lie group. To see this, see what happens when we take $i\lie{u}(n).$         
		\item Let $S^1$ be the circle embedded as a submanifold of $\C.$ Then $S^1$ carries a Lie group structure by writing its entries in polar coordinates. Define the \textbf{Torus} $\mathbb{T}^n=\prod^n S^1.$ This carries a natural Lie group structure under component-wise multiplication. Its lie algebra is $i\R^n.$  
	\end{enumerate}
\end{example}

\subsection{Lie Algebras Generally} 
Lie algebras are significantly easier to deal with than Lie groups because they are essentially generalized vector spaces. Therefore, we want to understand the structure of various types Lie algebras so that we may possibly deduce some information about the associated Lie group. 

\begin{definition}
	A \textbf{Lie subalgebra} (normally shortened to simply subalgebra) of a Lie algebra $\lie{g}$ is a vector subspace $\lie{h}$ such that $[\lie{h},\lie{h}]\subseteq \lie{h}$ where the bracket of Lie algebras is shorthand for the set of all $[X,Y].$ An \textbf{ideal} of $\lie{g}$ is a subset $\lie{i}$ such that $[\lie{g},\lie{i}]\subseteq \lie{i}.$ A subalgebra $\lie{a}$ is called \textbf{abelian} if $[\lie{a},\lie{a}]=0.$  
\end{definition}
We will denote ideals of $\lie{g}$ as $\lie{i}\trianglerighteq \lie{g}$ and subalgebras as $\lie{h}\subseteq \lie{g}.$ Notice that $[\lie{i},\lie{g}]\subseteq \lie{i}$ is equivalent to the definition given above as this amounts to putting a negative sign everywhere, but $-\lie{i}=\lie{i}.$   
\begin{proposition}
	If $\lie{g}$ is a Lie algebra and $\lie{i}$ is an ideal, then $\lie{g}/\lie{i}$ has the structure of a lie algebra. 
\end{proposition}
\begin{proof}
	As a set, $\lie{g}/\lie{i}$ is simply the vector space quotient. To show that the Lie bracket descends to the quotient, we consider two classes $X+\lie{i},Y+\lie{i}\in \lie{g}/\lie{i}.$ Then \[ [X+\lie{i},Y+\lie{i}]=[X,Y]+\lie{i}\]
	by the bilinearity of the bracket. It then follows immediately that this bracket satisfies the Jacobi identity. Hence, $\lie{g}/\lie{i}$ is a lie algebra.   
\end{proof}
Similar to the case of ideals of a ring, it can be shown (quite easily) that any ideal can be realized as the kernel of a Lie algebra homomorphism, namely $\varphi:\lie{g}\to \lie{g}/\lie{i}.$    
\begin{definition}
	A Lie algebra $\lie{g}$ is $\textbf{simple}$ if it has no non-zero proper ideals. It is \textbf{semisimple} if it has no non-zero solvable ideals. We say that a Lie group $G$ is semisimple (resp. simple) if $\lie{g}$ is semisimple (resp. simple).   
\end{definition} 
A fact which we will not prove is that all semisimple Lie algebras have no center, and therefore all semisimple Lie groups have a $0$-dimensional center. Further, one can prove (say by Cartan's criterion for semisimplicity) that all semisimple Lie algebras can be realized as a direct sum of simple lie algebras \cite[Chapter 1]{Knapp96}

Semisimple Lie groups are of interest to many areas of mathematics and are fairly well understood. The small piece of the theory of lie groups that we need for the rest of this section is the \textit{representation theory of semisimple Lie groups and Lie algebras}. Before we get into this, we want to understand where representation theory comes from in the first place. Why might we care about representations? Suppose $G$ is a finite group (not assumed to be of Lie type) and let $G$ act on a set $X.$ Denote by $\script{F}(X)$ the set of all complex valued functions on $X.$ Then $\script{F}(X)$ is naturally a $\C$-vector space under point-wise addition and scalar multiplication. We can extend the action of $G$ on $X$ to an action on all of $\script{F}(X)$ by \[ (g\cdot f)(x)=f(g^{-1}\cdot x)\]
This representation will break up into a direct sum of irreducible representations of $G$ with some multiplicities (by Maschke's Theorem). Precisely how this representation breaks up tells us something about the structure of $X.$ In particular, if we put some conditions on the functions (that they are all $L^2$ for instance) then we can better understand $X$ and its symmetries. This has a similar flavour to understanding $\Aut{X}$ for $X$ in an arbitrary category. 

\subsection{A Theorem of Lie} 
\begin{definition}
	Let $\lie{g}$ be a Lie algebra over an arbitrary field. The \textbf{commutator series} for $\lie{g}$ is defined by $\lie{g}^1=[\lie{g},\lie{g}]$ and $\lie{g}^{n+1}=[\lie{g}^n,\lie{g}^n].$ We get a chain of Lie subalgebras \[ \lie{g}^0=\lie{g}\supseteq \lie{g}^1\supseteq \lie{g}^2\supseteq...\]
	We say that $\lie{g}$ is \textbf{solvable} if $\lie{g}^n=0$ for some $n.$ 
\end{definition}
\begin{definition}
	Let $\lie{g}$ be a Lie algebra over an arbitrary field. The \textbf{lower central series} for $\lie{g}$ is defined by $\lie{g}_1=[\lie{g},\lie{g}]$ and $\lie{g}_{n+1}=[\lie{g},\lie{g}_n].$ We get a chain of ideals \[ \lie{g}_0=\lie{g}\supseteq \lie{g}_1\supseteq \lie{g}_2\supseteq...\]
	We say that $\lie{g}$ is \textbf{nilpotent} if $\lie{g}_n=0$ for some $n.$ 
\end{definition}
\begin{corollary}
	If $\lie{g}$ is nilpotent then it is solvable. 
\end{corollary}

\begin{lemma}
	Every subalgebra of a solvable (resp. nilpotent)  Lie algebra is solvable (resp. nilpotent). 
\end{lemma}
\begin{proof}
	Clearly, for each $\lie{h}\subseteq \lie{g}$ the commutator series satisfies $[\lie{h},\lie{h}]\subseteq [\lie{g},\lie{g}].$ 
\end{proof}

\begin{theorem}[Lie's Theorem]
	Let $\lie{g}$ be a complex solvable Lie algebra and $(\pi,V)$ a representation. Then there exists a simultaneous eigenvector for all elements in $\pi(\lie{g}).$ 
\end{theorem}	
This implies, for instance, that all elements of $\pi(\lie{g})$ act by upper triangular matrices on any $\pi(\lie{g})$ invariant subspaces. With the diagonal entries being the generalized eigenvalues of the matrices.     

\noindent For proofs of this theorem see \cite{Knapp86} or \cite{Bump2013}.

\subsection{Weights, Roots, and the Weyl Group}  
For this entire section, all statements not proven are presented in \cite{Knapp2005} with incredible detail. 
\begin{definition}
	Let $\lie{g}$ be a Lie algebra and $(\pi,V)$ a representation.  For $\alpha\in \lie{g}^*$ put \[V_\alpha= \{v\in V: (\pi(H)-\alpha(H)1)^nv=0, \forall H\in \lie{g}, n=n(v,H)\}\]
	 If $V_\alpha\neq 0,$ then $V_\alpha$ is called a \textbf{generalized weight space} and $\alpha$ a \textbf{weight}. We will denote the set of weights by $\Lambda(\lie{g},\pi).$ 
 \end{definition}
 If $V$ is finite dimensional then $\pi(H)-\alpha(H)1$ only has $0$ as a generalized eigenvalue and is nilpotent on $V_\alpha$ via the theory of Jordan normal forms. Therefore, we may assume that $n(v,H)=\dim V.$ In this case, we would like to somehow deduce information about $\pi$ from the generalized weight spaces.  
 
\begin{theorem}
	Let $\lie{h}$ be a nilpotent lie algebra and $(\pi,V)$ a finite dimensional complex representation. Then there are finitely many generalized weights of $\pi.$ Further, each generalized weight space is stable under $\pi(\lie{h})$ and $V=\bigoplus_{\alpha\in \Lambda(\lie{h},\pi)} V_\alpha$ 
\end{theorem}
\begin{proof}
	We first prove that $V_\alpha$ is invariant under $\pi(\lie{h}).$ Fix $H\in \lie{h}.$ Then put \[ V_{\alpha,H}=\{ v\in V: (\pi(H)-\alpha(H)1)^nv=0, n=n(v)\}\]
	Now, by construction $V_\alpha=\bigcap_{H\in \lie{h}} V_{\alpha,H}.$ It suffices to prove that $V_{\alpha,H}$ is $\pi(\lie{h})$-invariant. 
	
	Now, as $\lie{h}$ is nilpotent, $\ad H$ is nilpotent for all $H.$ Put \[ \lie{h}_{(m)}=\{ Y\in \lie{h}: (\ad H)^mY=0\}\]
	so that $\lie{h}=\bigcup_{m=o}^{\dim \lie{h}} \lie{h}_{(m)}.$ We prove that $\pi(Y)V_{\alpha,H}\subseteq V_{\alpha,H}$ for $Y\in \lie{h}_{(m)}$ by induction on $m.$ 
	
	For the case of $m=0$ this is trivial as $\lie{h}_{(m)}=0.$ Therefore, assume that this holds for all $Z\in \lie{h}_{(m-1)}.$ If $Y\in \lie{h}_{(m)},$ then $[H,Y]\in \lie{h}_{(m-1)}$ be construction. Therefore, \[ (\pi(H)-\alpha(H)1)\pi(Y)=\pi(Y)(\pi(H)-\alpha(H))+\pi([H,X]) \]
  and \begin{align*}
  	(\pi(H)-\alpha(H)1)^2\pi(Y)&=(\pi(H)-\alpha(H)1)\pi(Y)(\pi(H)\\&-\alpha(H)1)+(\pi(H)-\alpha(H)1)\pi([H,Y])\\
	&=\pi(Y)(\pi(H)-\alpha(H)1)^2\\
	&+ (\pi(H)-\alpha(H)1)\pi([H,Y])+\pi([H,Y])(\pi(H)-\alpha(H)1)
  \end{align*}	    
  Iterating this computation, we get the general formula \begin{align*} (\pi(H)-\alpha(H)1)^\ell\pi(Y)&=\pi(Y) (\pi(H)-\alpha(H)1)^{\ell}\\&+\sum_{s=0}^{\ell-1} (\pi(H)-\alpha(H)1)^{\ell-1-s}\pi([H,Y])(\pi(H)-\alpha(H)1)^s \end{align*}
For $v\in V_{\alpha,H},$ we know that $(\pi(H)-\alpha(H)1)^Nv=0$ for $N\geq \dim V.$ Take $\ell=2N$ in the above expression and apply it to $v.$ The only terms which survive are those for which $s<N.$ In this case, $\ell-1-s\geq N$ and therefore $(\pi(H)-\alpha(H)1)^sv\in V_{\alpha,H},$ $\pi([H,Y])$ preserves $V_{\alpha,H}$ be the induction hypothesis, and \[ (\pi(H)-\alpha(H)1)^{\ell-1-s}\pi([H,Y])(\pi(H)-\alpha(H)1)^sv=0\]
Hence, $(\pi(H)-\alpha(H)1)^\ell\pi(Y)v=0$ and thus $V_{\alpha,H}$ is stable under $\pi(Y).$ This completes the induction and $V_\alpha$ is invariant under $\pi(\lie{h}).$ 

Now we can obtain the decomposition. Let $H_1,...,H_d$ be a basis for $\lie{h}.$ The Jordan decomposition for $\pi(H_1)$ gives a generalized eigenspace decomposition that we can write as \[ V=\bigoplus_\lambda V_{\lambda,H_1}\]
We can regard the complex numbers $\lambda$ as running over all values of $\alpha(H_1)$ for $\alpha\in \lie{h}^*$ arbitrary. Therefore, we can re-write the decomposition as \[ V=\bigoplus_{\alpha(H_1), \alpha\in \lie{h}^*} V_{\alpha(H_1),H_1}\]
However, $V_{\alpha(H_1),H_1}=V_{\alpha,H_1}$ which we defined at the beginning of the proof. Therefore, each of these spaces is stable under $\pi(\lie{h}).$ Therefore, we can further decompose it under $\pi(H_2)$ to get \[ V=\bigoplus_{\alpha(H_1)} \bigoplus_{\alpha(H_2)} \left( V_{\alpha,H_1}\cap V_{\alpha,H_2} \right) \]
we iterate this for all basis elements of $\lie{h}$ to get \[ V=\bigoplus_{\alpha(H_1),...,\alpha(H_d)} \left( \bigcap_{j=1}^d V_{\alpha,H_j} \right) \]
with each of these spaces $\pi(\lie{h})$-invariant. By Lie's theorem, each $\pi(H_i)$ acts simultaneously by an upper-triangular matrices on $\bigcap^d V_{\alpha,H_i}$ with diagonal entries evidently $\alpha(H_i).$ Then $\pi(\sum c_iH_i)$ acts by $\sum c_i\alpha(H_i).$ Thus, if we define $\alpha(\sum c_iH_i)=\sum c_i\alpha(H_i),$ we see that $  \bigcap^d V_{\alpha,H_i}=V_\alpha$ and $V=\bigoplus V_\alpha.$ 
In particular there are only finitely many $\alpha$ which satisfy this property. This completes the proof.  
\end{proof}

Now let $\lie{g}$ be a semisimple Lie algebra and $\lie{h}$ a nilpotent subalgebra. Let $\lie{h}^*$ denote its dual space. Then for all $\lambda\in \lie{h}^*,$ define \[ \lie{g}_\lambda=\{X\in \lie{g}: (\ad H-\lambda(H)1)^nX=0, \forall H\in \lie{h}, n=n(X,H)\} \]
As $\lie{h}$ is nilpotent, we know that $\lie{g}=\bigoplus_{\lambda\in \lie{h}^*} \lie{g}_\lambda.$ Further, there exist finitely many $\lambda$ such that $\lie{g}_\lambda$ is non-zero. Let $\Delta(\lie{g},\lie{h})$ be the set of weights.  

\begin{proposition}
	In the setting above: \begin{enumerate}
		\item $\lie{g}=\bigoplus_{\alpha\in \Delta(\lie{g},\lie{h})} \lie{g}_\alpha$
		\item $[\lie{g}_\alpha,\lie{g}_\beta]\subseteq \lie{g}_{\alpha+\beta}$ (this space is understood to be $0$ if $\alpha+\beta\not\in \Delta(\lie{g},\lie{h}).$)
		\item $\lie{h}\subseteq \lie{g}_0 $
	\end{enumerate}
\end{proposition}
\begin{proof}
	This all follows from the previous theorem by replacing $V$ with $\lie{g}.$ 
\end{proof}

\begin{definition}
	A nilpotent Lie subalgebra $\lie{h}$ is a \textbf{Cartan subalgebra} if $\lie{h}=\lie{g}_0.$ 
\end{definition}
This definition in general is hard to check. Therefore, we would like an equivalent way of defining Cartan subalgebras so that this condition is not too abstract. 
\begin{proposition}
	Let $\lie{g}$ be a Lie algebra and $\lie{h}$ a nilpotent subalgebra. Then $\lie{h}$ is a Cartan subalgebra if and only if $N_\lie{g}(\lie{h})=\lie{h}.$ This is the \textbf{normalizer} of $\lie{h}$ and is $\{X\in \lie{g}: [X,\lie{h}]\subseteq \lie{h}\}.$  
\end{proposition}  
\begin{proof}
	See \cite{Knapp2005} 
\end{proof}

\begin{theorem}
	Let $\lie{g}$ be a complex finite-dimensional Lie algebra. Then there exists a Cartan subalgebra $\lie{h}\subseteq \lie{g}.$ Further, every Cartan subalgebra is conjugate.  
\end{theorem}
\begin{proof}
	See \cite{Knapp86}, \cite{Knapp2005}, and \cite{Helgason1978} for separate proofs of this theorem. 
\end{proof}

For the remainder of this section, we shall only give sketches of the proofs for the big theorems as there are much more important topics to cover. For a full treatment see \cite[Chapter 7]{Lorenz2018}. 
\begin{definition}
	Let $\lie{g}$ be a complex semisimple Lie algebra and $\lie{h}$ a Cartan subalgebra. We call the weights of the adjoint representation of $\lie{h}$ on $\lie{g}$ \textbf{roots}. The decomposition \[ \lie{g}=\lie{h} \ds \bigoplus_{\alpha\in \Delta(\lie{g},\lie{h})} \lie{g}_\alpha\]
	is called the \textbf{root space decomposition}.  
\end{definition}
We want to understand $\Delta(\lie{g},\lie{h}).$  
\begin{proposition}
Consider the situation above.
	\begin{enumerate}
		\item If $\alpha,\beta\in \Delta\cup \{0\}$ and $\alpha+\beta\neq 0,$ then $B(\lie{g}_\alpha,\lie{g}_\beta)=0.$ 
		\item If $\alpha\in \Delta\cup \{0\},$ then $B$ is non-singular on $\lie{g}_\alpha\times \lie{g}_{-\alpha}.$ 
		\item If $\alpha\in \Delta$ then $-\alpha\in \Delta.$
		\item $B|_{\lie{h}\times\lie{h}}$ is non-degenerate and thus for each $\alpha$ there exists $H_\alpha$ so that $B(H_\alpha,H)=\alpha(H).$ 
		\item $\Delta$ spans $\lie{h}^*.$ 
	\end{enumerate}
\end{proposition}
\begin{proof}
	See \cite[Chapter 2]{Knapp2005}
\end{proof}

The following proposition reduces the case of the root space decomposition nicely. 
\begin{proposition}
	If $\alpha\in \Delta,$ then $\dim \lie{g}_\alpha=1.$ Further $n\alpha\not\in \Delta$ for $n\geq 2.$ 
\end{proposition}
\begin{proof}
	See \cite[Chapter 2]{Knapp2005}
\end{proof}

All of this together shows that $\Delta(\lie{g},\lie{h})$ is an abstract, reduced root system. We can thus define a notion of \textit{positivity}. 
\begin{definition}
	Let $V$ be a finite dimensional inner product space. Fix a spanning set $\varphi_1,...,\varphi_m.$ Then a vector $\varphi$ is \textbf{positive} (denoted $\varphi>0$) if there exists an integer $k\geq 1$ such that $\ip{\varphi,\varphi_i}=0$ for $1\leq i\leq k-1$ and $\ip{\varphi,\varphi_i}>0$ for $i\geq k.$ 
\end{definition}

\begin{lemma}\label{Positivity}
	If $\varphi\in \Delta,$ the one of $\varphi$ or $-\varphi$ is positive. 
\end{lemma}
\begin{proof}
	See \cite[Chapter 7]{Lorenz2018}.  
\end{proof}

\begin{definition}
	A \textbf{basis} $\Pi$ for $\Delta$ is a choice of of elements such that \begin{enumerate}
		\item $\Pi$ is a basis of $\lie{h}^*.$
		\item For any $\beta\in \Delta,$ we can write $\beta=\sum n_i\alpha_i$ with $\alpha_i\in\Pi$ and $n_i\in \Z$ all positive or negative by Lemma \ref{Positivity}. 
	\end{enumerate}
	We call elements in $\Pi$ simple, and normally say choose a \textit{simple system} for $\Delta.$ 
\end{definition}

\begin{definition}
	Let $\alpha,\beta\in \lie{h}^*.$ We define an inner product on $\lie{h}^*$ by $(\alpha,\beta)=2\frac{\ip{\alpha,\beta}}{\ip{\beta,\beta}}=2\frac{||\alpha||}{||\beta||} \cos \theta$ where $\theta$ is the angle between the functionals. Then the \textbf{reflection} of $\beta$ by $\alpha$, denoted $s_\alpha\beta$ is defined by \[ s_\alpha \beta=\beta-(\beta,\alpha)\alpha\]
	The \textbf{Weyl group} is \[ W(\lie{g}):=\langle s_\alpha:\alpha\in \Delta\rangle \]
\end{definition}

\begin{theorem}
	$W(\lie{g})$ acts transitively on the set of simple systems for $\Delta.$ 
\end{theorem}
\begin{proof}
	See \cite[Chapter 2, Section 6]{Knapp2005}
\end{proof}
This final theorem eases the concern that picking positive elements is arbitrary and could possibly lead to different results. 

Now, let $\alpha\in \Delta$ and put $\lie{h}^\circ=\lie{h}-\bigcup_{\alpha\in \Delta} \alpha^\perp.$ The connected components of $\lie{h}^\circ$ are called \textbf{Weyl chambers} and given a choice of simple system $\Pi,$ there is a natural choice of Weyl chamber associated to $\Pi$ called the \textbf{positive Weyl Chamber} \[ \script{C}(\Pi)=\{ \alpha\in \lie{h}^*: (\alpha,\beta)>0, \forall \beta\in \Delta^+\}=\{ \alpha\in \lie{h}^*: (\alpha,\beta)>0, \forall \beta\in \Pi\}\]
Associated to any $\Delta(\lie{g},\lie{h})$ is a lattice $\Lambda=\{\alpha\in \lie{h}^*: (\alpha,\beta)\in \Z, \forall \beta\in \Delta\}.$ This is the \textbf{weight lattice} associated to $\Delta.$
\begin{definition}
	An element $\alpha\in \lie{h}^*$ is called \textbf{dominant and algebraically integral} if $\alpha\in \Lambda\cap \overline{\script{C}(\Pi)}.$ 
\end{definition}

\subsection{Universal Enveloping Algebra} 
Lie algebras are easier to deal with than Lie groups, but still the fact that they are non-associative makes the situation a bit difficult. What we would like is to fine an \textit{associative algebra} $A$ such that the representation theory of $\lie{g}$ is the same as the representation theory of $A$ in some semi-canonical sense. As a first guess, we could take the tensor algebra. Let $\lie{g}$ be a complex Lie algebra assumed to be finite dimensional (this construction works for the infinite dimensional case as well).  Let $T^\bullet(\lie{g})=\bigoplus_{\N} \lie{g}^{\tensor k}$ denote the tensor algebra of $\lie{g}.$ This does not force the resulting map $A\to \End(V)$ to be a Lie algebra homomorphism and thus is not the correct choice. Therefore, let \[ U(\lie{g})=T^\bullet(\lie{g})/\ip{X\tensor Y-Y\tensor X-[X,Y]} \]
with $X,Y\in \lie{g}.$ This is the \textbf{universal enveloping algebra} of $\lie{g}.$ Then the canonical map $i:\lie{g}\to U(\lie{g})$ is a lie algebra homomorphism. It is universal in the sense that given any unital associative algebra $A$ and a Lie algebra homomorphism $\lie{g}\to A$ there is a unique Lie algebra homomorphism so that the following diagram commutes \[ \begin{tikzcd}
U(\lie{g}) \arrow[r, "\hat{\varphi}", dotted] & A \\
\lie{g} \arrow[ru, "\varphi"'] \arrow[u, "i"] &  
\end{tikzcd}\]
The following theorem gives an algebraic description of the universal enveloping algebra. 

\begin{theorem}[Poincar\'e-Birkhoff-Witt] 
	Let $\lie{g}$ be a complex Lie algebra with basis $\{X_i\}.$ Then the monomials \[ X_1^{p_1}...X_n^{p_n}\]
	form a basis for $U(\lie{g)}$. If in addition we assume $\lie{g}$ is semisimple, then 
	Let $\{X_{-\alpha},H_{\alpha},X_{\alpha}\}$ be a basis for $\lie{g}$ with respect to a set of roots $\Delta(\lie{g},\lie{h})$ and a choice of simple system $\Pi.$ Then the monomials \[ X_{-\alpha_1}^{i_1}...X_{-\alpha_p}^{i_p}H_{\alpha_1}^{j_1}...H_{\alpha_q}^{j_q}X_{\alpha_1}^{k_1}...X_{\alpha_r}^{k_r}\] 
	form a basis for $U(\lie{g}).$  
\end{theorem}

\begin{corollary}
The canonical map $i:\lie{g}\to U(\lie{g})$ is an injective Lie algebra homomorphism.
\end{corollary}

\begin{proposition}
	Every representation of $\lie{g}$ extends to a representation of $U(\lie{g})$ and every $U(\lie{g})$-module descends to a representation of $\lie{g}.$
 \end{proposition}
\begin{proof}
	The inclusion of $U(\lie{g})$-modules into $\lie{g}$-representations is done by the corollary above. Therefore, it suffices to show that every $\lie{g}$-representation extends to an associative algebra homomorphism $U(\lie{g})\to \End(V).$ Any representation $\lie{g}\to \End(V)$ can be extended to an algebra homomorphism $T^\bullet(\lie{g})\to \End(V).$ The kernel of this map contains the ideal defining $U(\lie{g})$ and therefore descends to a map $U(\lie{g})\to \End(V).$ 
\end{proof} 

We want to give a more analytic interpretation of the universal enveloping algebra. Let $G$ be a semisimple (or reductive) lie group with Lie algebra $\lie{g}.$ Then $G$ acts on the space of smooth functions $C^\infty(G)$ in two ways \begin{align*} L(g)f(x)=f(g^{-1}x)  && R(g)f(x)=f(xg)
\end{align*}
An easy consequence of the definitions the differentiated action $d\lambda$ commutes with the differentiated action $d\rho.$ Therefore \[ L(g) dR(X)=dR(X)L(g)\]
for all $X\in \lie{g}$ and $g\in G.$ This exhibits $\lie{g}$ as \textit{left invariant differential operators} on $G.$ In fact, it is a faithful representation $\lie{g}\to \End(C^\infty(G)).$ We can extend this action to $U(\lie{g})$ and thereby realizing $U(\lie{g})$ as a ring of left invariant differential operators on $G.$ As it turns out, much of the representation theory of $G$ is determined by how certain differential operators (namely the Laplacian or Casimir element) act on representation. If the representation is irreducible for instance, then the center $Z(\lie{g})$ of the universal enveloping algebra acts by scalars. This parametrized the irreducible representations of $G.$ 

\subsection{Verma Modules}    
Let $\lie{g}$ be a complex semisimple lie algebra with cartan subalgebra $\lie{h}$ and root system $\Delta:=\Delta(\lie{g},\lie{h}).$ Let $\Delta^+$ denote the set of positive roots and $\Pi$ a system of simple ones. 

It is known that the finite dimensional representation theory of semisimple lie algebras is semisimple. In the case of complex representations, we have that for every finite dimensional representation $\varphi:\lie{g}\to \lie{gl}(V)=\End_\C(V),$ we can decompose $V=\bigoplus V_i$ where each $V_i$ is irreducible. Therefore we want to classify all irreducible finite dimensional representations and this will yield \textit{all} finite dimensional representations of $\lie{g}.$ We have the following theorem which does precisely this. 

\begin{theorem}[Theorem of Highest Weights]
	Let $\lie{g}$ be a complex semisimple Lie algebra, $\lie{h}$ a Cartan subalgebra and $\Delta(\lie{g},\lie{h})$ the roots with respect to $\lie{h}.$ Let $\script{C}^+$ be the positive Weyl chamber. Then the irreducible, finite-dimensional representations of $\lie{g}$ stand in one-one correspondence with the set of algebraically integral, dominant weights. The correspondence is given in one direction by $V\mapsto \lambda$ its highest weight.   
\end{theorem}

 The difficult step in the proof of this theorem is the construction of the correspondence in the $``\leftarrow"$ direction. To do this, we must build finite dimensional irreducible representations which have highest weight $\lambda.$ These are seen as quotients of Verma modules (to be defined below), which are infinite dimensional representations of $\lie{g}$ that are universal in some sense(see Proposition \ref{VermaProp}). 
 
The setup to the construction of such representations makes use of the root space decomposition of $\lie{g}.$ If $\alpha\in \Delta,$ define \[ \lie{g}_\alpha:=\{ X\in \lie{g}| (\ad H-\alpha(H)1)^nX=0, \forall\; H\in \lie{h}, \text{and some } n=n(\lie{h},X)\} \]
Then it is easy to see that \[ \lie{g}=\bigoplus_{\alpha\in \Delta} \lie{g}_\alpha=\lie{h}\oplus \bigoplus_{\alpha\neq 0} \lie{g}_\alpha \]
By definition the zero root space is the Cartan subalgebra.   
If we pick an order on $\lie{h}^*$ we can then decompose $\lie{g}$ further into positive and negative root spaces \begin{align*} \lie{n}=\bigoplus_{\alpha\in \Delta^+}\lie{g}_\alpha && \lie{n}^-=\bigoplus_{\alpha\in \Delta^+}\lie{g}_{-\alpha }
\end{align*}
These are both lie subalgebras by construction. 
\begin{definition}
	The lie subalgebra constructed by all of the non-negative roots is called the \textbf{Borel subalgebra} of $\lie{g}.$ We denote this as \[\lie{b}=\lie{h}\ds \lie{n}\]
Any lie subalgebra $\lie{p}$ such that $\lie{b}\subseteq \lie{p}\subsetneq \lie{g}$ is called a \textbf{parabolic subalgebra}. 	
\end{definition}

Before we head into the theory of highest weight modules, we recall some facts about $\lie{sl}_2(\C).$ If we let $\{e,f,h\}$ be a basis, then on any irreducible finite dimensional representation we have a weight space decomposition and the basis elements act in the following way 
\[ \begin{tikzcd}
	\vdots  \arrow[d,bend left,"f"] \\
	 \operatorname{\;\;}u_{i}\operatorname{\;\;} \arrow[d,bend left,"f"] \arrow[u,bend left,"e"] \arrow[loop left,"h"]\\
	 u_{i-1} \arrow[d,bend left,"f"] \arrow[u,bend left,"e"] \arrow[loop left,"h"]\\
	\vdots \arrow[u,bend left,"e"]
\end{tikzcd}
\]
As the representation is finite, there exists some $u_m$ such that $e(u_m)=0.$ We say that $u_m$ is the highest weight vector of this representation. In this same style we have the following definition. 
\begin{definition}
	Let $V$ be a left $U(\lie{g})$-module. A vector $v\in V$ is called a \textbf{highest weight vector} if $\lie{n}(v)=0.$ The left $U(\lie{g})$-submodule generated by a highest weight vector is called a \textbf{highest weight module}.    
\end{definition}  
The following proposition gives some properties of highest weight modules. 
\begin{proposition}
	Let $M$ be a highest weight module for $U(\lie{g}),$ and let $v$ be a highest weight vector generating $M.$ Suppose $v$ is of weight $\lambda.$ Then the following hold:
	\begin{enumerate}
		\item $M=U(\lie{n}^-)v$
		\item $M=\bigoplus_{\mu\in \lie{h}^*}M_\mu$ with each $M_\mu$ finite-dimensional and with $\dim_\C M_\lambda=1.$ 
		\item Every weight of $M$ is of the form $\lambda-\sum n_i\alpha_i$ with $\alpha_i\in \Pi$ and $n_i\in \Z^+.$		
	\end{enumerate}
\end{proposition}
\begin{proof}\text{}

	(a) As above, we have the decomposition $\lie{g}=\lie{b}\ds \lie{n}^-.$ The Poincar\'e-Birkoff-Witt Theorem gives a basis for $U(\lie{g})$ which gives us the decomposition \[U(\lie{g})=U(\lie{b})\tensor U(\lie{n}^-)=U(\lie{n})\tensor U(\lie{h})\tensor U(\lie{n}^-).\]
On the vector $v,$ $U(\lie{b})$ acts by scalars. This follows from the fact that $U(\lie{n})v=0$ and $U(\lie{h})$ does not increase or decrease the weight. Therefore $U(\lie{g})v=U(\lie{n}^-)v$ and as $M$ is generated by $v.$ We conclude that $M=U(\lie{n}^-)v.$ 

(b,c) It is clear that $\bigoplus M_\mu$ is stable under the left $U(\lie{g})$ action. As $v\in \bigoplus M_\mu,$ we have that $M\subseteq \bigoplus M_\mu.$ It is true by construction that $\bigoplus M_\mu\subseteq M$ and therefore $M=\bigoplus M_\mu.$ By $(a)$ we know that $M=U(\lie{n}^-)v.$ For any monomial $E_{-\beta_1}^{i_1}...E_{-\beta_k}^{i_k},$ this element acts on $M_\mu$ with weight $\mu-\sum^k i_j\beta_j.$ As $\lambda$ is the highest weight, we have that there are finitely many ways to write $\mu=\lambda-\sum i_j\beta_j$ and a unique way to write $\lambda.$ Therefore $M_\mu$ is finite-dimensional and $M_\lambda$ is 1-dimensional. The weights are all $\lambda-\sum i_j\beta_j=\lambda-\sum n_i\alpha_i$ as $\beta_p=\sum n_{i_p} \alpha_{i}$ for $\alpha_i\in \Pi.$ This completes the proof.     
\end{proof}

We will define Verma modules shortly. These will turn out to be highest weight modules which are universal in some sense. Before then, let $\lambda\in \lie{h}^*,$ and put $\delta=\frac{1}{2}\sum_{\alpha\in \Delta^+} \alpha.$ We can make $\C$ into a $U(\lie{b})$-module by defining how elements of $\lie{h}$ and $\lie{n}$ act and then by the Poincar\'e-Birkoff-Witt Theorem we will have defined how $U(\lie{b})$ acts.  
Define the action of $\lie{b}$ on $\C$ by \begin{align*}
	Hz&=(\lambda-\delta)(H)z & \forall\; H\in \lie{h}\\
	Xz&=0 & \forall\; X\in \lie{n}
\end{align*}
We denote $\C$ under this action as $\C_{\lambda-\delta}.$ Define a functor $\Ind_\lie{b}^\lie{g}: \operatorname{}_{U(\lie{b})}\textbf{Mod}\to  \operatorname{}_{U(\lie{g})}\textbf{Mod}$ by \[ V\mapsto U(\lie{g})\tensor_{U(\lie{b})} V\]
More generally for any lie subalgebra we can define $\Ind.$ This has a natural partner $\res$ which treats any $U(\lie{g})$-module as a module over the universal enveloping algebra of the subalgebra.     
\begin{definition}
	The \textbf{Verma module} corresponding to the weight lambda is \[V(\lambda)=\Ind^\lie{g}_\lie{b}(\C_{\lambda-\delta})=U(\lie{g})\tensor_{U(\lie{b})} \C_{\lambda-\delta}\] 
\end{definition}
The following theorem characterizes Verma modules. Using these modules, one can prove the $``\leftarrow"$ direction of the theorem of highest weights. 

\begin{proposition}\label{VermaProp}
	Let $\lambda\in \lie{h}^*.$ \begin{enumerate}
		\item $V(\lambda)$ is a highest weight module with weight $\lambda-\delta$ generated by $1\tensor 1.$ 
		\item Let $M$ be a highest weight module of weight $\lambda-\delta.$ Then there exists a unique $U(\lie{g})$-module map $\psi:V(\lambda)\to M$ with $\psi(1\tensor 1)=v$ with $\psi$ onto. It is injective if and only if $\ker \psi=0.$ 
	\end{enumerate}
\end{proposition} 
We will not prove this theorem as part (a) follows from the construction of $V(\lambda).$ Part $(b)$ follows from the universal mapping property for tensor products. Notice that $V(\lambda)$ is infinite dimensional over $\C.$ 

\begin{proposition}
	Let $\lambda\in \lie{h}^*$, $V(\lambda)$ the associated Verma module, and $S$ the sum of all proper $U(\lie{g})$ submodules of $V(\lambda).$ Then $L(\lambda)=V(\lambda)/S$ is an irreducible $U(\lie({g})$-module and is a highest weight module with weight $\lambda-\delta.$ 
\end{proposition}
This follows immediately from the definition and the fact that the image of $1\tensor 1$ in $L(\lambda)$ is non-zero. The following theorem completes the proof of the Theorem of Highest Weights.

\begin{theorem}
	Let $\lambda\in \lie{h}^*$ such that $\lambda$ is real on $\lie{h}_0,$ dominant, and algebraically integral. Then $L(\lambda+\delta)$ is an irreducible finite-dimensional representation of $\lie{g}$ with highest weight $\lambda.$   
\end{theorem}
For a proof of this see \cite[Chapter V, Section 3]{Knapp2005}.
\begin{remark}
	The exact same result holds on the group level as well. There, the proof follows from the theorem on the level of Lie algebras by differentiating the representations and then following the same steps. The only difference is the replacement of algebraically integral with analytically integral (defined below). For more details see \cite[Chapter IV, Section 7]{Knapp86}.  
\end{remark}

Now that we know these representations exist and are parametrized by dominant, algebraically integral weights, we want to find an explicit realization of the $L(\lambda+\delta).$ 
To do this, we make use of the theory of holomorphic vector bundles.  

\section{Compact Groups and Tori}
The key to understanding a majority of the representation theory of reductive, semisimple, or  compact Lie groups is the existence of a \textit{Haar Measure}. This is a left invariant Borel measure on $G.$ The existence of such a measure implies, as an example, that all representations of compact Lie groups can be taken to be unitary without a loss of generality. Additionally, combined with the Iwasawa decomposition, we get a variety of strong results. This will play a key role in the proof of the Borel-Weil theorem. 
Let us first show that such a measure exists. 

Let $G$ be a Lie group of dimension $n$ with Lie algebra $\lie{g}.$ Then as $T_1(G)=\lie{g}$ and there is an isomorphism $\lie{g}\to \Gamma_L(G,TG)$ the set of left-invariant smooth vector fields on $G.$ From this we conclude that $G$ is parallelizable. For this reason, we know that there exists an $n-$form $\omega\in \Omega^n(G)$ such that $\omega$ is positive relative to a chosen atlas on $G$, is nowhere vanishing, and. is left-invariant. Further, by the Riesz Representation theorem, there exists a Borel measure $d\mu_{\omega}$ on $G$ such that $\int_G f \omega=\int_G f  d\mu_{\omega}$ for all $f\in C_c(G).$  
\begin{lemma}
	$d\mu_{\omega}$ is left invariant in the sense that $d\mu_{\omega}(L_gE)=d\mu_{\omega}(E)$ for all Borel sets $E\subseteq G$ and all $g\in G.$ 
\end{lemma}   
\begin{proof}
	As $\omega$ is left-invariant, we know that $L^*_g\omega=\omega.$ Therefore, we have that \[ \int_G f \omega=\int_G f(gx) L_g^*\omega=\int_G f(gx)d\mu_\omega(x)=\int_G f(x)d\mu_{\omega}(x)\]
	Hence, $d\mu_\omega$ is left-invariant. If $K\subseteq G$ is compact, we apply the above integral formula to all $f\geq 1_K.$ Taking the infimum over these. functions we see that $d\mu_{\omega}(L_g^*K)=d\mu_\omega(K).$ Since $G$ has a countable base, $d\mu_\omega$ is regular and the lemma follows.  
\end{proof}

\begin{definition}
	A left-invariant, positive, Borel measure on $G$ is called a \textbf{left Haar measure}.  
\end{definition}
\begin{proposition}
	Every left Haar measure on $G$ is proportional. 
\end{proposition}
\begin{proof}
	See \cite[Theorem 8.23]{Knapp2005}. 
\end{proof} 

We could have equivalently defined \textbf{right Haar measures}. For most groups these are different from the left Haar measures. Let $d_lx$ denote a left Haar measure and $d_rx$ a right Haar measure. Notice that $L_g$ and $R_g$ commute with one another. Then, for any $t\in G,$ the measure $d_l(\cdot t)$ is a left Haar measure. For this reason, we get a function $\Delta:G\to \R^+$ called the \textbf{modular homomorphism} which satisfies \[ d_l(\cdot t)=\Delta^{-1}(t) d_l(\cdot) \]
This is a smooth function. 
\begin{lemma}
	$\Delta(t)=1$ for all $t\in K$ a compact subgroup of $G.$
\end{lemma}
\begin{proof}
	As $\Delta$ is smooth, $\Delta(K)$ is a compact subgroup of $\R^+.$ Therefore $\Delta(K)=\{1\}.$
\end{proof}
\begin{definition}
	A Lie group $G$ is called \textbf{unimodular} if $\Delta=1.$ Equivalently, if $d_r(x)=d_l(x).$   
\end{definition}
We now want to know what groups are unimodular. Then, when integration arises on these groups we do not have to worry about the choice of Haar measure.  
\begin{theorem}
	The following groups are unimodular: 
	\begin{enumerate}
		\item Compact  groups
		\item semisimple groups
		\item Reductive groups
	\end{enumerate}
\end{theorem}
We will not prove this as it requires the development of reductive lie groups which we do not present. See \cite{Knapp2005} for a proof in full generality. 

Now we turn to general representation theory for compact groups. A \textbf{representation} is a continuous group homomorphism $\Pi:K\to \Aut{V}$ for some Hilbert space $V.$ (The assumption that $V$ is a Hilbert space is unnecessary for $\dim V<\infty$. As we want the greatest generality, we do not place this finiteness assumption on $V.$) A representation is called \textbf{unitary} if $\Pi(k)$ us a unitary operator for all $k\in K.$ 
\begin{lemma} 
	Let $K$ be a compact Lie group and $(\Pi,V)$ a representation. Then there exists a Hermitian inner product $\ip{,}$ on $V$ so that the representation is unitary. 
\end{lemma}
\begin{proof}
	As $K$ is compact, every continuous function is integrable. Define \[ (u,v)=\int_K \ip{\Pi(k)u,\Pi(k)v} dk\]
	where $dk$ is the Haar measure on $K.$ Then it is obvious that each $\Pi(k')$ is a unitary operator with respect to this new Hermitian inner-product. Further, by the Principal of Uniform Boundedness we conclude that the topology on $V$ is the same as the topology generated by $\ip{,}.$  
\end{proof}

Therefore, we can assume without a loss of generality that every representation of a compact Lie group is unitary. Another interesting feature of compact Lie groups is the existence of a maximal abelian subgroup. 
\begin{proposition}[Cartan]
	Let $K$ be a compact, connected Lie group. Then there exists a maximal abelian subgroup which can be identified as a torus. Further, every maximal torus is conjugate. 
\end{proposition}  
\begin{proof}
	See \cite{Bump2013}. 	
\end{proof}

In a similar style to semisimple Lie algebras, we can define roots with respect to $\lie{t}$ the Lie algebra of $T\subseteq K$ a maximal torus. As $\lie{t}$ is abelian, the adjoint representation on $\lie{k}$ breaks up (as a direct sum) into one-dimensional irreducible representations. Each of these representations corresponds to a linear functional on $\lie{t}.$ We define roots as the those characters which yield non-zero spaces $\lie{k}_\alpha.$  

\begin{definition}
	Let $\lambda\in \lie{t}^*.$ Then we say $\lambda$ is \textbf{analytically integral} if for every $H\in \lie{t}$ with $\exp H=1$ then $\lambda(H)\in 2\pi i\Z.$ By a simple argument it can be shown that this condition is equivalent to the existence of a character $\xi_\lambda:T\to \C^\times$ such that $\xi_{\lambda}(\exp H)=e^{\lambda(H)}$ for all $H\in \lie{t}.$  
\end{definition}
\begin{proposition}
	If $\lambda$ is analytically integral, then $\lambda$ is algebraically integral. That is \[ (\lambda,\alpha)\in \Z, \text{ for each }\alpha\in \Delta(\lie{k},\lie{t})\] 
\end{proposition}
\begin{proof}
	See \cite{Knapp86}. 
\end{proof}

\section{Complex Lie Groups} 
We now depart from compact groups momentarily to set up the remaining background for the Borel-Weil theorem. 
\subsection{Complexification} 
Let $G$ be a real Lie group. We would like to find a complex Lie group $G_\C$ which extends $G$ in some meaningful way. 

\begin{definition}
	The \textbf{complexification} of a real Lie group $G$ is a complex Lie group $G_\C,$ together with an analytic map $G\to G_\C$ such that the Lie algebra of $G_\C$ is \[ \lie{g}_\C=\lie{g}\tensor_\R \C\]
	and $G_\C$ is universal in the following sense: if $H$ is a complex Lie group, and $\varphi:G\to H$ is a smooth homomorphism, then there exists a unique holomorphic homomorphism $G_\C\to H$ making the appropriate diagram commute.  
\end{definition} 
\begin{remark}
	Note that not all Lie groups admit a complexification. In fact, the double (unversal) cover of $SL(2;\R)$ does not admit a complexification. Even if a complexification exists, it is not necessarily unique up to isomorphism.  
\end{remark}
The following theorem gives us another convenient property of compact groups: they always admit a complexification! 
\begin{theorem}
	Let $K$ be a compact Lie group. Then $K$ admits a complexification which is unique up to isomorphism. 
\end{theorem}
\begin{proof}
	See \cite[Theorem 4.69 and Proposition 7.5]{Knapp2005}
\end{proof}

It turns out then that the finite-dimensional complex representations of compact Lie groups are is bijective correspondence with finite-dimensional holomorphic representations of $K_\C.$ 
Irreducibility need not be preserved by restriction. 	

We now come to arguably the most important decomposition of complex Lie algebras and the Lie groups associated to them. It is responsible for nearly all of the structure theory for semisimple Lie groups. 
\begin{theorem}[Iwasawa Decomposition]
	Let $\lie{g}$ be a real semisimple Lie algebra and $G$ a connected Lie group with Lie algebra $\lie{g}.$ Then there exist Lie subalgebras $\lie{k},\lie{a},\lie{n}$ and associated analytic subgroups $K,A,N,$ such that \[ \lie{g}=\lie{k}\ds \lie{a}\ds \lie{n}\]
	and \[ G=KAN\]	
	where $K$ is compact, $A$ is abelian, and $N$ is nilpotent and similarly for the lie algebras. 
\end{theorem}
\begin{proof}
	For the lie algebra decomposition, let $(\lie{g},\theta)$ by a semisimple Lie algebra together with a Cartan involution. Put $\lie{g}=\lie{k}\ds \lie{p}$ the associated Cartan decomposition and $\lie{h}_{\lie{p}}$ be a maximal abelian subspace of $\lie{p}.$ As $\lie{h}_\lie{p}$ is maximal abelian, we can simultaneously diagonalize all elements $\ad H, H\in \lie{h}_\lie{p}.$ Let \[ \lie{g}_{\lambda}=\{ X\in \lie{g}: [X,H]=\lambda(H)X, \forall H\in \lie{h}_\lie{p}, \lambda\in \lie{h}_\lie{p}^*\}\]
	Notice that $\theta(\lie{g}_\lambda)=\lie{g}_{-\lambda}.$ 
	Pick an ordering on $\lie{h}_\lie{p}^*$ and let $\lie{n}=\bigoplus_{\alpha>0} \lie{g}_\alpha.$ 
	Since $\lie{h}_\lie{p}$ is $\theta$-invariant and maximal abelian, we have that \[ \lie{g}_0=(\lie{g}_0\cap \lie{k})+\lie{h}_\lie{p}\]
	Now if $X\in \bigoplus_{\alpha<0} \lie{g}_\alpha$ we can write it as $X=X+\theta(X)-\theta(X).$ This decomposition has $X\in \lie{k}\ds \lie{n}$. Therefore, we have a decomposition \[ \lie{g}=\lie{k}+\lie{h}_{\lie{p}}+\lie{n} \]
	Applying $\theta$ we conclude that this decomposition is direct. 
	
	For the Lie group decomposition see \cite{Helgason1978}. 
\end{proof}	
\noindent This theorem also holds in the complex case. There is some slight modification that needs to be done to the proof above, but the big steps are identical. 

\begin{example}
	\begin{enumerate}
		\item Let $\lie{g}=\lie{sl}_n(\R).$ $SO(n)\hookrightarrow SL_n(\R)$ is a maximal compact subgroup and therefore $\lie{so}(n)$ is the corresponding compact lie algebra. Let $\lie{a}$ be the traceless diagonal matrixes and $|lie{n}$ be strictly upper triangular matrices. Then \[ \lie{sl}_n(\R)=\lie{so}(n)\ds \lie{a}\ds \lie{n}\]
		We can equivalently realize this on the group level as $SL_n(\R)=SO(n)\cdot T\cdot N$ where $N$ is upper triangular matrices and $T$ is the maximal torus. Notice that this is equivalent to the Gram-Schmidt orthogonalization of a matrix in $\lie{sl}_n.$   
		
		Now lets consider the Cartan decomposition of $\lie{sl}_n(\R)=\lie{so}(n)\ds \lie{p}$ where $\lie{p}$ are symmetric matrices. Notice that $\lie{so}(n)$ appears in both decompositions yet for the Cartan decomposition we have no lie algebra structure on $\lie{p}$. This should not be surprising however as both decompositions are equivalences as vector spaces. 
		
		\item Now consider $\lie{sp}_{2n}(\C).$   We have that \[ \lie{k}=\left\{ \Mat{U&V\\-\bar{V} &\bar{U} }: U \text{ skew-Hermitian}, V \text{ symmetric} \right\} \]
		Similar to $\lie{sl}_n$ we have $\lie{a}=\left\{ \Mat{A& 0\\ 0& -A}: A \text{ real diagonal matrix} \right\}$ which are the diagonal matrices and the nilpotent lie algebra are all upper triangular matrices, but now we can decompose them further into \[ \lie{n}=\left\{ \Mat{Z_1& Z_2\\ 0& -Z_1^T}: Z_1 \text{ strictly upper triangular}, Z_2 \text{ symmetric}  \right\}\]
		Then $\lie{sp}_{2n}(\C)=\lie{k}\ds \lie{a} \ds \lie{n}.$ 
	\end{enumerate}
\end{example} 

\begin{theorem}\label{Borel_Subgroup_Embedding} 
	Let $G$ be the complexification of a compact Lie group $K,$ and $T\subseteq K$ a maximal torus, $T_\C$ its complexification. Let $\lie{g}=\lie{k}_\C$ be the complexified Lie algebra of $\lie{k}$ and $\lie{t}_\C$ the Lie algebra of $T_\C.$ Denote the set of roots of $\lie{g}$ with respect to $\lie{t}_\C$ by $\Delta.$ Fix an ordering on $\lie{t}_\C^*$ and write $\Delta^+$ the set of positive roots. Denote by $\lie{n}=\bigoplus_{\alpha\in \Delta^+} \lie{g}_\alpha$ and let $\lie{b}=\lie{t}_\C\ds \lie{n}.$ If we denote by $N=\exp(\lie{n})$ and $B=T_\C N.$ Then $N$ and $B$ are closed subgroups of $G.$ Further, there exists $n>0$ such that $G\hookrightarrow GL_n(\C)$ such that $K$ consists of unitary matrices, $T_\C$ consists of diagonal matrices, and $B$ consists of upper triangular matrices.   
\end{theorem}
\begin{proof}
	Let $\pi:K\to \Aut{V}$ be a faithful unitary representation. By the definition of the complexification, we can extend $\pi$ to a holomorphic representation (also denoted $\pi$) $G\to \Aut{V}.$ Clearly, the Lie algebra $\lie{b}$ is solvable as $[\lie{b},\lie{b}]=\lie{n}$ and $\lie{n}$ is nilpotent. By Lie's Theorem, we may find a basis of $V$ such that $d\pi(X)$ is upper-triangular for all $X\in \lie{b}.$ 
	
	Identify $G$ with its imagine in $GL_n(\C)$ and its Lie algebra as a Lie subalgebra of $\lie{gl}_n(\C).$ Thus, we write $X$ instead of $\pi(X)$ and regard it as a matrix. Now, as each $X\in \lie{n}$ is nilpotent we know that \[ \exp(X)=I_n+X+\frac{1}{2} X^2+...+\frac{1}{n!}X^n\]
	Therefore, $Y-I_n$ is a sum of strictly upper triangular matrices and is therefore a strictly upper triangular matrix, hence nilpotent. Reversing the exponential series, we have that $X=\log(\exp(X))$ where we define $\log(Y)=\sum \frac{(-1)^{k-1}}{k}(Y-I_n)^k $  	
	for $Y$ an upper triangular unipotent matrix. In this case, the sum is finite. This defines a continuous map $\lie{n}\to N$ which is an inverse to $\exp.$ Therefore $\lie{n}\to N$ is a homeomorphism. Let $\lie{n}'$ be the Lie subalgebra of $\lie{gl}_n(\C)$ of upper-triangular nilpotent matrices and $\lambda_1,...,\lambda_r$ a set of linear functionals on $\lie{n}'$ such that $\lie{n}=\bigcap_{i} \ker \lambda_i.$ Then $N$ is characterized as the set of $A\in GL_n(\C)$ such that \[ \lambda_i(\log(g))=0\]
	These comprise a set of polynomial equations characterizing $N$ as a closed subgroup (subvariety) of $GL_n(\C).$ 
	
	Now, since $[\lie{t}_\C,\lie{n}]\subseteq \lie{n}$, we know that $T_\C$ normalizes $N$ and thus $B=T_\C N$ is a close subgroup of $GL_n(\C).$ Further its Lie algebra is $\lie{b}$ by construction. This completes the proof.   
\end{proof}

The group $B$ is a bit too big for the Iwasawa decomposition of $G$ above. Let $\lie{a}=i\lie{t}.$ It is the Lie algebra of a closed, connected subgroup $A$ or $T.$ If we embed $K$ and $G$ into $GL_n(\C),$ then $T$ is the group of diagonal matrices and $A$ is the group of diagonal matrices with positive real entries. Put $B_0=AN.$ Then by the Iwasawa decomposition $G=KB_0$ as a direct product.

\begin{corollary}\label{Complex_Structure}
	Let $K$ be a compact Lie group and $T$ a maximal torus. If we denote by $G$ the complexification of $K,$ then there is a bijection $K/T\cong G/B$ where $B=T_\C N.$ This gives $K/T$ the structure of a complex manifold. 
\end{corollary}  
\begin{proof}
	 From the Iwasawa decomposition, we have that $G=KB$ with $B\cap K=T.$ Note that this decomposition is not direct as $\lie{b}+\lie{k}=\lie{g}$ is not a direct sum. Then we have a diffeomorphism \[ G/B\to K/T\]
	  which is $K$-equivariant. Now as $G$ is complex Lie group and $B$ is a complex analytic submanifold, the quotient $G/B$ has the structure of a complex manifold. Further, the action of $K$ on $K/T$ is via holomorphic maps. 
\end{proof}
As well will see later, the proof of the Borel-Weil theorem uses the Iwasawa decomposition in a fundamental way. In fact, nearly all of the structure theory for semisimple Lie groups is due to the Iwasawa decomposition.

\section{Vector Bundles}

\begin{definition}
	Let $M$ be a complex manifold. We call a triple $(E,\pi, V)$ consisting of a complex manifold, a holomorphic projection map, and a complex vector space \textbf{holomorphic vector bundle}  of rank $\dim V$ over $M$ if: 
	\begin{enumerate}
		\item $\pi:E\to M$ is surjective and a local isomorphism. 
		\item There exist biholomorphic $\pi^{-1}(U)\to U\times V.$
		\item The fibre $\pi^{-1}(p)\cong {p}\times V\cong V$ is endowed with a vector space structure.   
	\end{enumerate}
\end{definition}
Similarly, we could have defined vector bundles as $E=\coprod_{p\in M} V_p$ where $V_p=\{p\}\times V.$ In this sense, we see that $TM$ and $T^*M$ are vector bundles over smooth manifolds. Similar to those, $\Gamma(M,E)$ is a $\mathcal{O}_M$-module. The main purpose of this section is to understand transformations on bundles and transformations between them. 
\begin{definition}
	Let $(E,\pi)$ and $(E',\pi')$ be holomorphic vector bundles over $M$ and $M'$ respectively. Then a \textbf{holomorphic bundle homomorphism} is a map $F:E\to E'$ such that there exists a map $f:M\to M'$ and the following diagram commutes: 
	\[
	\begin{tikzcd}
E \arrow[d, "\pi"'] \arrow[r, "F"] & E' \arrow[d, "\pi'"] \\
M \arrow[r, "f"']                    & M'                      
\end{tikzcd}
	\]
\end{definition} 
	
\begin{proposition}
	If $F$ is holomorphic, then $f$ is holomorphic.  
\end{proposition}	
\begin{proof}
	$f=\pi'_M\comp F\comp \zeta$ where $\zeta$ is the zero section. This is a composition of holomorphic maps and therefore holomorphic. 
\end{proof}
This lets us define a category $\textbf{Bun}_{H}(M)$ whose objects are holomorphic vector bundles over $M$ and where morphisms are holomorphic bundle homomorphisms. The forgetful functor \[U:\textbf{Bun}_{H}(M)\to \textbf{Man}_\C\] (with $\textbf{Man}_\C$ the category of complex manifolds) is faithful. In general, it is not full as there exist holomorphic maps $E\to E'$ which do not commute with the projection maps. We will denote by $\textbf{Bun}_H(M)^{<\infty}$ the category of finite rank vector bundles. In some more recent treatments of this material (say in \cite{Wedhorn2016}) this category is treated as finite locally free sheaves over $M.$ This is not useful for the theory presented below.   

\begin{example} 
	Let $TM$ denote the real tangent bundle for the complex manifold $M.$ It is a rank $2\dim M$ real vector bundle. The complex structure on $M$ induces an almost complex structure $J$ on $TM.$ This induces an endomorphism $J:TM\to TM$ such that $J^2=-1.$ This can thus be extended to a endomorphism $TM\tensor \C\to TM\tensor \C$ defined on fibres by $J(X+iY)=J(X)+iJ(Y).$ As $J^2=-1,$ we get a decomposition of $TM\tensor \C$ into two eigenspaces for $J$ corresponding to the eigenvalues $i$ and $-i.$ Then \[ TM\tensor \C=TM_i\ds TM_{-i}\]
	Then $TM_i$ is the holomorphic tangent bundle to $M.$ The bundle $TM_{-i}$ is called the anti-holomorphic tangent bundle.   
\end{example}

If $E$ and $E'$ are holomorphic vector bundles on a complex manifold $M,$ denote their space of holomorphic sections by $\Gamma(E)$ and $\Gamma(E').$ If $F:E\to E'$ is a bundle homomorphism, it induces a map \[ \widetilde{F}:\Gamma(E)\to \Gamma(E')\]
given by \[ \widetilde{F}(\sigma)(p)=F(\sigma(p))\]
Because a bundle homomorphism is linear on fibres, $\widetilde{F}$ is $\C$-linear on sections. 

We now want to construct some holomorphic vector bundles on a complex Lie group and on complex homogeneous spaces $G/H.$ 

\begin{proposition}\label{Associated_Bundle}
	Let $G$ be a complex Lie group and $(\pi,W)$ a complex representation of a closed subgroup $H.$ Then there exists a holomorphic vector bundle $V$ over $G/H$ such that $G$ acts on the space of sections.  
\end{proposition}
\begin{proof}
	The canonical map $G\to G/H$ is a principal $H$-bundle. Any complex representation $\pi:H\to GL(W)$ induces an action of $H$ on the space $G\times W$ by \[(g,w)\cdot h = (gh,\pi(h^{-1})w)\]
	Then put $V=G\times_H W=(G\times W)/H.$ Then $[gh,w]=[g,\pi(h)w]\in V.$ The map $q:V\to G/H$ given by $[g,w]\mapsto gH$ is well defined, surjective, and $q^{-1}(gH)\cong W.$ This is a fibre bundle with transition maps given by the transition maps for the principal bundle. Further, as the fibres are complex vector spaces and the canonical map is holomorphic, we have that $V$ is a holomorphic vector bundle over $G/H.$ Let $\Gamma(G/H,V)$ denote the set of sections $s:G/H\to V.$ We can identify \[\Gamma(G/H,V)\overset{\sim}{\longrightarrow} \script{F}_{H,\pi}:=\{ f:G\to V | f(gh)=\pi(h)^{-1}f(g) \} \]
	Then $G$ acts on this space by $g\cdot f(x)=f(g^{-1}x).$ This completes the proof. 
\end{proof}
Even for one dimensional representations $\chi$ of $H,$ the space $\script{F}_{H,\chi}$ is unbelievably massive. We may home that if we restrict to some subset (say impose more restrictions on $f\in \script{F}_{H,\chi}$) then we may be able to get a handle on what these representations are. As it will turn out in the next section, we can restrict ourselves to \textit{holomorphic sections} of $V.$ This restriction will turn out to be enough to realize all of the finite dimensional irreducible representations of $K$ a compact Lie group and $G=K_\C$ its complexification.

\subsection{Flag Manifolds} 
In this short subsection, we shall show that there is some interesting geometry happening behind the scenes here involving the quotients $G/B$ or more generally $G/P$ for any closed group containing $B.$ This is done through the language of \textit{flag manifolds}. Before we get to flag manifolds, we need to discuss the Grassmann manifolds (also called Grassmannians). 
\begin{definition}
	Let $V$ be a real (or complex) vector space of dimension $n.$ The \textbf{Grassmannian of k-planes} in $V$ is the set of all $k$-dimensional subspaces in $V$ and is denoted $\Gr(k,V).$   
\end{definition} 
Let $G=\Aut{V}$ be the group of automorphisms of $V.$ By choosing a basis for $V,$ we can identify $\Aut{V}\cong GL_n(\R)$ (resp. $GL_n(\C)$).  Now, let $A$ and $A'$ be two different elements of $\Gr(k,V).$ By choosing bases and extending these to full bases of $V,$ we can find a matrix $X\in GL_n(\R)$ such that $XA=A'.$ Therefore, $GL_n(\R)$ acts transitively on $\Gr(k,V).$ Let $\{v_1,...,v_n\}$ be the basis of $V$ and $S=\Span[\R]{v_1,...,v_k}$ be the \textit{standard} $k$-plane. Then the isotropy subgroup of $S$ is the closed subgroup \[     H=\left\{ \Mat{P& Q\\ 0 & R}: P\in GL_k(\R), Q\in M_{k,n-k}(\R), R\in GL_{n-k}(\R)     \right\} \]
This gives an identification $\Gr(k,V)=GL_n(\R)/H.$ We call $H$ a \textit{parabolic subgroup} of $G.$ This exhibits $\Gr(k,V)$ as a real (resp. complex) manifold.  

Now let $(n_1,....,n_j)\in \Z^j\;$ $j\leq n$ be an increasing tuple of integers with $n_j=n=\dim V.$ A \textbf{flag} of type $(n_1,...,n_j)$ is a chain of subspaces \[ 0=V_0\subseteq  V_1\subseteq V_2\subseteq... \subseteq V_j=V\]
with $\dim V_i=n_i.$ Equivalently, we could require that $\dim V_i/V_{i-1}=n_i-n_{i-1}.$ A \textit{full flag} corresponds to the tuple $(1,2,3,...,n)$ and thus a chain \[ 0=V_0\subseteq V_1\subseteq ... \subseteq V_n=V\]
and $\dim V_i/V_{i-1}=1.$   

\begin{definition}
	The \textbf{partial flag manifold} of type $(n_1,...,n_j)$ is the collection of all flags of type $(n_1,...,n_j)$ in $V$ and is denoted $\Fl(n_1,...,n_j; V).$ The \textbf{full flag manifold} of $V$ will be denoted $\Fl(V).$  
\end{definition}   
By choosing a basis for $V$ and thus identifying it with $\R^n,$ we have a natural action of $GL_n(\R)$ on $\Fl(n_1,...,n_j;V).$ Now, let $F$ and $F'$ be two distinct flags. There exists $X\in GL_n(\R)$ such that $XF=F'$ and the action is transitive. The stabilizer of $F$ is a closed subgroup $P$ of $GL_n(\R)$ and we identify $\Fl(n_1,...,n_j;V)=G/P.$ This exhibits $\Fl(n_1,...,n_j;V)$ as a smooth manifold. The stabilizer of the standard full flag is the subgroup $B$ of upper-triangular matrices. Thus $\Fl(V)=GL_n(\R)/B.$

\begin{remark}
	The groups $P$ and $B$ are called the \textbf{standard parabolic} and \textbf{standard Borel} subgroups respectively. An alternative definition of the standard Borel subgroup is as a standard \textit{minimal} parabolic subgroup. We call the conjugates of $B$, Borel subgroups and the conjugates of $P$ parabolic subgroups. Notice that every parabolic subgroup contains a Borel subgroup.  
\end{remark}

In the case of a complex vector space, we see that $\Fl(V)=GL_n(\C)/B.$ By Corollary \ref{Complex_Structure}, we can realize $\Fl(V)=K/T$ for $K=U(n).$ In more generality, for a connected Lie group $C,$ there exists a maximal torus $S$ and the quotient space $C/S$ is a \textit{flag manifold}. 

\begin{example}
	\begin{enumerate}
		\item As seen above, if $V$ is a complex vector space then $\Gr(k,V)$ is a flag manifold corresponding to the tuple $(k,n)\in \Z^2.$ It is realized as the quotient $GL_n(\C)/H$ with $H$ the complex analog of the group defined above. 
		\item Let $\mathbb{CP}^n$ (or $\mathbb{P}^n(\mathbb{C})$) denote the orbit space $(\C^{n+1}-\{0\})/\C^\times.$ This is realized as the space of all lines in $\C^{n+1}.$ In the language we have seen above, we can realize this as $\Gr(1,\C^{n+1}).$  
	\end{enumerate}
\end{example} 

\begin{definition}
	Let $G$ be a complex connected Lie group and $H$ a closed subgroup. Then $G/H$ is a complex homogeneous space. Let $p:V\to G/H$ be a holomorphic vector bundle. $V$ is \textbf{homogeneous} if the group of bundle automorphisms act transitively on the set of fibres of $V.$ We call $V$ \textbf{homogeneous with respect to $G$} if the $G$ action on $G/H$ lifts to a $G$ action on $V$ by bundle automorphisms. We will sometimes refer to these as $G$-homogeneous vector bundles.     
\end{definition}

\noindent Let us now characterize all vector bundles on flag manifolds which are homogeneous with respect to $K_\C$. 
\begin{proposition}
	Let $K$ be a compact, connected Lie group and $G$ its complexification. Let $(\pi,W)$ be a representation of a parabolic subgroup $P\subseteq G.$ Then this gives rise to a  holomorphic vector bundle over the partial flag manifold $G/P$ which is homogeneous with respect to $G.$ Further, every holomorphic vector bundle which is homogeneous with respect to $G$ arises in this way.     
\end{proposition}
\begin{proof}
	The existence of such a vector bundle was proven in Proposition \ref{Associated_Bundle}. The homogeneity condition is readily checked. Therefore, we shall show that every homogeneous vector bundle arises in this way. Let $V$ be a $G$-homogeneous vector bundle and $V_P$ the fibre $p^{-1}(P).$ $V_P$ comes naturally equipped with the structure of a representation $P\to \Aut{V_P}.$  The map \[ \mu:G\times V_H\to V\]
	defined by $\mu(g,z)=g\cdot z$ is surjective as $G$ acts transitively on $G/P.$ The fibres of $\mu$ are precisely the $P$ orbits on $G\times V_P$ via the diagonal action \[ (g,z)\mapsto (gp^{-1},p\cdot z) \]
	  Therefore, we may represent any element uniquely as an equivalence $[g,z]$ where $[gp,z]=[g,p\cdot z].$ Hence, we can make the identification $V=G\times_P V_P.$ This completes the proof.  
 \end{proof}

\section{Borel-Weil Theorem} 
We will motivate the theorem by starting with some facts about $G=GL_n(\C).$ The natural action of $G$ on $\C^n-\{0\}$ commutes with the action of $\C^\times$ and therefore descends to an action on $\mathbb{CP}^{n-1}.$ Moreover this action is transitive. Now, the isotropy subgroup of the class $[0:....:0:1]$ in $G$ consists of all $g\in G$ such that $g\cdot (0,...,0,1)^T=(0,...,0,\lambda)^T$ for $\lambda\in \C^\times.$ Let $Q$ be this group. Then \[ Q=\left\{ \Mat{A& 0\\ w^T & \lambda}\right\}\cap GL_n(\C)   \]
with $\lambda\in \C$, $w\in \C^{n-1}$, and $A\in M_{n-1}(\C).$ Then $Q$ is a complex subgroup of $G$ as its Lie algebra is complex. Therefore the quotient $G/Q$ becomes complex manifold which is biholomorphic to $\mathbb{CP}^{n-1}.$ 

Now fix $N\geq 0$ and put $\chi:Q\to \C^\times$ a character of $Q$ of the form \[ \chi \Mat{A& 0\\ w^T & \lambda}=\lambda^{-N} \] 
Then $\chi$ induces a holomorphic action of $\Q$ on $\C$ by $q\cdot z=\chi(q)z.$ Using this, we can build the associated bundle $G\times_Q \C\to G/Q$ in the style of the previous section. Now per the proof of Proposition \ref{Associated_Bundle}, we can identify the $C^\infty$ sections of this bundle with the space of functions \[ \script{F}_{Q,\chi}^\infty=\left \{ f:G\to \C \; \vline \;  f(gq)=\chi(q)^{-1}f(g), f \text{ smooth}     \right\} \]

Now, let $V_N$ be the space of homogenous polynomials of degree $N$ in $n$ complex variables. Then for any $f\in V_N$ define \[ \varphi_f(g)=f\left(g\Mat{0\\ \vdots\\ 1} \right) \]
Then if $q\in Q,$ we have that \[ \varphi_f(gq)=f\left(gq\Mat{0\\ \vdots\\ 1} \right)=\lambda^N\varphi_f(g)\]
Therefore $\varphi_f\in \script{F}_{Q,\chi}^\infty.$ In fact, this is holomorphic and therefore $\varphi_f\in  \script{F}_{Q,\chi}^{Hol},$ the space of holomorphic sections. For the rest of this section, let $\ell=\Mat{0\\ \vdots \\ 1}.$  

\begin{proposition}
	The only holomorphic sections of $G\times_Q\C\to G/Q$ are those $\varphi_f.$ 
\end{proposition} 
\begin{proof}
	Let $\varphi:G\to \C$ be the function corresponding to a holomorphic section of the bundle. We want to define a polynomial $P(z_1,...,z_n)$ on $\C^n-\{0\}.$ Let $g\in G$ be such that $g\ell=\Mat{z_1\\\vdots \\ z_n}.$ Then define $P(z_1,...,z_n)=\varphi(g).$ To see this is well-defined, let $g'$ be another element of $G$ satisfying  $g'\ell=\Mat{z_1\\\vdots \\ z_n}.$
	Then $g^{-1}g'$ stabilizes $\ell$ and therefore is en element $q$ of $Q.$ Writing $g'=gq,$ we have that $\varphi(g')=\varphi(g)$ and $P$ is well-defined. Moreover, by construction $P$ is homogeneous of degree $N.$ Since we can define $P$ using open sets of $G,$ we have that $P$ is holomorphic on $\C^n-\{0\}.$ The homogeneity condition implies that $P$ is bounded near $0.$ Hence, $P$ admits a holomorphic extension to $\C^n.$ Now, the $C^\infty$. behavior, combined with the homogeneity implies that \[ |P(\mathbold{z})|\leq C|\mathbold{z}|^N\]
	and similarly \[ |\partial^\alpha_\mathbold{z} P(\mathbold{z})|\leq C_{\alpha}|\mathbold{z}|^{N-|\alpha|}\] 
	for any multi-index $\alpha$ and $\mathbold{z}\in \C^n-\{0\}.$ If $|\alpha|>N,$ then $\partial^\alpha P$ vanishes at $\infty$ and by Liouville's theorem, is $0.$ Therefore, the Taylor expansion of $P$ about $0$ vanishes for all degrees $>N.$ Hence, $P$ is a polynomial.  
\end{proof}

This implies that the representation of $G$ on $V_N$ can be realized as the space of sections $\script{F}_{Q,\chi}^{Hol}.$ In different terminology, we say that $V_N=\Ind_Q^G(\chi)$ is the \textbf{induced representation} of $G$ from the representation $\chi$ of $Q.$ Now, we can turn to the general situation.

Let $K$ be a compact lie group with maximal torus $T$. If $G=K_\C$ is the complexification, then the Iwasawa decomposition implies that $G=KAN$ and $B=T_\C \overline{N}$ where $\overline{N}$ are the lower-triangular nilpotent matrices. Then by Corollary \ref{Complex_Structure}, we know that $G/B\cong K/T$ and both are complex manifolds. For any character $\lambda:T\to \C^\times,$ we can extend $\lambda$ to be a character of $T_\C$ and then to $B$ by declaring $\chi(\bar{n})=1.$ Therefore, we get two line bundles \[ G\times_B \C\cong K\times_T \C\]    
which are isomorphic as complex manifolds.

\begin{theorem}[Borel-Weil]\label{Borel-Weil}
	Let $K$ be a compact, connected Lie group and $T\subseteq K$ be a maximal torus. Let $G=K_\C$ be the complexification and $B=MA\overline{N}$ a Borel subgroup. Then the irreducible finite dimensional representations of $K$ stand in one-to-one correspondence with the dominant, analytically integral weights $\lambda\in \lie{t}^*$ with the correspondence given by \[ \lambda\mapsto \Gamma_H(K/T,L_{\lambda})\cong \script{F}_{B,\chi_{\lambda}}^{Hol}\]
	where $\Gamma_{H}(K/T,L_{\lambda})$ denotes the set of holomorphic sections of the bundle and \[\script{F}_{B,\chi_{\lambda}}^{Hol}=\left \{f:G\to \C \;\vline\; f(gb)=\chi_{\lambda}(b)^{-1}f(g), f \text{ holomorphic}  \right\}\]
	with $\chi_{\lambda}$ the character of $B$ associated to the analytically integral weight $\lambda.$ 
\end{theorem}
We present a combination of the proofs presented in \cite{Knapp86}, \cite{Helgason1978}, and \cite{Helgason2008}. The proof will proceed in two main steps: 1) show that $\Gamma_H(K/T,L_\lambda)$ is a finite dimensional and 2) show it is irreducible. Throughout the proof, we shall make use of the isomorphism $\Gamma_H(K/T,L_\lambda)\to \script{F}_{T,\chi_\lambda}^{Hol}\cong  \script{F}_{B,\chi_\lambda}^{Hol}.$  \\

\begin{remark}
	Another way of thinking about this theorem is as a classification result for various sheaves on the flag varieties (manifolds) $\Fl(\C^n).$ Every finite rank vector bundle on $\Fl(\C^n)$ corresponds to a finite locally free sheaf with the correspondence given by taking global sections. The theorem above classifies all of the line bundles (considered as sheaves) on $\Fl(\C^n)$ which admit global sections. 
\end{remark}

  	The Lie algebra of $G$ has a Cartan decomposition $\lie{g}_\C=\lie{k}\ds i\lie{k}$ corresponding to the Cartan involution $\theta:\lie{g}\to \lie{g}.$ Let $\Theta$ be the corresponding involution of $G.$ This gives an Iwasawa decomposition $\lie{g}=\lie{t}\ds \lie{a}\ds \lie{n}$. Pick a maximal abelian subalgebra $i\lie{a}$ of $i\lie{k}$ and form $\lie{m}=Z_\lie{k}(\lie{a})$ the centralizer of $\lie{a}$ in $\lie{k}.$ Then $\lie{m}$ is a Cartan subalgebra of $\lie{k}$ and $\lie{m}_\C=\lie{a}\ds \lie{m}.$ With respect to the roots $\Delta(\lie{k}_\C,\lie{m}_\C),$ put \[ \lie{b}=\lie{m}\ds \lie{a}\ds \bigoplus_{\alpha\in \Delta^+} \lie{k}_{-\alpha}\]
	Then $B=MA\overline{N}$ is the corresponding Iwasawa decomposition of the Borel subgroup. 
	
	Now, let $\lambda\in \lie{t}^*$ be a dominant, analytically integral weight and $(\Phi_\lambda,V)$ the irreducible, finite dimensional highest weight representation of $K$ with highest weight $\lambda.$ Let $v_\lambda\in V$ be a highest weight vector. This representation extends to a holomorphic representation (also denoted $\Phi_\lambda$) of $G$ via the universal property of the complexification. For each $v\in V,$ define a function $\psi_v(x)$ on $G$ by \[ \psi_v(x)=(\Phi_\lambda(x)^{-1}v,v_\lambda)\] where $(,)$ is the inner product on $V$ induced via the isomorphism with $\C^n.$  
	\begin{lemma}\label{Borel_Lemma_1} 
		For each $v\in V,$ $\psi_v\in \script{F}_{B,\chi_{\lambda}}^{Hol}.$ Moreover if $L$ denotes the left regular action, then $L(k)\psi_v=\psi_{\Phi_{\lambda}(k)v}$ and the collection $\{\psi_v:v\in V\}$ is an irreducible subrepresentation of $\script{F}_{B,\chi_\lambda}^{Hol}$ which is equivalent to $\Phi_\lambda.$ 
	\end{lemma}      
\begin{proof}[Proof of Lemma 5.3] 
	Let $\varphi_\lambda$ be the differential of $\Phi_\lambda.$ Since $\Phi_\lambda$ is unitary on $V,$ $\varphi_\lambda$ is skew-hermitian on $\lie{k}$ and complex-linear on $\lie{g}.$ Therefore, $\varphi_\lambda(\theta X)=-\varphi_{\lambda}(X)^*$ and $\Phi_\lambda(\Theta x)=\Phi_\lambda(x^{-1})^*$ for all $X\in \lie{g}$ and $x\in G.$ Now if $b\in B=MA\overline{N}$ we have that for all $x\in G$ \begin{align*}
		\psi_v(xma\bar{n})&=(\Phi_{\lambda}(ma\bar{n})^{-1}\Phi_\lambda(x)v,v_\lambda) && \\
		&=(\Phi_\lambda(x)^{-1}v,\Phi_\lambda(ma^{-1}n)v_\lambda) && \text{as } \Theta(ma\bar{n})=ma^{-1}n\in MAN \\
		&=(\Phi_\lambda(x)^{-1}v,\Phi_\lambda(ma^{-1})v_\lambda) && \text{as $v_\lambda$ is a highest weight vector}\\
		&=(\Phi_\lambda(x)^{-1}v,\chi_{\lambda}(m)\chi_\lambda(a)^{-1}v_\lambda) && \text{as $v_\lambda$ has weight $\lambda$}\\
		&=\overline{\chi_\lambda(m)}\chi_\lambda(a)^{-1}(\Phi_\lambda(x)^{-1}v,v_\lambda) && \\
		&=\chi_\lambda(b)^{-1}\psi_v(x)
	\end{align*} 
	Further, It is clearly holomorphic as is defined by a holomorphic representation. Hence, $\psi_v\in \script{F}_{B,\chi_{\lambda}}^{Hol}.$ Finally, \begin{align*}
		\psi_{\Phi_\lambda(k)v}(x)&=(\Phi_\lambda(x)^{-1}\Phi_\lambda(k)v,v_\lambda)\\
		&=(\Phi_\lambda(k^{-1}x)^{-1}v,v_\lambda)\\
		&=\psi_v(k^{-1}x)=L(k)\psi_v(x)
	\end{align*}
	This completes the proof of the lemma.  
\end{proof}	
Now we wish to show that $V\to \script{F}_{B,\chi_{\lambda}}^{Hol}$ is onto. Put $\psi_\lambda:=\psi_{v_\lambda}$ and $\script{F}_{\lambda}:=\script{F}_{B,\chi_\lambda}^{Hol}.$
\begin{lemma}\label{Borel_Lemma_2}
	Let $F\in \script{F}_\lambda.$ Then \[ \int_M F(mxm^{-1})dm=F(1)\psi_\lambda(x)\]
	for all $x\in G.$ ($dm$ is the normalized Haar measure on $M.$)
\end{lemma}
The idea of the proof is to show that the left side is a multiple of $F(1)$ independent of $F.$ This multiple is a power series in $x$ and evaluating at $F=\psi_\lambda,$ we see that they are equal near 1. By holomorphicity, the functions are thus equal everywhere. 
\begin{proof}[Proof of Lemma 5.4]
	Let $X\in \lie{g}$ and $\widetilde{X}$ the corresponding left invariant vector field on $G.$ Since $F$ is holomorphic, it is real-analytic and thus the Taylor series of $F$ converges to $F$ is a neighbourhood of $1.$ Thus \[ F(\exp X)=\sum \frac{1}{n!} (\widetilde{X}F)(1)\]
	Conjugating by $m$ and integrating, we see that \[ \int_M F(m\exp X m^{-1})dm=\sum \frac{1}{n!} \left( \left\{ \int_M \Ad(m)\widetilde{X}^n dm  \right\}F\right) (1)\]
	
	Now let $\{X_{\alpha},H_\alpha,X_{-\alpha}\}$ be a basis of $\lie{g}$ with respect to a positive choice of roots. Writing $X$ in terms of this basis and expanding, we get integrals of monomials. The coefficients can be factored out as $\widetilde{X}$ is complex-linear as an endomorphism of $\script{F}_\lambda.$ Now, by the Poincar\'e-Birkhoff-Witt Theorem, we can rewrite the expression as a linear combination of $\Ad(m)$ and monomials of the form  $X_{-\alpha_1}^{i_1}...X_{-\alpha_p}^{i_p}H_{\alpha_1}^{j_1}...H_{\alpha_q}^{j_q}X_{\alpha_1}^{k_1}...X_{\alpha_r}^{k_r}.$ Then integrals of each monomial is now $\Ad(m)$-invariant and a monomial. If this new monomial has no $X_{-\alpha}$ term for $\alpha\in \Delta^+$ then by $\Ad(m)$-invariance it cannot have any $X_{\alpha}$ term. 
	
	On the other hand, any $\Ad(m)$-invariant polynomial cannot have any $X_{-\alpha}$ terms as the vector field $\widetilde{X_{-\alpha}}F=0$ by the fact that $\exp t X_{-\alpha}\in \overline{N}.$ Hence, all of the $\Ad(m)$-invariant polynomials lie in $U(\lie{m}_\C)$ and as $\exp\lie{m}_\C= MA \subseteq B,$ each member of $U(\lie{m}_\C)$ acts by scalars depending only on $\lambda.$ Hence, any expression of the form $H_{\alpha_1}^{j_1}...H_{\alpha_n}^{j_n}F(1)$ is a scalar multiple of $F(1)$ independent of $F.$ This implies the lemma. 
\end{proof}
Now we can prove Theorem 5.2 in a few easy steps. 
\begin{proof}[Proof of Theorem 5.2]	
	Define an inner product on $\script{F}_\lambda$ by \[ \ip{F_1,F_2}=\int_K F_1(k)\overline{F_2(k)}dk \]
	\begin{claim}
		$|F(1)|\leq||\psi_\lambda||^{-1}\cdot ||F||$
	\end{claim} 
	In fact \begin{align*}
		||F||&=\int_K |F(k)|^2 dk=\int_K |F(mkm^{-1})|^2dk\\
		&=\int_K \int_M |F(mkm^{-1})|^2dm dk\\
		&\geq \int_K \left( \int_M |F(mkm^{-1})|dm \right)^2dk\\
		&=|F(1)|^2\int_K |\psi_\lambda(k)|^2 dk\\
		&=|F(1)|^2||\psi_v||^2
	\end{align*}
	As a direct corollary of this, for every compact $E\subseteq G,$ there exists a $C_E<\infty$ such that \[ |F(x)|\leq C_E||F||\]
	for all $F\in \script{F}_\lambda$ and $x\in E.$ Therefore, $\script{F}_\lambda$ is complete (Cauchy sequences converge on compact sets by the previous line and their limit is holomorphic and satisfies the desired relation). Now, $\script{F}_\lambda$ is finite-dimensional, as it is a locally compact Banach space. 
	
	It remains to be shown that $\script{F}_\lambda$ is irreducible as a representation of $K.$ Let $U\subseteq \script{F}_\lambda$ be a closed, invariant subspace. Then for $F\neq 0$ on $U,$ by applying some $L(k)$, we can assume that $F(1)\neq 0.$ Therefore by completeness \[ \int_M \overline{\chi_{\lambda}(m)}L(m)F dm \]
	is an element of $U.$ However, Lemma \ref{Borel_Lemma_2} says that this is equal to $F(1)\psi_\lambda.$ Hence, $\psi_\lambda \in U.$ Similarly, we see that $\psi_\lambda\in U^\perp.$ This is a contradiction and thus $U=0$ or $U^\perp=0.$ Hence, $\script{F}_\lambda$ is an irreducible, finite-dimensional representation of $K.$ By Lemma \ref{Borel_Lemma_1} the map $V\to \script{F}_\lambda$ is a $K$-equivariant isomorphism. This completes the proof.   
\end{proof}

This result shows us that we can derive some algebraic information from a geometric object. In the language of Chapter 3, this theorem can be restated as $H^0(G/B,\script{F}_\lambda)\neq 0$ if and only if $\lambda$ is dominant and analytically integral. In fact, a stronger form of this theorem due to Bott \cite{Fulton_Harris2004} says that the sheaf cohomology of the associated bundle is non-zero is only one degree. This surprising appearance of sheaf cohomology indicates that it may prove to be useful in understanding the sheaf $G$ of chapter $4$ as well as understanding $C^\infty(\mu(-))$ as a $G$-module. Some care needs to be taken here as we do not know much about the category $G$-\textbf{Mod}. In fact, the case of $O_X$-modules for a locally ringed space may deviate highly from this situation in some critical ways. One being that there is no reason \textit{a priori} that $G_x$ is a local ring. We do not provide a resolution to this here and thus there is still much work to be done.

\section{Vector Fields for Noisy odors}
One major deficit of the model in Chapter 4 is its dependence on the odor source representations to be clean and precise. What should happen if an odor is presented in an environment which is particularly noisy? For example, consider a fox in the wilderness. If the fox is eating a meal the odors are in high concentration and thus can be distinguished. If instead it is trotting along and the odor of rabbit wafts through the air, how may it determine what this odor is? There are clearly many other odors present in the second situation and thus should make identification nearly impossible. This contradicts experimental and observational evidence however! We know that foxes can find their prey with minimal odor stimulation; this implies the existence of some mechanism which produces a "best guess" for what a given noisy odor may be. As we shall see below, there is a naive way of modeling such a problem which we conjecture is indeed the correct approach. This naive method relies on \textit{vector fields} on $S$ and generates an attractor basin for the various odors. This has been shown to have some relation to \textit{\u Cech cohomology} which can be viewed as a refinement of sheaf cohomology. This ties together all of the ideas presented.    
We will not go through the construction of \u Cech cohomology as it is a bit involved and the main idea behind it is to serve as a computational tool for sheaf cohomology on suitably nice spaces (of which manifolds happen to fit).   

\subsection{Flows} 
In general, the theory of flows is a generalization of the theory of Ordinary differential equations. Now, the equations are defined on manifolds by vector fields $\xi:M\to TM.$ We shall not do the general case here but refer the reader to \cite{Lee2012}. Our situation is significantly eased as $R'$ and therefore $S$ are assumed to be diffeomorphic to open submanifolds of $\R^n$ and therefore $TR'\cong TS\cong S\times \R^n.$ So there exists vector fields $\{V_1,...,V_n\}$ which span the tangent space at each $s\in S.$ As a result, there exists a non-trivial vector field $\xi$ which is \textit{complete} meaning that every trajectory can be given $\R$ as a domain. By trajectory we mean a smooth map $\gamma:\R\to S$ such that $\gamma'(t)=\xi(\gamma(t)).$ In general, this is the solution of a differential equation and these trajectories are called \textit{maximal}. 
\begin{definition}
	Let $\xi\in \Gamma(S,TS).$ The \textbf{flow} of $\xi$ is the mapping $\theta:S\times \R\to S$ given by $(s,t)\mapsto \gamma_s(t)$ where $\gamma_s(0)=s$ and $\gamma$ is the maximal trajectory. 
\end{definition}      

We can use flows to understand noisy inputs into the olfactory system. Let $K$ be the collection of $s\in S$ such that $s$ is a local maximum of the function $f$ defining $S.$ That is, these are the "tops" of the peaks. Now define a smooth vector field on $S$ which makes $K$ an attractor. Then the attractor basin is the disjoint union of a finite number of contractible open sets. What we would like to know is that the attractor basin is a cover of $S$ so that any point can be draw into one of the peaks and identified as in the classification scheme of Chapter 4. This would allow us to identify any noisy odor (one for which $\widetilde{U_x}$ is particularly large) with some degree of accuracy. Sadly, this cannot be guaranteed as a cover would rely on exposure to an enormous number of different odors (then we can assume that the $U_x$ form a cover of $R'$ and thus $\widetilde{U_x}$ is a cover of $S$ themselves) or some nearly equivalent requirement. As a consolation, we can still identify noisy odors which fall within the attractor basin of the learned odors. 

Let us connect the idea of flows to representations. Let $X\in \Gamma(M,TM)$ be a complete smooth vector field and $\theta:\R\times M\to M$ the associated flow. This is equivalent to defining an action of the Lie group $(\R,+)$ on $M$ and therefore a non-linear representation of $\R.$ Here, as diffeomorphisms of $M.$ By differentiating this action, we get a non-linear Lie algebra representation $\R\to \lie{X}(M).$ As we can realize flows as solutions to certain non-linear partial differential equations, we can equivalently understand theses solutions by understanding the corresponding representation on either the Lie group or Lie algebra level. This is one reason representation theory my play a key role in the further development of this theory and for understanding the identification of noisy odors.

\clearpage
\begin{center}
    \thispagestyle{empty}
    \vspace*{\fill}
    \textit{Thanks for reading!}
    \vspace*{\fill}
\end{center}
\clearpage

\bibliographystyle{alpha}
\bibliography{Thesis_References}

\end{document}